\documentclass[submission]{eptcs}
\input{NAPRC.sty}
\ifpdf
  \usepackage{underscore}         
  \usepackage[T1]{fontenc}        
\else
  \usepackage{breakurl}           
\fi

\title{Improving the Fidelity of CNOT Circuits on NISQ Hardware}

\author{Dohun Kim\footnote{Department of Physics and Astronomy, Seoul National University, Seoul, South Korea.}\, \footnote{Institute of Applied Physics, Seoul National University, Seoul, South Korea.} \email{dohunkim@snu.ac.kr}\and Minyoung Kim \footnotemark[1] \,\footnotemark[2] \email{kimmy1110@snu.ac.kr} \and Sarah Meng Li\footnote{Informatics Institute, University of Amsterdam, Amsterdam, Netherlands.}\; \footnote{Department of Combinatorics \& Optimization, University of Waterloo, Waterloo, Canada.}\; \footnote{Institute for Quantum Computing , University of Waterloo, Waterloo, Canada.}\, \footnote{Perimeter Institute for Theoretical Physics, Waterloo, Canada.} \email{sarah.li@uva.nl}\and  Michele Mosca \footnotemark[4]\, \footnotemark[5] \,\footnotemark[6] \email{michele.mosca@uwaterloo.ca}
}

\def\titlerunning{Improving the Fidelity of CNOT Circuits on NISQ Hardware}
\def\authorrunning{Kim, D., Kim, M., Li, S. M., Mosca, M.}
\begin{document}
\maketitle

\begin{abstract}

We introduce an improved CNOT synthesis algorithm that considers nearest-neighbour interactions and CNOT gate error rates in noisy intermediate-scale quantum (NISQ) hardware. Compared to IBM's Qiskit compiler, it improves the fidelity of a synthesized CNOT circuit by about 2 times on average (up to 9 times). It lowers the synthesized CNOT count by a factor of 13 on average (up to a factor of 162).

Our contribution is twofold. First, we define a $\Cost$ function by approximating the average gate fidelity $F_{avg}$. According to the simulation results, $\Cost$ fits the error probability of a noisy CNOT circuit, $\Prob = 1 - F_{avg}$, much tighter than the commonly used cost functions. On IBM's fake Nairobi backend, it matches $\Prob$ to within $10^{-3}$. On other backends, it fits $\Prob$ to within $10^{-1}$. 
$\Cost$ accurately quantifies the dynamic error characteristics and shows remarkable scalability. Second, we propose a noise-aware CNOT routing algorithm, NAPermRowCol, by adapting the leading Steiner-tree-based connectivity-aware CNOT synthesis algorithms. A weighted edge is used to encode a CNOT gate error rate and $\Cost$-instructed heuristics are applied to each reduction step. NAPermRowCol does not use ancillary qubits and is not restricted to certain initial qubit maps. Compared with algorithms that are noise-agnostic, it improves the fidelity of a synthesized CNOT circuit across varied NISQ hardware. Depending on the benchmark circuit and the IBM backend selected, it lowers the synthesized CNOT count up to $56.95\%$ compared to ROWCOL and up to $21.62\%$ compared to PermRowCol. It reduces the synthesis $\Cost$ up to $25.71\%$ compared to ROWCOL and up to $9.12\%$ compared to PermRowCol. Our method can be extended to route a more general quantum circuit, giving a powerful new tool for compiling on NISQ devices.
\end{abstract}


{\hypersetup{hidelinks}
\tableofcontents}

\section{Introduction}
The current phase of quantum technology is called "Noisy Intermediate-Scale Quantum" (NISQ)~\cite{preskill2018quantum}. It is defined by the number of physical qubits (i.e., approximately between 50 and a few hundred) and their ``noisy'' nature (e.g., erroneous gate operation). Recently, the focus has shifted from merely increasing the number of qubits to enhancing their quality and error-correction capabilities~\cite{google2023suppressing,bluvstein2023logical,bravyi2023high,google1,google2,ryan2021realization,sundaresan2023demonstrating,xu2023constant}. This marks a transition toward more reliable and practical quantum computing, namely, the ``Logical Intermediate-Scale Quantum'' (LISQ)~\cite{lisq}.

Despite these pivotal changes, the current and near-future landscape is still dominated by NISQ devices~\cite{ibm,ibm0,arute2019quantum,google,intel,intel0}. They allow researchers to explore the potential of quantum computers and carry out various tasks such as quantum simulation~\cite{barratt2021parallel}, combinatorial optimization~\cite{gemeinhardt2023quantum}, cryptography~\cite{harshvardhan2023simulating}, and quantum chemistry \cite{bharti2022noisy}. 
Therefore, it is necessary and important to improve compilation methods tailored to NISQ devices and minimize the resource overhead~\cite{cai2020resource,hansen2023resource,volume}.

\textbf{Quantum circuit routing} is the problem of mapping a logical circuit to NISQ hardware. The connectivity of NISQ architecture restricts a two-qubit gate to adjacent qubits, while various noise sources impact the reliability of quantum computations. An optimal solution to quantum circuit routing minimizes the resource overhead and maximizes the success probability of running a NISQ-executable circuit~\cite{de2022advances,maslov2008quantum,murali2019noise}. 

\subsection{Problem of Interest}
The main source of error in a NISQ device comes from entangling gates, whose error rates differ by a huge margin. These gates serve as a major component of many quantum programs. Although CNOT gates are not the native entangling operations in most NISQ hardware, they are commonly used in theoretical quantum circuits and algorithm designs. Since logical quantum circuits are transpiled to the device’s native gate set and a CNOT gate is equivalent to a native entangling gate up to local operations, we consider CNOT gates for simplicity. For instance, IBM's quantum devices use the Echoed Cross-Resonance (ECR) gate as their entangling operation. It is equivalent to a CNOT gate up to single-qubit pre-rotations~\cite{ECRGate,waring2024noise}. Since CNOT error rate is about an order of magnitude higher than the single-qubit gate error rate~\cite{IBMCalibration,murali2019noise}, the errors introduced by pre-rotations are relatively small. As a result, the CNOT error rate is about the same magnitude as that of an ECR gate. 

\textbf{Noise-aware CNOT circuit routing} is a subproblem of quantum circuit routing where the logical circuits are composed of only CNOT gates. It accounts for the hardware connectivity and CNOT error rate to successfully route a CNOT circuit in a scalable manner~\cite{chen2023nearest,zhu2022physical}. Although it is not immediately applicable for mapping a quantum program to NISQ devices, we can always decompose a quantum circuit into a layer of CNOT gates, followed by single-qubit gates, and then another layer of CNOT gates and so on~\cite{gheorghiu2022reducing}. Hence, in this paper, we focus on routing a noisy CNOT circuit, with the goal of reducing the resource overhead and improving the execution success probability. 

CNOT circuits have well-behaved mathematical structures. The output parity terms of an $n$-qubit CNOT circuit correspond to an $n\times n$ parity matrix over $\F_2$, a non-singular binary square matrix. By using Gaussian elimination, we can decompose the parity matrix into a sequence of row operations, each of which corresponds to a CNOT gate~\cite{amy2018controlled,2008PMH}. After concatenating these gates, we obtain a circuit with the same semantics as the input CNOT circuit (\cref{subsec:synthesis}). This process is also called \textbf{CNOT circuit synthesis}. 

\textbf{Connectivity-aware CNOT circuit synthesis} (a.k.a., CNOT circuit routing) takes a logical CNOT circuit and a uniform edge-weighted connectivity graph as inputs. It ignores the error distribution of different CNOT gates and returns a sequence of physically allowed CNOT operations. 
In literature, the connectivity constraint is also called the nearest-neighbour (NN) interaction. SWAP-based synthesis is one of the predominant methods to route a quantum circuit by relocating logical qubits in quantum registers~\cite{li2019tackling,sivarajah2020tket,qiskit,mcts}. A major downside is that a SWAP gate equals three CNOT gates, so adding SWAP gates to quantum circuits results in an explosion of CNOT count.

Alternatively, Steiner-tree-based synthesis uses the parity matrix representation of a CNOT circuit and existing heuristics to optimize the synthesized CNOT count (\cref{subsec:routeCNOT}). As a variant of a minimum-spanning tree, a Steiner tree finds the shortest path to connect a given set of vertices, corresponding to the shortest sequence of CNOT gates to route a subcircuit. Compared to simply inserting SWAP gates, reducing connectivity-aware CNOT circuit synthesis to a Steiner tree problem suppresses the CNOT explosion~\cite{gheorghiu2022reducing,griendli2022,meijer-vandegriend2020architecture,kissinger2020cnot,nash2020quantum,vandaele2022phase}. 

For a quantum computer to be powerful, we should consider not only the idealized computation model, but also the imperfections and variations in the real system. Most of the leading Steiner-tree-based algorithms assume a uniform error distribution across the quantum system and use the synthesized CNOT count as their cost function. In practice, the error rate of each CNOT gate varies significantly depending on the coupled qubits' unique properties, their system positions, and the nature of the interactions they participate in. If certain qubits or connections have higher error rates, using them frequently might decrease the execution's accuracy, even when the overall gate count is reduced. For example, a routing path that minimizes CNOT count may choose more expensive edges (i.e., CNOT gates with higher error rates). It is unclear how closely the synthesized CNOT count aligns with the accumulated error in a noisy CNOT circuit. 

\subsection{Our Contributions}
In this work, we focus on improving the fidelity of a NISQ-executable CNOT circuit. First, we approximate its average gate fidelity and propose a scalable $\Cost$ function to gauge its quality. In \cref{sec:ApproximatedCostFunction}, we formally derive $\Cost$ by assuming CNOT gates are not parallelized and there is no noise on idle qubits. While these seem to be strong assumptions, we show that compared to the intuitively defined cost functions~\cite{chen2023nearest,murali2019noise,zhu2022physical}, $\Cost$ fits the error probability of a noisy CNOT circuit tightly, for varied hardware topology and error distribution (\cref{subsec: compare functions}). To the best of our knowledge, no one before us has investigated what is a good and scalable approximation of a CNOT circuit's reliability. $\Cost$ is the first attempt towards an efficient and accurate quantification of a noisy quantum circuit's reliability, rather than simply summing up the gate error rates or estimating the error probability disregarding the system size.

Next, we apply $\Cost$ to make a noise-aware adaptation of the algorithm PermRowCol~\cite{griendli2022}, which is the leading Steiner-tree-based connectivity-aware CNOT circuit synthesis algorithm. In \cref{sec:algorithm}, we propose the algorithm NAPermRowCol to account for the connectivity and noise constraints. It not only returns a synthesized circuit with allowed CNOT operations and increased reliability, but also reduces the CNOT count by factoring out SWAP gates. Our technique can be summarized in two key points. First, it uses a noise-aware heuristic for pivot selection before each reduction step. Second, it prioritizes the cheapest way to route a noisy CNOT subcircuit. More precisely, it minimizes the cost to (1) propagate a parity 1 from one of all terminal nodes to a fixed Steiner node in each column reduction's first step, or (2) propagate the parity of a fixed Steiner node to one of all terminal nodes in each row reduction's first step.

Compared with GENNS~\cite{zhu2022physical}, NAPermRowCol is not restricted to certain initial qubit maps. Compared with Qiskit~\cite{li2019tackling}, it uses no ancilla and is thus more resource-efficient. Compared to the leading CNOT routing algorithms, our benchmark results show that NAPermRowCol consistently improves the fidelity of a synthesized CNOT circuit across varied topologies (\cref{subsec: compare algorithms}). Moreover, it reduces the synthesized CNOT count, shortening the circuit execution time. 

\subsection{Open Problems and Discussions}
While noise-aware CNOT circuit routing is not immediately applicable to NISQ devices, our results pave the way to a fully noise-aware routing strategy. Here, we outline how NAPermRowCol could be extended to route a noisy quantum circuit over a universal gate set.

On one hand, we need an efficient cost function to quantify a noisy circuit's reliability. For noisy CNOT circuits, we use a generalized Pauli channel to model noise, assuming CNOT gates are not parallelized and idle qubits are noiseless. Although our empirical study shows that $\Cost$ approximates $\Prob$ closely, these simplifications make our cost function less general. A major room for improvement is to drop these assumptions. 

For example, the generalized Pauli channels are a widely applicable model in quantum information, but they do not encompass all possible completely positive maps (e.g., complex multi-qubit interactions, non-Pauli types of noise and decoherence). Modelling a noisy quantum circuit with a more general channel representation is an important avenue for future work. Additionally, we should allow gate parallelization and account for errors resulting from the T1/T2 time. Moreover, we should consider single-qubit gate errors and readout errors arising from qubit state measurement. Depending on the quantum computing platform, we should also factor in error sources such as thermal relaxation and crosstalk. 

In summary, the new cost function should accurately and efficiently quantify the dynamic error characteristics of a NISQ-executable circuit. Next, we can combine it with various circuit synthesis frameworks~\cite{gheorghiu2022reducing,martiel2022architecture,murphy2023global} to generalize NAPermRowCol and route an arbitrary quantum program on NISQ hardware. This includes designing noise-aware heuristics on a vertex-and-edge-weighted Steiner tree. We will then enhance the routed circuit's reliability by reducing the cost evaluation. 

On the other hand, a routing algorithm's performance may decrease as the system scales with the qubit count and circuit complexity. For example, NAPermRowCol shows modest improvements when synthesizing a CNOT circuit over more than 16 qubits. This means we may need to refine the noise-aware heuristics to improve a routing algorithm's scalability.


\subsection{Related Work}
Substantial progress was made to understand noise within a quantum system~\cite{gupta2020adaptive,hamilton2020scalable,horodecki1998general,nielsen2002simple,schumacher1996sending,shaib2023efficient}. The problem of quantum circuit routing was introduced in \cite{maslov2008quantum}. Since then, numerous papers have appeared studying this problem~\cite{chen2023nearest,CalAware,maslov2008quantum,murali2019noise,niu2020hardware,quetschlich2023compiler,waring2024noise,zhu2022physical}. Most of them use intuitive cost functions without formal basis. As cost function is a vital metric to improve the reliability of a NISQ-executable circuit, it is important to quantify a system's noise accurately and efficiently. In what follows, we describe in more detail those closely related to what we do.

\cite{maslov2008quantum} proposes a routing strategy tailored to Nuclear Magnetic Resonance (NMR) quantum computers. It uses circuit runtime as a cost function and inserts SWAP gates for quantum circuit routing. When testing its performance and scalability, the authors assume a linear architecture with a uniform interaction time between adjacent qubits. This is no longer a reasonable assumption for the leading NISQ devices. Moreover, NMR has largely been overshadowed by more practical quantum technologies, whose hardware topology is better suited for reliable and scalable quantum computing. For example, as of 2021, the topology of all active IBM quantum computers is based around the heavy-hex lattice. In each cell of the lattice, qubits are arranged in a hexagon, with an additional qubit on each edge~\cite{heavyhex}. Therefore, it is of pressing importance to develop noise-aware circuit routing strategies tailored to these state-of-the-art NISQ architectures.

\cite{waring2024noise} does not employ a cost function but instead suggests a noise-aware partition of a weighted connectivity graph as a preprocessor for IBM's Qiskit transpiler. According to user-defined error thresholds, it gets rid of disconnected vertices and graph components with high CNOT or readout error rates. Compared to Qiskit's inherent binary classification on whether a qubit is faulty or not, it takes advantage of the quantum processor's noise profile to make adaptive topology-pruning decisions.

\cite{zhu2022physical} proposes the algorithm GENNS to route a CNOT circuit on IBM's NISQ devices. It enhances the routed circuit's reliability by accounting for the NN interactions and CNOT error rates. These restrictions are encoded in an edge-weighted connectivity graph. Among all pairs of connected qubits, GENNS uses the sum of edge weights as a cost function and applies the Floyd-Warshall algorithm~\cite{floyd1962algorithm} to find the shortest path. Since the cost function is proposed without a formal basis, GENNS may not return the most reliable routed circuit. Moreover, it is restricted to a feasible initial qubit map, or its reduction step terminates upon an invalid row operation. In the empirical study, GENNS is benchmarked with relatively short CNOT circuits (containing up to 256 gates) on 5- and 20-qubit backends. It is unclear how scalable and adaptive GENNS is for synthesizing a large CNOT circuit on varied IBM's backends.


\cite{chen2023nearest,murali2019noise} propose comprehensive compiling strategies that are also customized for IBM's architecture. They encapsulate a noise-aware initial qubit mapping and a subsequent routing of the mapped circuit. \cite{chen2023nearest} concentrates on routing a noisy CNOT circuit and proposes a Steiner-tree-based synthesis algorithm. It uses path fidelity (the product of CNOT success rates) as its cost function to instruct path selections in a Steiner tree. However, it is not clear how this cost function is related to the routed circuit reliability. \cite{murali2019noise} offers a collection of optimization- and heuristic-based methods to map an arbitrary quantum program to NISQ hardware. It accounts for CNOT and readout errors, as well as the connectivity and gate scheduling constraints. Its technique can be summarized in two key points. First, it reduces the problem of finding an optimal initial qubit mapping to a constrained optimization problem. Based on the linearized reliability score (i.e., the logarithm of the product of the CNOT gate and measurement success rates), it leverages the quantum analogue of the Satisfiability Modulo Theory (SMT) solver to find an optimal solution. Second, it proposes greedy heuristics that have comparable performance to the SMT-based methods, with improved scalability. Both approaches use SWAP gates for circuit routing and focus on an obsolete IBM grid topology in their empirical study. 

\section{Preliminaries}
\label{sec:preliminaries}
Here we review the core concepts for synthesizing a CNOT circuit on NISQ hardware. In \cref{subsec:conventions}, we introduce notions and conventions that will be used in this paper. In \cref{subsec:synthesis}, we define the parity matrix representation of a CNOT circuit and use it to synthesize a noiseless CNOT circuit without any connectivity constraint. This is also known as the ``CNOT circuit synthesis''. In \cref{subsec:routeCNOT}, we introduce the Steiner tree, a variant of the minimum spanning tree, and use it to synthesize a noiseless CNOT circuit based on a hardware topology. This is also known as the ''connectivity-aware CNOT circuit synthesis''. In \cref{subsec:errormitigation}, we define the average gate fidelity for a noisy quantum circuit and motivate noise-aware CNOT circuit routing on NISQ hardware.

\subsection{Notations and Conventions}
\label{subsec:conventions}
$I$ denotes an identity operator (a.k.a, identity matrix), whose dimension can be inferred from the context. $\mathbb{C}$ denotes the set of complex numbers, $\N$ denotes the set of nonnegative integers, $\N^{\neq 0}=\N \setminus\{0\}$, and $\Z_q$ denotes the set of integers $\{0, 1, \ldots, q-1\}$. LHS (RHS) is short for the ``lefthand (righthand) side'' of an equation. $\Tr[A]$ denotes the trace of a matrix $A$. It is the sum of elements on the main diagonal of $A$. $A^{\top}$ denotes the transpose of matrix $A$. The Pauli matrices are $2\times 2$ unitary operators acting on a single qubit. Let $i$ be the imaginary unit, $i^4 = 1$.
\[
I = \begin{bmatrix} 1 & 0 \\ 0 & 1
\end{bmatrix}, \qquad X = \begin{bmatrix} 0 & 1 \\ 1 & 0
\end{bmatrix}, \qquad Y = \begin{bmatrix} 0 & -i \\ i & 0
\end{bmatrix}, \qquad Z = \begin{bmatrix} 1 & 0 \\ 0 & -1
\end{bmatrix}.
\]

By direct computation, $\Tr[I]=2$, $\Tr[X] = \Tr[Y] = \Tr[Z] = 0$.

\begin{definition}
    For $n \in \N^{\neq 0}$, let $\mathcal{C}_n$ be the $n$-qubit Clifford group. $\mathcal{C}_n$ is generated by H, S, and CNOT gates through tensor product and composition. $\mathcal{C}_n = \langle H, \;S, \;CNOT\rangle$, where
    
    \begin{equation}
      H = \frac{1}{\sqrt{2}}\begin{bmatrix}
        1 & \;\;\;1\\
        1 & -1
    \end{bmatrix}, \qquad S = \begin{bmatrix}
        1 & 0\\
        0 & i
    \end{bmatrix}, \qquad \CNOT = \begin{bmatrix}
        1 & 0 & 0 & 0\\
        0 & 1 & 0 & 0\\
        0 & 0 & 0 & 1\\
        0 & 0 & 1 & 0
    \end{bmatrix}.
    \label{unitaries}
    \end{equation}
\end{definition}

\begin{remark}
    $\mathcal{C}_1 = \langle H,S\rangle.$
\end{remark}

\textbf{CNOT gate} is short for the ``controlled-not gate''. It is a quantum gate acting on two qubits. Its unitary matrix is shown in \cref{unitaries}. A \textbf{CNOT circuit} is a circuit composed of only CNOT gates. On top of the unitary matrix representation, there are different ways to represent a CNOT circuit. In this paper, they are used based on the problem that we are trying to solve. The parity matrix is often used to represent a logical CNOT circuit~(\cref{def:parity row}). The Kraus decomposition (\cref{def:multiequbit}) and the superoperator representation (\cref{lem: superoperator not unitary}) are often used to represent a noisy CNOT circuit. 

When the qubit count $n$ increases, the parity matrix size grows polynomially with $n$, while the unitary matrix size grows exponentially. Therefore, the parity matrix representation is scalable and convenient for different instances of the CNOT routing problem. For example, it is used in CNOT circuit synthesis (\cref{subsec:synthesis}), connectivity-aware CNOT circuit synthesis (\cref{subsec:routeCNOT}), and noise-aware CNOT circuit routing (\cref{sec:algorithm}). 

\subsection{CNOT Circuit Synthesis}
\label{subsec:synthesis}
CNOT circuit synthesis takes a parity matrix as input and returns a CNOT circuit performing the desired operation~\cite{alber2001quantum,2008PMH,shende2003synthesis}. It allows the coupling of any pair of qubits and aims to reduce the synthesized gate count. In what follows, we show that every CNOT circuit can be uniquely represented by a non-singular binary square matrix, namely, the parity matrix. 

\paragraph{The parity matrix representation of a CNOT circuit}

Consider computational basis states $\ket{c}$ and $\ket{t}$, $c,t\in \Z_2$. Let $\oplus$ denote the bitwise XOR operation. $\ket{c}$ and $\ket{t}$ are called the \textbf{control and target qubit states}. When $\ket{c} = \ket{0}$, CNOT acts trivially on $\ket{t}$. Otherwise, CNOT acts as a NOT gate and flips $\ket{t}$. For convenience, the operation in \cref{eqn: cnot} is denoted as CNOT(c, t). 
\begin{align}
    \CNOT\ket{c}\ket{t} = \ket{c}\ket{c\oplus t}.
\label{eqn: cnot}
\end{align}

\begin{definition}
Let \vect{A} be a binary square matrix. In the matrices below, we use $'$ to distinguish between the indices for rows and columns. This convention will become useful in the CNOT circuit synthesis. When denoting a column operation C($i'$, $j'$) or column $j'$, we drop $'$ for brevity.
\begin{itemize}
    \item $R(c, t)\cdot\vect{A}$ denotes a \textbf{row operation} on \vect{A}, where row $c$ is added to row $t$ modulo $2$ and row $c$ remains unchanged. Let $\llbracket R(c, t) \rrbracket$ be the matrix representation of $R(c, t)$. Its diagonal components are equal to $1$. Its off-diagonal components are equal to $0$, except for the entry on row $t$ and column $c$. 
    \[
    R(c, t)\cdot\vect{A} =\llbracket R(c, t)\rrbracket\vect{A}, \qquad \llbracket R(c, t) \rrbracket=\begin{blockarray}{ccccccc}
& \cdots & c' & \cdots & t' & \cdots \\
\begin{block}{c(cccccc)}
\vdots & & \vdots & & \vdots & \\
c & \cdots & 1 & \cdots & 0 & \cdots\\
\vdots & & \vdots & & \vdots &  \\
t & \cdots & \color{IMSRed}1 & \cdots & 1 & \cdots \\
\vdots & & \vdots & & \vdots &  \\
\end{block}
\end{blockarray}.
    \]
    \item $\vect{A}\cdot C(c, t)$ denotes a \textbf{column operation} on \vect{A}, where column $c$ is added to column $t$ modulo $2$ and column $c$ remains unchanged. Let $\llbracket C(c, t) \rrbracket$ be the matrix representation of $C(c, t)$. Its diagonal components are equal to $1$. Its off-diagonal components are equal to $0$, except for the entry on row $c$ and column $t$. 
    \[
    \vect{A}\cdot C(c, t) = \vect{A}\llbracket C(c, t) \rrbracket,\qquad \llbracket C(c, t) \rrbracket = \begin{blockarray}{ccccccc}
& \cdots & c' & \cdots & t' & \cdots \\
\begin{block}{c(cccccc)}
\vdots & & \vdots & & \vdots & \\
c & \cdots & 1 & \cdots & \color{IMSRed}1 & \cdots\\
\vdots & & \vdots & & \vdots &  \\
t & \cdots & 0 & \cdots & 1 & \cdots \\
\vdots & & \vdots & & \vdots &  \\
\end{block}
\end{blockarray}.
    \]
\end{itemize}
\label{def:row column operation}
\end{definition}


\begin{definition}
A CNOT circuit over $n$ qubits can be uniquely represented by an $n\times n$ binary square matrix, namely, the parity matrix. Let $0 \leq i,j \leq n-1$. The $i$-th row represents the $i$-the input qubit. The $j$-th column represents the parity term on the $j$-th output qubit~\cite{alber2001quantum,shende2003synthesis,2008PMH}. 
    \label{def:parity row}
\end{definition}

\begin{example}
    The parity matrix of an $n$-qubit empty circuit is an $n \times n$ identity matrix $I$. 
\end{example}

\begin{example}
    \cref{fig:CNOT-circuit} is a circuit consisting of one CNOT gate, CNOT(c, t). \cref{fig:paritymatrix cnot} shows its parity matrix representation. CNOT(c, t) corresponds to performing a column operation C(c, t) on $I$, or performing a row operation R(t, c) on $I$. 
    \[
    \vect{A}= I\llbracket C(c, t) \rrbracket= \llbracket R(t, c) \rrbracket I.
    \]

    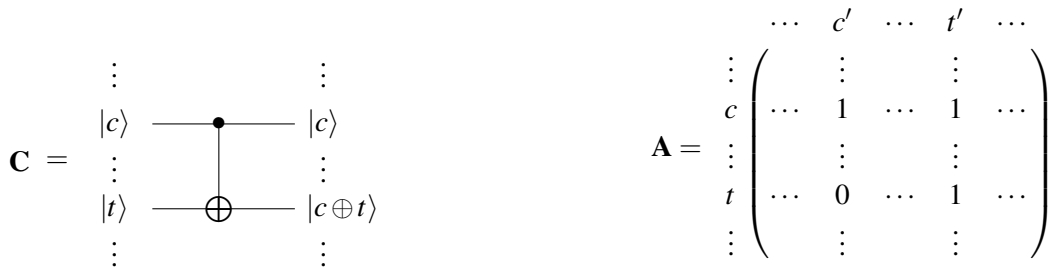
\begin{figure}[H]
\begin{subfigure}{.4\textwidth}
\centering
\begin{tikzpicture}
	\begin{pgfonlayer}{nodelayer}
		\node [style=none] (2) at (4, 0.75) {};
		\node [style=none] (3) at (7.75, 0.75) {};
		\node [style=none] (4) at (4, -1.5) {};
		\node [style=none] (5) at (7.75, -1.5) {};
		\node [style=none] (7) at (3, 0.75) {$|c\rangle$};
		\node [style=none] (8) at (3, -1.5) {$|t\rangle$};
		\node [style=none] (12) at (3, 2.25) {$\vdots$};
		\node [style=none] (13) at (3, -0.25) {$\vdots$};
		\node [style=none] (14) at (3, -2.5) {$\vdots$};
		\node [style=none] (15) at (5.75, -1.5) {$\bigoplus$};
		\node [style=none] (16) at (5.75, 0.75) {$\bullet$};
		\node [style=none] (17) at (8.5, 0.75) {$|c\rangle$};
		\node [style=none] (18) at (9, -1.5) {$|c\oplus t\rangle$};
		\node [style=none] (21) at (8.5, 2.25) {$\vdots$};
		\node [style=none] (22) at (8.5, -0.25) {$\vdots$};
		\node [style=none] (23) at (8.5, -2.5) {$\vdots$};
		\node [style=none] (26) at (0.5, -0.25) {$\mathbf{C}$};
		\node [style=none] (27) at (1.5, -0.25) {$=$};
	\end{pgfonlayer}
	\begin{pgfonlayer}{edgelayer}
		\draw (2.center) to (3.center);
		\draw (4.center) to (5.center);
		\draw (16.center) to (15.center);
	\end{pgfonlayer}
\end{tikzpicture}

\subcaption{\vect{C} is an $n$-qubit CNOT circuit with one CNOT gate.}
  \label{fig:CNOT-circuit}
\end{subfigure}
\begin{subfigure}{.6\textwidth}
  \centering
\[
\vect{A}= \begin{blockarray}{ccccccc}
& \cdots & c' & \cdots & t' & \cdots \\
\begin{block}{c(cccccc)}
\vdots &  & \vdots & & \vdots & \\
c & \cdots & 1 & \cdots & 1 & \cdots\\
\vdots & & \vdots & & \vdots &  \\
t & \cdots & 0 & \cdots & 1 & \cdots \\
\vdots & & \vdots & & \vdots & \\
\end{block}
\end{blockarray}
\]
\vspace{-.8cm}
  \subcaption{\vect{A} is the parity matrix of \vect{C}.}
  \label{fig:paritymatrix cnot}
\end{subfigure}
\caption{The action of a CNOT gate corresponds to a column (row) operation on $I$. This allows us to derive the parity matrix of a CNOT circuit.}
\label{fig:paritydef}
\end{figure}
\end{example}

\paragraph{Derive the parity matrix of a CNOT circuit}
\label{par:derivation}
To see the correspondence between a CNOT circuit and its parity matrix, consider a $4$-qubit CNOT circuit in \cref{fig:4qcircuit}. Denote the initial state on each qubit wire by $\ket{0}$, $\ket{1}$, $\ket{2}$, and $\ket{3}$. On the righthand side of \vect{C}, the output parity terms are $\ket{0 \oplus 2}$, $\ket{0 \oplus 3}$, $\ket{0 \oplus 1}$, and $\ket{1 \oplus 2 \oplus 3}$. They are expressed by $4$-dimensional binary vectors $\vect{b}_0$, $\vect{b}_1$, $\vect{b}_2$, and $\vect{b}_3$, where row $i$ denotes the participation of input qubit $i$ in the parity term. More precisely, $\vect{b}_{ij} = 1$ indicates that input qubit $i$ participates in the parity term $j$. $\vect{b}_{ij} = 0$ indicates otherwise. In \cref{fig:paritymatrix 4q}, column $j$ in \vect{A} corresponds to $\vect{b}_j$.

\[
\vect{b}_0 = \begin{bmatrix}
    1\\
    0\\
    1\\
    0
\end{bmatrix},\qquad \vect{b}_1 = \begin{bmatrix}
    1\\
    0\\
    0\\
    1
\end{bmatrix}, \qquad \vect{b}_2 = \begin{bmatrix}
    1\\
    1\\
    0\\
    0
\end{bmatrix}, \qquad \vect{b}_3 = \begin{bmatrix}
    0\\
    1\\
    1\\
    1
\end{bmatrix}.
\]


We can also express the entanglement in \vect{C} as a bipartite graph. In \cref{fig:bipartite graph 1,fig:bipartite graph 2}, the inputs on wires $0$, $1$, $2$, and $3$ are entangled in the parity terms $0 \oplus 2$, $0 \oplus 3$, $0 \oplus 1$, and $1 \oplus 2 \oplus 3$ on output wires $0'$, $1'$, $2'$, and $3'$ respectively. We use $'$ to denote an output wire of a quantum circuit. In the purple dashed box, $W_{in} = \{0,1,2,3\}$ denotes the input qubits. In the pink dashed box, $W_{out}=\{0',1',2',3'\}$ denotes the output qubits. $W_{in}$ and $W_{out}$ are two disjoint independent sets. The connectivity between them denotes the information propagation in \vect{C}. It is represented by the biadjacency matrix \vect{A}, which is the parity matrix of \vect{C}.

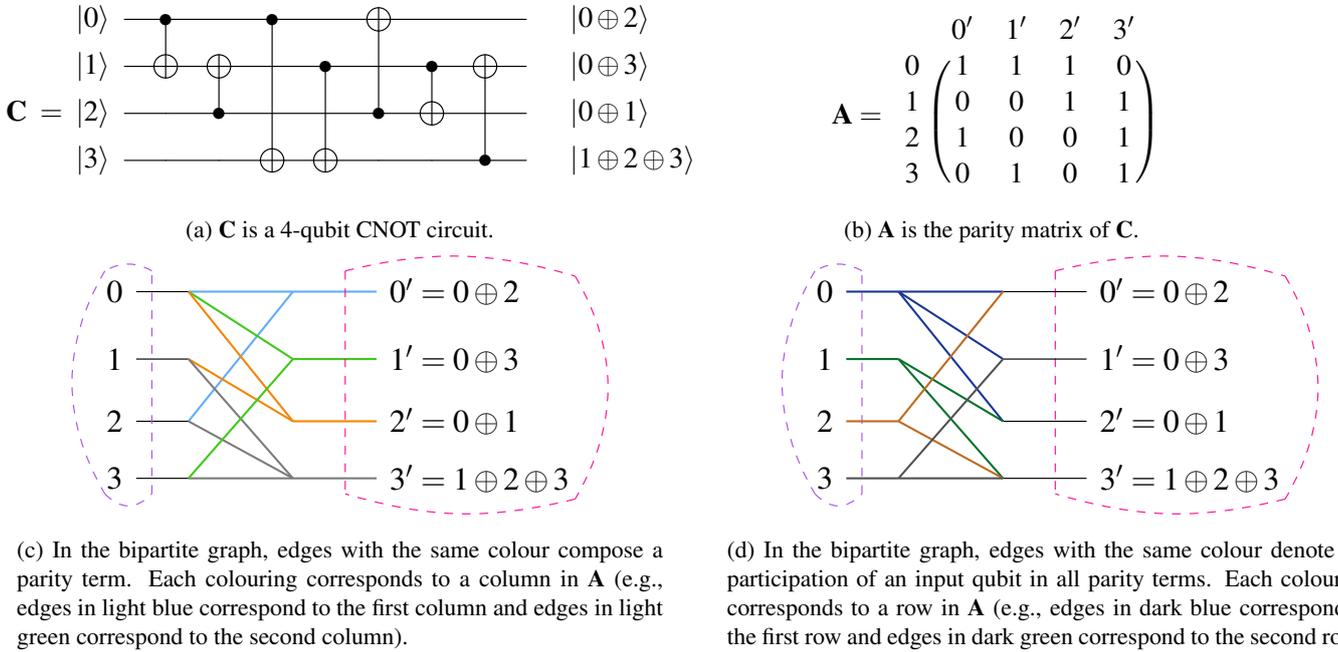
\begin{figure}[H]
\begin{subfigure}{.5\textwidth}
  \centering
\[
\Qcircuit @C=1em @R=0.8em @!R {
\lstick{\ket{0}}& \ctrl{1} & \qw & \ctrl{3} & \qw & \targ & \qw  & \qw & \qw &\rstick{\ket{0\oplus 2}}\\
\lstick{\ket{1}}& \targ & \targ & \qw & \ctrl{2} & \qw & \ctrl{1}  & \targ & \qw&\rstick{\ket{0\oplus 3}}\\
 \lstick{\vect{C} \; = \; \ket{2}}& \qw & \ctrl{-1} & \qw & \qw & \ctrl{-2}  & \targ & \qw & \qw &\rstick{\ket{0\oplus 1}}\\
\lstick{\ket{3}}& \qw & \qw & \targ & \targ & \qw & \qw & \ctrl{-2} & \qw &\rstick{\ket{1\oplus 2 \oplus 3}}}
\]
\subcaption{\vect{C} is a $4$-qubit CNOT circuit.}
  \label{fig:4qcircuit}
\end{subfigure}
\begin{subfigure}{.5\textwidth}
  \centering
\[
\vect{A}=\begin{blockarray}{ccccc}
& 0'& 1'& 2' & 3'\\
\begin{block}{c(cccc)}
0 & 1 & 1 & 1 & 0 \\
1 & 0 & 0 & 1 & 1 \\
2 & 1 & 0 & 0 & 1 \\
3 & 0 & 1 & 0 & 1 \\
\end{block}
\end{blockarray}
\]
\vspace{-.7cm}
  \subcaption{\vect{A} is the parity matrix of \vect{C}.}
  \label{fig:paritymatrix 4q}
\end{subfigure}
\begin{subfigure}{.5\textwidth}
    \centering
\scalebox{1.1}{\begin{tikzpicture}[tikzfig]
	\begin{pgfonlayer}{nodelayer}
		\node [style=none, label={left:$0$}] (0) at (-7.5, 1.25) {};
		\node [style=none] (1) at (-5, 1.25) {};
		\node [style=none] (2) at (0, 1.25) {};
		\node [style=none, label={right:$0' = 0 \oplus 2$}] (3) at (4, 1.25) {};
		\node [style=none, label={left:$1$}] (4) at (-7.5, -2) {};
		\node [style=none] (5) at (-5, -2) {};
		\node [style=none] (6) at (0, -2) {};
		\node [style=none, label={right:$1'=0 \oplus 3$}] (7) at (4, -2) {};
		\node [style=none, label={left:$2$}] (8) at (-7.5, -5) {};
		\node [style=none] (9) at (-5, -5) {};
		\node [style=none] (10) at (0, -5) {};
		\node [style=none, label={right:$2'=0 \oplus 1$}] (11) at (4, -5) {};
		\node [style=none, label={left:$3$}] (12) at (-7.5, -7.75) {};
		\node [style=none] (13) at (-5, -7.75) {};
		\node [style=none] (14) at (0, -7.75) {};
		\node [style=none, label={right:$3'=1 \oplus 2 \oplus 3$}] (15) at (4, -7.75) {};
		\node [style=none] (18) at (-9, 2.25) {};
		\node [style=none] (19) at (-6.75, 2.25) {};
		\node [style=none] (20) at (-9, -8.5) {};
		\node [style=none] (21) at (-6.75, -8.5) {};
		\node [style=none] (22) at (-7.5, 1.25) {};
		\node [style=none] (23) at (2.5, 2.25) {};
		\node [style=none] (24) at (13.5, 2) {};
		\node [style=none] (26) at (13.5, -8.75) {};
		\node [style=none] (27) at (2.5, -8.5) {};
	\end{pgfonlayer}
	\begin{pgfonlayer}{edgelayer}
		\draw (0.center) to (1.center);
		\draw (4.center) to (5.center);
		\draw (8.center) to (9.center);
		\draw (12.center) to (13.center);
		\draw [style=green line] (6.center) to (7.center);
		\draw [style=new edge style 0] (10.center) to (11.center);
		\draw [style=gray line] (14.center) to (15.center);
		\draw [style=violet box, bend right] (18.center) to (20.center);
		\draw [style=violet box, bend right=45, looseness=1.25] (20.center) to (21.center);
		\draw [style=violet box] (21.center) to (19.center);
		\draw [style=violet box, bend left=330] (19.center) to (18.center);
		\draw [style=white] (2.center) to (3.center);
		\draw [style=white] (1.center) to (2.center);
		\draw [style=white] (2.center) to (9.center);
		\draw [style=green line] (1.center) to (6.center);
		\draw [style=new edge style 0] (1.center) to (10.center);
		\draw [style=new edge style 0] (5.center) to (10.center);
		\draw [style=gray line] (5.center) to (14.center);
		\draw [style=gray line] (9.center) to (14.center);
		\draw [style=gray line] (13.center) to (14.center);
		\draw [style=green line] (13.center) to (6.center);
		\draw [style=pink box] (23.center) to (27.center);
		\draw [style=pink box, bend left=15] (23.center) to (24.center);
		\draw [style=pink box, bend left] (24.center) to (26.center);
		\draw [style=pink box, bend left=15] (26.center) to (27.center);
	\end{pgfonlayer}
\end{tikzpicture}}
  \subcaption{In the bipartite graph, edges with the same colour compose a parity term. Each colouring corresponds to a column in \vect{A} (e.g., edges in light blue correspond to the first column and edges in light green correspond to the second column).}
  \label{fig:bipartite graph 1}
\end{subfigure}
\qquad 
\begin{subfigure}{.5\textwidth}
    \centering
\scalebox{1.1}{\begin{tikzpicture}[tikzfig]
	\begin{pgfonlayer}{nodelayer}
		\node [style=none, label={left:$0$}] (0) at (-7.5, 1.25) {};
		\node [style=none] (1) at (-5, 1.25) {};
		\node [style=none] (2) at (0, 1.25) {};
		\node [style=none, label={right:$0' = 0 \oplus 2$}] (3) at (4, 1.25) {};
		\node [style=none, label={left:$1$}] (4) at (-7.5, -2) {};
		\node [style=none] (5) at (-5, -2) {};
		\node [style=none] (6) at (0, -2) {};
		\node [style=none, label={right:$1'=0 \oplus 3$}] (7) at (4, -2) {};
		\node [style=none, label={left:$2$}] (8) at (-7.5, -5) {};
		\node [style=none] (9) at (-5, -5) {};
		\node [style=none] (10) at (0, -5) {};
		\node [style=none, label={right:$2'=0 \oplus 1$}] (11) at (4, -5) {};
		\node [style=none, label={left:$3$}] (12) at (-7.5, -7.75) {};
		\node [style=none] (13) at (-5, -7.75) {};
		\node [style=none] (14) at (0, -7.75) {};
		\node [style=none, label={right:$3'=1 \oplus 2 \oplus 3$}] (15) at (4, -7.75) {};
		\node [style=none] (18) at (-9, 2.25) {};
		\node [style=none] (19) at (-6.75, 2.25) {};
		\node [style=none] (20) at (-9, -8.5) {};
		\node [style=none] (21) at (-6.75, -8.5) {};
		\node [style=none] (22) at (-7.5, 1.25) {};
		\node [style=none] (23) at (2.5, 2.25) {};
		\node [style=none] (24) at (13.5, 2) {};
		\node [style=none] (26) at (13.5, -8.75) {};
		\node [style=none] (27) at (2.5, -8.5) {};
	\end{pgfonlayer}
	\begin{pgfonlayer}{edgelayer}
		\draw [style=violet box, bend right] (18.center) to (20.center);
		\draw [style=violet box, bend right=45, looseness=1.25] (20.center) to (21.center);
		\draw [style=violet box] (21.center) to (19.center);
		\draw [style=violet box, bend left=330] (19.center) to (18.center);
		\draw [style=pink box] (23.center) to (27.center);
		\draw [style=pink box, bend left=15] (23.center) to (24.center);
		\draw [style=pink box, bend left] (24.center) to (26.center);
		\draw [style=pink box, bend left=15] (26.center) to (27.center);
		\draw [style=edge] (2.center) to (3.center);
		\draw [style=edge] (6.center) to (7.center);
		\draw [style=edge] (10.center) to (11.center);
		\draw [style=edge] (14.center) to (15.center);
		\draw [style=dark blue edge] (22.center) to (1.center);
		\draw [style=dark blue edge] (1.center) to (2.center);
		\draw [style=dark blue edge] (1.center) to (6.center);
		\draw [style=dark blue edge] (1.center) to (10.center);
		\draw [style=new edge style 1] (5.center) to (10.center);
		\draw [style=new edge style 1] (4.center) to (5.center);
		\draw [style=new edge style 1] (5.center) to (14.center);
		\draw [style=dark orange] (9.center) to (2.center);
		\draw [style=dark orange] (9.center) to (14.center);
		\draw [style=dark orange] (8.center) to (9.center);
		\draw [style=darg gray] (13.center) to (6.center);
		\draw [style=darg gray] (13.center) to (14.center);
		\draw [style=darg gray] (12.center) to (13.center);
	\end{pgfonlayer}
\end{tikzpicture}}
  \subcaption{In the bipartite graph, edges with the same colour denote the participation of an input qubit in all parity terms. Each colouring corresponds to a row in \vect{A} (e.g., edges in dark blue correspond to the first row and edges in dark green correspond to the second row).}
  \label{fig:bipartite graph 2}
\end{subfigure}
\caption{A $4$-qubit CNOT circuit \vect{C} is uniquely represented by a $4 \times 4$ parity matrix \vect{A}. Row $i$ denotes the state on the input qubit wire $i$. Column $j$ is the parity term $\vect{b}_j$ on the output qubit wire $j'$. The bipartite graph interpretation of \vect{C} shows the information propagation in the CNOT circuit. Different ways of colouring edges help us interpret a row and column in \vect{A}.}
\label{fig: example circuit matrix correspondence}
\end{figure}

\begin{definition}
Let \vect{C} be an $n$-qubit CNOT circuit with parity matrix \vect{A}. Right-concatenate \vect{C} with a sequence of $t$ CNOT gates, $t \in \N^{\neq 0}$. Let the resulting circuit be $\vect{C}_{\text{syn}}$. Without parallelizing CNOT gates, for $i_k, j_k \in \Z_n$, $k \in \Z_t$, 
\[
\vect{C}_{\text{syn}} = \vect{C}\circ\CNOT(j_0, i_0)\circ \CNOT(j_1, i_1)\circ \ldots \circ\CNOT(j_{t-1}, i_{t-1}).
\]

In the circuit diagram, it is visualized as follows.
\[
\begin{tikzpicture}
	\begin{pgfonlayer}{nodelayer}
		\node [style=none] (0) at (2.25, 1.75) {};
		\node [style=none] (1) at (4.25, 1.75) {};
		\node [style=none] (2) at (2.25, -1) {};
		\node [style=none] (3) at (4.25, -1) {};
		\node [style=none] (4) at (3.25, 0.25) {$\mathbf{C}$};
		\node [style=none] (5) at (4.25, 1.25) {};
		\node [style=none] (7) at (5.75, 1.25) {};
		\node [style=none] (8) at (4.25, -0.5) {};
		\node [style=none] (9) at (5.75, -0.5) {};
		\node [style=none] (10) at (-3.75, 1.75) {};
		\node [style=none] (11) at (0.75, 1.75) {};
		\node [style=none] (12) at (-3.75, -1) {};
		\node [style=none] (13) at (0.75, -1) {};
		\node [style=none] (14) at (-1.5, 0.25) {CNOT($j_0$, $i_0$)};
		\node [style=none] (15) at (0.75, 1.25) {};
		\node [style=none] (16) at (2.25, 1.25) {};
		\node [style=none] (17) at (0.75, -0.5) {};
		\node [style=none] (18) at (2.25, -0.5) {};
		\node [style=none] (19) at (-3.75, 1.75) {};
		\node [style=none] (20) at (-3.75, -1) {};
		\node [style=none] (21) at (-9.75, 1.75) {};
		\node [style=none] (22) at (-5.25, 1.75) {};
		\node [style=none] (23) at (-9.75, -1) {};
		\node [style=none] (24) at (-5.25, -1) {};
		\node [style=none] (25) at (-7.5, 0.25) {CNOT($j_1$, $i_1$)};
		\node [style=none] (26) at (-5.25, 1.25) {};
		\node [style=none] (27) at (-3.75, 1.25) {};
		\node [style=none] (28) at (-5.25, -0.5) {};
		\node [style=none] (29) at (-3.75, -0.5) {};
		\node [style=none] (30) at (-4.5, 0.5) {$\vdots$};
		\node [style=none] (33) at (1.5, 0.5) {$\vdots$};
		\node [style=none] (34) at (5, 0.5) {$\vdots$};
		\node [style=none] (35) at (-11.25, -0.5) {};
		\node [style=none] (36) at (-9.75, -0.5) {};
		\node [style=none] (37) at (-11.25, 1.25) {};
		\node [style=none] (38) at (-9.75, 1.25) {};
		\node [style=none] (39) at (-20, 1.75) {};
		\node [style=none] (40) at (-14, 1.75) {};
		\node [style=none] (41) at (-20, -1) {};
		\node [style=none] (42) at (-14, -1) {};
		\node [style=none] (43) at (-17, 0.25) {CNOT($j_{t-1}$, $i_{t-1}$)};
		\node [style=none] (44) at (-14, 1.25) {};
		\node [style=none] (45) at (-12.75, 1.25) {};
		\node [style=none] (46) at (-14, -0.5) {};
		\node [style=none] (47) at (-12.75, -0.5) {};
		\node [style=none] (48) at (-13.25, 0.5) {$\vdots$};
		\node [style=none] (49) at (-10.5, 0.5) {$\vdots$};
		\node [style=none] (50) at (-21.5, 1.25) {};
		\node [style=none] (51) at (-20, 1.25) {};
		\node [style=none] (52) at (-21.5, -0.5) {};
		\node [style=none] (53) at (-20, -0.5) {};
		\node [style=none] (54) at (-20.75, 0.5) {$\vdots$};
		\node [style=none] (55) at (-12, 1.25) {$\ldots$};
		\node [style=none] (56) at (-12, -0.5) {$\ldots$};
		\node [style=none] (57) at (-24, 0.25) {$\mathbf{C}_{syn}$};
		\node [style=none] (58) at (-22.75, 0.25) {$=$};
	\end{pgfonlayer}
	\begin{pgfonlayer}{edgelayer}
		\draw [style=edge] (0.center) to (1.center);
		\draw [style=edge] (1.center) to (3.center);
		\draw [style=edge] (3.center) to (2.center);
		\draw [style=edge] (0.center) to (2.center);
		\draw [style=edge] (5.center) to (7.center);
		\draw [style=edge] (8.center) to (9.center);
		\draw [style=edge] (10.center) to (11.center);
		\draw [style=edge] (11.center) to (13.center);
		\draw [style=edge] (13.center) to (12.center);
		\draw [style=edge] (10.center) to (12.center);
		\draw [style=edge] (15.center) to (16.center);
		\draw [style=edge] (17.center) to (18.center);
		\draw [style=edge] (19.center) to (20.center);
		\draw [style=edge] (21.center) to (22.center);
		\draw [style=edge] (22.center) to (24.center);
		\draw [style=edge] (24.center) to (23.center);
		\draw [style=edge] (21.center) to (23.center);
		\draw [style=edge] (26.center) to (27.center);
		\draw [style=edge] (28.center) to (29.center);
		\draw [style=edge] (35.center) to (36.center);
		\draw [style=edge] (37.center) to (38.center);
		\draw [style=edge] (39.center) to (40.center);
		\draw [style=edge] (40.center) to (42.center);
		\draw [style=edge] (42.center) to (41.center);
		\draw [style=edge] (39.center) to (41.center);
		\draw [style=edge] (44.center) to (45.center);
		\draw [style=edge] (46.center) to (47.center);
		\draw [style=edge] (50.center) to (51.center);
		\draw [style=edge] (52.center) to (53.center);
	\end{pgfonlayer}
\end{tikzpicture}
\]

Let $\vect{A}_{\text{syn}}$ be the parity matrix of $\vect{C}_{\text{syn}}$. $\vect{A}_{\text{syn}} = R(i_{t-1}, j_{t-1})\cdot\ldots \cdot R(i_1, j_1)\cdot R(i_0, j_0)\cdot \vect{A}.$
\label{def: CNOT composition}
\end{definition}

\paragraph{Exactly synthesize a CNOT circuit}
Next, we establish a one-to-one correspondence between a CNOT circuit and its parity matrix. \cref{lem:fullrank} shows that every parity matrix has full rank. Based on this linear algebraic property, \cref{lem: reverse engineering} proposes a way to exactly synthesize a CNOT circuit. 

\begin{lemma}
    The parity matrix of an $n$-qubit CNOT circuit is an $n\times n$ binary matrix of full rank.
    \label{lem:fullrank}
\end{lemma}
\begin{proof}
    By \cref{def: CNOT composition}, an $n$-qubit CNOT circuit corresponds to applying a sequence of row operations to $I$. Let the resulting matrix be \vect{A}. Since a square matrix is non-singular if and only if it is row-equivalent to an identity matrix, \vect{A} is a binary square matrix of full rank. By \cref{def:row column operation,def:parity row}, \vect{A} is the parity matrix of \vect{C}. Hence, the parity matrix of a CNOT circuit is a binary square matrix of full rank. 
\end{proof}

Let \vect{A} be the parity matrix of an $n$-qubit CNOT circuit \vect{C}. By \cref{lem:fullrank}, \vect{A} is a binary square matrix of full-rank. We can apply Gaussian elimination to find a sequence of row operations that sends \vect{A} to the identity matrix $I$. Since each row operation corresponds to a CNOT gate, we can obtain a CNOT circuit $\vect{C}_{syn}$ by concatenating the corresponding CNOT operations designated by the Gaussian elimination. This process is called \textbf{CNOT circuit synthesis}. $\vect{C}_{syn}$ is the synthesized CNOT circuit of $\vect{A}$. 

\begin{lemma}
 For $t\in\N^{\neq 0}$, $i_k, j_k \in \Z_n$, $k \in \Z_t$, let $R(i_0, j_0), \; R(i_1, j_1), \; \ldots, \; R(i_{t-1}, j_{t-1})$ be a sequence of row operation on \vect{A} such that
\begin{align}
    R(i_{t-1}, j_{t-1})\cdot \ldots \cdot R(i_1, j_1)\cdot R(i_0, j_0)\cdot \vect{A}=I.
    \label{eqn: row operation}
\end{align}
Then $\vect{C}_{syn} = \CNOT(j_{t-1}, i_{t-1}) \circ \ldots \circ \CNOT(j_1, i_1)\circ\CNOT(j_0, i_0)$ is a circuit representation of \vect{A}, so $\vect{C}$ and $\vect{C}_{syn}$ have the same semantics. 
\label{lem: reverse engineering}
\end{lemma}
\begin{proof}

By \cref{def: CNOT composition}, \cref{eqn: row operation} is expressed as 
\begin{align}
    \vect{C}\circ \bigl(\CNOT(j_0, i_0) \circ \CNOT(j_1, i_1) \circ \ldots \circ \CNOT(j_{t-1}, i_{t-1})\bigr) = I.
    \label{eqn: identity}
\end{align}

Right-multiplying \cref{eqn: identity} by $\bigl(\CNOT(j_0, i_0) \circ \CNOT(j_1, i_1) \circ \ldots \circ \CNOT(j_{t-1}, i_{t-1})\bigr)^{-1}$ yields
\begin{align*}
    \vect{C} &= \bigl(\CNOT(j_0, i_0) \circ \CNOT(j_1, i_1) \circ \ldots \circ \CNOT(j_{t-1}, i_{t-1})\bigr)^{-1}\\
    &=\CNOT(j_{t-1}, i_{t-1}) \circ \ldots \circ \CNOT(j_1, i_1)\circ\CNOT(j_0, i_0)\eqqcolon \vect{C}_{syn}.
\end{align*}

In the circuit diagram, $\vect{C}_{syn}$ is visualized as follows.
\[
\begin{tikzpicture}
	\begin{pgfonlayer}{nodelayer}
		\node [style=none] (10) at (-2.25, 1.75) {};
		\node [style=none] (11) at (3.75, 1.75) {};
		\node [style=none] (12) at (-2.25, -1) {};
		\node [style=none] (13) at (3.75, -1) {};
		\node [style=none] (14) at (0.75, 0.25) {CNOT($j_{t-1}$, $i_{t-1}$)};
		\node [style=none] (15) at (3.75, 1.25) {};
		\node [style=none] (17) at (3.75, -0.5) {};
		\node [style=none] (19) at (-2.25, 1.75) {};
		\node [style=none] (20) at (-2.25, -1) {};
		\node [style=none] (21) at (-11.25, 1.75) {};
		\node [style=none] (22) at (-6.75, 1.75) {};
		\node [style=none] (23) at (-11.25, -1) {};
		\node [style=none] (24) at (-6.75, -1) {};
		\node [style=none] (25) at (-9, 0.25) {CNOT($j_1$, $i_1$)};
		\node [style=none] (26) at (-6.75, 1.25) {};
		\node [style=none] (27) at (-5.5, 1.25) {};
		\node [style=none] (28) at (-6.75, -0.5) {};
		\node [style=none] (29) at (-5.5, -0.5) {};
		\node [style=none] (30) at (-6, 0.5) {$\vdots$};
		\node [style=none] (35) at (-12.75, -0.5) {};
		\node [style=none] (36) at (-11.25, -0.5) {};
		\node [style=none] (37) at (-12.75, 1.25) {};
		\node [style=none] (38) at (-11.25, 1.25) {};
		\node [style=none] (39) at (-17.25, 1.75) {};
		\node [style=none] (40) at (-12.75, 1.75) {};
		\node [style=none] (41) at (-17.25, -1) {};
		\node [style=none] (42) at (-12.75, -1) {};
		\node [style=none] (43) at (-15, 0.25) {CNOT($j_0$, $i_0$)};
		\node [style=none] (44) at (-12.75, 1.25) {};
		\node [style=none] (45) at (-12.75, 1.25) {};
		\node [style=none] (46) at (-12.75, -0.5) {};
		\node [style=none] (47) at (-12.75, -0.5) {};
		\node [style=none] (49) at (-12, 0.5) {$\vdots$};
		\node [style=none] (50) at (-18.75, 1.25) {};
		\node [style=none] (51) at (-17.25, 1.25) {};
		\node [style=none] (52) at (-18.75, -0.5) {};
		\node [style=none] (53) at (-17.25, -0.5) {};
		\node [style=none] (54) at (-18, 0.5) {$\vdots$};
		\node [style=none] (55) at (-4.5, 1.25) {$\ldots$};
		\node [style=none] (56) at (-4.5, -0.5) {$\ldots$};
		\node [style=none] (57) at (-21.25, 0.25) {$\mathbf{C}_{syn}$};
		\node [style=none] (58) at (-20, 0.25) {$=$};
		\node [style=none] (59) at (-3.5, -0.5) {};
		\node [style=none] (60) at (-2.25, -0.5) {};
		\node [style=none] (61) at (-3.5, 1.25) {};
		\node [style=none] (62) at (-2.25, 1.25) {};
		\node [style=none] (63) at (-3.5, 1.25) {};
		\node [style=none] (64) at (-3.5, -0.5) {};
		\node [style=none] (65) at (-2.75, 0.5) {$\vdots$};
		\node [style=none] (66) at (3.75, -0.5) {};
		\node [style=none] (67) at (5.25, -0.5) {};
		\node [style=none] (68) at (3.75, 1.25) {};
		\node [style=none] (69) at (5.25, 1.25) {};
		\node [style=none] (70) at (3.75, 1.25) {};
		\node [style=none] (71) at (3.75, -0.5) {};
		\node [style=none] (72) at (4.5, 0.5) {$\vdots$};
	\end{pgfonlayer}
	\begin{pgfonlayer}{edgelayer}
		\draw [style=edge] (10.center) to (11.center);
		\draw [style=edge] (11.center) to (13.center);
		\draw [style=edge] (13.center) to (12.center);
		\draw [style=edge] (10.center) to (12.center);
		\draw [style=edge] (19.center) to (20.center);
		\draw [style=edge] (21.center) to (22.center);
		\draw [style=edge] (22.center) to (24.center);
		\draw [style=edge] (24.center) to (23.center);
		\draw [style=edge] (21.center) to (23.center);
		\draw [style=edge] (26.center) to (27.center);
		\draw [style=edge] (28.center) to (29.center);
		\draw [style=edge] (35.center) to (36.center);
		\draw [style=edge] (37.center) to (38.center);
		\draw [style=edge] (39.center) to (40.center);
		\draw [style=edge] (40.center) to (42.center);
		\draw [style=edge] (42.center) to (41.center);
		\draw [style=edge] (39.center) to (41.center);
		\draw [style=edge] (50.center) to (51.center);
		\draw [style=edge] (52.center) to (53.center);
		\draw [style=edge] (59.center) to (60.center);
		\draw [style=edge] (61.center) to (62.center);
		\draw [style=edge] (66.center) to (67.center);
		\draw [style=edge] (68.center) to (69.center);
	\end{pgfonlayer}
\end{tikzpicture}
\]
\end{proof}


\paragraph{Derive the parity matrix of a SWAP circuit}
A \textbf{permutation matrix} is a binary square matrix equivalent to an identity matrix up to row and column permutation. An $n \times n$ permutation matrix \vect{P} represents a permutation of $n$ elements. Let M be an $n\times n$ matrix. \vect{P}M permutes the rows of M, while M\vect{P} permutes the columns of M. 

A \textbf{SWAP circuit} is a circuit over SWAP gates. It can be obtained by relabelling the qubit wires. Since a SWAP gate is equal to three CNOT gates, a SWAP circuit is a CNOT circuit whose parity matrix coincides with a permutation matrix. In \cref{fig:circuit matrix correspondence}, we use a $4$-qubit SWAP circuit $\vect{C}_{\text{SWAP}}$ to show that its parity matrix is exactly the permutation matrix for $\vect{C}_{\text{SWAP}}$.

\begin{figure}[H]
\begin{subfigure}{.65\textwidth}
  \centering
  \scalebox{.7}{\begin{tikzpicture}
	\begin{pgfonlayer}{nodelayer}
		\node [style=none, label={left:$0$}] (0) at (-6, 1.25) {};
		\node [style=none] (1) at (-3.5, 1.25) {};
		\node [style=none] (2) at (0, 1.25) {};
		\node [style=none, label={right:$0'$}] (3) at (2.5, 1.25) {};
		\node [style=none, label={left:$1$}] (4) at (-6, -1) {};
		\node [style=none] (5) at (-3.5, -1) {};
		\node [style=none] (6) at (0, -1) {};
		\node [style=none, label={right:$1'$}] (7) at (2.5, -1) {};
		\node [style=none, label={left:$2$}] (8) at (-6, -3.25) {};
		\node [style=none] (9) at (-3.5, -3.25) {};
		\node [style=none] (10) at (0, -3.25) {};
		\node [style=none, label={right:$2'$}] (11) at (2.5, -3.25) {};
		\node [style=none, label={left:$3$}] (12) at (-6, -5.25) {};
		\node [style=none] (13) at (-3.5, -5.25) {};
		\node [style=none] (14) at (0, -5.25) {};
		\node [style=none, label={right:$3'$}] (15) at (2.5, -5.25) {};
		\node [style=none] (16) at (-11, -2) {$\vect{C}_{\text{SWAP}}$};
		\node [style=none] (17) at (-8.5, -2) {$=$};
	\end{pgfonlayer}
	\begin{pgfonlayer}{edgelayer}
		\draw (0.center) to (1.center);
		\draw (2.center) to (3.center);
		\draw (4.center) to (5.center);
		\draw (6.center) to (7.center);
		\draw (8.center) to (9.center);
		\draw (10.center) to (11.center);
		\draw (12.center) to (13.center);
		\draw (14.center) to (15.center);
		\draw (1.center) to (10.center);
		\draw (5.center) to (2.center);
		\draw (13.center) to (6.center);
		\draw (9.center) to (14.center);
	\end{pgfonlayer}
\end{tikzpicture}}
  \caption{Circuit $\vect{C}_{\text{SWAP}}$ is composed of a sequence of SWAP gates. As a result, the inputs on wires $0$, $1$, $2$, and $3$ are mapped to output wires $2'$, $0'$, $3'$, and $1'$ respectively. Equivalently, we can think of $\vect{C}_{\text{SWAP}}$ as a bipartite graph, with an input part $W_{in} = \{0,1,2,3\}$ and an output part $W_{out}=\{0',1',2',3'\}$. $W_{in}$ and $W_{out}$ are two disjoint independent sets. The connectivity between them shows the information propagation within $\vect{C}_{\text{SWAP}}$.}
  \label{fig:permutation2}
\end{subfigure}
\quad
\begin{subfigure}{.3\textwidth}
  \[
  \vect{A}=\begin{blockarray}{ccccc}
     & 0' & 1' & 2' & 3'\\
    \begin{block}{c(cccc)}
      0 & 0 & 0 & 1 & 0 \\
      1 & 1 & 0 & 0 & 0 \\
      2 & 0 & 0 & 0 & 1 \\
      3 & 0 & 1 & 0 & 0 \\
    \end{block}
    \end{blockarray}
  \]
  \vspace{-1 em}
  \caption{\vect{A} is the parity matrix of $\vect{C}_{\text{SWAP}}$. It is the biadjacency matrix that describes the connectivity between $W_{in}$ and $W_{out}$ in the bipartite graph. It is known as a permutation matrix.}
  \label{fig: the parity matrix for permutation.}
\end{subfigure}
\caption{The parity matrix of an $n$-qubit SWAP circuit is an $n \times n$ permutation matrix.}
\label{fig:circuit matrix correspondence}
\end{figure}

\paragraph{Synthesize a CNOT circuit up to permutation}
Lastly, we generalize \cref{lem: reverse engineering} to CNOT circuit synthesis up to permutation. This establishes the soundness of algorithm PermRowCol~\cite{griendli2022}, based on which we develop the noise-aware CNOT routing algorithm NAPermRowCol in \cref{sec:algorithm}.

\begin{lemma}
Let \vect{A} be the parity matrix of an $n$-qubit CNOT circuit \vect{C}. For $t\in\N^{\neq 0}$ and $i_{t-1}, j_{t-1} \in \Z_n$, let $R(i_0, j_0), \; R(i_1, j_1), \; \ldots, \; R(i_{t-1}, j_{t-1})$ be a sequence of row operation on \vect{A} such that
\begin{align}
    R(i_{t-1}, j_{t-1})\cdot \ldots \cdot R(i_1, j_1)\cdot R(i_0, j_0)\cdot \vect{A}=\vect{P},
    \label{eqn: row operation2}
\end{align}
\vect{P} is the permutation matrix of a SWAP circuit $\vect{C}_{\text{SWAP}}$. Then 
\[
\vect{C}_{\text{fully syn}} = \vect{C}_{\text{SWAP}} \circ \vect{C}_{syn}, \qquad \vect{C}_{syn} = \CNOT(j_{t-1}, i_{t-1}) \circ \ldots \circ \CNOT(j_1, i_1)\circ\CNOT(j_0, i_0).
\]

$\vect{C}_{\text{fully syn}}$ is a circuit representation of \vect{A}, so $\vect{C}$ and $\vect{C}_{\text{fully syn}}$ have the same semantics. In other words, $\vect{C}_{\text{syn}}$ is the synthesized CNOT circuit of $\vect{A}$ up to permutation.
\label{lem: reverse engineering permutation}
\end{lemma}
\begin{proof}
By \cref{def: CNOT composition}, \cref{eqn: row operation2} is expressed as 
\begin{align}
    \vect{C}\circ \bigl(\CNOT(j_0, i_0) \circ \CNOT(j_1, i_1) \circ \ldots \circ \CNOT(j_{t-1}, i_{t-1})\bigr) = \vect{C}_{\text{SWAP}}.
    \label{eqn: identity2}
\end{align}

Right-multiplying \cref{eqn: identity2} by $\bigl(\CNOT(j_0, i_0) \circ \CNOT(j_1, i_1) \circ \ldots \circ \CNOT(j_{t-1}, i_{t-1})\bigr)^{-1}$ yields
\begin{align*}
    \vect{C} &= \vect{C}_{\text{SWAP}}\circ \bigl(\CNOT(j_0, i_0) \circ \CNOT(j_1, i_1) \circ \ldots \circ \CNOT(j_{t-1}, i_{t-1})\bigr)^{-1}\\
    &= \vect{C}_{\text{SWAP}}\circ \CNOT(j_{t-1}, i_{t-1}) \circ \ldots \circ \CNOT(j_1, i_1)\circ\CNOT(j_0, i_0)\\
    &= \vect{C}_{\text{SWAP}}\circ \vect{C}_{syn} \eqqcolon \vect{C}_{\text{fully syn}}.
\end{align*}

In the circuit diagram, $\vect{C}_{\text{fully syn}}$ is visualized as follows.
\[
\scalebox{.9}{\begin{tikzpicture}
	\begin{pgfonlayer}{nodelayer}
		\node [style=none] (10) at (-2.25, 1.75) {};
		\node [style=none] (11) at (3.75, 1.75) {};
		\node [style=none] (12) at (-2.25, -1) {};
		\node [style=none] (13) at (3.75, -1) {};
		\node [style=none] (14) at (0.75, 0.25) {CNOT($j_{t-1}$, $i_{t-1}$)};
		\node [style=none] (15) at (3.75, 1.25) {};
		\node [style=none] (17) at (3.75, -0.5) {};
		\node [style=none] (19) at (-2.25, 1.75) {};
		\node [style=none] (20) at (-2.25, -1) {};
		\node [style=none] (21) at (-11.25, 1.75) {};
		\node [style=none] (22) at (-6.75, 1.75) {};
		\node [style=none] (23) at (-11.25, -1) {};
		\node [style=none] (24) at (-6.75, -1) {};
		\node [style=none] (25) at (-9, 0.25) {CNOT($j_1$, $i_1$)};
		\node [style=none] (26) at (-6.75, 1.25) {};
		\node [style=none] (27) at (-5.5, 1.25) {};
		\node [style=none] (28) at (-6.75, -0.5) {};
		\node [style=none] (29) at (-5.5, -0.5) {};
		\node [style=none] (30) at (-6, 0.5) {$\vdots$};
		\node [style=none] (35) at (-12.75, -0.5) {};
		\node [style=none] (36) at (-11.25, -0.5) {};
		\node [style=none] (37) at (-12.75, 1.25) {};
		\node [style=none] (38) at (-11.25, 1.25) {};
		\node [style=none] (39) at (-17.25, 1.75) {};
		\node [style=none] (40) at (-12.75, 1.75) {};
		\node [style=none] (41) at (-17.25, -1) {};
		\node [style=none] (42) at (-12.75, -1) {};
		\node [style=none] (43) at (-15, 0.25) {CNOT($j_0$, $i_0$)};
		\node [style=none] (44) at (-12.75, 1.25) {};
		\node [style=none] (45) at (-12.75, 1.25) {};
		\node [style=none] (46) at (-12.75, -0.5) {};
		\node [style=none] (47) at (-12.75, -0.5) {};
		\node [style=none] (49) at (-12, 0.5) {$\vdots$};
		\node [style=none] (50) at (-18.75, 1.25) {};
		\node [style=none] (51) at (-17.25, 1.25) {};
		\node [style=none] (52) at (-18.75, -0.5) {};
		\node [style=none] (53) at (-17.25, -0.5) {};
		\node [style=none] (54) at (-18, 0.5) {$\vdots$};
		\node [style=none] (55) at (-4.5, 1.25) {$\ldots$};
		\node [style=none] (56) at (-4.5, -0.5) {$\ldots$};
		\node [style=none] (57) at (-22, 0.25) {$\mathbf{C}_{\text{fully syn}}$};
		\node [style=none] (58) at (-20, 0.25) {$=$};
		\node [style=none] (59) at (-3.5, -0.5) {};
		\node [style=none] (60) at (-2.25, -0.5) {};
		\node [style=none] (61) at (-3.5, 1.25) {};
		\node [style=none] (62) at (-2.25, 1.25) {};
		\node [style=none] (63) at (-3.5, 1.25) {};
		\node [style=none] (64) at (-3.5, -0.5) {};
		\node [style=none] (65) at (-2.75, 0.5) {$\vdots$};
		\node [style=none] (66) at (3.75, -0.5) {};
		\node [style=none] (67) at (5.25, -0.5) {};
		\node [style=none] (68) at (3.75, 1.25) {};
		\node [style=none] (69) at (5.25, 1.25) {};
		\node [style=none] (70) at (3.75, 1.25) {};
		\node [style=none] (71) at (3.75, -0.5) {};
		\node [style=none] (72) at (4.5, 0.5) {$\vdots$};
		\node [style=none] (73) at (5.25, 1.75) {};
		\node [style=none] (74) at (9.75, 1.75) {};
		\node [style=none] (75) at (5.25, -1) {};
		\node [style=none] (76) at (9.75, -1) {};
		\node [style=none] (77) at (7.5, 0.25) {$\mathbf{C}_{\text{SWAP}}$};
		\node [style=none] (78) at (9.75, 1.25) {};
		\node [style=none] (79) at (11, 1.25) {};
		\node [style=none] (80) at (9.75, -0.5) {};
		\node [style=none] (81) at (11, -0.5) {};
		\node [style=none] (82) at (10.5, 0.5) {$\vdots$};
		\node [style=none] (83) at (5.25, -0.5) {};
		\node [style=none] (84) at (5.25, 1.25) {};
	\end{pgfonlayer}
	\begin{pgfonlayer}{edgelayer}
		\draw [style=edge] (10.center) to (11.center);
		\draw [style=edge] (11.center) to (13.center);
		\draw [style=edge] (13.center) to (12.center);
		\draw [style=edge] (10.center) to (12.center);
		\draw [style=edge] (19.center) to (20.center);
		\draw [style=edge] (21.center) to (22.center);
		\draw [style=edge] (22.center) to (24.center);
		\draw [style=edge] (24.center) to (23.center);
		\draw [style=edge] (21.center) to (23.center);
		\draw [style=edge] (26.center) to (27.center);
		\draw [style=edge] (28.center) to (29.center);
		\draw [style=edge] (35.center) to (36.center);
		\draw [style=edge] (37.center) to (38.center);
		\draw [style=edge] (39.center) to (40.center);
		\draw [style=edge] (40.center) to (42.center);
		\draw [style=edge] (42.center) to (41.center);
		\draw [style=edge] (39.center) to (41.center);
		\draw [style=edge] (50.center) to (51.center);
		\draw [style=edge] (52.center) to (53.center);
		\draw [style=edge] (59.center) to (60.center);
		\draw [style=edge] (61.center) to (62.center);
		\draw [style=edge] (66.center) to (67.center);
		\draw [style=edge] (68.center) to (69.center);
		\draw [style=edge] (73.center) to (74.center);
		\draw [style=edge] (74.center) to (76.center);
		\draw [style=edge] (76.center) to (75.center);
		\draw [style=edge] (73.center) to (75.center);
		\draw [style=edge] (78.center) to (79.center);
		\draw [style=edge] (80.center) to (81.center);
	\end{pgfonlayer}
\end{tikzpicture}}
\]
\end{proof}

\subsection{Connectivity-Aware CNOT Circuit Synthesis}
\label{subsec:routeCNOT}

Consider a NISQ hardware with at least $n$ physical qubits. \textbf{Connectivity-aware CNOT synthesis} accounts for the hardware topology and synthesizes an $n$-qubit CNOT circuit based on an $n \times n$ parity matrix. In the meantime, it reduces the synthesized gate count. In \cref{subsubsec:steinertree}, we introduce the Steiner tree, which is used to encode the connectivity constraint and efficiently synthesize a CNOT circuit. In \cref{subsubsec:routewithSteiner}, we reduce the connectivity-aware CNOT circuit synthesis to a Steiner tree problem. As a concrete example, we explain how algorithm PermRowCol~\cite{griendli2022} accounts for the connectivity constraint and synthesizes a CNOT circuit up to permutation. 
 
\subsubsection{Steiner Tree}
\label{subsubsec:steinertree}
\begin{definition}
    A \textbf{graph} $G$ is defined by an ordered pair $(V_G,E_G)$. $V_G$ is a set of \textbf{vertices} and $E_G$ is a set of \textbf{edges}. Each edge is defined as $e=(u,v)$, where $u,v\in V_G$. The \textbf{degree} of a vertex is the number of edges that are incident to this vertex.
\end{definition}

\begin{definition}
Let $G = (V_G, E_G)$ be a graph.
\begin{itemize}
    \item G is \textbf{edge-weighted} if it has a weight function $\omega_G: E_G \rightarrow \mathbb{R}$. 
    \item G is \textbf{undirected} if none of its edges have orientations. For all $u,v \in V_G$, $(u,v) = (v,u)$.
    \item G is \textbf{connected} if every vertex in the graph is reachable from any other vertex by traversing a sequence of edges (i.e., a \textbf{path}). A \textbf{disconnected graph} is a graph that is not connected.
    \item $v \in V_G$ is a \textbf{cut vertex} if its removal disconnects $G$. It is a \textbf{non-cut vertex} if otherwise. 
    \item $G$ is \textbf{acyclic} if it has no cycle.
    \item $G$ has a self-loop if $(u, u) \in E_G$ for some $u \in V_G$.
    \item $G$ is a \textbf{tree} if it is undirected, connected, and acyclic.
\end{itemize}
 
\end{definition}

In this paper, we consider only the undirected edge-weighted connected graph and use $G = (V_G,E_G,\omega_G)$ to denote such graphs.
\begin{definition}
    A graph is \textbf{simple} if it is undirected with all edge weights equal to $1$; it has at most one edge between two distinct vertices with no self-loops. 
\end{definition}

\begin{definition}
    A simple graph is \textbf{complete} if every pair of distinct vertices is connected by a unique edge. Otherwise, the graph is \textbf{incomplete}.
\end{definition}

\begin{definition}
Let $G = (V_G,E_G,\omega_G)$. 
\begin{itemize}
    \item  $H=(V_H,E_H,\omega_{H})$ is a \textbf{subgraph} of $G$, denoted as $H \subseteq G$, if $V_H\subseteq V_G$, $E_H\subseteq E_G$, and $\omega_H(e) = \omega_G(e)$ for all $e \in E_H$. 
    \item $T = (V_T, E_T, \omega_T)$ is a \textbf{minimum spanning tree} of $G$ if $T \subseteq G$,  $V_T = V_G$, with $\sum_{e\in E_T} \omega_T(e)$ minimal.
\end{itemize}
\end{definition}

\begin{definition}
Let $G = (V_G,E_G,\omega_{G})$, $S \subseteq V_G$. A \textbf{Steiner tree} $T=(V_T,E_T,\omega_T)$ is a subgraph of $G$ such that $S \subseteq V_T$ with $\sum_{e\in E_T} \omega_T(e)$ minimal. $S$ is called the \textbf{terminal}, the vertices in $S$ are called the \textbf{terminal nodes}, and those in $V_T\setminus S$ are called the \textbf{Steiner nodes}. A solution to the \textbf{Steiner tree problem} \Steiner(G, S) is a Steiner tree $T$ of $G$ with $S$ as its leaves.
 \label{defn:steinerTree}
\end{definition}

A Steiner tree is a variation of a minimum-spanning tree. A Steiner tree problem involves finding a minimum-weight tree that spans a given set of vertices (i.e., the terminal nodes). Solutions to \Steiner(G, S) are not unique. Consider an example in \cref{fig:steiner}, where $G$ is a $12$-vertex grid, $S = \{2,3,7,11\}$, and $\omega_G:\Z_{12} \rightarrow \{1\}$. In \cref{fig5:a}, terminal nodes are coloured in red. \cref{fig5:b,fig5:c} show two distinct solutions to \Steiner(G,S). The edges of the Steiner trees $T_0$ and $T_1$ are highlighted in green. The sets of Steiner nodes $V_{T_0}\setminus S$ and $V_{T_1}\setminus S$ can be read off from each graph.

\begin{figure}[H]
\begin{subfigure}{.3\textwidth}
 \centering
\begin{tikzpicture}[scale=0.9]
 \draw  (1,1) circle [radius=0.2];
  \node at (1,1) {\scriptsize$0$};
  \draw (1.2,1)--(1.8,1); 
  \draw (1.2,0)--(1.8,0); 
  \draw (2,1) circle [radius=0.2];
  \node at (2,1) {\scriptsize$1$};
  \draw (2.2,1)--(2.8,1); 
  \draw (2.2,0)--(2.8,0); 
  \draw [fill=IMSRed] (3,1) circle [radius=0.2];
  \node [IMSWhite] at (3,1) {\scriptsize$2$};
  \draw [fill=IMSRed] (1,0) circle [radius=0.2];
  \node [IMSWhite] at (1,0) {\scriptsize$3$};
  \draw (2,0) circle [radius=0.2];
  \node at (2,0) {\scriptsize$4$};
  \draw (3,0) circle [radius=0.2];
  \node at (3,0) {\scriptsize$5$};
  \draw (1,0.8)--(1,0.2);    
  \draw (3,0.8)--(3,0.2);    
  \draw (2,0.8)--(2,0.2); 
  \draw (1,-1) circle [radius=0.2];
  \node at (1,-1) {\scriptsize$6$};
  \draw (1,-.2)--(1,-.8);    
  \draw [fill=IMSRed] (2,-1) circle [radius=0.2];
  \node [IMSWhite] at (2,-1) {\scriptsize$7$};
  \draw (1.2,-1)--(1.8,-1);
  \draw (2,-.2)--(2,-.8);    
  \node at (3,-1) {\scriptsize$8$};
  \draw (3,-.2)--(3,-.8);    
  \draw (3,-1) circle [radius=0.2];
  \node at (1,-2) {\scriptsize$9$};
  \draw (1,-2) circle [radius=0.2];
  \draw (1,-1.2)--(1,-1.8);    
  \draw  (2,-2) circle [radius=0.2];
  \node at (2,-2) {\scriptsize$10$};
  \draw (2,-1.2)--(2,-1.8);    
  \draw (2.2,-1)--(2.8,-1);    
  \draw [fill=IMSRed] (3,-2) circle [radius=0.2];
  \node [IMSWhite] at (3,-2) {\scriptsize$11$};
  \draw (3,-1.2)--(3,-1.8);    
  \draw (1.2,-2)--(1.8,-2);    
  \draw (2.2,-2)--(2.8,-2);    
\end{tikzpicture}
\subcaption{$S=\{2,3,7,11\}$.}
\label{fig5:a}
\end{subfigure}
  \begin{subfigure}{.3\textwidth}
  \centering
  \begin{tikzpicture}[scale=0.9]
 \draw  (1,1) circle [radius=0.2];
  \node at (1,1) {\scriptsize$0$};
  \draw (1.2,1)--(1.8,1); 
  \draw [ultra thick,IMSGreen] (1.2,0)--(1.8,0); 
  \draw (2,1) circle [radius=0.2];
  \node at (2,1) {\scriptsize$1$};
  \draw [ultra thick,IMSGreen] (2.2,1)--(2.8,1); 
  \draw (2.2,0)--(2.8,0); 
  \draw [fill=IMSRed] (3,1) circle [radius=0.2];
  \node [IMSWhite] at (3,1) {\scriptsize$2$};
  \draw [fill=IMSRed] (1,0) circle [radius=0.2];
  \node [IMSWhite] at (1,0) {\scriptsize$3$};
  \draw (2,0) circle [radius=0.2];
  \node at (2,0) {\scriptsize$4$};
  \draw (3,0) circle [radius=0.2];
  \node at (3,0) {\scriptsize$5$};
  \draw (1,0.8)--(1,0.2);    
  \draw (3,0.8)--(3,0.2);    
  \draw [ultra thick,IMSGreen] (2,0.8)--(2,0.2); 
  \draw (1,-1) circle [radius=0.2];
  \node at (1,-1) {\scriptsize$6$};
  \draw (1,-.2)--(1,-.8);    
  \draw [fill=IMSRed] (2,-1) circle [radius=0.2];
  \node [IMSWhite] at (2,-1) {\scriptsize$7$};
  \draw (1.2,-1)--(1.8,-1);
  \draw [ultra thick,IMSGreen] (2,-.2)--(2,-.8);    
  \node at (3,-1) {\scriptsize$8$};
  \draw (3,-.2)--(3,-.8);    
  \draw (3,-1) circle [radius=0.2];
  \node at (1,-2) {\scriptsize$9$};
  \draw (1,-2) circle [radius=0.2];
  \draw (1,-1.2)--(1,-1.8);    
  \draw  (2,-2) circle [radius=0.2];
  \node at (2,-2) {\scriptsize$10$};
  \draw [ultra thick,IMSGreen] (2,-1.2)--(2,-1.8);    
  \draw (2.2,-1)--(2.8,-1);    
  \draw [fill=IMSRed] (3,-2) circle [radius=0.2];
  \node [IMSWhite] at (3,-2) {\scriptsize$11$};
  \draw (3,-1.2)--(3,-1.8);    
  \draw (1.2,-2)--(1.8,-2);    
  \draw [ultra thick,IMSGreen] (2.2,-2)--(2.8,-2);    
\end{tikzpicture}
\subcaption{$V_{T_0}\setminus S=\{1,4,10\}$.}
\label{fig5:b}
\end{subfigure}
  \begin{subfigure}{.3\textwidth}
  \centering
\begin{tikzpicture}[scale=0.9]
 \draw  (1,1) circle [radius=0.2];
  \node at (1,1) {\scriptsize$0$};
  \draw (1.2,1)--(1.8,1); 
  \draw [ultra thick,IMSGreen] (1.2,0)--(1.8,0); 
  \draw (2,1) circle [radius=0.2];
  \node at (2,1) {\scriptsize$1$};
  \draw (2.2,1)--(2.8,1); 
  \draw [ultra thick,IMSGreen] (2.2,0)--(2.8,0); 
  \draw [fill=IMSRed] (3,1) circle [radius=0.2];
  \node [IMSWhite] at (3,1) {\scriptsize$2$};
  \draw [fill=IMSRed] (1,0) circle [radius=0.2];
  \node [IMSWhite] at (1,0) {\scriptsize$3$};
  \draw (2,0) circle [radius=0.2];
  \node at (2,0) {\scriptsize$4$};
  \draw (3,0) circle [radius=0.2];
  \node at (3,0) {\scriptsize$5$};
  \draw (1,0.8)--(1,0.2);    
  \draw [ultra thick,IMSGreen] (3,0.8)--(3,0.2);    
  \draw (2,0.8)--(2,0.2); 
  \draw (1,-1) circle [radius=0.2];
  \node at (1,-1) {\scriptsize$6$};
  \draw (1,-.2)--(1,-.8);    
  \draw [fill=IMSRed] (2,-1) circle [radius=0.2];
  \node [IMSWhite] at (2,-1) {\scriptsize$7$};
  \draw (1.2,-1)--(1.8,-1);
  \draw (2,-.2)--(2,-.8);    
  \node at (3,-1) {\scriptsize$8$};
  \draw [ultra thick,IMSGreen] (3,-.2)--(3,-.8);    
  \draw (3,-1) circle [radius=0.2];
  \node at (1,-2) {\scriptsize$9$};
  \draw (1,-2) circle [radius=0.2];
  \draw (1,-1.2)--(1,-1.8);    
  \draw  (2,-2) circle [radius=0.2];
  \node at (2,-2) {\scriptsize$10$};
  \draw (2,-1.2)--(2,-1.8);    
  \draw [ultra thick,IMSGreen] (2.2,-1)--(2.8,-1);    
  \draw [fill=IMSRed] (3,-2) circle [radius=0.2];
  \node [IMSWhite] at (3,-2) {\scriptsize$11$};
  \draw [ultra thick,IMSGreen] (3,-1.2)--(3,-1.8);    
  \draw (1.2,-2)--(1.8,-2);    
  \draw (2.2,-2)--(2.8,-2);    
\end{tikzpicture}
\subcaption{$V_{T_1}\setminus S=\{4,5,8\}$.}
\label{fig5:c}
\end{subfigure}
  \caption{Solutions to a Steiner tree problem \Steiner(G,S) are not unique. (a) $G$ is a 12-vertex grid with $S = \{2,3,7,11\}$. The terminal nodes are coloured in red. (b) $T_0$ is a solution to \Steiner(G,S). The edges of the Steiner tree $T_0$ are coloured in green. Its Steiner nodes are the intermediary nodes in the Steiner tree, $V_{T_0}\setminus S=\{1,4,10\}$. (c) $T_1$ is an alternative solution to \Steiner(G, S). Its set of Steiner nodes is $V_{T_1}\setminus S=\{4,5,8\}$.}
\label{fig:steiner} 
\end{figure}
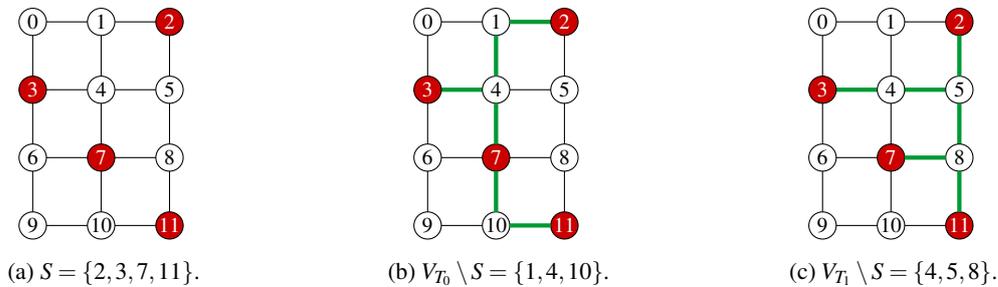

Computing Steiner trees is NP-hard and the related decision problem is NP-complete~\cite{1972K}. There are different heuristic algorithms to compute approximate Steiner trees~\cite{2013BGRS, 1992HR, 2005RZ}. The choice of heuristics depends on their use case, as well as the trade-off between the problem's size and the algorithm's runtime.

\subsubsection{Reduction to Steiner Tree Problem: PermRowCol}
\label{subsubsec:routewithSteiner}

In superconducting quantum computers, qubits are arranged in a 2D grid. Each qubit can only interact with its nearest neighbours~\cite{google2023suppressing,arute2019quantum,gambetta2017building}. This is called the \textbf{NN interaction} and is modelled by a \textbf{connectivity graph}, $G=(V_G, E_G, \omega_G)$, $\omega_G:E_G \rightarrow \R$. Each vertex corresponds to a physical qubit, and each edge represents an entangling gate that can be performed on the qubits corresponding to its endpoints. Connectivity-aware CNOT circuit synthesis maps a logical CNOT circuit to NISQ hardware by accounting for its underlying topology, but assuming no noise in the system. That is, $G$ is simple, $\omega_G:E_G \rightarrow \{1\}$. 

Let \vect{C} be a logical CNOT circuit and \vect{A} be its parity matrix. The synthesized circuit $\vect{C}_{\text{syn}}$ contains only CNOT gates acting on adjacent physical qubits in the NISQ hardware. Moreover, its CNOT count should be as few as possible. In \cite{gheorghiu2022reducing,griendli2022,kissinger2020cnot,nash2020quantum,wu2023optimization}, this problem is reduced to a Steiner tree problem. Here, we use the algorithm PermRowCol~\cite{griendli2022} to demonstrate the problem reduction. In \cref{sec:algorithm}, it is generalized to NAPermRowCol to account for noise in the system. In both cases, we assume an arbitrary initial qubit mapping, and it is illustrated below.



\paragraph{An arbitrary initial qubit mapping}

Let $n$ be the number of logical qubits in \vect{C} and $m=\lvert V_G\rvert$, $n \leq m$. Consider an initial qubit map where logical qubit $i$ is mapped to quantum register $j$. Formally, this is defined by an injective function, $\Phi:\; \mathbb{Z}_n \rightarrow \mathbb{Z}_m$.  $\Phi(\mathbb{Z}_n)$ corresponds to a connected subregion in the NISQ hardware. In what follows, we use vertices, physical qubits, and quantum registers interchangeably.

To illustrate the basics of PermRowCol, we use the $4$-qubit CNOT circuit in \cref{fig:4q circuit 2} with a linear topology in \cref{fig:linear topology} as an example. According to \cref{tab: initial qubit map}, the logical qubit $i$ is mapped to the quantum register $j$ by $\Phi$, $\Phi: \Z_4 \rightarrow \Z_4$. This means we need to measure quantum register $\Phi(i)$ to return the state of logical qubit $i$. For clarity, we use a green label to denote a logical qubit and a blue label to denote a physical qubit.

\begin{figure}[H]
\begin{subfigure}{.5\textwidth}
\centering
\begin{tabular}{|l|c|c|c|c|}
\hline
Logical qubit $\color{IMSGreen}i$  & $\color{IMSGreen}0$ & $\color{IMSGreen}1$ & $\color{IMSGreen}2$ & $\color{IMSGreen}3$ \\ \hline
Quantum register $\color{IMSBlue}j$ & \color{IMSBlue}3 & \color{IMSBlue}0 & \color{IMSBlue}1 & \color{IMSBlue}2\\ \hline
\end{tabular}
\caption{Due to the initial qubit map, each logical qubit (in green) is mapped to a vertex (in blue).}
\label{tab: initial qubit map}
\end{subfigure}
\qquad
\begin{subfigure}{.5\textwidth}
     \centering
     \includegraphics[scale =0.35]{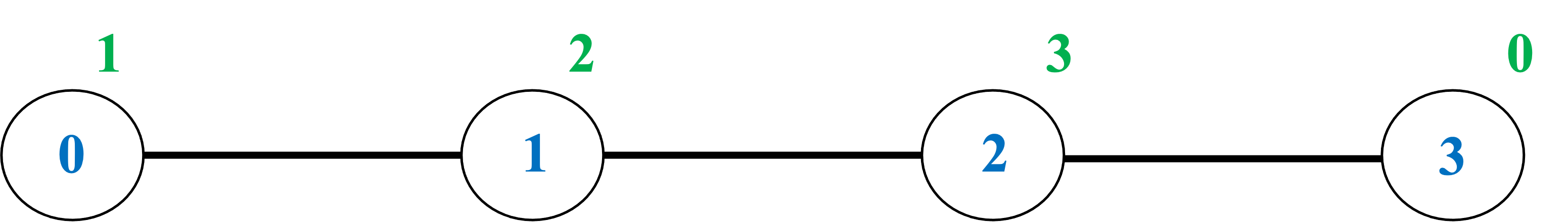}
    \caption{G is the linear graph. Each edge has a unit weight, which is omitted here. $\{\color{IMSBlue}0,3\color{black}\}$ is the set of non-cut vertices.}
    \label{fig:linear topology}
\end{subfigure}
\caption{Visualize the initial qubit map on a hardware topology.}
\end{figure}

\begin{remark}
    The x-indexed row (a.k.a., row x) is a row whose index is x. It is not necessarily the x-th row in the matrix. For example, in \cref{fig:parity matrix}, the $3$-index row of $\vect{A}_0$ is the first row of $\vect{A}_0$. This convention also applies to the indexing of columns. On the physical layer, CNOT(\color{IMSBlue}i\color{black}, \color{IMSBlue}j\color{black}) denotes a CNOT gate acting on physical qubits labelled by \color{IMSBlue}i \color{black} and \color{IMSBlue}j\color{black}. Correspondingly, R(\color{IMSBlue}j\color{black}, \color{IMSBlue}i\color{black}) on $\vect{A}_0$ denotes a row operation of adding the \color{IMSBlue}j\color{black}-indexed row to the \color{IMSBlue}i\color{black}-indexed row, while keeping the \color{IMSBlue}j\color{black}-indexed row unchanged.
\end{remark}

\cref{fig:4q circuit 3} demonstrates the consequence of mapping each input qubit of \vect{C} to a quantum register in \cref{fig:linear topology}. For brevity, we drop the ket notations in every circuit diagram and use $'$ to distinguish output wires from input wires.

\begin{definition}
    Given a hardware topology $G$, CNOT(\color{IMSBlue}i\color{black}, \color{IMSBlue}j\color{black}) (or R(\color{IMSBlue}j\color{black}, \color{IMSBlue}i\color{black}) on $\vect{A}_0$) is \textbf{allowed} if quantum registers \color{IMSBlue}i \color{black} and \color{IMSBlue}j \color{black} are connected. It is \textbf{not allowed} if otherwise.
\end{definition}

In \cref{fig:4q circuit 3}, some of the CNOT operations (annotated by $\color{IMSRed}\times$) are not allowed because they do not act on adjacent physical qubits. For example, CNOT(\color{IMSBlue}3\color{black}, \color{IMSBlue}0\color{black}) is not allowed because vertices \color{IMSBlue}3 \color{black} and \color{IMSBlue}0 \color{black} are not adjacent in \cref{fig:linear topology}. In light of this, we synthesize this circuit using PermRowCol, after which $\vect{A}_0$ is reduced to the permutation matrix \vect{P},

\[
\vect{P} = \begin{blockarray}{ccccc}
     & \color{IMSBlue}3'\color{black} & \color{IMSBlue}0'\color{black} & \color{IMSBlue}1'\color{black} & \color{IMSBlue}2'\color{black}\\
    \begin{block}{c(cccc)}
     \color{IMSBlue}3\color{black} & 0 & 1 & 0 & 0 \\
     \color{IMSBlue}0\color{black}& 0 & 0 & 1 & 0 \\
     \color{IMSBlue}1\color{black}& 0 & 0 & 0 & 1 \\
     \color{IMSBlue}2\color{black}& 1 & 0 & 0 & 0 \\
    \end{block}
    \end{blockarray}.
\]

\begin{figure}[H]
\begin{subfigure}{1\textwidth}
  \centering
  \scalebox{1}{\begin{tikzpicture}
	\begin{pgfonlayer}{nodelayer}
		\node [style=none] (0) at (-6.75, 3) {};
		\node [style=none] (1) at (10.25, 3) {};
		\node [style=none] (2) at (-6.75, 1.5) {};
		\node [style=none] (3) at (10.25, 1.5) {};
		\node [style=none] (4) at (-6.75, 0) {};
		\node [style=none] (5) at (10.25, 0) {};
		\node [style=none] (6) at (-6.75, -1.5) {};
		\node [style=none] (7) at (10.25, -1.5) {};
		\node [style=none] (8) at (-4.5, 1.5) {$\bigoplus$};
		\node [style=none] (9) at (-4.5, 3) {$\bullet$};
		\node [style=none] (10) at (-2.75, 1.5) {$\bigoplus$};
		\node [style=none] (11) at (-2.75, 0) {$\bullet$};
		\node [style=none] (12) at (-0.75, -1.5) {$\bigoplus$};
		\node [style=none] (13) at (-0.75, 3) {$\bullet$};
		\node [style=none] (14) at (1, -1.5) {$\bigoplus$};
		\node [style=none] (15) at (1, 1.5) {$\bullet$};
		\node [style=none] (16) at (3, 3) {$\bigoplus$};
		\node [style=none] (17) at (3, 0) {$\bullet$};
		\node [style=none] (18) at (5, 0) {$\bigoplus$};
		\node [style=none] (19) at (5, 1.5) {$\bullet$};
		\node [style=none] (20) at (7, 1.5) {$\bigoplus$};
		\node [style=none] (21) at (7, -1.5) {$\bullet$};
		\node [style=none] (22) at (-7.5, 3) {$\color{IMSGreen}0$};
		\node [style=none] (23) at (-7.5, 1.5) {$\color{IMSGreen}1$};
		\node [style=none] (24) at (-7.5, 0) {$\color{IMSGreen}2$};
		\node [style=none] (25) at (-7.5, -1.5) {$\color{IMSGreen}3$};
		\node [style=none] (26) at (-5.25, 4) {};
		\node [style=none] (27) at (-5.25, -2.25) {};
		\node [style=none] (28) at (8.75, 4) {};
		\node [style=none] (29) at (8.75, -2.25) {};
		\node [style=none] (30) at (1.25, 4.75) {The logical circuit $\mathbf{C}$};
		\node [style=none] (31) at (12.25, 3) {$\color{IMSGreen}0'\color{black} = \color{IMSGreen}0\color{black} \oplus \color{IMSGreen}2$};
		\node [style=none] (32) at (12.25, 1.5) {$\color{IMSGreen}1'\color{black} = \color{IMSGreen}0 \color{black}\oplus \color{IMSGreen}3$};
		\node [style=none] (33) at (12.25, 0) {$\color{IMSGreen}2'\color{black} = \color{IMSGreen}0\color{black} \oplus \color{IMSGreen}1$};
		\node [style=none] (34) at (12.75, -1.5) {$\color{IMSGreen}3'\color{black} = \color{IMSGreen}1\color{black} \oplus \color{IMSGreen}2\color{black} \oplus \color{IMSGreen}3$};
	\end{pgfonlayer}
	\begin{pgfonlayer}{edgelayer}
		\draw [style=edge] (0.center) to (1.center);
		\draw [style=edge] (2.center) to (3.center);
		\draw [style=edge] (4.center) to (5.center);
		\draw [style=edge] (6.center) to (7.center);
		\draw [style=edge] (9.center) to (8.center);
		\draw [style=edge] (11.center) to (10.center);
		\draw [style=edge] (13.center) to (12.center);
		\draw [style=edge] (15.center) to (14.center);
		\draw [style=edge] (17.center) to (16.center);
		\draw [style=edge] (19.center) to (18.center);
		\draw [style=edge] (21.center) to (20.center);
		\draw [style=edge] (26.center) to (28.center);
		\draw [style=edge] (26.center) to (27.center);
		\draw [style=edge] (27.center) to (29.center);
		\draw [style=edge] (29.center) to (28.center);
	\end{pgfonlayer}
\end{tikzpicture}}
  \caption{The logical circuit \vect{C} is composed of a sequence of CNOT gates. The label of each wire is coloured green to indicate that it is a logical qubit. We use $'$ to distinguish output wires from input wires.}
  \label{fig:4q circuit 2}
\end{subfigure}%

\vspace{1 em}

\begin{subfigure}{1\textwidth}
  \centering
  \scalebox{1}{\begin{tikzpicture}
	\begin{pgfonlayer}{nodelayer}
		\node [style=none] (0) at (-6.75, 3) {};
		\node [style=none] (1) at (10.25, 3) {};
		\node [style=none] (2) at (-6.75, 1.5) {};
		\node [style=none] (3) at (10.25, 1.5) {};
		\node [style=none] (4) at (-6.75, 0) {};
		\node [style=none] (5) at (10.25, 0) {};
		\node [style=none] (6) at (-6.75, -1.5) {};
		\node [style=none] (7) at (10.25, -1.5) {};
		\node [style=none] (8) at (-4.5, 1.5) {$\bigoplus$};
		\node [style=none] (9) at (-4.5, 3) {$\bullet$};
		\node [style=none] (10) at (-2.75, 1.5) {$\bigoplus$};
		\node [style=none] (11) at (-2.75, 0) {$\bullet$};
		\node [style=none] (12) at (-0.75, -1.5) {$\bigoplus$};
		\node [style=none] (13) at (-0.75, 3) {$\bullet$};
		\node [style=none] (14) at (1, -1.5) {$\bigoplus$};
		\node [style=none] (15) at (1, 1.5) {$\bullet$};
		\node [style=none] (16) at (3, 3) {$\bigoplus$};
		\node [style=none] (17) at (3, 0) {$\bullet$};
		\node [style=none] (18) at (5, 0) {$\bigoplus$};
		\node [style=none] (19) at (5, 1.5) {$\bullet$};
		\node [style=none] (20) at (7, 1.5) {$\bigoplus$};
		\node [style=none] (21) at (7, -1.5) {$\bullet$};
		\node [style=none] (22) at (-8.5, 3) {$\color{IMSBlue}3\color{black}=\Phi(\color{IMSGreen}0\color{black})$};
		\node [style=none] (23) at (-8.5, 1.5) {$\color{IMSBlue}0\color{black}=\Phi(\color{IMSGreen}1\color{black})$};
		\node [style=none] (24) at (-8.5, 0) {$\color{IMSBlue}1\color{black}=\Phi(\color{IMSGreen}2\color{black})$};
		\node [style=none] (25) at (-8.5, -1.5) {$\color{IMSBlue}2\color{black}=\Phi(\color{IMSGreen}3\color{black})$};
		\node [style=none] (26) at (-5.25, 5) {};
		\node [style=none] (27) at (-5.25, -2.25) {};
		\node [style=none] (28) at (8.75, 5) {};
		\node [style=none] (29) at (8.75, -2.25) {};
		\node [style=none] (30) at (1.25, 5.75) {$\mathbf{C}$ after the initial qubit mapping};
		\node [style=none] (31) at (14.75, 3) {$\color{IMSBlue}3'\color{black} = \Phi(\color{IMSGreen}0\color{black}) \oplus \Phi(\color{IMSGreen}2\color{black}) = \color{IMSBlue}3\color{black} \oplus \color{IMSBlue}1\color{black}$};
            \node [style=none] (32) at (14.75, 1.5) {$\color{IMSBlue}0'\color{black} = \Phi(\color{IMSGreen}0\color{black}) \oplus \Phi(\color{IMSGreen}3\color{black}) = \color{IMSBlue}3\color{black} \oplus \color{IMSBlue}2\color{black}$};
             \node [style=none] (33) at (14.75, 0) {$\color{IMSBlue}1'\color{black} = \Phi(\color{IMSGreen}0\color{black}) \oplus \Phi(\color{IMSGreen}1\color{black}) = \color{IMSBlue}3\color{black} \oplus \color{IMSBlue}0\color{black}$};
		\node [style=none] (34) at (16.5, -1.5) {$\color{IMSBlue}2'\color{black} = \Phi(\color{IMSGreen}1\color{black}) \oplus \Phi(\color{IMSGreen}2\color{black}) \oplus \Phi(\color{IMSGreen}3\color{black}) = \color{IMSBlue}0\color{black} \oplus \color{IMSBlue}1\color{black} \oplus \color{IMSBlue}2\color{black}$};
		\node [style=none] (35) at (-4.5, 4.25) {$\color{IMSRed}\times$};
		\node [style=none] (36) at (1, 4.25) {$\color{IMSRed}\color{IMSRed}\times$};
		\node [style=none] (37) at (3, 4.25) {$\color{IMSRed}\times$};
		\node [style=none] (38) at (7, 4.25) {$\color{IMSRed}\times$};
	\end{pgfonlayer}
	\begin{pgfonlayer}{edgelayer}
		\draw [style=edge] (0.center) to (1.center);
		\draw [style=edge] (2.center) to (3.center);
		\draw [style=edge] (4.center) to (5.center);
		\draw [style=edge] (6.center) to (7.center);
		\draw [style=edge] (9.center) to (8.center);
		\draw [style=edge] (11.center) to (10.center);
		\draw [style=edge] (13.center) to (12.center);
		\draw [style=edge] (15.center) to (14.center);
		\draw [style=edge] (17.center) to (16.center);
		\draw [style=edge] (19.center) to (18.center);
		\draw [style=edge] (21.center) to (20.center);
		\draw [style=edge] (26.center) to (28.center);
		\draw [style=edge] (26.center) to (27.center);
		\draw [style=edge] (27.center) to (29.center);
		\draw [style=edge] (29.center) to (28.center);
	\end{pgfonlayer}
\end{tikzpicture}}
  \caption{After applying the initial qubit map $\Phi$, the inputs of \vect{C} are mapped to each quantum register in \cref{fig:linear topology}. On the physical layer, the labels of each wire are coloured blue. We use $'$ to distinguish output wires from input wires. Based on the linear topology, some of the CNOT operations are not allowed (annotated by $\color{IMSRed}\times$) because they do not act on adjacent qubits.}
  \label{fig:4q circuit 3}
\end{subfigure}%

\vspace{.5 em}

\begin{subfigure}{1\textwidth}
    \[
  \vect{A}=\begin{blockarray}{ccccc}
     & \color{IMSGreen}0'\color{black} & \color{IMSGreen}1'\color{black} & \color{IMSGreen}2'\color{black} & \color{IMSGreen}3'\color{black}\\
    \begin{block}{c(cccc)}
      \color{IMSGreen}0\color{black} & 1 & 1 & 1 & 0 \\
      \color{IMSGreen}1\color{black} & 0 & 0 & 1 & 1 \\
      \color{IMSGreen}2\color{black} & 1 & 0 & 0 & 1 \\
      \color{IMSGreen}3\color{black} & 0 & 1 & 0 & 1 \\
    \end{block}
    \end{blockarray}\qquad \xlongrightarrow{\Phi}\qquad \vect{A}_0=\begin{blockarray}{ccccc}
     & \color{IMSBlue}3'\color{black} & \color{IMSBlue}0'\color{black} & \color{IMSBlue}1'\color{black} & \color{IMSBlue}2'\color{black}\\
    \begin{block}{c(cccc)}
      \color{IMSBlue}3\color{black} & 1 & 1 & 1 & 0 \\
      \color{IMSBlue}0\color{black} & 0 & 0 & 1 & 1 \\
      \color{IMSBlue}1\color{black} & 1 & 0 & 0 & 1 \\
      \color{IMSBlue}2\color{black} & 0 & 1 & 0 & 1 \\
    \end{block}
    \end{blockarray}
  \]
  
  \vspace{-1.5 em}
  
  \caption{$\vect{A}$ and $\vect{A}_0$ are the parity matrices of \vect{C} before and after the initial qubit mapping.}
  \label{fig:parity matrix}
\end{subfigure}%
\caption{$\Phi$ maps the inputs of logical CNOT circuit \vect{C} to quantum registers in NISQ hardware. Accordingly, \vect{C}'s parity matrix \vect{A} is updated to $\vect{A}_0$. Connectivity-aware CNOT circuit synthesis is carried out on $\vect{A}_0$ with the hardware topology $G$ in \cref{fig:linear topology}.}
\label{fig: map before after}
\end{figure} 

\paragraph{PermRowCol reduces a parity matrix to a permutation matrix}
In \cref{subsec:synthesis}, Gaussian elimination reduces a parity matrix to an identity matrix. PermRowCol, however, reduces it to a permutation matrix. This means at each reduction step, we have more freedom to choose a pivot row and column. Here, we use the example in \cref{fig: map before after} to provide detailed explanations of PermRowCol. \cref{fig:routed circuit} demonstrates the synthesized circuit $\vect{C}_{\text{syn}}$ up to permutation of quantum registers. 

In $\vect{C}_{\text{syn}}$, every CNOT operation is allowed. To recover the initial qubit mapping, $\vect{C}_{\text{syn}}$ is concatenated with a circuit over SWAP gates, $\vect{C}_{\text{SWAP}}$, whose parity matrix is \vect{P}. In the end, the fully synthesized circuit for \vect{C} is expressed as $\vect{C}_{\text{SWAP}}\circ \vect{C}_{\text{syn}}$. We can read off an evolved state from a quantum register that is annotated on the right. For example, the evolved state $\color{IMSGreen}0\color{black} \oplus \color{IMSGreen}2\color{black}$ of the input wire $\color{IMSGreen}0\color{black}$ can be obtained by measuring quantum register $\color{IMSBlue}3\color{black}$ (denotes as $\color{IMSBlue}3''\color{black}$). Without $\vect{C}_{\text{SWAP}}$, we can obtain this parity term by measuring quantum register $\color{IMSBlue}2\color{black}$ (denotes as $\color{IMSBlue}2'\color{black}$). This is equivalent to relabeling quantum registers according to \vect{P}.

In a nutshell, PermRowCol reduces the CNOT count of a synthesized circuit by factoring out SWAP gates after synthesizing a parity matrix. It offloads the task of quantum computing onto classical processing and reduces the computation resources in NISQ compilation. As a result, $\vect{C}_{\text{syn}}$ contains fewer CNOT gates than the fully synthesized circuit $\vect{C}_{\text{SWAP}}\circ \vect{C}_{\text{syn}}$. 

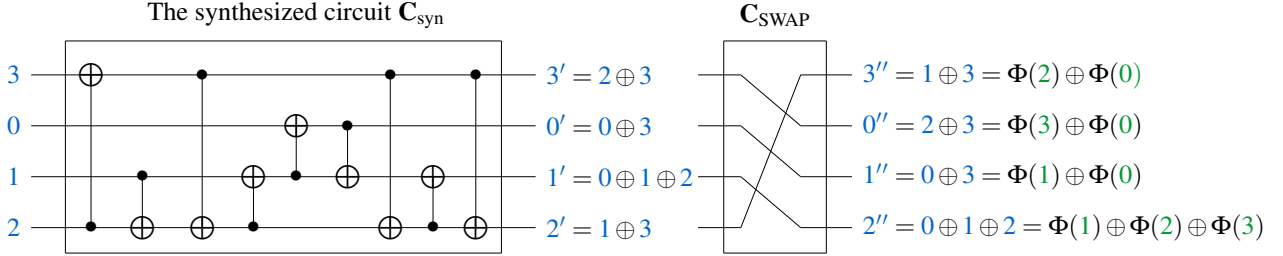
\begin{figure}[H]
  \centering
  \scalebox{0.9}{\begin{tikzpicture}
	\begin{pgfonlayer}{nodelayer}
		\node [style=none, label={left:$\color{IMSBlue}3$}] (0) at (2, 3) {};
		\node [style=none, label={right:$\color{IMSBlue}3'\color{black} = \color{IMSBlue}2\color{black} \oplus \color{IMSBlue}3$}] (1) at (16.75, 3) {};
		\node [style=none, label={left:$\color{IMSBlue}0$}] (2) at (2, 1.5) {};
		\node [style=none, label={right:$\color{IMSBlue}0'\color{black} = \color{IMSBlue}0\color{black} \oplus \color{IMSBlue}3$}] (3) at (16.75, 1.5) {};
		\node [style=none, label={left:$\color{IMSBlue}1$}] (4) at (2, 0) {};
		\node [style=none, label={right:$\color{IMSBlue}1'\color{black} = \color{IMSBlue}0\color{black} \oplus \color{IMSBlue}1\color{black} \oplus \color{IMSBlue}2$}] (5) at (16.75, 0) {};
		\node [style=none, label={left:$\color{IMSBlue}2$}] (6) at (2, -1.5) {};
		\node [style=none, label={right:$\color{IMSBlue}2'\color{black} = \color{IMSBlue}1\color{black} \oplus \color{IMSBlue}3$}] (7) at (16.75, -1.5) {};
		\node [style=none] (8) at (3.75, 3) {$\bigoplus$};
		\node [style=none] (9) at (3.75, -1.5) {$\bullet$};
		\node [style=none] (10) at (5.25, -1.5) {$\bigoplus$};
		\node [style=none] (11) at (5.25, 0) {$\bullet$};
		\node [style=none] (12) at (7, -1.5) {$\bigoplus$};
		\node [style=none] (13) at (7, 3) {$\bullet$};
		\node [style=none] (14) at (8.5, 0) {$\bigoplus$};
		\node [style=none] (15) at (8.5, -1.5) {$\bullet$};
		\node [style=none] (16) at (12.5, -1.5) {$\bigoplus$};
		\node [style=none] (17) at (12.5, 3) {$\bullet$};
		\node [style=none] (18) at (13.75, 0) {$\bigoplus$};
		\node [style=none] (19) at (13.75, -1.5) {$\bullet$};
		\node [style=none] (20) at (15, -1.5) {$\bigoplus$};
		\node [style=none] (21) at (15, 3) {$\bullet$};
		\node [style=none] (26) at (3, 4) {};
		\node [style=none] (27) at (3, -2.25) {};
		\node [style=none] (28) at (15.75, 4) {};
		\node [style=none] (29) at (15.75, -2.25) {};
		\node [style=none] (30) at (9.75, 4.75) {The synthesized circuit $\vect{C}_{\text{syn}}$};
		\node [style=none] (35) at (9.75, 1.5) {$\bigoplus$};
		\node [style=none] (36) at (9.75, 0) {$\bullet$};
		\node [style=none] (37) at (11.25, 0) {$\bigoplus$};
		\node [style=none] (38) at (11.25, 1.5) {$\bullet$};
		\node [style=none] (39) at (21.5, 3) {};
		\node [style=none] (40) at (22.75, 3) {};
		\node [style=none] (41) at (21.5, 1.5) {};
		\node [style=none] (42) at (22.75, 1.5) {};
		\node [style=none] (43) at (21.5, 0) {};
		\node [style=none] (44) at (22.75, 0) {};
		\node [style=none] (45) at (21.5, -1.5) {};
		\node [style=none] (46) at (22.75, -1.5) {};
		\node [style=none] (47) at (24.5, 3) {};
		\node [style=none, label={right:$\color{IMSBlue}3{'}{'}\color{black}=\color{IMSBlue}1\color{black} \oplus \color{IMSBlue}3\color{black} = \Phi(\color{IMSGreen}2\color{black}) \oplus \Phi(\color{IMSGreen}0)$}] (48) at (26, 3) {};
		\node [style=none] (49) at (24.5, 1.5) {};
		\node [style=none, label={right:$\color{IMSBlue}0{'}{'}\color{black}=\color{IMSBlue}2\color{black} \oplus \color{IMSBlue}3\color{black} = \Phi(\color{IMSGreen}3\color{black}) \oplus \Phi(\color{IMSGreen}0\color{black})$}] (50) at (26, 1.5) {};
		\node [style=none] (51) at (24.5, 0) {};
		\node [style=none, label={right:$\color{IMSBlue}1{'}{'}\color{black}=\color{IMSBlue}0\color{black} \oplus \color{IMSBlue}3\color{black} = \Phi(\color{IMSGreen}1\color{black}) \oplus \Phi(\color{IMSGreen}0\color{black})$}] (52) at (26, 0) {};
		\node [style=none] (53) at (24.5, -1.5) {};
		\node [style=none, label={right:$\color{IMSBlue}2{'}{'}\color{black}=\color{IMSBlue}0\color{black} \oplus \color{IMSBlue}1\color{black} \oplus \color{IMSBlue}2\color{black} = \Phi(\color{IMSGreen}1\color{black}) \oplus \Phi(\color{IMSGreen}2\color{black}) \oplus \Phi(\color{IMSGreen}3\color{black})$}] (54) at (26, -1.5) {};
		\node [style=none] (55) at (25.25, 4) {};
		\node [style=none] (56) at (25.25, -2.25) {};
		\node [style=none] (57) at (22.25, 4) {};
		\node [style=none] (58) at (22.25, -2.25) {};
		\node [style=none] (59) at (23.75, 4.75) {$\mathbf{C}_{\text{SWAP}}$};
	\end{pgfonlayer}
	\begin{pgfonlayer}{edgelayer}
		\draw [style=edge] (0.center) to (1.center);
		\draw [style=edge] (2.center) to (3.center);
		\draw [style=edge] (4.center) to (5.center);
		\draw [style=edge] (6.center) to (7.center);
		\draw [style=edge] (9.center) to (8.center);
		\draw [style=edge] (11.center) to (10.center);
		\draw [style=edge] (13.center) to (12.center);
		\draw [style=edge] (15.center) to (14.center);
		\draw [style=edge] (17.center) to (16.center);
		\draw [style=edge] (19.center) to (18.center);
		\draw [style=edge] (21.center) to (20.center);
		\draw [style=edge] (26.center) to (28.center);
		\draw [style=edge] (26.center) to (27.center);
		\draw [style=edge] (27.center) to (29.center);
		\draw [style=edge] (29.center) to (28.center);
		\draw [style=edge] (36.center) to (35.center);
		\draw [style=edge] (38.center) to (37.center);
		\draw (39.center) to (40.center);
		\draw (41.center) to (42.center);
		\draw (43.center) to (44.center);
		\draw (45.center) to (46.center);
		\draw (47.center) to (48.center);
		\draw (49.center) to (50.center);
		\draw (51.center) to (52.center);
		\draw (53.center) to (54.center);
		\draw (40.center) to (49.center);
		\draw (42.center) to (51.center);
		\draw (44.center) to (53.center);
		\draw (47.center) to (46.center);
		\draw [style=edge] (56.center) to (55.center);
		\draw [style=edge] (58.center) to (57.center);
		\draw (57.center) to (55.center);
		\draw (58.center) to (56.center);
	\end{pgfonlayer}
\end{tikzpicture}}
  \caption{The fully synthesized circuit of \vect{C} is expressed as $\vect{C}_{\text{SWAP}}\circ \vect{C}_{\text{syn}}$. The label of each wire is coloured blue (green) to indicate that it is a quantum register (logical qubit). $'$ and $''$ are used to distinguish output wires from input wires. For example, after $\vect{C}_{\text{syn}}$, quantum register \color{IMSBlue}2 \color{black} carries a parity term, $\color{IMSBlue}2' \color{black}=\color{IMSBlue}1 \color{black}\oplus \color{IMSBlue}3 \color{black}$. After $\vect{C}_{\text{SWAP}}$, this parity term is mapped to quantum register \color{IMSBlue}3 \color{black}, $\color{IMSBlue}3'' \color{black}=\color{IMSBlue}1 \color{black}\oplus \color{IMSBlue}3 \color{black}$. Equivalently, $\vect{C}_{\text{SWAP}}$ relabels quantum register \color{IMSBlue}2 \color{black} to \color{IMSBlue}3\color{black}. This means after $\vect{C}_{\text{SWAP}}\circ \vect{C}_{\text{syn}}$, we can measure quantum register \color{IMSBlue}3 \color{black} to obtain $\color{IMSBlue}1\color{black} \oplus \color{IMSBlue}3\color{black} = \Phi(\color{IMSGreen}2\color{black}) \oplus \Phi(\color{IMSGreen}0\color{black})$. In other words, the synthesized circuit successfully recovers the parity term in the output wire \color{IMSGreen}0 \color{black} of \vect{C}, up to permuting quantum registers.}
  \label{fig:routed circuit}
\end{figure}%

\paragraph{The technicality of PermRowCol}

$\vect{A}_0$ is reduced to \vect{P} through a sequence of reduction steps. Before each reduction, a pivot row and column (denoted as $r$ and $c$) are selected based on the binary structure of $\vect{A}_0$ and the connectivity of G. In $\vect{A}_0$, suppose that row r corresponds to the $i$-th row and column c corresponds to the $j$-th column, $i,j \in \Z_n$. That is, $\Phi(i) = r$ and $\Phi(j) = c$. For $t \in \Z_n$, let $e_t$ be the basis vector (i.e., the t-th column of I). A reduction step involves two actions by applying a sequence of row operations: a \textbf{column reduction} which transforms column c to $e_i$ and a \textbf{row reduction} which transforms row r to $e_j^\top$. After a reduction step, the reduced row r and column c are removed from $\vect{A}_0$. Accordingly, vertex $\Phi(i)$ is removed from $G$. The algorithm terminates when there is one vertex left in G.

\begin{procedure}
To select a pivot row, proceed as follows.
    \begin{enumerate}
    \item Among all rows in $\vect{A}_0$, find the set of rows whose indices correspond to the non-cut vertices in $G$. Let it be $R_0$.
    \item Among all rows in $R_0$, find the set of rows with the minimum Hamming weight. Let it be $R_1$.
    \item Return an arbitrary row in $R_1$. Let its index be r. 
\end{enumerate}
\label{proc:select a pivot row}
\end{procedure}

\begin{procedure}
Given the pivot row r, to select a pivot column, proceed as follows.
   \begin{enumerate}
    \item Among all columns in $\vect{A}_0$, find the set of columns with a non-zero entry at row r. Let it be $C_0$.
    \item Among all columns in $C_0$, find the set of columns with the minimum Hamming weight. Let it be $C_1$.
    \item Return an arbitrary column in $C_1$. Let its index be c.
\end{enumerate}
 \label{proc:select a pivot column}
\end{procedure}

\paragraph{Column reduction}
Given a pivot column c, let $S_0$ be the set of its rows with a parity of $1$. In the trivial case, column c has a hamming weight of 1 and $S_0 = \{r\}$. This means it is already the basis vector $e_i$. Otherwise, conduct \cref{proc:column walk} to perform column reduction.

\begin{procedure}
Let $r \in S_0$ and $\lvert S_0 \rvert > 1$. Build a Steiner tree $T_0$ of G where $r$ is the root and $S_0$ is the terminal. Carry out two traversals in $T_0$, each of which returns a sequence of row operations. 
    \begin{enumerate}
    \item Traverse $T_0$ from the leaves to the root. For each Steiner node v, add its child c to it. This corresponds to performing R(c, v) on $\vect{A}_0$.
    \item Traverse $T_0$ from the leaves to the root and add every parent p to its child c. This corresponds to performing R(p, c) on $\vect{A}_0$.
\end{enumerate}
\label{proc:column walk}
\end{procedure}

Before the first traversal, the leaves of $T_0$ correspond to the rows (excluding row r) in column c that have a parity of 1, while the Steiner nodes correspond to the rows that have a parity of 0. After the first traversal, all Steiner nodes will carry a parity of 1. After the second traversal, the parity 1 from the root (i.e. row r) will be propagated to every other node in $T_0$. As a result, every row in column c will have a parity of 0 except for row r. Let $\ell_0$ be a word composed of all row operations output from the two traversals in $T_0$. $t \in \N$,
 \[
 \ell_0 = R(i_0, j_0)R(i_1, j_1)\ldots R(i_{t-1}, j_{t-1}).
 \]
 
After column reduction, column c is reduced to $e_i$ and $\vect{A}_0$ is updated as

 \[
 \vect{A}_{0}\leftarrow  R(i_{t-1}, j_{t-1})\ldots \cdot R(i_1, j_1)\cdot R(i_0, j_0)\cdot \vect{A}_0.
 \]

\paragraph{Row reduction} 
In the trivial case, the pivot row r has a hamming weight of 1 and $S_1 = \{r\}$. This means it is already the basis vector $e_j^\top$. Otherwise, to reduce row r, we start by solving a system of linear equations: finding rows $r_k$ in $\vect{A}_0$ such that $\bigoplus r_k  = e_j^\top \oplus r$. Let $S_1$ be the set of these indices $k$ including $r$. 

\begin{procedure}
Build a Steiner tree $T_1$ where $r$ is the root and $S_1$ is the terminal. Carry out two traversals in $T_1$, each of which returns a sequence of row operations. 
\begin{enumerate}
    \item Traverse $T_1$ from the root to the leaves and add every Steiner node v to its parent p. This corresponds to performing R(v, p) on $\vect{A}_0$.
    \item Traverse $T_1$ from the leaves to the root and add every child c to its parent p. This corresponds to performing R(c, p) on $\vect{A}_0$.
\end{enumerate}
    \label{proc:row walk}
\end{procedure}

Summing the parities of all rows in $S_1$ implies that all columns in row $r$ carry a $0$ except for column $c$. 

\begin{align}
    \bigoplus_{k \in S_1} r_k = e_j^\top.
    \label{eqn:cancellation}
\end{align}

After the second traversal, the parity on each terminal node is propagated to the root and added together. Since the Steiner nodes are added twice modulo 2 throughout the two traversals, they do not participate in the parity sum of \cref{eqn:cancellation}. Let $\ell_1$ be a word composed of all row operations output from the two traversals in $T_1$. $t' \in \N$,

\[
\ell_1 = R(i_{0'}, j_{0'})R(i_{1'}, j_{1'})\ldots R(i_{t'-1}, j_{t'-1}).
\]

After row reduction, row r is reduced to $e_j^\top$ and $\vect{A}_0$ is updated as

 \[
 \vect{A}_{0}\leftarrow  R(i_{t'-1}, j_{t'-1})\ldots \cdot R(i_{1'}, j_{1'})\cdot R(i_{0'}, j_{0'})\cdot \vect{A}_0.
 \]

\paragraph{Update the parity matrix and the connectivity graph after a reduction step}
After a column reduction, the participation of all input registers except for quantum register r is removed from the output register c. After a row reduction, the participation of input register r is removed from all output registers except for the output register c. As a result, the input register r and output register c are no longer coupled with any other registers. This is illustrated in \cref{fig: illus first reduction step,fig: second reduction,fig: third reduction}.

In the meantime, a parity term is eliminated to a single parity. Leveraging the reversibility of quantum gates and the relabelling of quantum registers, we can recover the evolved logical state in the designated quantum register. After that, row r and column c are removed from $\vect{A}_0$, and vertex r is removed from G. This means at each reduction step, our instance size is getting smaller. When PermRowCol terminates, we can recover the permutation matrix \vect{P} by assembling the reduced row and column from each reduction step.

\paragraph{Demonstrate a reduction step}
To illustrate PermRowCol's technicality, we use it to synthesize the logical circuit in \cref{fig:4q circuit 2} according to the hardware topology in \cref{fig:linear topology}. Since this is carried out on the physical layer, there is no ambiguity and we will no longer use coloured labels in the descriptions below.

\cref{fig:first reduction step-column reduction} demonstrates the column reduction in the first reduction step, where column 1 is eliminated. In \cref{fig:topology}, Steiner tree $T_0$ is constructed based on the pivot column, with vertex 0 as the root and vertex 3 as the leave. To reduce column 1, we carry out two traversals. The first bottom-up traversal returns two row operations R(3,2) and R(2,1). The second bottom-up traversal returns three row operations R(2,3), R(1,2), and R(0,1). Hence $\ell_0 = R(3,2)R(2,1)R(2,3)R(1,2)R(0,1)$. In \cref{fig:row operations column reduction}, they are performed on $\vect{A}_0$ and we get the evolved parity matrix $\vect{A}_0^0$ after the column reduction.

\begin{figure}[H]
    \begin{subfigure}{1\textwidth}
    \[
 \vect{A}_0=\begin{blockarray}{cccccc}
     & \color{IMSBlue}3'\color{black} & \color{IMSBlue}0'\color{black} & \color{IMSRed}1'\color{black} & \color{IMSBlue}2'\color{black}& \color{IMSGray} \text{Candidate rows}\\
    \begin{block}{c(cccc)c}
      \color{IMSBlue}3\color{black} & 1 & 1 & \color{IMSRed}1 & 0 & 3\\
      \color{IMSRed}0 & \color{IMSRed}0 & \color{IMSRed}0 & \color{IMSRed}1 & \color{IMSRed}1 & 2\\
      \color{IMSBlue}1\color{black} & 1 & 0 & \color{IMSRed}0 & 1 \\
      \color{IMSBlue}2\color{black} & 0 & 1 & \color{IMSRed}0 & 1 \\
    \end{block}
    \color{IMSGray}\text{Candidate columns} & & & 2 & 3\\
    \end{blockarray}
  \]

    \vspace{-1.5 em}
  
  \caption{For the first reduction step, the pivot row and column are highlighted in red.}
  \label{fig:parity matrix 2}
\end{subfigure}%

\vspace{1.5 em}

\begin{subfigure}{1\textwidth}
  \centering
 \includegraphics[scale =0.35]{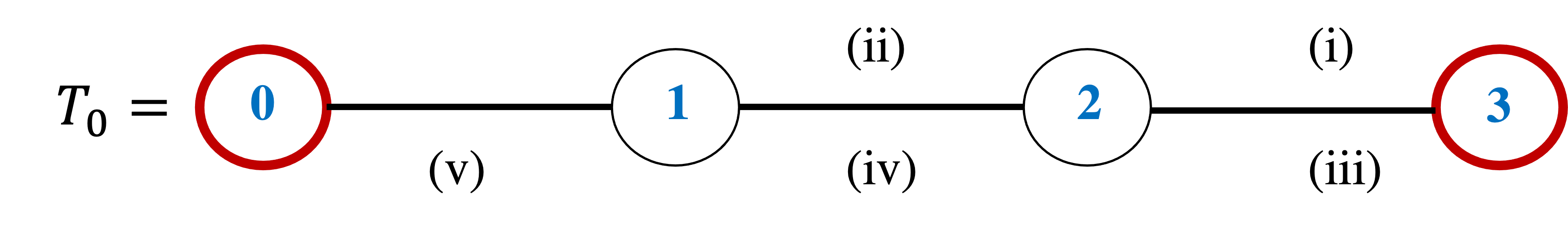}
 \caption{Steiner tree $T_0$ is constructed for the column reduction. Vertex \color{IMSBlue}0 \color{black} is the root and vertex \color{IMSBlue}3 \color{black} is the leaf. $S = \{\color{IMSBlue}0\color{black},  \color{IMSBlue}3\color{black}\}$. Two traversals of $T_0$ return a sequence of row operations, (i) $\sim$ (v).}
  \label{fig:topology}
\end{subfigure}%

 \vspace{1 em}
 
\begin{subfigure}{1\textwidth}
\[
\vect{A}_0 \xrightarrow[R(\color{IMSBlue}3\color{black},\color{IMSBlue}2\color{black})]{\text{(i)}} \begin{blockarray}{ccccc}
     & \color{IMSBlue}3'\color{black} & \color{IMSBlue}0'\color{black} & \color{IMSBlue}1'\color{black} & \color{IMSBlue}2'\color{black}\\
    \begin{block}{c(cccc)}
      \color{IMSBlue}3\color{black} & 1 & 1 & 1 & 0 \\
      \color{IMSBlue}0\color{black} & 0 & 0 & 1 & 1 \\
      \color{IMSBlue}1\color{black} & 1 & 0 & 0 & 1 \\
      \color{IMSBlue}2\color{black} & 1 & 0 & 1 & 1 \\
    \end{block}
    \end{blockarray}\xrightarrow[R(\color{IMSBlue}2\color{black},\color{IMSBlue}1\color{black})]{\text{(ii)}}\begin{blockarray}{ccccc}
     & \color{IMSBlue}3'\color{black} & \color{IMSBlue}0'\color{black} & \color{IMSBlue}1'\color{black} & \color{IMSBlue}2'\color{black}\\
    \begin{block}{c(cccc)}
      \color{IMSBlue}3\color{black} & 1 & 1 & 1 & 0 \\
      \color{IMSBlue}0\color{black} & 0 & 0 & 1 & 1 \\
      \color{IMSBlue}1\color{black} & 0 & 0 & 1 & 0 \\
      \color{IMSBlue}2\color{black} & 1 & 0 & 1 & 1 \\
    \end{block}
    \end{blockarray}\xrightarrow[R(\color{IMSBlue}2\color{black},\color{IMSBlue}3\color{black})]{\text{(iii)}}\begin{blockarray}{ccccc}
     & \color{IMSBlue}3'\color{black} & \color{IMSBlue}0'\color{black} & \color{IMSBlue}1'\color{black} & \color{IMSBlue}2'\color{black}\\
    \begin{block}{c(cccc)}
      \color{IMSBlue}3\color{black} & 0 & 1 & 0 & 1 \\
      \color{IMSBlue}0\color{black} & 0 & 0 & 1 & 1 \\
      \color{IMSBlue}1\color{black} & 0 & 0 & 1 & 0 \\
      \color{IMSBlue}2\color{black} & 1 & 0 & 1 & 1 \\
    \end{block}
    \end{blockarray}\xrightarrow[R(\color{IMSBlue}1\color{black},\color{IMSBlue}2\color{black})]{\text{(iv)}}
\]
\[
\begin{blockarray}{ccccc}
     & \color{IMSBlue}3'\color{black} & \color{IMSBlue}0'\color{black} & \color{IMSBlue}1'\color{black} & \color{IMSBlue}2'\color{black}\\
    \begin{block}{c(cccc)}
      \color{IMSBlue}3\color{black} & 0 & 1 & 0 & 1 \\
      \color{IMSBlue}0\color{black} & 0 & 0 & 1 & 1 \\
      \color{IMSBlue}1\color{black} & 0 & 0 & 1 & 0 \\
      \color{IMSBlue}2\color{black} & 1 & 0 & 0 & 1 \\
    \end{block}
    \end{blockarray}\xrightarrow[R(\color{IMSBlue}0\color{black},\color{IMSBlue}1\color{black})]{\text{(v)}}\begin{blockarray}{ccccc}
     & \color{IMSBlue}3'\color{black} & \color{IMSBlue}0'\color{black} & \color{IMSBlue}1'\color{black} & \color{IMSBlue}2'\color{black}\\
    \begin{block}{c(cccc)}
      \color{IMSBlue}3\color{black} & 0 & 1 & 0 & 1 \\
      \color{IMSBlue}0\color{black} & 0 & 0 & 1 & 1 \\
      \color{IMSBlue}1\color{black} & 0 & 0 & 0 & 1 \\
      \color{IMSBlue}2\color{black} & 1 & 0 & 0 & 1 \\
    \end{block}
    \end{blockarray} = \vect{A}_0^0.
\]

\vspace{-1.5 em}

\subcaption{After traversing $T_0$, perform a sequence of row operations on $\vect{A}_0$. Column \color{IMSBlue}1 \color{black} is reduced to $e_1$ and $\vect{A}_0^0$ is the evolved parity matrix.}
\label{fig:row operations column reduction}
\end{subfigure}
    \caption{Illustrate the column reduction in the first reduction step.}
    \label{fig:first reduction step-column reduction}
\end{figure}

After that, \cref{fig:first reduction step-row reduction} demonstrates the row reduction in the first reduction step, where row 0 is eliminated. In \cref{fig:topology2}, Steiner tree $T_1$ is constructed based on the pivot row in $\vect{A}_0^0$, since $r \oplus r_1 = e_2^\top$. The tree traversals only return a row operation R(1,0). Hence $\ell_1 = R(1,0)$. In \cref{fig:row operations row reduction}, it is performed on $\vect{A}_0^0$ and we get the evolved parity matrix $\vect{A}_0^1$ after the row reduction.

\begin{figure}[H]
\begin{subfigure}{1\textwidth}
  \centering
 \includegraphics[scale =0.35]{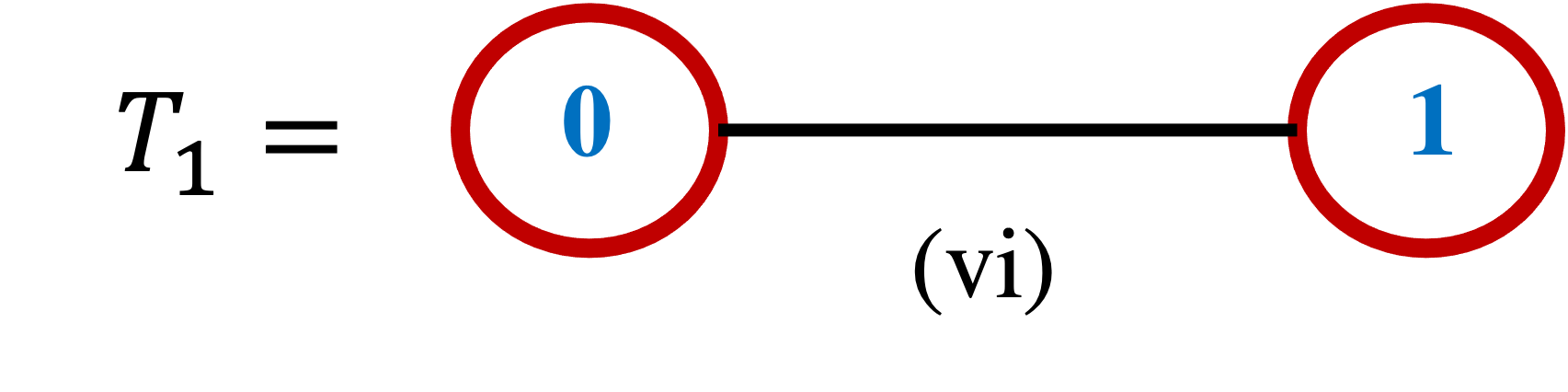}
 \caption{Steiner tree $T_1$ is constructed for the row reduction. Vertex \color{IMSBlue}0 \color{black} is the root and vertex \color{IMSBlue}1 \color{black} is the leaf. $S = \{\color{IMSBlue}0\color{black},  \color{IMSBlue}1\color{black}\}$. Two traversals of $T_1$ return one row operation, (vi).}
  \label{fig:topology2}
\end{subfigure}%

 \vspace{1 em}

\begin{subfigure}{1\textwidth}
\[
\vect{A}_0^0 \xrightarrow[R(\color{IMSBlue}1\color{black},\color{IMSBlue}0\color{black})]{\text{(vi)}}\begin{blockarray}{ccccc}
     & \color{IMSBlue}3'\color{black} & \color{IMSBlue}0'\color{black} & \color{IMSBlue}1'\color{black} & \color{IMSBlue}2'\color{black}\\
    \begin{block}{c(cccc)}
      \color{IMSBlue}3\color{black} & 0 & 1 & 0 & 1 \\
      \color{IMSBlue}0\color{black} & 0 & 0 & 1 & 0 \\
      \color{IMSBlue}1\color{black} & 0 & 0 & 0 & 1 \\
      \color{IMSBlue}2\color{black} & 1 & 0 & 0 & 1 \\
    \end{block}
    \end{blockarray} = \vect{A}_0^1.
\]

\vspace{-1.5 em}

\subcaption{After traversing $T_1$, perform a row operation on $\vect{A}_0^0$. Row \color{IMSBlue}0 \color{black} is reduced to $e_2^\top$ and $\vect{A}_0^1$ is the evolved parity matrix.}
\label{fig:row operations row reduction}
\end{subfigure}
    \caption{Illustrate the row reduction in the first reduction step.}
    \label{fig:first reduction step-row reduction}
\end{figure}

\cref{fig: illus first reduction step} articulates the intuition for a column and row reduction using the bipartite graph introduced in \cref{par:derivation}. Since the input register 0 and output register 1 are no longer coupled with any register, row 0 and column 1 are removed from $\vect{A}_0^1$. The parity matrix $\vect{A}_0$ is updated in \cref{fig: the updated parity matrix}. In the meantime, vertex 0 is removed from G and the connectivity graph is updated in \cref{fig:linear topology 2}. \cref{fig: second reduction,fig: third reduction} demonstrate details of the remaining reduction steps.

\begin{figure}[H]
\begin{subfigure}{1\textwidth}
\centering
    \scalebox{0.9}{\begin{tikzpicture}[tikzfig]
	\begin{pgfonlayer}{nodelayer}
		\node [style=none, label={left:$\color{IMSBlue}3$}] (0) at (-26.75, 1.25) {};
		\node [style=none] (1) at (-5, 1.25) {};
		\node [style=none] (2) at (0, 1.25) {};
		\node [style=none, label={right:$\color{IMSBlue}3'\color{black} = \color{IMSBlue}2$}] (3) at (4, 1.25) {};
		\node [style=none, label={left:$\color{IMSBlue}0$}] (4) at (-26.75, -2) {};
		\node [style=none] (5) at (-5, -2) {};
		\node [style=none] (6) at (0, -2) {};
		\node [style=none, label={right:$\color{IMSBlue}0'\color{black}= \color{IMSBlue}3$}] (7) at (4, -2) {};
		\node [style=none, label={left:$\color{IMSBlue}1$}] (8) at (-26.75, -5) {};
		\node [style=none] (9) at (-5, -5) {};
		\node [style=none] (10) at (0, -5) {};
		\node [style=none, label={right:$\color{IMSBlue}1'\color{black}=\color{IMSBlue}0$}] (11) at (4, -5) {};
		\node [style=none, label={left:$\color{IMSBlue}2$}] (12) at (-26.75, -7.75) {};
		\node [style=none] (13) at (-5, -7.75) {};
		\node [style=none] (14) at (0, -7.75) {};
		\node [style=none, label={right:$\color{IMSBlue}2'\color{black}= \color{IMSBlue}0\color{black} \oplus \color{IMSBlue}1\color{black} \oplus \color{IMSBlue}2\color{black} \oplus \color{IMSBlue}3\color{black}$}] (15) at (4, -7.75) {};
		\node [style=none] (22) at (-26.75, 1.25) {};
		\node [style=none] (27) at (-18.75, -7.75) {$\bigoplus$};
		\node [style=none] (28) at (-18.75, 1.25) {$\bullet$};
		\node [style=none] (29) at (-24.25, -2) {$\bigoplus$};
		\node [style=none] (30) at (-24.25, -5) {$\bullet$};
		\node [style=none] (31) at (-21.5, -5) {$\bigoplus$};
		\node [style=none] (32) at (-21.5, -7.75) {$\bullet$};
		\node [style=none] (33) at (-16.25, -7.75) {$\bigoplus$};
		\node [style=none] (34) at (-16.25, -5) {$\bullet$};
		\node [style=none] (35) at (-13.75, 1.25) {$\bigoplus$};
		\node [style=none] (36) at (-13.75, -7.75) {$\bullet$};
		\node [style=none] (37) at (-11, 1.25) {};
		\node [style=none] (38) at (-11, -2) {};
		\node [style=none] (39) at (-11, -5) {};
		\node [style=none] (40) at (-11, -7.75) {};
		\node [style=none] (41) at (-7, 1.25) {};
		\node [style=none] (42) at (-7, -2) {};
		\node [style=none] (43) at (-7, -5) {};
		\node [style=none] (44) at (-7, -7.75) {};
		\node [style=none] (45) at (-9.25, 1.25) {$\color{IMSBlue}1$};
		\node [style=none] (46) at (-9.25, -2) {$\color{IMSBlue}0\color{black} \oplus \color{IMSBlue}1$};
		\node [style=none] (47) at (-9.25, -5) {$\color{IMSBlue}1\color{black} \oplus \color{IMSBlue}2$};
		\node [style=none] (48) at (-9.25, -7.75) {$\color{IMSBlue}1\color{black} \oplus \color{IMSBlue}3$};
		\node [style=none] (49) at (14.25, -3) {$=$};
		\node [style=none, label={left:$\color{IMSBlue}3$}] (50) at (19.5, 1.25) {};
		\node [style=none, label={left:$\color{IMSBlue}0$}] (51) at (19.5, -2) {};
		\node [style=none, label={left:$\color{IMSBlue}1$}] (52) at (19.5, -5) {};
		\node [style=none, label={left:$\color{IMSBlue}2$}] (53) at (19.5, -7.75) {};
		\node [style=none] (54) at (19.5, 1.25) {};
		\node [style=none] (55) at (22.75, 1.25) {};
		\node [style=none] (56) at (27.75, 1.25) {};
		\node [style=none, label={right:$\color{IMSBlue}3'\color{black} = \color{IMSBlue}2$}] (57) at (31.75, 1.25) {};
		\node [style=none] (58) at (22.75, -2) {};
		\node [style=none] (59) at (27.75, -2) {};
		\node [style=none, label={right:$\color{IMSBlue}0'\color{black}= \color{IMSBlue}3$}] (60) at (31.75, -2) {};
		\node [style=none] (61) at (22.75, -5) {};
		\node [style=none] (62) at (27.75, -5) {};
		\node [style=none, label={right:$\color{IMSBlue}1'\color{black}=\color{IMSBlue}0$}] (63) at (31.75, -5) {};
		\node [style=none] (64) at (22.75, -7.75) {};
		\node [style=none] (65) at (27.75, -7.75) {};
		\node [style=none, label={right:$\color{IMSBlue}2'\color{black}= \color{IMSBlue}0\color{black} \oplus \color{IMSBlue}1\color{black} \oplus \color{IMSBlue}2\color{black} \oplus \color{IMSBlue}3\color{black}$}] (66) at (31.75, -7.75) {};
	\end{pgfonlayer}
	\begin{pgfonlayer}{edgelayer}
		\draw [style=edge] (2.center) to (3.center);
		\draw [style=edge] (6.center) to (7.center);
		\draw [style=red line] (10.center) to (11.center);
		\draw [style=edge] (14.center) to (15.center);
		\draw [style=edge] (1.center) to (2.center);
		\draw [style=edge] (1.center) to (6.center);
		\draw [style=red line,,dashed] (1.center) to (10.center);
		\draw [style=red line] (5.center) to (10.center);
		\draw [style=edge] (5.center) to (14.center);
		\draw [style=edge] (9.center) to (2.center);
		\draw [style=edge] (9.center) to (14.center);
		\draw [style=edge] (13.center) to (6.center);
		\draw [style=edge] (13.center) to (14.center);
		\draw [style=edge] (28.center) to (27.center);
		\draw [style=edge] (30.center) to (29.center);
		\draw [style=edge] (32.center) to (31.center);
		\draw [style=edge] (34.center) to (33.center);
		\draw [style=edge] (36.center) to (35.center);
		\draw [style=edge] (22.center) to (37.center);
		\draw [style=red line] (4.center) to (38.center);
		\draw [style=edge] (8.center) to (39.center);
		\draw [style=edge] (12.center) to (40.center);
		\draw [style=edge] (41.center) to (1.center);
		\draw [style=edge] (42.center) to (5.center);
		\draw [style=edge] (43.center) to (9.center);
		\draw [style=edge] (44.center) to (13.center);
		\draw [style=edge] (56.center) to (57.center);
		\draw [style=edge] (59.center) to (60.center);
		\draw [style=red line] (62.center) to (63.center);
		\draw [style=edge] (65.center) to (66.center);
		\draw [style=edge] (55.center) to (59.center);
		\draw [style=red line] (58.center) to (62.center);
		\draw [style=edge] (58.center) to (65.center);
		\draw [style=edge] (61.center) to (65.center);
		\draw [style=edge] (64.center) to (65.center);
		\draw [style=edge] (55.center) to (54.center);
		\draw [style=red line] (58.center) to (51.center);
		\draw [style=edge] (61.center) to (52.center);
		\draw [style=edge] (64.center) to (53.center);
		\draw [style=edge] (64.center) to (56.center);
		\draw [style=edge] (55.center) to (65.center);
	\end{pgfonlayer}
\end{tikzpicture}}
    \caption{In the column reduction, a sequence of row operations are performed on $\vect{A}_0$. They correspond to left-applying a sequence of CNOT gates to circuit \vect{C} after the initial qubit mapping, whose information propagation is illustrated by the bipartite graph (LHS). The pivot row corresponds to the input register \color{IMSBlue}0\color{black}. The pivot column corresponds to the output register \color{IMSBlue}1\color{black}. Relevant qubit wires are highlighted in red. After the column reduction, the updated coupling of circuit $\vect{C}_0^0$ (RHS) is described by $\vect{A}_0^0$. Except for quantum register \color{IMSBlue} 0 \color{black} (the red solid wire), the column reduction removes all other input registers (the red dashed wire) from output register \color{IMSBlue}1\color{black}.}
    \label{fig:the first column reduction}
\end{subfigure}

\vspace{1 em}

\begin{subfigure}{1\textwidth}
\centering
    \scalebox{1}{\begin{tikzpicture}[tikzfig]
	\begin{pgfonlayer}{nodelayer}
		\node [style=none, label={left:$\color{IMSBlue}3$}] (0) at (-16, 1) {};
		\node [style=none, label={left:$\color{IMSBlue}0$}] (4) at (-16, -2) {};
		\node [style=none, label={left:$\color{IMSBlue}1$}] (8) at (-16, -4.75) {};
		\node [style=none, label={left:$\color{IMSBlue}2$}] (12) at (-16, -7.5) {};
		\node [style=none] (22) at (-16, 1) {};
		\node [style=none] (33) at (-13, -4.75) {$\bigoplus$};
		\node [style=none] (34) at (-13, -2) {$\bullet$};
		\node [style=none] (37) at (-10, 1) {};
		\node [style=none] (38) at (-10, -2) {};
		\node [style=none] (39) at (-10, -4.75) {};
		\node [style=none] (40) at (-10, -7.5) {};
		\node [style=none] (45) at (-8.25, 1) {$\color{IMSBlue}3$};
		\node [style=none] (46) at (-8.25, -2) {$\color{IMSBlue}0$};
		\node [style=none] (47) at (-8.25, -4.75) {$\color{IMSBlue}0 \color{black}\oplus \color{IMSBlue}1\color{black}$};
		\node [style=none] (48) at (-8.25, -7.5) {$\color{IMSBlue}2$};
		\node [style=none] (49) at (15.75, -3) {$=$};
		\node [style=none, label={left:$\color{IMSBlue}3$}] (50) at (23, 1) {};
		\node [style=none, label={left:$\color{IMSBlue}0$}] (51) at (23, -2) {};
		\node [style=none, label={left:$\color{IMSBlue}1$}] (52) at (23, -4.75) {};
		\node [style=none, label={left:$\color{IMSBlue}2$}] (53) at (23, -7.5) {};
		\node [style=none] (54) at (23, 1) {};
		\node [style=none] (55) at (26.25, 1) {};
		\node [style=none] (56) at (31.25, 1) {};
		\node [style=none, label={right:$\color{IMSBlue}3'\color{black} = \color{IMSBlue}2\color{black}$}] (57) at (35.25, 1) {};
		\node [style=none] (58) at (26.25, -2) {};
		\node [style=none] (59) at (31.25, -2) {};
		\node [style=none, label={right:$\color{IMSBlue}0'\color{black}= \color{IMSBlue}3$}] (60) at (35.25, -2) {};
		\node [style=none] (61) at (26.25, -4.75) {};
		\node [style=none] (62) at (31.25, -4.75) {};
		\node [style=none, label={right:$\color{IMSBlue}1'\color{black}=\color{IMSBlue}0$}] (63) at (35.25, -4.75) {};
		\node [style=none] (64) at (26.25, -7.5) {};
		\node [style=none] (65) at (31.25, -7.5) {};
		\node [style=none, label={right:$\color{IMSBlue}2'\color{black}= \color{IMSBlue}1\color{black} \oplus \color{IMSBlue}2\color{black} \oplus \color{IMSBlue}3\color{black}$}] (66) at (35.25, -7.5) {};
		\node [style=none] (67) at (-5.5, 1) {};
		\node [style=none] (68) at (-5.5, -2) {};
		\node [style=none] (69) at (-5.5, -4.75) {};
		\node [style=none] (70) at (-5.5, -7.5) {};
		\node [style=none] (71) at (-5.5, 1) {};
		\node [style=none] (72) at (-2.25, 1) {};
		\node [style=none] (73) at (2.75, 1) {};
		\node [style=none, label={right:$\color{IMSBlue}3' \color{black}= \color{IMSBlue}2$}] (74) at (6.75, 1) {};
		\node [style=none] (75) at (-2.25, -2) {};
		\node [style=none] (76) at (2.75, -2) {};
		\node [style=none, label={right:$\color{IMSBlue}0' \color{black}= \color{IMSBlue}3$}] (77) at (6.75, -2) {};
		\node [style=none] (78) at (-2.25, -4.75) {};
		\node [style=none] (79) at (2.75, -4.75) {};
		\node [style=none, label={right:$\color{IMSBlue}1' \color{black}=\color{IMSBlue}0$}] (80) at (6.75, -4.75) {};
		\node [style=none] (81) at (-2.25, -7.5) {};
		\node [style=none] (82) at (2.75, -7.5) {};
		\node [style=none, label={right:$\color{IMSBlue}2'\color{black}= \color{IMSBlue}1 \color{black}\oplus \color{IMSBlue}2 \color{black}\oplus \color{IMSBlue}3$}] (83) at (6.75, -7.5) {};
	\end{pgfonlayer}
	\begin{pgfonlayer}{edgelayer}
		\draw [style=edge] (34.center) to (33.center);
		\draw [style=edge] (22.center) to (37.center);
		\draw [style=red line] (4.center) to (38.center);
		\draw [style=edge] (8.center) to (39.center);
		\draw [style=edge] (12.center) to (40.center);
		\draw [style=edge] (56.center) to (57.center);
		\draw [style=edge] (59.center) to (60.center);
		\draw [style=edge] (65.center) to (66.center);
		\draw [style=edge] (55.center) to (59.center);
		\draw [style=edge] (61.center) to (65.center);
		\draw [style=edge] (64.center) to (65.center);
		\draw [style=edge] (55.center) to (54.center);
		\draw [style=edge] (61.center) to (52.center);
		\draw [style=edge] (64.center) to (53.center);
		\draw [style=edge] (64.center) to (56.center);
		\draw [style=edge] (55.center) to (65.center);
		\draw [style=edge] (73.center) to (74.center);
		\draw [style=edge] (76.center) to (77.center);
		\draw [style=red line] (79.center) to (80.center);
		\draw [style=edge] (82.center) to (83.center);
		\draw [style=edge] (72.center) to (76.center);
		\draw [style=red line] (75.center) to (79.center);
		\draw [style=red line, dashed] (75.center) to (82.center);
		\draw [style=edge] (78.center) to (82.center);
		\draw [style=edge] (81.center) to (82.center);
		\draw [style=edge] (72.center) to (71.center);
		\draw [style=red line] (75.center) to (68.center);
		\draw [style=edge] (78.center) to (69.center);
		\draw [style=edge] (81.center) to (70.center);
		\draw [style=edge] (81.center) to (73.center);
		\draw [style=edge] (72.center) to (82.center);
		\draw [style=red line] (51.center) to (58.center);
		\draw [style=red line] (62.center) to (58.center);
		\draw [style=red line] (62.center) to (63.center);
	\end{pgfonlayer}
\end{tikzpicture}}
    \caption{In the row reduction, a row operation R(\color{IMSBlue}1\color{black},\color{IMSBlue}0\color{black}) is performed on $\vect{A}_0^0$. It corresponds to left-applying CNOT(\color{IMSBlue}0\color{black},\color{IMSBlue}1\color{black}) to circuit $\vect{C}_0^0$, whose information propagation is illustrated by the bipartite graph (LHS). After the row reduction, the updated parity matrix $\vect{A}_0^1$ corresponds to the updated circuit $\vect{C}_0^1$, whose information propagation is illustrated by the bipartite graph (RHS). The row reduction removes input register \color{IMSBlue} 0 \color{black} from all other output registers (the red dashed wire), except output register \color{IMSBlue}1\color{black} (the red solid line).}
    \label{fig:the first row reduction}
\end{subfigure}
\caption{The intuition behind the first reduction step.}
\label{fig: illus first reduction step}
\end{figure}

Lastly, we comment on the heuristic algorithm that we choose to create a Steiner tree with a terminal S for a simple connected graph. This is equivalent to constructing a minimum spanning tree over S using Dijkstra's Shortest Path algorithm~\cite{hougardy2015dijkstra}. For each terminal node, its shortest path to the root is used to construct the minimum spanning tree. After that, the calculation of total edge weight considers the paths between the constructed tree and the terminals that have not yet been added to the spanning tree.


\begin{figure}[H]
\begin{subfigure}{.68\textwidth}
    \[
 \vect{A}_{0}=\begin{blockarray}{ccccc}
     & \color{IMSBlue}3'\color{black} & \color{IMSBlue}0'\color{black} & \color{IMSRed}2'\color{black}& \color{IMSGray} \text{Candidate rows}\\
    \begin{block}{c(ccc)c}
      \color{IMSBlue}3\color{black} & 0 & 1 & \color{IMSRed}1 & 2\\
      \color{IMSRed}1\color{black} & \color{IMSRed}0 & \color{IMSRed}0 & \color{IMSRed}1 & 1\\
      \color{IMSBlue}2\color{black} & 1 & 0 & \color{IMSRed}1 \\
    \end{block}
    \color{IMSGray}\text{Candidate columns} & & & 3\\
    \end{blockarray}
  \]
  
  \vspace{-2 em}
  
  \caption{For the second reduction step, the pivot row and column are highlighted in red.}
  \label{fig: the updated parity matrix}
\end{subfigure}%
\;
\begin{subfigure}{.3\textwidth}
    \centering
     \includegraphics[scale =0.35]{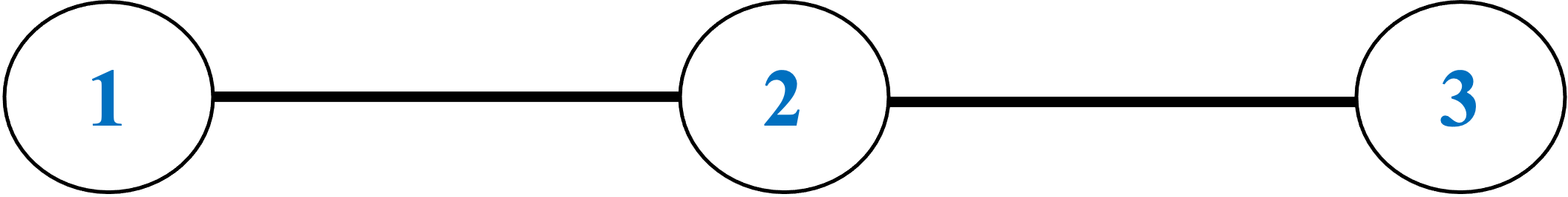}
    \caption{G is the updated connectivity graph after removing vertex \color{IMSBlue}0\color{black}.}
    \label{fig:linear topology 2}
\end{subfigure}%

\vspace{1 em}

\begin{subfigure}{1\textwidth}
  \centering
 \includegraphics[scale =0.35]{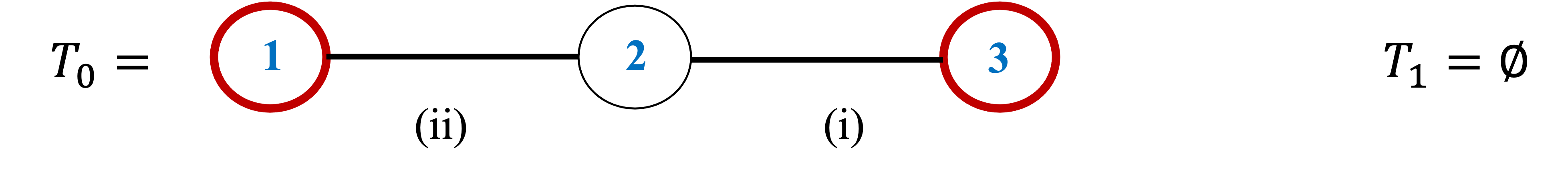}
 \caption{Steiner tree $T_0$ is constructed for the column reduction. Vertex \color{IMSBlue}1 \color{black} is the root and vertex \color{IMSBlue}3 \color{black} is the leaf. $S = \{\color{IMSBlue}1\color{black},  \color{IMSBlue}3\color{black}\}$. Two traversals of $T_0$ return a sequence of row operations (annotated by (i) $\sim$ (ii)).}
  \label{fig:topology3}
\end{subfigure}%

 \vspace{1 em}
 
\begin{subfigure}{1\textwidth}
\[
\vect{A}_{0} \xrightarrow[R(\color{IMSBlue}2\color{black},\color{IMSBlue}3\color{black})]{\text{(i)}}\begin{blockarray}{cccc}
     & \color{IMSBlue}3'\color{black} & \color{IMSBlue}0'\color{black} & \color{IMSBlue}2'\color{black}\\
    \begin{block}{c(ccc)}
      \color{IMSBlue}3\color{black} & 1 & 1 & 0 \\
      \color{IMSBlue}1\color{black} & 0 & 0 & 1 \\
      \color{IMSBlue}2\color{black} & 1 & 0 & 1 \\
    \end{block}
    \end{blockarray}\xrightarrow[R(\color{IMSBlue}1\color{black},\color{IMSBlue}2\color{black})]{\text{(ii)}}\begin{blockarray}{cccc}
     & \color{IMSBlue}3'\color{black} & \color{IMSBlue}0'\color{black} & \color{IMSBlue}2'\color{black}\\
    \begin{block}{c(ccc)}
      \color{IMSBlue}3\color{black} & 1 & 1 & 0 \\
      \color{IMSBlue}1\color{black} & 0 & 0 & 1 \\
      \color{IMSBlue}2\color{black} & 1 & 0 & 0 \\
    \end{block}
    \end{blockarray}=\vect{A}_0^0=\vect{A}_0^1.
\]

\vspace{-1.5 em}

\subcaption{After traversing $T_0$, perform a sequence of row operations on $\vect{A}_0$. Column \color{IMSBlue}2 \color{black} is reduced to $e_2$ and $\vect{A}_0^0$ is the evolved parity matrix. Since row \color{IMSBlue}1 \color{black} is also reduced to $e_3^\top$, $\vect{A}_0^0=\vect{A}_0^1$ and Steiner tree $T_1 = \emptyset$. No more row reduction is needed.}
\label{fig:row operations2}
\end{subfigure}

 \vspace{1 em}
 
\begin{subfigure}{1\textwidth}
\centering
    \scalebox{1}{\begin{tikzpicture}[tikzfig]
	\begin{pgfonlayer}{nodelayer}
		\node [style=none, label={left:$\color{IMSBlue}3$}] (0) at (-18.5, 1) {};
		\node [style=none, label={left:$\color{IMSBlue}1$}] (4) at (-18.5, -2) {};
		\node [style=none, label={left:$\color{IMSBlue}2$}] (12) at (-18.5, -5) {};
		\node [style=none] (22) at (-18.5, 1) {};
		\node [style=none] (33) at (-14, -5) {$\bigoplus$};
		\node [style=none] (34) at (-14, 1) {$\bullet$};
		\node [style=none] (37) at (-11, 1) {};
		\node [style=none] (38) at (-11, -2) {};
		\node [style=none] (40) at (-11, -5) {};
		\node [style=none] (45) at (-8.5, 1) {$\color{IMSBlue}3$};
		\node [style=none] (46) at (-8.5, -2) {$\color{IMSBlue}1\color{black}\oplus\color{IMSBlue} 2$};
		\node [style=none] (48) at (-8.5, -5) {$\color{IMSBlue}2\color{black} \oplus \color{IMSBlue}3$};
		\node [style=none] (49) at (15, -2.25) {$=$};
		\node [style=none] (50) at (-5.5, 1) {};
		\node [style=none] (52) at (-5.5, -2) {};
		\node [style=none] (53) at (-5.5, -5) {};
		\node [style=none] (54) at (-5.5, 1) {};
		\node [style=none] (55) at (-2.25, 1) {};
		\node [style=none] (56) at (2.75, 1) {};
		\node [style=none, label={right:$\color{IMSBlue}3'\color{black} = \color{IMSBlue}2\color{black} \oplus \color{IMSBlue}3\color{black}$}] (57) at (6.75, 1) {};
		\node [style=none] (59) at (2.75, -2) {};
		\node [style=none, label={right:$\color{IMSBlue}0'\color{black}= \color{IMSBlue}3\color{black}$}] (60) at (6.75, -2) {};
		\node [style=none] (61) at (-2.25, -2) {};
		\node [style=none] (64) at (-2.25, -5) {};
		\node [style=none] (65) at (2.75, -5) {};
		\node [style=none, label={right:$\color{IMSBlue}2'\color{black}= \color{IMSBlue}1\color{black}$}] (66) at (6.75, -5) {};
		\node [style=none] (67) at (-16.25, -2) {$\bigoplus$};
		\node [style=none] (68) at (-16.25, -5) {$\bullet$};
		\node [style=none] (69) at (18.25, 1) {};
		\node [style=none, label={left:$\color{IMSBlue}1\color{black}$}] (70) at (18.25, -2) {};
		\node [style=none, label={left:$\color{IMSBlue}2\color{black}$}] (71) at (18.25, -5) {};
		\node [style=none, label={left:$\color{IMSBlue}3\color{black}$}] (72) at (18.25, 1) {};
		\node [style=none] (73) at (21.5, 1) {};
		\node [style=none] (74) at (26.5, 1) {};
		\node [style=none, label={right:$\color{IMSBlue}3'\color{black} = \color{IMSBlue}2\color{black} \oplus \color{IMSBlue}3\color{black}$}] (75) at (30.5, 1) {};
		\node [style=none] (76) at (26.5, -2) {};
		\node [style=none, label={right:$\color{IMSBlue}0'\color{black}= \color{IMSBlue}3$}] (77) at (30.5, -2) {};
		\node [style=none] (78) at (21.5, -2) {};
		\node [style=none] (79) at (21.5, -5) {};
		\node [style=none] (80) at (26.5, -5) {};
		\node [style=none, label={right:$\color{IMSBlue}2'\color{black}= \color{IMSBlue}1$}] (81) at (30.5, -5) {};
	\end{pgfonlayer}
	\begin{pgfonlayer}{edgelayer}
		\draw [style=edge] (34.center) to (33.center);
		\draw [style=edge] (22.center) to (37.center);
		\draw [style=red line] (4.center) to (38.center);
		\draw [style=edge] (12.center) to (40.center);
		\draw [style=edge] (56.center) to (57.center);
		\draw [style=edge] (59.center) to (60.center);
		\draw [style=red line] (65.center) to (66.center);
		\draw [style=edge] (55.center) to (59.center);
		\draw [style=red line] (61.center) to (65.center);
		\draw [style=red line, dashed] (64.center) to (65.center);
		\draw [style=edge] (55.center) to (54.center);
		\draw [style=red line] (61.center) to (52.center);
		\draw [style=edge] (64.center) to (53.center);
		\draw [style=edge] (64.center) to (56.center);
		\draw [style=red line, dashed] (55.center) to (65.center);
		\draw [style=edge] (68.center) to (67.center);
		\draw [style=edge] (74.center) to (75.center);
		\draw [style=edge] (76.center) to (77.center);
		\draw [style=edge] (73.center) to (76.center);
		\draw [style=edge] (73.center) to (72.center);
		\draw [style=edge] (79.center) to (71.center);
		\draw [style=edge] (79.center) to (74.center);
		\draw (73.center) to (74.center);
		\draw [style=red line] (70.center) to (78.center);
		\draw [style=red line] (78.center) to (80.center);
		\draw [style=red line] (80.center) to (81.center);
	\end{pgfonlayer}
\end{tikzpicture}}
    \caption{The intuition behind the second reduction step.}
\end{subfigure}    
\caption{Illustrate the second reduction step.}
\label{fig: second reduction}
\end{figure}
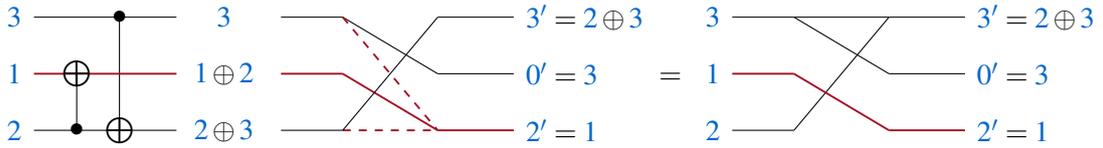 

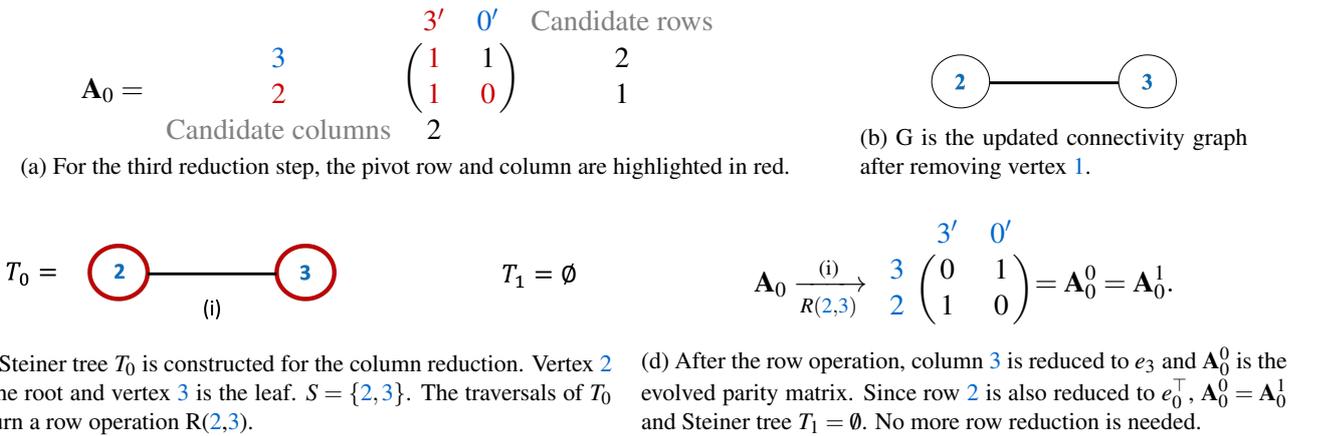
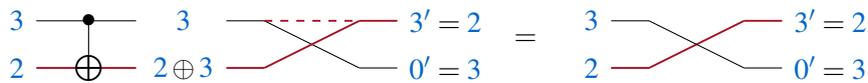
\begin{figure}[H]
\begin{subfigure}{.68\textwidth}
    \[
 \vect{A}_{0}=\begin{blockarray}{cccc}
     & \color{IMSRed}3'\color{black} & \color{IMSBlue}0'\color{black} &  \color{IMSGray} \text{Candidate rows}\\
    \begin{block}{c(cc)c}
      \color{IMSBlue}3\color{black} & \color{IMSRed}1 & 1 & 2\\
      \color{IMSRed}2\color{black} & \color{IMSRed}1 & \color{IMSRed}0  & 1\\
    \end{block}
    \color{IMSGray}\text{Candidate columns} & 2 & & \\
    \end{blockarray}
  \]
  
  \vspace{-2 em}
  
  \caption{For the third reduction step, the pivot row and column are highlighted in red.}
  \label{fig: the final parity matrix}
\end{subfigure}%
\;
\begin{subfigure}{.3\textwidth}
    \centering
     \includegraphics[scale =0.35]{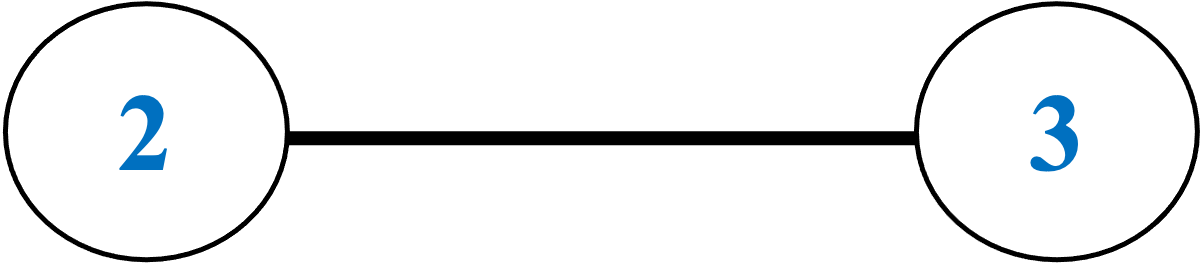}
    \caption{G is the updated connectivity graph after removing vertex \color{IMSBlue}1\color{black}.}
    \label{fig:linear topology 3}
\end{subfigure}%

\vspace{1 em}

\begin{subfigure}{.5\textwidth}
  \centering
 \includegraphics[scale =0.35]{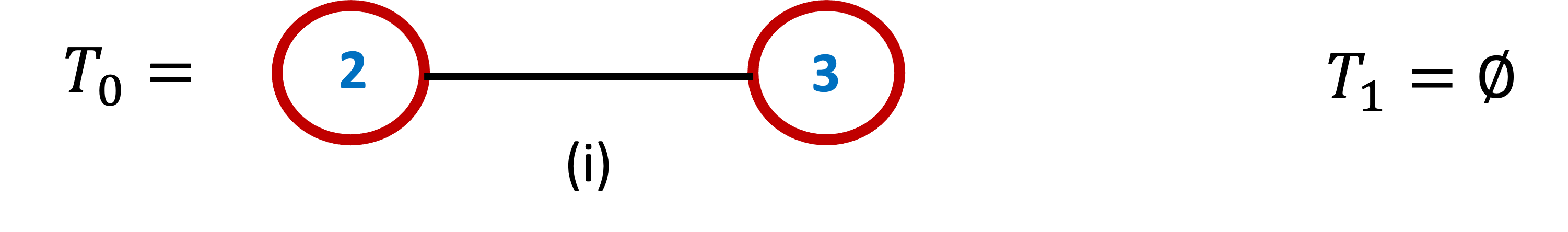}
 \caption{Steiner tree $T_0$ is constructed for the column reduction. Vertex \color{IMSBlue}2 \color{black} is the root and vertex \color{IMSBlue}3 \color{black} is the leaf. $S = \{\color{IMSBlue}2\color{black},  \color{IMSBlue}3\color{black}\}$. The traversals of $T_0$ return a row operation R(\color{IMSBlue}2\color{black},\color{IMSBlue}3\color{black}).}
  \label{fig:topology4}
\end{subfigure}%
\quad
\begin{subfigure}{.5\textwidth}
\[
\vect{A}_{0} \xrightarrow[R(\color{IMSBlue}2\color{black},\color{IMSBlue}3\color{black})]{\text{(i)}}\begin{blockarray}{ccc}
     & \color{IMSBlue}3'\color{black} & \color{IMSBlue}0'\color{black}\\
    \begin{block}{c(cc)}
      \color{IMSBlue}3\color{black} & 0 & 1 \\
      \color{IMSBlue}2\color{black} & 1 & 0 \\
    \end{block}
    \end{blockarray}=\vect{A}_0^0=\vect{A}_0^1.
\]
\vspace{-1.5 em}
\subcaption{After the row operation, column \color{IMSBlue}3 \color{black} is reduced to $e_3$ and $\vect{A}_0^0$ is the evolved parity matrix. Since row \color{IMSBlue}2 \color{black} is also reduced to $e_0^\top$, $\vect{A}_0^0=\vect{A}_0^1$ and Steiner tree $T_1 = \emptyset$. No more row reduction is needed.}
\label{fig:row operations3}
\end{subfigure}

\vspace{1 em}

\begin{subfigure}{1\textwidth}
\centering
    \scalebox{1}{\begin{tikzpicture}[tikzfig]
	\begin{pgfonlayer}{nodelayer}
		\node [style=none, label={left:$\color{IMSBlue}3$}] (0) at (-14.5, 1) {};
		\node [style=none, label={left:$\color{IMSBlue}2$}] (12) at (-14.5, -1.5) {};
		\node [style=none] (22) at (-14.5, 1) {};
		\node [style=none] (33) at (-11.75, -1.5) {$\bigoplus$};
		\node [style=none] (34) at (-11.75, 1) {$\bullet$};
		\node [style=none] (37) at (-9.25, 1) {};
		\node [style=none] (40) at (-9.25, -1.5) {};
		\node [style=none] (45) at (-6.75, 1) {$\color{IMSBlue}3$};
		\node [style=none] (48) at (-6.75, -1.5) {$\color{IMSBlue}2\color{black} \oplus \color{IMSBlue}3$};
		\node [style=none] (49) at (11.25, 0) {$=$};
		\node [style=none] (69) at (-4.5, 1) {};
		\node [style=none] (71) at (-4.5, -1.5) {};
		\node [style=none] (72) at (-4.5, 1) {};
		\node [style=none] (73) at (-2.5, 1) {};
		\node [style=none] (74) at (2.5, 1) {};
		\node [style=none, label={right:$\color{IMSBlue}3'\color{black} = \color{IMSBlue}2$}] (75) at (4.5, 1) {};
		\node [style=none] (76) at (2.5, -1.5) {};
		\node [style=none, label={right:$\color{IMSBlue}0'\color{black}= \color{IMSBlue}3$}] (77) at (4.5, -1.5) {};
		\node [style=none] (79) at (-2.5, -1.5) {};
		\node [style=none] (80) at (15.75, 1) {};
		\node [style=none, label={left:$\color{IMSBlue}2$}] (81) at (15.75, -1.5) {};
		\node [style=none, label={left:$\color{IMSBlue}3$}] (82) at (15.75, 1) {};
		\node [style=none] (83) at (17.75, 1) {};
		\node [style=none] (84) at (22.75, 1) {};
		\node [style=none, label={right:$\color{IMSBlue}3'\color{black} = \color{IMSBlue}2$}] (85) at (24.75, 1) {};
		\node [style=none] (86) at (22.75, -1.5) {};
		\node [style=none, label={right:$\color{IMSBlue}0'\color{black}= \color{IMSBlue}3$}] (87) at (24.75, -1.5) {};
		\node [style=none] (88) at (17.75, -1.5) {};
	\end{pgfonlayer}
	\begin{pgfonlayer}{edgelayer}
		\draw [style=edge] (34.center) to (33.center);
		\draw [style=edge] (22.center) to (37.center);
		\draw [style=red line] (12.center) to (40.center);
		\draw [style=red line] (74.center) to (75.center);
		\draw [style=edge] (76.center) to (77.center);
		\draw [style=edge] (73.center) to (76.center);
		\draw [style=edge] (73.center) to (72.center);
		\draw [style=red line] (79.center) to (71.center);
		\draw [style=red line] (79.center) to (74.center);
		\draw [style=red line, dashed] (73.center) to (74.center);
		\draw [style=red line] (84.center) to (85.center);
		\draw [style=red line] (88.center) to (81.center);
		\draw [style=red line] (88.center) to (84.center);
		\draw [style=edge] (82.center) to (83.center);
		\draw [style=edge] (83.center) to (86.center);
		\draw [style=edge] (86.center) to (87.center);
	\end{pgfonlayer}
\end{tikzpicture}}
    \caption{The intuition behind the third reduction step.}
\end{subfigure}
\caption{Illustrate the third reduction step.}
\label{fig: third reduction}
\end{figure}

\subsection{A Primer for Noise-Aware CNOT Circuit Routing}
\label{subsec:errormitigation}

In NISQ architectures, it is inevitable to have several error sources during the execution of a quantum circuit. Therefore, it is important to reduce their effects as much as possible. \textbf{Noise-aware CNOT circuit routing} accounts for both the CNOT gate error rates and the constraint on the NN interactions. It takes a parity matrix and an undirected edge-weighted connected graph as input. It outputs a synthesized circuit composed of allowed CNOT operations, with improved reliability.

Let $\rho$ denote the density matrix describing the initial state of a quantum system. Let $n$ denote the number of qubits in a quantum system, and $t = 2^n$ denote the dimension of its state space. In \cref{subsubcec:noisychannel}, we review a noisy quantum channel and its \textbf{superoperator representation}. In \cref{subsubsec:average gate fidelity}, we introduce a metric called the \textbf{average gate fidelity}. Compared to the CNOT count and the accumulated CNOT gate error rate, it accurately quantifies the reliability of executing a CNOT circuit on NISQ hardware.

\subsubsection{Quantum Channel in Superoperator Representation}
\label{subsubcec:noisychannel}

\begin{definition}
    A \textbf{quantum channel} $\mathcal{E}$ is a completely positive trace-preserving map between spaces of operators. 
    Let $\{M_k\}$ be a set of its Kraus operators, $\sum_k M_k^{\dagger}M_k = I$. The action of $\mathcal{E}$ on $\rho$ can be expressed in terms of the Kraus decomposition,
\[
\mathcal{E}(\rho) = \sum_{k}M_k\rho M^{\dagger}_k.
\]
\label{def:channel}
\end{definition}

We use a quantum channel to describe the evolution of a quantum state. It provides a convenient mathematical framework for us to characterize noise in a system. We start by representing a single-qubit noisy channel $\mathcal{E}$ using the Kraus decomposition. Without loss of generality, consider a single-qubit Pauli channel~\cite{basile2023quantum}. 

\begin{definition}
Let $\mathcal{B}_1 = \{I, X, Y, Z\}$ be a $1$-qubit Pauli basis. For $n \in \N^{\neq 0}$, let $\mathcal{B}_n$ be an $n$-qubit Pauli basis.
    \[
    \mathcal{B}_n = \Biggl\{\bigotimes_{j=0}^{n-1}B_j; \; B_j \in \mathcal{B}_1\Biggr\}.
    \]
\end{definition}

\begin{definition}
    For $n \in \N^{\neq 0}$, let $\mathcal{P}_n$ be the $n$-qubit Pauli group. For $j \in \Z_n$, $P_j$ denotes the single-qubit Pauli operator acting on qubit $j$. 
    \[
    \mathcal{P}_n = \Biggl\{\bigotimes_{j=0}^{n-1}i^cP_j;\;\; P_j \in \mathcal{B}_1, \;\;c \in \mathbb{Z}_4\Biggr\}.
    \]
    \label{def:pauli}
\end{definition}
\begin{lemma}
    For any $Q \in \mathcal{P}_n \setminus \{i^cI;\;\;c\in\Z_4\}$, $ \Tr[Q] = 0$.
    \label{lem:traceless}
\end{lemma}
\begin{proof}
By \cref{def:pauli}, $Q =\bigotimes_{j=0}^{n-1}i^cP_j,\; P_j \in \mathcal{B}_1, \; c \in \Z_4$. Since $Q \notin \{i^cI;\;\;c\in\Z_4\}$, there exists some $k \in \Z_n$ such that $P_k \neq I$. Then, $\Tr[P_k] = 0$ implies that
    \[
    \Tr(Q) = \Tr\Biggl[\bigotimes_{j=0}^{n-1}i^cP_j\Biggr] = i^c\prod_{j=0}^{n-1}\Tr[P_j] = 0.
    \]
\end{proof}

\begin{definition}
    For $k\in \Z_4$, $E_k \in \mathcal{B}_1$ with $E_0 = I, \;\; E_1 = X, \;\; E_2 = Y, \;\; E_3 = Z$. $P_k$ is the probability distribution of $E_k$. $0 \leq P_k \leq 1$ and $\sum P_k = 1$. $M_k$ is the Kraus operator such that
\[
M_k = \sqrt{P_k}E_k, \qquad \sum_{k=0}^{3}M_k^{\dagger}M_k = I.
\]
A single-qubit Pauli channel $\mathcal{E}$ is defined as
\begin{align}
    \mathcal{E}(\rho) = \sum_{k=0}^3 P_k E_k \rho E_k^\dagger. 
    \label{eqn: error channel and Pauli with index}
\end{align}
\label{def:singlequbit}
\end{definition}

\begin{remark}
    When $P_1 = P_2 = P_3$, $\mathcal{E}$ is called the \textbf{depolarizing channel}.
\end{remark}

\begin{definition}
For $n \in \N^{\neq 0}$,  $k \in \Z_{4^n}$, $E_k \in \mathcal{B}_n$. $P_k$ is the probability distribution of $E_k$. $0 \leq P_k \leq 1$ and $\sum P_k = 1$. $M_k$ is the Kraus operator such that
\[
M_k = \sqrt{P_k}E_k, \qquad \sum_{k=0}^{4^n - 1}M_k^{\dagger}M_k = I.
\]
An $n$-qubit Pauli channel $\mathcal{E}$ is defined as
\begin{align}
    \mathcal{E}(\rho) = \sum_{k=0}^{4^n - 1} P_k E_k \rho E_k^\dagger. \label{eqn:multi-qubit Pauli channel}
\end{align}
    \label{def:multiequbit}
\end{definition}

\begin{definition}
For $m \in \N^{\neq 0}$, let $A$ be an $m \times m$ matrix.
\[
A = [a_{i,j}]=\begin{bmatrix}
    a_{0,0} & a_{0,1} & \cdots & a_{0,m-1}\\
    a_{1,0} & a_{1,1} & \cdots & a_{1,m-1}\\
    \vdots & \vdots & \ddots & \vdots \\
    a_{m-1,0} & a_{m-1,1} & \cdots & a_{m-1,m-1}
\end{bmatrix},\qquad i,j \in \Z_m.
\]
$a_{i,j}$ denotes the matrix element on row $i$ and column $j$. $\vectz{A}$ is a \textbf{column-vectorized matrix} obtained from stacking each column of $A$ on top of one another. 
\[
\vectz{A} \coloneqq (a_{0,0}\;\; a_{1,0}\;\; \ldots\;\; a_{m-1,0}\;\;a_{0,1}\;\; a_{1,1}\;\;\ldots\;\; a_{m-1,1}\;\;\ldots\;\;a_{0,m-1}\;\; a_{1,m-1}\;\; \ldots\;\; a_{m-1,m-1})^{\top}.
\]
    \label{def:vectorization}
\end{definition}
Vectorization converts an operator to a vector. It is linear in that for any two matrices $A$ and $B$ of the same dimension,
\[
\vectz{A+B} = \vectz{A} + \vectz{B}.
\]

With the Kronecker product, matrix vectorization has a convenient property and we will use it to derive the \textbf{superoperator representation} for an arbitrary linear operator.
\begin{lemma}[\cite{dhrymes2000matrix}]
Consider matrices $A$, $B$, and $C$ with compatible dimensions for matrix multiplications,
    \[
    \vectz{ABC} = (C^{\top}\otimes A)\vectz{B}.
    \]
    \label{lem:roth}
\end{lemma}

\begin{definition}
Let $\mathcal{A}$ be a quantum operation acting on a density matrix $\rho$. $\rho$ is of dimension $t\times t$, $t \in \N^{\neq 0}$. The superoperator representation $S_{\mathcal{A}}$ is given alongside the column-vectorized density matrix $\vectz{\rho}$ as
 \[
 \vectz{\mathcal{A}(\rho)} = S_{\mathcal{A}}\vectz{\rho}.
 \]
 $S_\mathcal{A}$ is of dimension $t^2 \times t^2$. 
 \label{def:superoperator}
\end{definition}
\begin{lemma}
Let $\mathcal{A}$ be a quantum operation acting on a density matrix $\rho$. When $\mathcal{A}$ is unitary, let $A$ be its unitary matrix representation. Then $S_{\mathcal{A}} = A^{*}\otimes A$.
    \label{lem: superoperator unitary}
\end{lemma}
\begin{proof}
Since $\mathcal{A}$ is unitary, $\mathcal{A}(\rho) = A\rho A^{\dagger}$. By \cref{lem:roth}, $\vectz{\mathcal{A}(\rho)} = \vectz{A\rho A^{\dagger}} = ((A^{\dagger})^\top\otimes A )\vectz{\rho} = A^{*}\otimes A\vectz{\rho}$. By \cref{def:superoperator}, $S_{\mathcal{A}} = A^{*}\otimes A$.
\end{proof}

\begin{lemma}
Let $\mathcal{A}$ be a quantum operation acting on a density matrix $\rho$. When $\mathcal{A}$ is not unitary, let $\{M_k\}$ be a set of its Kraus operators. Then $S_{\mathcal{A}} = \sum_k M_k^{*} \otimes M_k$.
\label{lem: superoperator not unitary}
\end{lemma}

\begin{proof}
    Since $\mathcal{A}$ is not unitary, $\mathcal{A}(\rho) = \sum_k M_k \rho M_k^{\dagger}$, with $\sum_kM_k^{\dagger}M_k = I$. By linearity and \cref{lem:roth},
\begin{align}
    \vectz{\mathcal{A}(\rho} = \bigl\lvert\sum_k M_k \rho M_k^{\dagger}\bigr\rangle\bigr\rangle = \sum_k\vectz{M_k \rho M_k^{\dagger}}=\sum_k M^{*}_k\otimes M_k\vectz{\rho}.
\end{align}
By \cref{def:superoperator}, $S_{\mathcal{A}} = \sum_k M_k^{*} \otimes M_k$.
\end{proof}


\begin{corollary}
For $E_k$, $P_k$, and $M_k$ in \cref{def:multiequbit}, let $S_{E_k} = E_k^{*}\otimes E_k$. Then
\[
S_{\mathcal{E}} = \sum_{k=0}^{4^n - 1} P_kS_{E_k}.
\]
\label{cor:sup multiqubit}
\end{corollary}

\begin{proof}
By \cref{lem: superoperator not unitary,def:multiequbit}, 
\[
S_{\mathcal{E}} = \sum_{k=0}^{4^n - 1}M_k^{*}\otimes M_k = \sum_{k=0}^{4^n - 1}P_k(E_k^{*}\otimes E_k).
\]
Since $E_k \in \mathcal{B}_n$ is unitary, by \cref{lem: superoperator unitary}, $S_{E_k} = E_k^{*}\otimes E_k$. Therefore, $S_{\mathcal{E}} = \sum_{k=0}^{4^n - 1}P_kS_{E_k}$.

\end{proof}

\begin{lemma}
For $E_j$, $E_k$ in \cref{def:multiequbit}, $t = 2^n$, and $n \in \N^{\neq 0}$, $\Tr[E_jE_k] =t\delta_{jk}$.
    \label{lem: helper}
\end{lemma}

\begin{proof}
    When $j = k$, $E_jE_k = E_j^2 = I$. Since $I$ is an identity matrix of dimension $t \times t$, $\Tr[E_iE_j] = \Tr[I] = t.$ Otherwise, $E_jE_k \notin  \{i^cI;\;c\in \Z_4\}$. By \cref{lem:traceless}, $\Tr[E_jE_k]  = 0$.
\end{proof}

\begin{lemma}
For $E_k$ in \cref{def:multiequbit}, $E_k \neq I$ and $t = 2^n$. Let $E_0 = I$. Then $ \Tr[S_{E_0}] = t^2$. For $k>0$, $\Tr[S_{E_k}] = 0$.
    \label{rmk:trace}
\end{lemma}
\begin{proof}
By \cref{def:superoperator}, $S_{E_0}$ is an identity matrix of dimension $t^2 \times t^2$. By direct computation, $\Tr[S_{E_0}] =\Tr[I] =t^2$. For $k>0$, since $E_k \neq I$, by \cref{lem:traceless,lem: superoperator unitary}, $\Tr[S_{E_k}] = \Tr[E_k^{*}\otimes E_k] = \Tr[E_k]^{*}\Tr[E_k] = 0$.
\end{proof}

\begin{lemma}
    Let $\mathcal{E}_0$ and $\mathcal{E}_1$ be two quantum channels with respective sets of Kraus operators $\{M_k; \; 0 \leq k < n_0, \; n_0 \in \N^{\neq 0}\}$ and  $\{N_\ell; \; 0 \leq \ell < n_1, \; n_1 \in \N^{\neq 0}\}$. $\mathcal{E}_0$ and $\mathcal{E}_1$ have compatible dimensions. $\sum M^\dagger_kM_k = \sum N^\dagger_\ell N_\ell = I$. Then $S_{\mathcal{E}_1 \circ \mathcal{E}_0}=S_{\mathcal{E}_1} S_{\mathcal{E}_0}$.
    \label{lem:trace composition}
\end{lemma}

\begin{proof}
    Let $\mathcal{K} = \{N_\ell  M_k; \; 0 \leq k < n_0, \; 0 \leq \ell < n_1, \; n_0, n_1 \in \N^{\neq 0}\}$. $\mathcal{K}$ satisfies the completeness equation since
\[
\sum_{\ell,k}(N_\ell M_k)^\dagger(N_\ell M_k) = \sum_{\ell,k}(M_k^\dagger N_\ell^\dagger) (N_\ell M_k) = \sum_{k}M_k^\dagger \Bigl(\sum_{\ell}N_\ell^\dagger N_\ell \Bigr)M_k = \sum_{k}M_k^\dagger M_k = I.
\]
Hence, $\mathcal{K}$ is a set of Kraus operators for $\mathcal{E}_1\circ \mathcal{E}_0$. By \cref{lem: superoperator not unitary}, $S_{\mathcal{E}_1\circ \mathcal{E}_0} = \sum_{\ell, k}(N_\ell M_k)^*\otimes(N_\ell M_k)$. Using the linearity and cyclicity of trace with the property that $\Tr[A\otimes B] = \Tr[A]\Tr[B]$,

\begin{align*}
    S_{\mathcal{E}_1\circ \mathcal{E}_0} &= \sum_{\ell, k}(N_\ell M_k)^*\otimes(N_\ell M_k)\\
    &= \sum_{\ell, k}(N^*_\ell M^*_k)\otimes(N_\ell M_k)\\
    &=\sum_{\ell, k} (N^*_\ell \otimes N_\ell)(M^*_k\otimes M_k)\\
    &= \biggl(\sum_\ell (N^*_\ell \otimes N_\ell)\biggr) \biggl(\sum_k(M^*_k\otimes M_k)\biggr) =S_{\mathcal{E}_1}S_{\mathcal{E}_0}.
\end{align*}
\end{proof}

\begin{lemma}
 Let $\mathcal{E}_0$ and $\mathcal{E}_1$ be two quantum channels with respective sets of Kraus operators $\{M_k; \; 0 \leq k < n_0, \; n_0 \in \N^{\neq 0}\}$ and  $\{N_\ell; \; 0 \leq \ell < n_1, \; n_1 \in \N^{\neq 0}\}$. $\mathcal{E}_0$ and $\mathcal{E}_1$ have compatible dimensions. $\sum M^\dagger_kM_k = I_0$ and $\sum N^\dagger_\ell N_\ell = I_1$, $I_0$ and $I_1$ are identity matrices.
    \[
    \Tr[S_{\mathcal{E}_0\otimes \mathcal{E}_1}] = \Tr[S_{\mathcal{E}_0}]\Tr[S_{\mathcal{E}_1}].
    \]
    \label{lem:trace tensor}
\end{lemma}

\begin{proof}
Let $\mathcal{K} = \{M_k \otimes N_\ell; \; 0 \leq k < n_0, \; 0 \leq \ell < n_1, \; n_0, n_1 \in \N^{\neq 0}\}$. $\mathcal{K}$ satisfies the completeness equation since
\[
\sum_{k,\ell}(M_k\otimes N_\ell)^\dagger(M_k\otimes N_\ell) = \sum_{k,\ell}(M_k^\dagger \otimes N_\ell^\dagger) (M_k\otimes N_\ell)= \sum_{k,\ell}(M_k^\dagger M_k) \otimes(N_\ell^\dagger N_\ell) = \Bigl(\sum_k M^\dagger_kM_k\Bigr) \otimes \Bigl(\sum_\ell N_\ell^\dagger N_\ell\Bigr) =I.
\]
Hence, $\mathcal{K}$ is a set of Kraus operators for $\mathcal{E}_0\otimes \mathcal{E}_1$. Using the linearity of trace and the property that $\Tr[A \otimes B] = \Tr[A]\Tr[B]$,
\begin{align*}
    \Tr[S_{\mathcal{E}_0\otimes \mathcal{E}_1}] &= \Tr\Bigl[ \sum_{k,\ell} (M_k \otimes N_\ell)^* \otimes (M_k \otimes N_\ell)\Bigr]\\
    &=\sum_{k,\ell}\Tr\Bigl[(M_k^* \otimes N_\ell^*)\otimes (M_k \otimes N_\ell)\Bigr]\\
    &=\sum_{k,\ell}\Tr[M_k^*]\Tr[N_\ell^*]\Tr[M_k]\Tr[N_\ell]\\
    &=\sum_{k,\ell}\Tr[M_k^*]\Tr[M_k]\Tr[N_\ell^*]\Tr[N_\ell]\\
    &=\sum_{k,\ell}\Tr[M_k^* \otimes M_k]\Tr[N_\ell^* \otimes N_\ell]\\
    &=\sum_k\Tr[M_k^* \otimes M_k] \sum_\ell\Tr[N_\ell^* \otimes N_\ell]\\
    &= \Tr\Bigl[\sum_k M_k^* \otimes M_k\Bigr]\Tr\Bigl[\sum_\ell N_\ell^* \otimes N_\ell]\Bigr]= \Tr[S_{\mathcal{E}_0}]\Tr[S_{\mathcal{E}_1}].
\end{align*}
\end{proof}

\subsubsection{Average Gate Fidelity}
\label{subsubsec:average gate fidelity}
In an ideal world, a quantum algorithm can be precisely implemented by a sequence of carefully selected gates, and the state evolution is described by a unitary transformation. 
Let $\mathcal{U}$ be the linear map describing its transformation and $U$ be its unitary matrix representation. Then $\mathcal{U}(\rho) = U\rho U^{\dagger}$ denotes the ideal evolved state of $\rho$. In reality, quantum operations are prone to errors and $\mathcal{U}$ is approximated by a noisy quantum channel $\mathcal{E}$. 
By \cref{def:channel}, $\mathcal{E}(\rho)$ denotes the actual evolved state of $\rho$,
\begin{align}
    \mathcal{E}(\rho) = \sum_k M_k\rho M^{\dagger}_k.
    \label{eqn: channel}
\end{align}

When $\mathcal{E}$ is "close" to $\mathcal{U}$, the resulting quantum evolution is closely aligned with the desired algorithm. Let $\sigma$ be the density operator describing the quantum state of another physical system. Recall that in~\cite{AverageGateFidelity}, the \textbf{state fidelity} of $\rho$ and $\sigma$ is defined as $F(\rho,\sigma)$, 


\begin{equation}
    F(\rho,\sigma) = \Biggl(Tr\biggl[\sqrt{\sqrt{\rho}\;\sigma\sqrt{\rho}}\biggr]\Biggr)^2.
    \label{eqn: state fidelity}
\end{equation}

Let $\mathcal{E}(\rho)$ denote the final state of $\rho$ after the action of $\mathcal{E}$. The \textbf{gate fidelity}, $F_{\mathcal{U},\mathcal{E}}(\rho)$, is defined as state fidelity of $\mathcal{E}(\rho)$ and $\mathcal{U}(\rho)$,
\begin{equation}
    F_{\mathcal{U},\mathcal{E}}(\rho) = \Biggl(Tr\biggl[\sqrt{\sqrt{\mathcal{U}(\rho)}\;\mathcal{E}(\rho)\sqrt{\mathcal{U}(\rho)}}\biggr]\Biggr)^2.
    \label{eqn: gate fidelity}
\end{equation}

When the input state is a pure state, $\rho =\ket{\psi}\bra{\psi}$ where $\ket{\psi}$ is the state vector. Using the cyclic property of trace, \cref{eqn: gate fidelity} is reduced to \cref{eqn: reduced gate fidelity}. 
\begin{equation}
    F_{\mathcal{U},\mathcal{E}}(\ket{\psi}\bra{\psi})=\biggl(Tr\Bigl[\sqrt{U\ket{\psi}\bra{\psi}U^{\dagger}\mathcal{E}(\ket{\psi}\bra{\psi})}\Bigr]\biggr)^2 =\bra{\psi}U^{\dagger}\mathcal{E}(\ket{\psi}\bra{\psi})U\ket{\psi}.
    \label{eqn: reduced gate fidelity}
\end{equation}

In \cref{eqn: average gate fidelity original}, we obtain the average gate fidelity by integrating over all pure input states~\cite{bowdrey2002fidelity}. It provides a concise measure of error independent of the input state. This is more accurate than merely counting the number of CNOT gates~\cite{griendli2022,gheorghiu2022reducing,nash2020quantum,kissinger2020cnot} or calculating the sum of CNOT gate error rates~\cite{wu2023optimization}. Hence, we use it as a cost function to gauge the proximity between the dynamic evolution and the desired evolution. 

\begin{align}
    F_{avg}(\mathcal{E}, \mathcal{U})  = \int d\psi F_{\mathcal{U},\mathcal{E}}(\ket{\psi}\bra{\psi}) = \int d\psi \bra{\psi}U^{\dagger}\mathcal{E}(\ket{\psi}\bra{\psi})U\ket{\psi}. 
    \label{eqn: average gate fidelity original} 
\end{align} 

Let $p$ be the error rate of a noisy quantum channel $\mathcal{E}$ whose ideal operation is $\mathcal{U}$. Then
\begin{equation}
    p = 1- F_{avg}(\mathcal{E}, \mathcal{U}).
    \label{eqn: gate error}
\end{equation} 

Next, we show that the average gate fidelity assesses how closely $\mathcal{E}$ approximates $\mathcal{U}$ independent of the input states. We start by simplifying \cref{eqn: average gate fidelity original} and let $\mathcal{E}' = \mathcal{U}^{-1}\circ\mathcal{E}$. Accordingly, the modified set of Kraus operators is $\{M'_k; \;\; M'_k = U^\dagger M_k\}$. To compute $F_{avg}(\mathcal{E}, \mathcal{U})$, it is equivalent to calculate $F_{avg}(\mathcal{E}', \mathcal{I})$.

\begin{lemma}
    \[
    F_{avg}(\mathcal{E}, \mathcal{U}) = F_{avg}(\mathcal{E}', \mathcal{I}).
    \]
    \label{lem: simplified}
\end{lemma}
\begin{proof}
Note that $\sum M_k^{\dagger}M_k = I$ and $U$ is unitary. Then $\{M'_k; \;\; M'_k = U^\dagger M_k\}$ is a valid set of Kraus operators, since 
\begin{align*}
    \sum M_k{'}^{\dagger}M'_k = \sum (U^\dagger M_k)^{\dagger}(U^\dagger M_k) = \sum M_k^{\dagger}(UU^{\dagger})M_k = \sum M_k^{\dagger}M_k = I.
\end{align*}
Based on \cref{eqn: channel,eqn: average gate fidelity original}, we have

    \begin{align*}
        F_{avg}(\mathcal{E}, \mathcal{U}) &= \int d\psi \bra{\psi}U^{\dagger}\Bigl(\sum M_k\ket{\psi}\bra{\psi}M_k^{\dagger}\Bigr)U\ket{\psi} \nonumber\\
        &= \int d\psi \bra{\psi}\Bigl(\sum U^{\dagger}M_k\ket{\psi}\bra{\psi}M_k^{\dagger}U\Bigr)\ket{\psi}\nonumber\\
        &= \int d\psi \bra{\psi}\Bigl(\sum M'_k\ket{\psi}\bra{\psi}M_k{'}^{\dagger}\Bigr)\ket{\psi}\nonumber\\
        &= \int d\psi \bra{\psi}\mathcal{E}'(\ket{\psi}\bra{\psi})\ket{\psi} =  F_{avg}(\mathcal{E}', \mathcal{I}).
    \end{align*}
\end{proof}

For brevity, we write $F_{avg}(\mathcal{E}')$ for $F_{avg}(\mathcal{E}', \mathcal{I})$. \cite{horodecki1998general,nielsen2002simple} provide an alternative expression for the average gate fidelity, as shown in \cref{eqn: average gate fidelity}. $F_{pro}(\mathcal{E}', \mathcal{I})$ is called the \textbf{process fidelity} (a.k.a., the \textbf{entanglement fidelity}) and it gauges the overlap between $\rho$ before and after the application of $\mathcal{E}'$. It describes how well the quantum information in a system and the entanglement with other systems are preserved~\cite{schumacher1996sending}.

     \begin{align}
F_{avg}(\mathcal{E}') &= \frac{t F_{pro}(\mathcal{E}')+1}{t+1}.
    \label{eqn: average gate fidelity}
    \end{align}

\begin{lemma}
  \lemsim
\label{lem: process fidelity with super operator}
\end{lemma}

Proof details could be found in \cref{sec: proof of claim}. Combining \cref{lem: simplified,lem: process fidelity with super operator}, we have
\begin{equation}
    F_{avg}(\mathcal{E}') = \frac{t F_{pro}(\mathcal{E}')+1}{t+1} = \frac{\Tr[S_{\mathcal{E}'}]}{t(t+1)} + \frac{1}{t+1}.
    \label{eqn: average fidelity simplified}
\end{equation}

Finally, by using \cref{eqn: average fidelity simplified,eqn: gate error}, we derive an equality that will be used in \cref{subsec:consecutive error channel}.

\begin{align}
    F_{avg}(\mathcal{E}) = 1-p = \frac{\Tr[S_{\mathcal{E}}]}{t(t + 1)}+ \frac{1}{t+1}.
\label{eqn:trace and epsilon}
\end{align}
\section{Quantify the Reliability of a Noisy CNOT Circuit}
\label{sec:ApproximatedCostFunction}
To gauge the quality of a noisy CNOT circuit, we define a cost function by approximating

\begin{definition}
    Let $\mathcal{E}$ be an error channel. The error probability of $\mathcal{E}$ is $\Prob(\mathcal{E})$,
    \[
    \Prob(\mathcal{E}) = 1 - F_{avg}(\mathcal{E}).
    \]
    \label{def: error prob and agf}
\end{definition}

Since the size of the superoperator in \cref{eqn:trace and epsilon} scales exponentially with the number of qubits in the system, exactly computing $F_{avg}(\mathcal{E})$ demands a substantial amount of computational resource. 
In light of this, we model a noisy CNOT circuit using both the \textbf{parallel error channel} and the \textbf{consecutive error channel}. Based on the NISQ specifics (the number of physical qubits, and the range of CNOT gate error rates), we find a tight approximation for the average gate fidelity and use it as our cost function. The higher it is, the worse the performance of a noise-aware CNOT circuit routing algorithm.

\begin{definition}
    Let $\mathbf{C}$ be an $n$-qubit quantum circuit with $m$ noisy gates. Numbering each gate from left to right, $\Prob(\mathcal{E}_i)$ is the error probability of the $i$-th noisy channel $\mathcal{E}_i$. Let $\Cost(\mathbf{C})$ be a cost function of $\mathbf{C}$.
    \[
    \Cost(\mathbf{C}) = 1 - \prod^{m-1}_{i=0}(1-\Prob(\mathcal{E}_i)).
    \]
    \label{def:cost function}
\end{definition}

In \cref{subsec:parallel error channel,subsec:consecutive error channel}, we characterize a noisy parallel and consecutive error channels using the superoperator representation and calculate their respective average gate fidelity. In \cref{subsec:costfunction}, we combine both to approximate the average gate fidelity of a noisy CNOT circuit and formally define its cost function based on \cref{def:cost function}. 

\subsection{Parallel Error Channel}
\label{subsec:parallel error channel}
A parallel error channel $\mathcal{E}_q$ is vertically composed of channels with arbitrary input and output dimensions. Let $t$ be its input dimension. As shown in \cref{fig:ParallelE}, $\mathcal{E}_0$ and $\mathcal{E}_1$ are two error channels followed by two ideal operations $\mathcal{U}_0$ and $\mathcal{U}_1$. $U_0$ and $U_1$ are their respective unitary matrix representations. Let $p_0$ and $p_1$ be the error rates of $\mathcal{E}_0$ and $\mathcal{E}_1$. Let $d_0$ and $d_1$ be the respectively input dimensions and $k \in \{0,1\}$. $t = d_0d_1$. By \cref{eqn:trace and epsilon}, we can express $\Tr[S_{\mathcal{E}_k}]$ in terms of $d_k$ and $p_k$.
\begin{align}
    \frac{\Tr[S_{\mathcal{E}_k}]}{d_k^2} = 1-\frac{d_k+1}{d_k}p_k. \label{eqn:trace and epsilon intermediate}
\end{align}

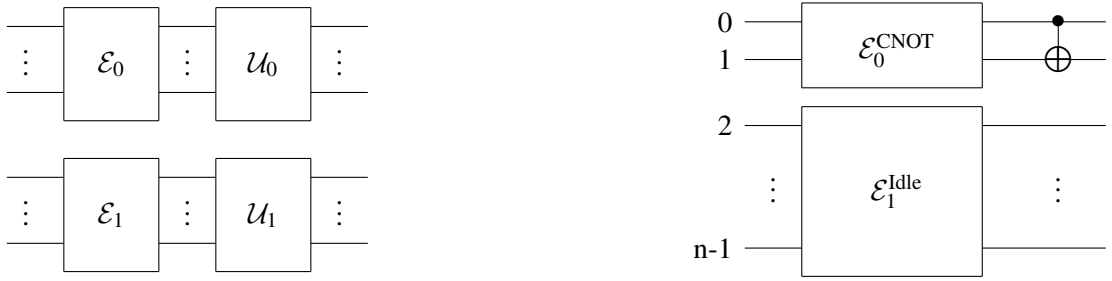
\begin{figure}[H]
\begin{subfigure}{.5\textwidth}
  \centering
  \scalebox{1}{\begin{tikzpicture}
	\begin{pgfonlayer}{nodelayer}
		\node [style=none] (0) at (-6.5, 3.5) {};
		\node [style=none] (1) at (-4, 3.5) {};
		\node [style=none] (2) at (-6.5, 0.5) {};
		\node [style=none] (3) at (-4, 0.5) {};
		\node [style=none] (4) at (-6.5, 3) {};
		\node [style=none] (5) at (-6.5, 1.25) {};
		\node [style=none] (6) at (-4, 3) {};
		\node [style=none] (7) at (-4, 1.25) {};
		\node [style=none] (8) at (-2.5, 3) {};
		\node [style=none] (9) at (-2.5, 1.25) {};
		\node [style=none] (10) at (-8, 3) {};
		\node [style=none] (11) at (-8, 1.25) {};
		\node [style=none] (12) at (-5.25, 2) {$\mathcal{E}_0$};
		\node [style=none] (13) at (-2.5, 3.5) {};
		\node [style=none] (14) at (0, 3.5) {};
		\node [style=none] (15) at (-2.5, 0.5) {};
		\node [style=none] (16) at (0, 0.5) {};
		\node [style=none] (17) at (-2.5, 3) {};
		\node [style=none] (18) at (-2.5, 1.25) {};
		\node [style=none] (19) at (0, 3) {};
		\node [style=none] (20) at (0, 1.25) {};
		\node [style=none] (21) at (1.5, 3) {};
		\node [style=none] (22) at (1.5, 1.25) {};
		\node [style=none] (23) at (-1.25, 2) {$\mathcal{U}_0$};
		\node [style=none] (24) at (-6.5, -0.5) {};
		\node [style=none] (25) at (-4, -0.5) {};
		\node [style=none] (26) at (-6.5, -3.5) {};
		\node [style=none] (27) at (-4, -3.5) {};
		\node [style=none] (28) at (-6.5, -1) {};
		\node [style=none] (29) at (-6.5, -2.75) {};
		\node [style=none] (30) at (-4, -1) {};
		\node [style=none] (31) at (-4, -2.75) {};
		\node [style=none] (32) at (-2.5, -1) {};
		\node [style=none] (33) at (-2.5, -2.75) {};
		\node [style=none] (34) at (-8, -1) {};
		\node [style=none] (35) at (-8, -2.75) {};
		\node [style=none] (36) at (-5.25, -2) {$\mathcal{E}_1$};
		\node [style=none] (37) at (-2.5, -0.5) {};
		\node [style=none] (38) at (0, -0.5) {};
		\node [style=none] (39) at (-2.5, -3.5) {};
		\node [style=none] (40) at (0, -3.5) {};
		\node [style=none] (41) at (-2.5, -1) {};
		\node [style=none] (42) at (-2.5, -2.75) {};
		\node [style=none] (43) at (0, -1) {};
		\node [style=none] (44) at (0, -2.75) {};
		\node [style=none] (45) at (1.5, -1) {};
		\node [style=none] (46) at (1.5, -2.75) {};
		\node [style=none] (47) at (-1.25, -2) {$\mathcal{U}_1$};
		\node [style=none] (48) at (-7.5, 2.25) {$\vdots$};
		\node [style=none] (49) at (-3.25, 2.25) {$\vdots$};
		\node [style=none] (50) at (0.75, 2.25) {$\vdots$};
		\node [style=none] (51) at (-7.5, -1.75) {$\vdots$};
		\node [style=none] (52) at (-3.25, -1.75) {$\vdots$};
		\node [style=none] (53) at (0.75, -1.75) {$\vdots$};
	\end{pgfonlayer}
	\begin{pgfonlayer}{edgelayer}
		\draw (0.center) to (1.center);
		\draw (0.center) to (2.center);
		\draw (2.center) to (3.center);
		\draw (3.center) to (1.center);
		\draw (4.center) to (10.center);
		\draw (5.center) to (11.center);
		\draw (6.center) to (8.center);
		\draw (7.center) to (9.center);
		\draw (13.center) to (14.center);
		\draw (13.center) to (15.center);
		\draw (15.center) to (16.center);
		\draw (16.center) to (14.center);
		\draw (19.center) to (21.center);
		\draw (20.center) to (22.center);
		\draw (24.center) to (25.center);
		\draw (24.center) to (26.center);
		\draw (26.center) to (27.center);
		\draw (27.center) to (25.center);
		\draw (28.center) to (34.center);
		\draw (29.center) to (35.center);
		\draw (30.center) to (32.center);
		\draw (31.center) to (33.center);
		\draw (37.center) to (38.center);
		\draw (37.center) to (39.center);
		\draw (39.center) to (40.center);
		\draw (40.center) to (38.center);
		\draw (43.center) to (45.center);
		\draw (44.center) to (46.center);
	\end{pgfonlayer}
\end{tikzpicture}}
\vspace{.1 cm}
  \caption{The error channel $\mathcal{E}_q$ is vertically composed of two error channels $\mathcal{E}_0$ and $\mathcal{E}_1$, followed the respective ideal operations $\mathcal{U}_0$ and $\mathcal{U}_1$. $\mathcal{E}_0$ has an error rate $p_0$ and $\mathcal{E}_1$ has an error rate $p_1$. Their input dimensions are $d_0$ and $d_1$.}
  \label{fig:ParallelE}
\end{subfigure}%
\qquad
\begin{subfigure}{.5\textwidth}
  \centering
  \scalebox{1}{\begin{tikzpicture}
	\begin{pgfonlayer}{nodelayer}
		\node [style=none] (0) at (2, 2.25) {};
		\node [style=none] (1) at (6.75, 2.25) {};
		\node [style=none] (2) at (2, 0) {};
		\node [style=none] (3) at (6.75, 0) {};
		\node [style=none] (4) at (2, 1.75) {};
		\node [style=none] (5) at (2, 0.75) {};
		\node [style=none] (6) at (6.75, 1.75) {};
		\node [style=none] (7) at (6.75, 0.75) {};
		\node [style=none, label={left: 0}] (8) at (0.5, 1.75) {};
		\node [style=none, label={left: 1}] (9) at (0.5, 0.75) {};
		\node [style=none] (10) at (4.5, 1) {$\mathcal{E}_0^{\text{CNOT}}$};
		\node [style=none] (11) at (10, 1.75) {};
		\node [style=none] (12) at (10, 0.75) {};
		\node [style=none] (13) at (8.75, 0.75) {$\bigoplus$};
		\node [style=none] (14) at (8.75, 1.75) {};
		\node [style=none] (15) at (2, -0.5) {};
		\node [style=none] (16) at (6.75, -0.5) {};
		\node [style=none] (17) at (2, -5) {};
		\node [style=none] (18) at (6.75, -5) {};
		\node [style=none] (19) at (2, -1) {};
		\node [style=none] (20) at (2, -4.25) {};
		\node [style=none] (21) at (6.75, -1) {};
		\node [style=none] (22) at (6.75, -4.25) {};
		\node [style=none, label={left: 2}] (23) at (0.5, -1) {};
		\node [style=none, label={left: n-1}] (24) at (0.5, -4.25) {};
		\node [style=none] (25) at (4.5, -2.75) {$\mathcal{E}^\text{Idle}_1$};
		\node [style=none] (26) at (10, -1) {};
		\node [style=none] (27) at (10, -4.25) {};
		\node [style=none] (30) at (8.75, 1.75) {$\bullet$};
		\node [style=none] (31) at (8.75, -2.5) {$\vdots$};
		\node [style=none] (32) at (1.25, -2.5) {$\vdots$};
	\end{pgfonlayer}
	\begin{pgfonlayer}{edgelayer}
		\draw (0.center) to (1.center);
		\draw (0.center) to (2.center);
		\draw (2.center) to (3.center);
		\draw (3.center) to (1.center);
		\draw (4.center) to (8.center);
		\draw (5.center) to (9.center);
		\draw (6.center) to (11.center);
		\draw (7.center) to (12.center);
		\draw (14.center) to (13.center);
		\draw (15.center) to (16.center);
		\draw (15.center) to (17.center);
		\draw (17.center) to (18.center);
		\draw (18.center) to (16.center);
		\draw (19.center) to (23.center);
		\draw (20.center) to (24.center);
		\draw (21.center) to (26.center);
		\draw (22.center) to (27.center);
	\end{pgfonlayer}
\end{tikzpicture}}
      \caption{Without loss of generality, suppose that $\mathcal{E}_q$ is composed of two parallel channels, $\mathcal{E}^{\text{CNOT}}_0$ and $\mathcal{E}^{\text{Idle}}_1$. $\mathcal{E}^{\text{CNOT}}_0$ is the error channel of a noisy CNOT gate. Its error rate is $p_0$. $\mathcal{E}^{\text{Idle}}_1$ is the error channel of $n-2$ idle qubits. Its error rate is $p_1$.}
    \label{fig:ParallelECNOT}
\end{subfigure}
\caption{The parallel error channels.}
\end{figure}

\begin{lemma}
    \lempara
    \label{lem:average fidelity parallel}
\end{lemma}

\begin{proof}
    Applying \cref{lem:trace tensor} with \cref{eqn:trace and epsilon,eqn:trace and epsilon intermediate}, we have

\begin{align}
    F_{avg}(\mathcal{E}_q) & =F_{avg}(\mathcal{E}_0 \otimes \mathcal{E}_1)= \frac{\Tr[S_{\mathcal{E}_0\otimes \mathcal{E}_1}]}{(d_0d_1)^2}\times \frac{d_0d_1}{d_0d_1+1} + \frac{1}{d_0d_1+1}\nonumber\\
     & =\frac{\Tr[S_{\mathcal{E}_0}]}{d_0^2}\frac{\Tr[S_{\mathcal{E}_1}]}{d_1^2}\times \frac{d_0d_1}{d_0d_1+1} + \frac{1}{d_0d_1+1}\nonumber\\
    & = 1-p_0 -p_1 + p_0p_1 + \frac{(1-d_1)p_0 + (1-d_0)p_1 + (d_0+d_1)p_0p_1}{d_0d_1+1}
     \label{eqn: approximated form of parallel channels}
\end{align}

Technical details can be found in \cref{sec: proof of claim}.
\end{proof}

\subsection{Consecutive Error Channel}
\label{subsec:consecutive error channel}
A consecutive error channel $\mathcal{E}_c$ is horizontally composed of channels with compatible input and output dimensions. Let $t$ be its input dimension. In \cref{fig:consecutiveE}, $\mathcal{E}_0$ and $\mathcal{E}_1$ are two error channels followed by two ideal operations $\mathcal{U}_0$ and $\mathcal{U}_1$. $U_0$ and $U_1$ are their respective unitary matrix representations. Let $p_0$ and $p_1$ be the error rates of $\mathcal{E}_0$ and $\mathcal{E}_1$. Let $d_0$ and $d_1$ be the respective input dimensions and $k \in \Z_2$. $t = d_0 = d_1$. 
\[
\mathcal{E}_c = \mathcal{U}_1 \circ \mathcal{E}_1\circ\mathcal{U}_0 \circ \mathcal{E}_0.
\]

Here $\circ$ denotes the composition of linear maps.
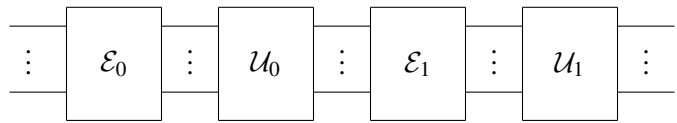
\begin{figure}[H]
  \centering
  \scalebox{1}{\begin{tikzpicture}
	\begin{pgfonlayer}{nodelayer}
		\node [style=none] (0) at (2, 3.5) {};
		\node [style=none] (1) at (4.5, 3.5) {};
		\node [style=none] (2) at (2, 0.5) {};
		\node [style=none] (3) at (4.5, 0.5) {};
		\node [style=none] (4) at (2, 3) {};
		\node [style=none] (5) at (2, 1.25) {};
		\node [style=none] (6) at (4.5, 3) {};
		\node [style=none] (7) at (4.5, 1.25) {};
		\node [style=none] (8) at (6, 3) {};
		\node [style=none] (9) at (6, 1.25) {};
		\node [style=none] (10) at (0.5, 3) {};
		\node [style=none] (11) at (0.5, 1.25) {};
		\node [style=none] (12) at (3.25, 2) {$\mathcal{E}_0$};
		\node [style=none] (13) at (6, 3.5) {};
		\node [style=none] (14) at (8.5, 3.5) {};
		\node [style=none] (15) at (6, 0.5) {};
		\node [style=none] (16) at (8.5, 0.5) {};
		\node [style=none] (17) at (6, 3) {};
		\node [style=none] (18) at (6, 1.25) {};
		\node [style=none] (19) at (8.5, 3) {};
		\node [style=none] (20) at (8.5, 1.25) {};
		\node [style=none] (21) at (10, 3) {};
		\node [style=none] (22) at (10, 1.25) {};
		\node [style=none] (23) at (7.25, 2) {$\mathcal{U}_0$};
		\node [style=none] (24) at (1, 2.25) {$\vdots$};
		\node [style=none] (25) at (5.25, 2.25) {$\vdots$};
		\node [style=none] (26) at (9.25, 2.25) {$\vdots$};
		\node [style=none] (27) at (10, 3.5) {};
		\node [style=none] (28) at (12.5, 3.5) {};
		\node [style=none] (29) at (10, 0.5) {};
		\node [style=none] (30) at (12.5, 0.5) {};
		\node [style=none] (31) at (10, 3) {};
		\node [style=none] (32) at (10, 1.25) {};
		\node [style=none] (33) at (12.5, 3) {};
		\node [style=none] (34) at (12.5, 1.25) {};
		\node [style=none] (35) at (14, 3) {};
		\node [style=none] (36) at (14, 1.25) {};
		\node [style=none] (37) at (11.25, 2) {$\mathcal{E}_1$};
		\node [style=none] (38) at (14, 3.5) {};
		\node [style=none] (39) at (16.5, 3.5) {};
		\node [style=none] (40) at (14, 0.5) {};
		\node [style=none] (41) at (16.5, 0.5) {};
		\node [style=none] (42) at (14, 3) {};
		\node [style=none] (43) at (14, 1.25) {};
		\node [style=none] (44) at (16.5, 3) {};
		\node [style=none] (45) at (16.5, 1.25) {};
		\node [style=none] (46) at (18, 3) {};
		\node [style=none] (47) at (18, 1.25) {};
		\node [style=none] (48) at (15.25, 2) {$\mathcal{U}_1$};
		\node [style=none] (49) at (13.25, 2.25) {$\vdots$};
		\node [style=none] (50) at (17.25, 2.25) {$\vdots$};
	\end{pgfonlayer}
	\begin{pgfonlayer}{edgelayer}
		\draw (0.center) to (1.center);
		\draw (0.center) to (2.center);
		\draw (2.center) to (3.center);
		\draw (3.center) to (1.center);
		\draw (4.center) to (10.center);
		\draw (5.center) to (11.center);
		\draw (6.center) to (8.center);
		\draw (7.center) to (9.center);
		\draw (13.center) to (14.center);
		\draw (13.center) to (15.center);
		\draw (15.center) to (16.center);
		\draw (16.center) to (14.center);
		\draw (19.center) to (21.center);
		\draw (20.center) to (22.center);
		\draw (27.center) to (28.center);
		\draw (27.center) to (29.center);
		\draw (29.center) to (30.center);
		\draw (30.center) to (28.center);
		\draw (33.center) to (35.center);
		\draw (34.center) to (36.center);
		\draw (38.center) to (39.center);
		\draw (38.center) to (40.center);
		\draw (40.center) to (41.center);
		\draw (41.center) to (39.center);
		\draw (44.center) to (46.center);
		\draw (45.center) to (47.center);
	\end{pgfonlayer}
\end{tikzpicture}}
  \caption{The error channel $\mathcal{E}_c$ is horizontally composed of two error channels $\mathcal{E}_0$ and $\mathcal{E}_1$, followed the respective ideal operations $\mathcal{U}_0$ and $\mathcal{U}_1$.}
  \label{fig:consecutiveE}
\end{figure}

\begin{lemma}
    When two noisy Clifford channels are composed horizontally, the identity in \cref{fig:ConsecutiveE2} holds.
    \label{lem:push}
\end{lemma}

\begin{figure}[H]
    \centering
    \scalebox{1}{\begin{tikzpicture}
	\begin{pgfonlayer}{nodelayer}
		\node [style=none] (51) at (-17.5, 3.5) {};
		\node [style=none] (52) at (-15.75, 3.5) {};
		\node [style=none] (53) at (-17.5, 0.5) {};
		\node [style=none] (54) at (-15.75, 0.5) {};
		\node [style=none] (55) at (-17.5, 3) {};
		\node [style=none] (56) at (-17.5, 1.25) {};
		\node [style=none] (57) at (-15.75, 3) {};
		\node [style=none] (58) at (-15.75, 1.25) {};
		\node [style=none] (59) at (-14.25, 3) {};
		\node [style=none] (60) at (-14.25, 1.25) {};
		\node [style=none] (61) at (-19, 3) {};
		\node [style=none] (62) at (-19, 1.25) {};
		\node [style=none] (63) at (-16.5, 2) {$\mathcal{E}_0$};
		\node [style=none] (64) at (-14.25, 3.5) {};
		\node [style=none] (65) at (-12.5, 3.5) {};
		\node [style=none] (66) at (-14.25, 0.5) {};
		\node [style=none] (67) at (-12.5, 0.5) {};
		\node [style=none] (68) at (-14.25, 3) {};
		\node [style=none] (69) at (-14.25, 1.25) {};
		\node [style=none] (70) at (-12.5, 3) {};
		\node [style=none] (71) at (-12.5, 1.25) {};
		\node [style=none] (72) at (-11, 3) {};
		\node [style=none] (73) at (-11, 1.25) {};
		\node [style=none] (74) at (-13.25, 2) {$\mathcal{U}_0$};
		\node [style=none] (75) at (-18.5, 2.25) {$\vdots$};
		\node [style=none] (76) at (-15, 2.25) {$\vdots$};
		\node [style=none] (77) at (-11.75, 2.25) {$\vdots$};
		\node [style=none] (78) at (-11, 3.5) {};
		\node [style=none] (79) at (-9.25, 3.5) {};
		\node [style=none] (80) at (-11, 0.5) {};
		\node [style=none] (81) at (-9.25, 0.5) {};
		\node [style=none] (82) at (-11, 3) {};
		\node [style=none] (83) at (-11, 1.25) {};
		\node [style=none] (84) at (-9.25, 3) {};
		\node [style=none] (85) at (-9.25, 1.25) {};
		\node [style=none] (86) at (-7.75, 3) {};
		\node [style=none] (87) at (-7.75, 1.25) {};
		\node [style=none] (88) at (-10, 2) {$\mathcal{E}_1$};
		\node [style=none] (89) at (-7.75, 3.5) {};
		\node [style=none] (90) at (-6, 3.5) {};
		\node [style=none] (91) at (-7.75, 0.5) {};
		\node [style=none] (92) at (-6, 0.5) {};
		\node [style=none] (93) at (-7.75, 3) {};
		\node [style=none] (94) at (-7.75, 1.25) {};
		\node [style=none] (95) at (-6, 3) {};
		\node [style=none] (96) at (-6, 1.25) {};
		\node [style=none] (97) at (-4.5, 3) {};
		\node [style=none] (98) at (-4.5, 1.25) {};
		\node [style=none] (99) at (-6.75, 2) {$\mathcal{U}_1$};
		\node [style=none] (100) at (-8.5, 2.25) {$\vdots$};
		\node [style=none] (101) at (-5.25, 2.25) {$\vdots$};
		\node [style=none] (102) at (-3.5, 2) {$=$};
		\node [style=none] (103) at (-1, 3.5) {};
		\node [style=none] (104) at (0.75, 3.5) {};
		\node [style=none] (105) at (-1, 0.5) {};
		\node [style=none] (106) at (0.75, 0.5) {};
		\node [style=none] (107) at (-1, 3) {};
		\node [style=none] (108) at (-1, 1.25) {};
		\node [style=none] (109) at (0.75, 3) {};
		\node [style=none] (110) at (0.75, 1.25) {};
		\node [style=none] (111) at (2.25, 3) {};
		\node [style=none] (112) at (2.25, 1.25) {};
		\node [style=none] (113) at (-2.5, 3) {};
		\node [style=none] (114) at (-2.5, 1.25) {};
		\node [style=none] (115) at (0, 2) {$\mathcal{U}_0$};
		\node [style=none] (116) at (2.25, 3.5) {};
		\node [style=none] (117) at (4, 3.5) {};
		\node [style=none] (118) at (2.25, 0.5) {};
		\node [style=none] (119) at (4, 0.5) {};
		\node [style=none] (120) at (2.25, 3) {};
		\node [style=none] (121) at (2.25, 1.25) {};
		\node [style=none] (122) at (4, 3) {};
		\node [style=none] (123) at (4, 1.25) {};
		\node [style=none] (124) at (5.5, 3) {};
		\node [style=none] (125) at (5.5, 1.25) {};
		\node [style=none] (126) at (3.25, 2) {$\mathcal{E}'_0$};
		\node [style=none] (127) at (-2, 2.25) {$\vdots$};
		\node [style=none] (128) at (1.5, 2.25) {$\vdots$};
		\node [style=none] (129) at (4.75, 2.25) {$\vdots$};
		\node [style=none] (130) at (5.5, 3.5) {};
		\node [style=none] (131) at (7.25, 3.5) {};
		\node [style=none] (132) at (5.5, 0.5) {};
		\node [style=none] (133) at (7.25, 0.5) {};
		\node [style=none] (134) at (5.5, 3) {};
		\node [style=none] (135) at (5.5, 1.25) {};
		\node [style=none] (136) at (7.25, 3) {};
		\node [style=none] (137) at (7.25, 1.25) {};
		\node [style=none] (138) at (8.75, 3) {};
		\node [style=none] (139) at (8.75, 1.25) {};
		\node [style=none] (140) at (6.5, 2) {$\mathcal{E}_1$};
		\node [style=none] (141) at (8.75, 3.5) {};
		\node [style=none] (142) at (10.5, 3.5) {};
		\node [style=none] (143) at (8.75, 0.5) {};
		\node [style=none] (144) at (10.5, 0.5) {};
		\node [style=none] (145) at (8.75, 3) {};
		\node [style=none] (146) at (8.75, 1.25) {};
		\node [style=none] (147) at (10.5, 3) {};
		\node [style=none] (148) at (10.5, 1.25) {};
		\node [style=none] (149) at (12, 3) {};
		\node [style=none] (150) at (12, 1.25) {};
		\node [style=none] (151) at (9.75, 2) {$\mathcal{U}_1$};
		\node [style=none] (152) at (8, 2.25) {$\vdots$};
		\node [style=none] (153) at (11.25, 2.25) {$\vdots$};
	\end{pgfonlayer}
	\begin{pgfonlayer}{edgelayer}
		\draw (51.center) to (52.center);
		\draw (51.center) to (53.center);
		\draw (53.center) to (54.center);
		\draw (54.center) to (52.center);
		\draw (55.center) to (61.center);
		\draw (56.center) to (62.center);
		\draw (57.center) to (59.center);
		\draw (58.center) to (60.center);
		\draw (64.center) to (65.center);
		\draw (64.center) to (66.center);
		\draw (66.center) to (67.center);
		\draw (67.center) to (65.center);
		\draw (70.center) to (72.center);
		\draw (71.center) to (73.center);
		\draw (78.center) to (79.center);
		\draw (78.center) to (80.center);
		\draw (80.center) to (81.center);
		\draw (81.center) to (79.center);
		\draw (84.center) to (86.center);
		\draw (85.center) to (87.center);
		\draw (89.center) to (90.center);
		\draw (89.center) to (91.center);
		\draw (91.center) to (92.center);
		\draw (92.center) to (90.center);
		\draw (95.center) to (97.center);
		\draw (96.center) to (98.center);
		\draw (103.center) to (104.center);
		\draw (103.center) to (105.center);
		\draw (105.center) to (106.center);
		\draw (106.center) to (104.center);
		\draw (107.center) to (113.center);
		\draw (108.center) to (114.center);
		\draw (109.center) to (111.center);
		\draw (110.center) to (112.center);
		\draw (116.center) to (117.center);
		\draw (116.center) to (118.center);
		\draw (118.center) to (119.center);
		\draw (119.center) to (117.center);
		\draw (122.center) to (124.center);
		\draw (123.center) to (125.center);
		\draw (130.center) to (131.center);
		\draw (130.center) to (132.center);
		\draw (132.center) to (133.center);
		\draw (133.center) to (131.center);
		\draw (136.center) to (138.center);
		\draw (137.center) to (139.center);
		\draw (141.center) to (142.center);
		\draw (141.center) to (143.center);
		\draw (143.center) to (144.center);
		\draw (144.center) to (142.center);
		\draw (147.center) to (149.center);
		\draw (148.center) to (150.center);
	\end{pgfonlayer}
\end{tikzpicture}}
  \caption{$\mathcal{E}_c$ is composed of two noisy Clifford gates. Let $n$ be the number of qubits. $p_0$ and $p_1$ are the error rates of $\mathcal{E}_0$ and $\mathcal{E}_1$. Their input dimensions are identical, $d_0=d_1 = 2^n$.  $U_0$ and $U_1$ are the respective unitary matrix representations of $\mathcal{U}_0$ and $\mathcal{U}_1$. $U_0, \; U_1 \in \mathcal{C}_n$.}
  \label{fig:ConsecutiveE2}
\end{figure}
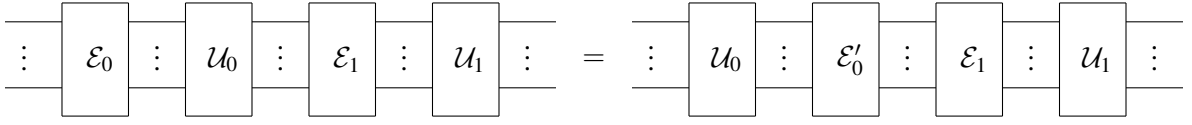
\begin{proof}
Let $\rho$ be the density matrix describing the initial state of the quantum system in \cref{fig:ConsecutiveE2}. It is sufficient to show that
\[
\mathcal{U}_0\circ \mathcal{E}_0(\rho) = \mathcal{E}'_0\circ \mathcal{U}_0(\rho).
\]
By \cref{def:multiequbit}, let $\{M_k\}$ be a set of Kraus operators of $\mathcal{E}_0$ such that $E_k \in \mathcal{B}_n$,
\[
M_k = \sqrt{P_k}E_k, \qquad \sum M_k^{\dagger}M_k = I, \qquad \sum P_k = 1, \qquad 0 \leq P_k \leq 1.
\]

Since $U_0 \in \mathcal{C}_n$ and $E_k \in \mathcal{P}_n$, $U_0 E_k U_0^\dagger = E'_k, \;\; E'_k \in \mathcal{P}_n$. Let $M_K' = \sqrt{P_k}E_k'$. It follows that
\begin{equation}
    U_0 E_k = E'_kU_0, \qquad U_0 M_k = M'_kU_0
    \label{eqv: Clifford conjugate}
\end{equation}

 Since $\sum P_kE_k^\dagger E_k = I$ and $U_0$ is unitary, to show that $\forall k, \; M'_k = \sqrt{P_k}E'_k$ are valid Kraus operators, we have
    
    \[
    \sum {M'}_k^{\dagger}M'_k = \sum P_i {E'}_k^{\dagger} E'_k = \sum P_i\bigl(U_0 E_k^\dagger U_0^\dagger \bigr) \bigl(U_0 E_k U_0^\dagger \bigr) = \sum P_i\bigl(U_0 E_k^\dagger E_k U_0^\dagger \bigr) = U_0 \Bigl( \sum P_k E_k^\dagger E_k\Bigr) U_0^\dagger = U_0 I U_0^\dagger = I.
    \]
Hence, there exists a channel $\mathcal{E}'_0$ over $E'_k$ such that $\mathcal{E}'_0(\rho) = \sum {M'}_k\rho {M'}^{\dagger}_k$. Therefore,
 \begin{align*}
      LHS = \mathcal{U}_0\circ\mathcal{E}_0(\rho) &=\mathcal{U}_0\Bigl(\mathcal{E}_0(\rho)\Bigr)= \sum \mathcal{U}_0\Bigl(M_k \rho M_k^\dagger\Bigr)\nonumber\\
      &= \sum U_0\Bigl( M_k\rho M_k^\dagger\Bigr)U_0^\dagger\nonumber\\
      & = \sum M'_kU_0\rho U_0^\dagger M_k^{'\dagger}\nonumber\\
      & = \sum M'_k\Bigl(\mathcal{U}_0 (\rho)\Bigr) M_k^{'\dagger}\nonumber\\
      & = \mathcal{E}'_0\circ \mathcal{U}_0(\rho) = RHS
\end{align*}
\end{proof}

\begin{remark}
    Let $m \in \N$, \cref{lem:push} generalizes to a horizontal composition of $m$ noisy Clifford channels. This can be proved by induction on the number of CNOT gates.
    \label{rmk: m horizontal channels}
\end{remark}

\begin{corollary}
    Up to Clifford conjugation, the composite error channel in \cref{fig:consecutiveE} is $\mathcal{E}_c = \mathcal{E}_1 \circ \mathcal{E}_0$. Moreover, 
    \[
    F_{avg}(\mathcal{E}_c) = \frac{\Tr[S_{\mathcal{E}_1}S_{\mathcal{E}_0}]}{t(t+1)} + \frac{1}{t+1}
    \]
    \label{cor:first form fidelity Ec}
\end{corollary}

\begin{proof}
    By \cref{eqn:trace and epsilon,lem:trace composition}, we can calculate the average gate fidelity $\mathcal{E}_c$ as

\begin{align*}
    F_{avg}(\mathcal{E}_c) =F_{avg}(\mathcal{E}_1 \circ \mathcal{E}_0)= \frac{\Tr[S_{\mathcal{E}_1 \circ \mathcal{E}_0}]}{t^2}\times \frac{t}{t+1} + \frac{1}{t+1}= \frac{\Tr[S_{\mathcal{E}_1}S_{\mathcal{E}_0}]}{t(t+1)} + \frac{1}{t+1}.
\end{align*}
\end{proof}

\begin{lemma}
    Let $I$ be an identity matrix of size $t^2 \times t^2$. For $k \in \Z_2$, $A_k = I - S_{\mathcal{E}_k}$.
    \begin{align*}
   \Tr[S_{\mathcal{E}_1}S_{\mathcal{E}_0}]  =  \Tr[S_{\mathcal{E}_0}]+\Tr[S_{\mathcal{E}_1}]-t^2+\Tr[A_1A_0] 
\end{align*}
\label{lem:superoperator}
\end{lemma}

\begin{proof}
    By the linearity of trace,
    \[
    \Tr[A_k] = \Tr[I - S_{\mathcal{E}_k}]= \Tr[I] - \Tr[S_{\mathcal{E}_k}] = t^2 - \Tr[S_{\mathcal{E}_k}].
    \]
It follows that    
 \begin{align*}
     \Tr[S_{\mathcal{E}_1}S_{\mathcal{E}_0}] & = \Tr[(I-A_1)(I-A_0)] \nonumber\\
     &= \Tr[I-A_1-A_0+A_1A_0]  \nonumber\\
     &= t^2 -\Tr[A_1]-\Tr[A_0] +\Tr[A_1A_0] \nonumber\\
     & =  \Tr[S_{\mathcal{E}_0}]+\Tr[S_{\mathcal{E}_1}]-t^2+\Tr[A_1A_0].
 \end{align*}

\end{proof}

\begin{lemma}

    \label{lem:bound}
  \lemtext
\end{lemma}

\begin{proof}
    Justification is detailed in \cref{sec: proof of claim}.
\end{proof}

\begin{lemma}
Let $\mathcal{E}_c$ be an $n$-qubit channel that is horizontally composed of two noisy channels $\mathcal{E}_0$ and $\mathcal{E}_1$, with input dimensions $d_0$ and $d_1$, error probabilities $p_0$ and $p_1$. $0 \leq p_0, p_1 \leq 1$. $\mathcal{E}_c = \mathcal{E}_1 \circ \mathcal{E}_0$, $t=d_0=d_1=2^n$. $t$ is the input dimension of $\mathcal{E}_c$. Then

\[
0 \leq F_{avg}(\mathcal{E}_c)-(1-p_0)(1-p_1) \leq 
\Bigl(1+\frac{1}{2^{n-1}}\Bigr)p_0p_1.
\]
\label{lem:approximate fidelity}
\end{lemma}
\begin{proof}
Based on \cref{eqn:trace and epsilon},
\[
\Tr[S_{\mathcal{E}_0}] = t(t+1)(1-p_0)-t, \qquad \Tr[S_{\mathcal{E}_1}] = t(t+1)(1-p_1)-t.
\]
By \cref{cor:first form fidelity Ec,lem:superoperator}, we have 
    \begin{align}
    F_{avg}(\mathcal{E}_c) & =\frac{\Tr[S_{\mathcal{E}_0}]+\Tr[S_{\mathcal{E}_1}]-t^2+\Tr[A_1A_0]}{t(t+1)} + \frac{1}{t+1}\nonumber\\
    &=\frac{t(t+1)(1-p_0)-t + t(t+1)(1-p_1)-t - t^2}{t(t+1)}+  \frac{\Tr[A_1A_0]}{t(t+1)}+ \frac{1}{t+1}\nonumber\\
      & =  1-p_0-p_1 +  \frac{\Tr[A_1A_0]}{t(t+1)}.\nonumber\\
 \label{eqn: helper2}
 \end{align}
 Combining \cref{lem:bound,eqn: helper2}, we get
 \begin{align}
 1-p_0-p_1+p_0p_1+ \frac{1}{t}p_0p_1 \leq F_{avg}(\mathcal{E}_c) \leq 1-p_0-p_1+2p_0p_1+\frac{2}{t}p_0p_1.\label{eqn: mid2}
\end{align}
Since $d = 2^n$,
\[
\frac{1}{2^n}p_0p_1\leq \bigl[F_{avg}(\mathcal{E}_c)-(1-p_0-p_1 + p_0p_1)\bigr] \leq p_0p_1 + \frac{2}{2^n}p_0p_1 .
\]
Therefore,
 \[
0 \leq F_{avg}(\mathcal{E}_c)-(1-p_0-p_1 + p_0p_1) \leq p_0p_1 + \frac{2}{2^n}p_0p_1 = \Bigl(1+\frac{1}{2^{n-1}}\Bigr)p_0p_1 .
\]
\end{proof}

\subsection{Approximate Average Gate Fidelity}
\label{subsec:costfunction}

To simplify discussions, we assume no parallelization of quantum gates and no noise on idle qubits. Then any noisy quantum circuit can be modelled as a horizontal composition of noisy gates. \cref{fig:noisy circuit} shows an example of a noisy CNOT circuits over $4$ qubits.

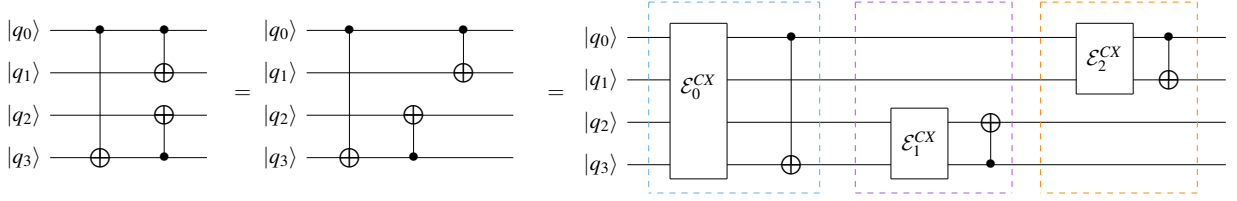
\begin{figure}[H]
    \centering
    \scalebox{.75}{\begin{tikzpicture}
	\begin{pgfonlayer}{nodelayer}
		\node [style=none, label={left:$|q_0\rangle$}] (0) at (1.75, 3) {};
		\node [style=none] (1) at (9, 3) {};
		\node [style=none, label={left:$|q_1\rangle$}] (2) at (1.75, 1.5) {};
		\node [style=none] (3) at (9, 1.5) {};
		\node [style=none, label={left:$|q_2\rangle$}] (4) at (1.75, 0) {};
		\node [style=none] (5) at (9, 0) {};
		\node [style=none, label={left:$|q_3\rangle$}] (6) at (1.75, -1.5) {};
		\node [style=none] (7) at (9, -1.5) {};
		\node [style=none] (8) at (3.25, 3) {$\bullet$};
		\node [style=none] (9) at (3.25, -1.5) {$\bigoplus$};
		\node [style=none] (10) at (5.5, -1.5) {$\bullet$};
		\node [style=none] (11) at (5.5, 0) {$\bigoplus$};
		\node [style=none] (12) at (7.25, 3) {$\bullet$};
		\node [style=none] (13) at (7.25, 1.5) {$\bigoplus$};
		\node [style=none, label={left:$|q_0\rangle$}] (14) at (-7.25, 3) {};
		\node [style=none] (15) at (-1.75, 3) {};
		\node [style=none, label={left:$|q_1\rangle$}] (16) at (-7.25, 1.5) {};
		\node [style=none] (17) at (-1.75, 1.5) {};
		\node [style=none, label={left:$|q_2\rangle$}] (18) at (-7.25, 0) {};
		\node [style=none] (19) at (-1.75, 0) {};
		\node [style=none, label={left:$|q_3\rangle$}] (20) at (-7.25, -1.5) {};
		\node [style=none] (21) at (-1.75, -1.5) {};
		\node [style=none] (22) at (-5.5, 3) {$\bullet$};
		\node [style=none] (23) at (-5.5, -1.5) {$\bigoplus$};
		\node [style=none] (24) at (-3.25, -1.5) {$\bullet$};
		\node [style=none] (25) at (-3.25, 0) {$\bigoplus$};
		\node [style=none] (26) at (-3.25, 3) {$\bullet$};
		\node [style=none] (27) at (-3.25, 1.5) {$\bigoplus$};
		\node [style=none] (28) at (10.5, 0.75) {$=$};
		\node [style=none, label={left:$|q_0\rangle$}] (29) at (13, 2.75) {};
		\node [style=none] (30) at (34, 2.75) {};
		\node [style=none, label={left:$|q_1\rangle$}] (31) at (13, 1.25) {};
		\node [style=none] (32) at (34, 1.25) {};
		\node [style=none, label={left:$|q_2\rangle$}] (33) at (13, -0.25) {};
		\node [style=none] (34) at (34, -0.25) {};
		\node [style=none, label={left:$|q_3\rangle$}] (35) at (13, -1.75) {};
		\node [style=none] (36) at (34, -1.75) {};
		\node [style=none] (37) at (18.75, 2.75) {$\bullet$};
		\node [style=none] (38) at (18.75, -1.75) {$\bigoplus$};
		\node [style=none] (39) at (25.75, -1.75) {$\bullet$};
		\node [style=none] (40) at (25.75, -0.25) {$\bigoplus$};
		\node [style=none] (41) at (32, 2.75) {$\bullet$};
		\node [style=none] (42) at (32, 1.25) {$\bigoplus$};
		\node [style=none] (43) at (14.5, 3.25) {};
		\node [style=none] (44) at (16.5, 3.25) {};
		\node [style=none] (45) at (14.5, -2.25) {};
		\node [style=none] (46) at (16.5, -2.25) {};
		\node [style=none] (47) at (14.5, 2.75) {};
		\node [style=none] (48) at (16.5, 2.75) {};
		\node [style=none] (49) at (14.5, 1.25) {};
		\node [style=none] (50) at (16.5, 1.25) {};
		\node [style=none] (51) at (14.5, -0.25) {};
		\node [style=none] (52) at (16.5, -0.25) {};
		\node [style=none] (53) at (14.5, -1.75) {};
		\node [style=none] (54) at (16.5, -1.75) {};
		\node [style=none] (55) at (15.5, 1) {$\mathcal{E}_0^{CX}$};
		\node [style=none] (56) at (22.25, 0.25) {};
		\node [style=none] (57) at (24.25, 0.25) {};
		\node [style=none] (58) at (22.25, -2.25) {};
		\node [style=none] (59) at (24.25, -2.25) {};
		\node [style=none] (60) at (22.25, -0.25) {};
		\node [style=none] (61) at (24.25, -0.25) {};
		\node [style=none] (62) at (23.25, -1) {$\mathcal{E}_1^{CX}$};
		\node [style=none] (63) at (22.25, -1.75) {};
		\node [style=none] (64) at (24.25, -1.75) {};
		\node [style=none] (65) at (28.75, 3.25) {};
		\node [style=none] (66) at (30.75, 3.25) {};
		\node [style=none] (67) at (28.75, 0.75) {};
		\node [style=none] (68) at (30.75, 0.75) {};
		\node [style=none] (69) at (28.75, 2.75) {};
		\node [style=none] (70) at (30.75, 2.75) {};
		\node [style=none] (71) at (29.75, 2) {$\mathcal{E}_2^{CX}$};
		\node [style=none] (72) at (28.75, 1.25) {};
		\node [style=none] (73) at (30.75, 1.25) {};
		\node [style=none] (123) at (13.75, 4) {};
		\node [style=none] (124) at (19.75, 4) {};
		\node [style=none] (125) at (19.75, -2.75) {};
		\node [style=none] (126) at (13.75, -2.75) {};
		\node [style=none] (131) at (21, 4) {};
		\node [style=none] (132) at (26.5, 4) {};
		\node [style=none] (133) at (21, -2.75) {};
		\node [style=none] (134) at (26.5, -2.75) {};
		\node [style=none] (135) at (27.5, 4) {};
		\node [style=none] (136) at (33, 4) {};
		\node [style=none] (137) at (27.5, -2.75) {};
		\node [style=none] (138) at (33, -2.75) {};
		\node [style=none] (149) at (-0.5, 0.75) {$=$};
	\end{pgfonlayer}
	\begin{pgfonlayer}{edgelayer}
		\draw (0.center) to (1.center);
		\draw (2.center) to (3.center);
		\draw (4.center) to (5.center);
		\draw (6.center) to (7.center);
		\draw (8.center) to (9.center);
		\draw (10.center) to (11.center);
		\draw (12.center) to (13.center);
		\draw (14.center) to (15.center);
		\draw (16.center) to (17.center);
		\draw (18.center) to (19.center);
		\draw (20.center) to (21.center);
		\draw (22.center) to (23.center);
		\draw (24.center) to (25.center);
		\draw (26.center) to (27.center);
		\draw (37.center) to (38.center);
		\draw (39.center) to (40.center);
		\draw (41.center) to (42.center);
		\draw (43.center) to (45.center);
		\draw (45.center) to (46.center);
		\draw (46.center) to (44.center);
		\draw (44.center) to (43.center);
		\draw (47.center) to (29.center);
		\draw (49.center) to (31.center);
		\draw (51.center) to (33.center);
		\draw (53.center) to (35.center);
		\draw (56.center) to (58.center);
		\draw (58.center) to (59.center);
		\draw (59.center) to (57.center);
		\draw (57.center) to (56.center);
		\draw (52.center) to (60.center);
		\draw (54.center) to (63.center);
		\draw (61.center) to (34.center);
		\draw (64.center) to (36.center);
		\draw (65.center) to (67.center);
		\draw (67.center) to (68.center);
		\draw (68.center) to (66.center);
		\draw (66.center) to (65.center);
		\draw (48.center) to (69.center);
		\draw (70.center) to (30.center);
		\draw (50.center) to (72.center);
		\draw (73.center) to (32.center);
		\draw [style=blue box] (123.center) to (126.center);
		\draw [style=blue box] (123.center) to (124.center);
		\draw [style=blue box] (124.center) to (125.center);
		\draw [style=blue box] (125.center) to (126.center);
		\draw [style=violet box] (131.center) to (133.center);
		\draw [style=violet box] (131.center) to (132.center);
		\draw [style=violet box] (132.center) to (134.center);
		\draw [style=violet box] (133.center) to (134.center);
		\draw [style=dashed not dense] (135.center) to (137.center);
		\draw [style=dashed not dense] (135.center) to (136.center);
		\draw [style=dashed not dense] (136.center) to (138.center);
		\draw [style=dashed not dense] (138.center) to (137.center);
	\end{pgfonlayer}
\end{tikzpicture}}
    \caption{Any noisy CNOT circuit is horizontally composed of noisy CNOT gates when we assume no gate parallelization.}
    \label{fig:noisy circuit}
\end{figure}

Combining the notions developed in \cref{subsec:parallel error channel,subsec:consecutive error channel}, \cref{def:generalization} specifies a cost function for a noisy CNOT circuit \vect{C}. By \cref{lem:push}, consider an equivalent circuit of \vect{C} where the noisy channel for each CNOT gate is placed adjacent to each other. Let $\mathcal{E}_c$ be this horizontally composed noisy channel. Let $m$ be the number of CNOT gates in the circuit. 

\paragraph{When a channel is composed of one noisy CNOT gate}
When $m=1$, based on \cref{lem:average fidelity parallel}, we can calculate the error probability of a noisy CNOT channel.
\begin{corollary}
    Let $\mathcal{E}_q$ be an $n$-qubit error channel for a noisy CNOT gate, with error rate $p$ and $(n-2)$ idle qubits. Then, $\Prob(\mathcal{E}_q) =  \bigl(1+\frac{2^{n-2}-1}{2^n+1}\bigr)p$.
\label{cor:one noisy cnot}
\end{corollary}

    \begin{proof}
         Without loss of generality, suppose that $\mathcal{E}_q$ is composed of two parallel channels, $\mathcal{E}^{\text{CNOT}}_0$ and $\mathcal{E}^{\text{Idle}}_1$, as shown in \cref{fig:ParallelECNOT}. $\mathcal{E}^{\text{CNOT}}_0$ is the error channel of a noisy CNOT gate. Its error rate is $p_0 = p$. $\mathcal{E}^{\text{Idle}}_1$ is the error channel of $n-2$ idle qubits. Its error rate is $p_1$. By assumption, $p_1 = 0$. Moreover, $d_0 = 2^2 = 4$, $d_1 = 2^{n-2}$. Substituting the variables in \cref{eqn: approximated form of parallel channels}, we have
         \begin{align*}
             F_{avg}(\mathcal{E}_q) & =1-p_0 -p_1 + p_0p_1 + \frac{(1-d_1)p_0 + (1-d_0)p_1 + (d_0+d_1)p_0p_1}{d_0d_1+1} \nonumber \\ 
             & = 1 - p-0 +0 +\frac{(1-2^{n-2})p+0+0}{2^2\times 2^{n-2}+1}\nonumber \\ 
             & = 1- \bigl(1+\frac{2^{n-2}-1}{2^n+1}\bigr)p
         \end{align*}
         
         By \cref{def: error prob and agf}, $\Prob(\mathcal{E}_q)= 1 - F_{avg}(\mathcal{E}_q) = \bigl(1+\frac{2^{n-2}-1}{2^n+1}\bigr)p$.
    \end{proof}
    
\paragraph{When a channel is composed of two noisy CNOT gates}

When $m=2$, based on \cref{cor:one noisy cnot,lem:approximate fidelity}, we can calculate the error probability of two horizontally composed noisy CNOT channels.

\begin{corollary} 
Let \vect{C} be an $n$-qubit circuit with two noisy CNOT gates. Let $\mathcal{E}_0$ and $\mathcal{E}_1$ be their respective error channels with error rates $p_0$ and $p_1$. Assume no parallelization of CNOT gates. In each channel, there are $(n-2)$ idle qubits. Let $\alpha = 1+\frac{2^{n-2}-1}{2^n+1}$. Then, 
    \[0 \leq F_{avg}(\mathcal{E}_c) - (1-\Prob(\mathcal{E}_0))(1-\Prob(\mathcal{E}_1))\leq \Bigl(1+\frac{1}{2^{n-1}}\Bigr)\alpha^2
    p_0p_1.  
    \]
        \label{cor:two channels agf}
\end{corollary}

    \begin{proof}
        By \cref{cor:one noisy cnot}, the error probabilities of $\mathcal{E}_0$ and $\mathcal{E}_1$ are
        \[\Prob(\mathcal{E}_0) = \alpha p_0,\quad \Prob(\mathcal{E}_1)=\alpha p_1.\] 
        Substituting $p_0$ and $p_1$ by $\Prob(\mathcal{E}_0)$ and $\Prob(\mathcal{E}_1)$ in \cref{lem:approximate fidelity}, we have
        \begin{align*}
         0 \leq LHS &\leq \Bigl(1+\frac{1}{2^{n-1}}\Bigr)(\alpha p_0)(\alpha p_1)=\Bigl(1+\frac{1}{2^{n-1}}\Bigr)\alpha^2
    p_0p_1 = RHS.  
            \end{align*}
    \end{proof}
\paragraph{When a channel is composed of multiple noisy CNOT gates} 
\cref{def:generalization} proposes a new cost function for a noisy CNOT circuit based on inspecting the approximated form of error probability in simpler cases. In \cref{subsec: compare functions}, we compare it with the sum of CNOT error rates and the average gate fidelity by running simulations with CNOT circuits of small size.
\begin{definition}
Let $d = 2^n$ be the input dimension of $\mathcal{E}_c$. For all $m \in \N$ and $k \in \Z_m$, $d_k=d$ is the input dimension of an error channel $\mathcal{E}_k$ consists of a noisy CNOT gate with error rate $p_k$. Let $\alpha = 1+\frac{2^{n-2}-1}{2^n+1}$. Assume no parallelization of CNOT gates. In each channel, there are $(n-2)$ idle qubits. Then,
\[
\Cost(\mathcal{E}_c)  =  1 - \prod_{i=0}^{m-1}(1 - \alpha p_i).
\]
\label{def:generalization}
\end{definition}

\begin{remark}
For all $n \in \N$, $1 < \alpha < \frac{5}{4}$, since
\[
0 < \frac{2^{n-2}-1}{2^n+1}, \quad 1+ \frac{2^{n-2}-1}{2^n+1} < 1+ \frac{2^{n-2}}{2^n} = 1+ \frac{1}{4} = \frac{5}{4}.
\]

For all $m \in \N$, $i \in \Z_m$, when $0 < p_i < 0.8$, $0 < \alpha p_i < 1$, since
\[
    0 < \alpha p_i < \frac{5}{4}p_i < 1.
    \]
    \label{rem:range of alpha}
\end{remark}

\cref{sec:backends} shows that in all benchmarked backends, the error rate of any CNOT gate is bounded by 0.1. By \cref{rem:range of alpha}, for any synthesized CNOT circuit $\vect{C}_{syn}$, $\Cost(\vect{C}_{syn})$ ranges between 0 and 1.
\section{Noise-Aware CNOT Circuit Routing: NAPermRowCol}
\label{sec:algorithm}

A \textbf{noise-aware CNOT circuit routing} algorithm maps a logical CNOT circuit \vect{C} to a NISQ hardware. It takes \vect{C}'s parity matrix \vect{A} and an undirected edge-weighted connected graph $G$ as inputs. A vertex in $G$ corresponds to a physical qubit. An edge in $G$ represents an allowed CNOT operation on the qubits corresponding to its endpoints. Its edge weight records the CNOT gate error rate. In \cref{subsubsec:routewithSteiner}, we introduce the technicality of the connectivity-aware CNOT synthesis algorithm PermRowCol~\cite{griendli2022}. Here, we propose a noise-aware CNOT circuit routing algorithm by adapting PermRowCol and utilizing the $\Cost$ evaluation. It reduces noise-aware CNOT circuit routing to a Steiner tree problem while accounting for nearest-neighbour interactions and CNOT gate error rates. Since its $\Cost$-instructed heuristics make PermRowCol aware of noises, it is named ``NAPermRowCol''. ``NA'' stands for ``Noise-Aware''.


In \cref{subsec: napermrowcol overview}, we offer an overview of NAPermRowCol, by presenting a comprehensive summary of its workflow and the intuition behind its technical aspects. In \cref{subsec: inner working}, we explain NAPerRowCol's noise-aware adaptation of PermRowCol. To keep things straightforward, we assume a naive qubit mapping strategy that assigns logical qubit $i$ to physical qubit $i$. It is important to mention that NAPermRowCol is designed to be compatible with an arbitrary initial qubit mapping strategy.


\subsection{The Workflow of NAPermRowCol}
\label{subsec: napermrowcol overview}
NAPermRowCol synthesizes \vect{A} by carrying out a sequence of reduction steps. Before each reduction step, a pivot row $r$ and a pivot column $c$ are selected based on the connectivity and edge weights of $G$, the $\Cost$ evaluation, and the binary structure of $\vect{A}$. A reduction step involves two actions: a \textbf{column reduction} which transforms column $c$ to a basis vector $e_r$, and a \textbf{row reduction} which transforms row $r$ to a transposed basis vector $e_c^\top$. Both reductions are achieved by applying a sequence of Steiner-tree-instructed row operations on $\vect{A}$. After a reduction step, vertex $r$ is removed from $G$. The algorithm terminates when there is one vertex left in $G$. In the meantime, $\vect{A}$ is reduced to a permutation matrix $\vect{P}$. Since each row operation corresponds to a CNOT gate, each reduction step outputs a CNOT sequence. NAPermRowCol concatenates these CNOT gates and returns a synthesized circuit $\vect{C}_{syn}$ composed of allowed CNOT operations, with \vect{P} for qubit relabeling. More precisely, $\vect{C}_{syn}$ is semantically equivalent to $\vect{C}$ up to permuting logical qubits in quantum registers. 

Since SWAP gates are factored out of $\vect{C}_{syn}$, the synthesized CNOT count is drastically eliminated. Since noise-aware greedy heuristics are applied to build and traverse a weighted Steiner tree, the cheapest solution is picked for each reduction step and the reliability of a synthesized CNOT subcircuit is maximized. As a result, NAPermRowCol produces a NISQ-executable CNOT circuit for $\vect{C}$ with reduced circuit execution time and enhanced fidelity. In what follows, we break down the crux of NAPermRowCol through a two-step explanation.

First, we explain how to find the cheapest sequence of row operations for a reduction step. Given a pivot column $c$, let $S_0$ be the set of rows that have a parity of $1$. $r \in S_0$ and $\lvert S_0 \rvert > 1$. Build a Steiner tree $T_0$ of $G$ where $r$ is the root and $S_0$ is the terminal. 
The Steiner nodes correspond to the rows that have a parity of 0. \cref{proc: first traverse column reduc} finds the cheapest path to move the parity of terminal nodes to Steiner nodes. After this traversal, all Steiner nodes will carry a parity of 1. Then traverse $T_0$ from the leaves to the root and add every parent $p$ to its child $c$. After the second traversal, the parity 1 from the root will be propagated to every other node in $T_0$. As a result, every row in column $c$ has a parity of 0 except for row $r$. Column $c$ is reduced to the basis vector $e_r$ and the column reduction is completed.

Given a pivot row $r$, we start by solving a system of linear equations and find rows in $\vect{A}$ such that $\bigoplus r_k  = e_c^\top \oplus r$. Let $S_1$ be the set of these indices $k$ including the pivot index $r$. Build a Steiner tree $T_1$ where $r$ is the root and $S_1$ is the terminal. Detailed in \cref{proc: first traverse row reduc}, the first traversal of $T_1$ is similar to \cref{proc: first traverse column reduc}, except that $u$ and $v$ have their roles exchanged. Next, traverse $T_1$ from the leaves to the root and add every child $c$ to its parent p. After the second traversal, the parity on each terminal node is propagated to the root and added together. Since the Steiner nodes are added twice modulo 2 throughout the two traversals, they do not participate in the desired parity summation. As a result, every column in row $r$ has a parity of 0 except for column $c$. Row $r$ is reduced to the basis vector $e_c^\top$ and the row reduction is completed.

Second, we explain the noise-aware pivot selection. Since each row corresponds to a physical qubit, the removal of vertex $r$ must not disconnect $G$. Among all non-cut vertices, \cref{proc: select pivot row} prioritizes rows with the lowest Hamming weight, followed by the rows tied to vertices having the lowest average incident edge weight. After selecting the pivot row, \cref{proc: select pivot column} exhaustively calculates the minimum $\Cost$ of each candidate column according to \cref{proc: first traverse column reduc}, and then picks the one that induces the cheapest solution as the pivot column.

Finally, we consolidate all components and outline the complete workflow of NAPermRowCol.

\begin{procedure}
To synthesize a CNOT circuit \vect{C} according to a NISQ architecture, let $\vect{A}$ be $\vect{C}$'s parity matrix. $G = (V_G, E_G, \omega_G)$ is an undirected edge-weighted connected graph characterizing the physical restrictions. $\omega_G: E_G \rightarrow \{x \in \R; \; 0 \leq x < 1\}$. For $e = (u, v) \in E_G$, $\omega_G(e)$ is the error rate of coupling physical qubits u and v. NAPermRowCol takes $\vect{A}$ and $G$ as inputs. When $\vect{A}$ is not a permutation matrix, proceed as follows.

\begin{enumerate}
    \item \cref{proc: select pivot row} picks a pivot row $r$.
    \item \cref{proc: select pivot column} picks a pivot column $c$.
    \item \cref{proc: first traverse column reduc} carries out a column reduction on $\vect{A}$. Let $\vect{A}^0$ be the transformed parity matrix.
    \item \cref{proc: first traverse row reduc} carries out a row reduction on $\vect{A}^0$. Let $\vect{A}^1$ be the transformed parity matrix.
    \item Remove vertex $r$ from G.
    \item Remove the pivot row and column from $\vect{A}^1$. Let $\vect{A}$ be the updated parity matrix.
    \item Go to step 1 until there is precisely one vertex left in G.
    \item Assemble the permutation matrix \vect{P} based on the reduced row and column in each reduction step.
    \item Concatenate row operations output from each reduction step. \cref{lem: reverse engineering permutation} returns a synthesized circuit $\vect{C}_{syn}$ with $\vect{P}$.
\end{enumerate}

\end{procedure}

Thanks to the implementation of the $\Cost$ function, NAPermRowCol is scalable and has the potential to route a more complicated quantum circuit on NISQ hardware. Compared to algorithm GENNS whose reduction step halts due to an invalid row operation~\cite{zhu2022physical}, it is not restricted to an initial qubit map. Compared to the Qiskit transpiler, it does not require ancillary qubits, making it more efficient in terms of resource usage and more suitable for large-scale noise-aware circuit routing.

\subsection{Noise-Aware Heuristics}
\label{subsec: inner working}
Similar to PermRowCol, NAPermRowCol proceeds by iteratively selecting a pivot row and column, then reducing them to basis vectors with a sequence of row operations. Given a pivot column $c$, let $S_0$ be the set of rows that have a parity of $1$. In the nontrivial case, $\lvert S_0\rvert > 1$. Let $T_0 = \Steiner(G,S_0)$. A traversal of $T_0$ is represented by an ordered set of edges, $w_{T_0}$. When $\lvert w_{T_0} \rvert = t$, it corresponds to applying $t$ row operations on $\vect{A}$. Let $p_i$ be the weight of the i-th edge in $w_{T_0}$, $i \in \Z_t$. The $\Cost$ of $w_{T_0}$ is calculated as 
\begin{equation}
    \Cost(w_{T_0})  =  1 - \prod_{i=0}^{m-1}(1 - \alpha p_i), \qquad \alpha = 1+\frac{2^{n-2}-1}{2^n+1}.
    \label{eqn:cost of path}
\end{equation}

Recall that there are two traversals in \cref{proc:column walk}. 
In NAPermRowCol, we optimize the first traversal based on the $\Cost$ function, keeping the second traversal the same as in PermRowCol.

\begin{procedure}
In $T_0 = \Steiner(G, S_0)$, to find the cheapest path to move the parity of terminal nodes to Steiner nodes, proceed as follows.
    \begin{enumerate}
    \item For a Steiner node $v \in V_{T_0} \setminus S_0$ that has not yet been picked: (1) For every terminal node $u \in S_0$, use Dijkstra's Algorithm to find the cheapest path from u to v according to the $\Cost$ evaluation in \cref{eqn:cost of path}. (2) Pick the cheapest one among them. Express it as an ordered set and add it to the collection of optimal solutions \Opt.
    \item Go to step 1 until all Steiner nodes are considered.
    \item Use the union operation for all ordered sets in $\Opt$ to find a combined solution.
    \item Return the ordered set of the combined solution with its $\Cost$ evaluation.
\end{enumerate}
\label{proc: first traverse column reduc}
\end{procedure}

For every candidate column, the traversal on $T_0$ and its associated $\Cost$ are saved in memory. This enables us to quickly access the $\Cost$ evaluation and corresponding row operations after a pivot column is selected. As a result, this enhances the time efficiency of NAPermRowCol.

\begin{lemma}
    Given $T_0 = \Steiner(G, S_0)$, \cref{proc: first traverse column reduc} is well-defined. In other words, it is not possible to encounter a situation depicted in \cref{fig:column reduction first traversal}, where the cheapest paths for different Steiner nodes cross over each other. By ``cross over'', we mean two paths share an edge $(a, b)$. One path goes from a to b, while the other goes from b to a.
    \begin{figure}[H]
        \centering
        \includegraphics[scale = .5]{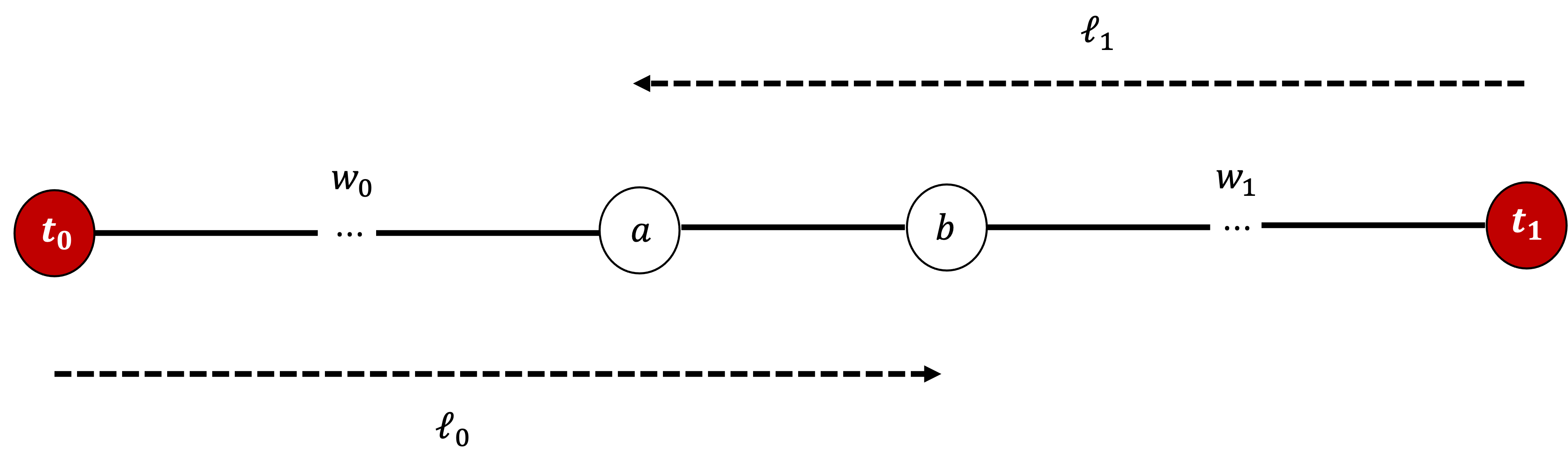}
        \caption{In $T_0 = \Steiner(G, S_0)$, $a$ and $b$ are Steiner nodes. $(a,b)$ denotes the edge between them with weight $q$. Among all terminal nodes, let $t_0$ and $t_1$ be the ones that have the cheapest paths $\ell_0$ and $\ell_1$ to $b$ and $a$ respectively. Starting from $t_0$, $w_0$ is the collection of edges on $\ell_0$ between $t_0$ and $a$. Starting from $t_1$, $w_1$ is the collection of edges on $\ell_1$ between $t_1$ and $b$.}
        \label{fig:column reduction first traversal}
    \end{figure}
    \label{lem: first traverse column reduc}
\end{lemma}

\begin{proof}
Suppose towards contradiction that the cheapest path $\ell_1$ for $a$ and the cheapest path $\ell_0$ for $b$ overlaps on edge $(a,b)$. According to \cref{proc: first traverse column reduc},
\begin{equation}
    \Cost(\ell_0) \leq \Cost(w_1), \qquad \Cost(\ell_1) \leq \Cost(w_0). \label{eqn: shortest paths}
\end{equation}

Suppose there are $m_0$ and $m_1$ edges in $w_0$ and $w_1$ respectively. $q$ is the edge weight of $(a,b)$. $p_i$ and $p_j$ denote the weight of the $i$-th edge in $w_0$ and the $j$-th edge in $w_1$. $0 \leq i < m_0$, $0 \leq j < m_1$. Let $\alpha$ be the constant specified in \cref{def:generalization}. Then 
\begin{align}
\Cost(w_0) &= 1 -\prod_{i=0}^{m_0-1}(1 - \alpha p_i), &\Cost(w_1) &= 1 -\prod_{j=0}^{m_1-1}(1 - \alpha p_j). \label{eqn:cost calculation without w}\\
\Cost(\ell_0) &= 1 -\Biggl(\;\prod_{i=0}^{m_0-1}(1 - \alpha p_i)\Biggr)(1-\alpha q), & \Cost(\ell_1) &= 1 -\Biggl(\;\prod_{j=0}^{m_1-1}(1 - \alpha p_j)\Biggr)(1-\alpha q). \label{eqn:cost calculation}
\end{align}

Combining \cref{eqn: shortest paths,eqn:cost calculation without w,eqn:cost calculation}, we have
\begin{align}
    1 -\Biggl(\;\prod_{i=0}^{m_0-1}(1 - \alpha p_i)\Biggr)(1-\alpha q) &\leq 1 -\prod_{j=0}^{m_1-1}(1 - \alpha p_j). \label{eqn: inequality 1}\\
    1 -\Biggl(\prod_{j=0}^{m_1-1}(1 - \alpha p_j)\Biggr)(1-\alpha q) & \leq 1 -\prod_{i=0}^{m_0-1}(1 - \alpha p_i). \label{eqn: inequality 2}
\end{align}

From \cref{eqn: inequality 1,eqn: inequality 2}, we have
\begin{align}
    \prod_{j=0}^{m_1-1}(1 - \alpha p_j) &\leq \Biggl(\;\prod_{i=0}^{m_0-1}(1 - \alpha p_i)\Biggr)(1-\alpha q) = \prod_{i=0}^{m_0-1}(1 - \alpha p_i) - \alpha q \Biggl(\prod_{i=0}^{m_0-1}(1 - \alpha p_i)\Biggr).\label{eqn: derivation 0}\\
    \prod_{i=0}^{m_0-1}(1 - \alpha p_i) &\leq \Biggl(\prod_{j=0}^{m_1-1}(1 - \alpha p_j)\Biggr)(1-\alpha q) = \prod_{j=0}^{m_1-1}(1 - \alpha p_j) - \alpha q\Biggl(\prod_{j=0}^{m_1-1}(1 - \alpha p_j)\Biggr).\label{eqn: derivation 1}
\end{align}

According to the empirical data of IBM's fake backends, the error rate of any CNOT gate is bounded by 0.1. By \cref{rem:range of alpha}, $0 < \alpha p_i, \alpha p_j, \alpha q < 1$. From \cref{eqn: derivation 0}, 
\begin{equation}
    0 < \alpha q \Biggl(\prod_{i=0}^{m_0-1}(1 - \alpha p_i)\Biggr) \leq \prod_{i=0}^{m_0-1}(1 - \alpha p_i)- \prod_{j=0}^{m_1-1}(1 - \alpha p_j).\label{eqn: last step 0}
\end{equation}

From \cref{eqn: derivation 1},
\begin{equation}
    \prod_{i=0}^{m_0-1}(1 - \alpha p_i)- \prod_{j=0}^{m_1-1}(1 - \alpha p_j)\leq -\alpha q\Biggl(\prod_{j=0}^{m_1-1}(1 - \alpha p_j)\Biggr) < 0.\label{eqn: last step 1}
\end{equation}

\cref{eqn: last step 0,eqn: last step 1} yields a contradiction. Hence, there is no crossover between the cheapest paths of different Steiner nodes.
\end{proof}

Recall that there are two traversals in \cref{proc:row walk}. In NAPermRowCol, we optimize the first traversal based on the $\Cost$ function, keeping the second traversal the same as in PermRowCol.

\begin{procedure}
In $T_1 = \Steiner(G, S_1)$, to find the cheapest path to move the parity of Steiner nodes to terminal nodes, proceed as follows.
\begin{enumerate}
    \item For a Steiner node $v \in V_{T_1} \setminus S_1$ that has not yet been picked: (1) For every terminal node $u \in S_1$, use Dijkstra's Algorithm to find the cheapest path from v to u according to the $\Cost$ evaluation in \cref{eqn:cost of path}. (2) Pick the cheapest one among them. Express it as an ordered set and add it to the collection of optimal solutions \Opt.
    \item Go to step 1 until all Steiner nodes are considered.
    \item Use the union operation for all ordered sets in $\Opt$ to find a combined solution.
    \item Return the ordered set of the combined solution with its $\Cost$ evaluation.
\end{enumerate}
\label{proc: first traverse row reduc}
\end{procedure}

\begin{lemma}
 Given $T_1 = \Steiner(G, S_1)$, \cref{proc: first traverse row reduc} is well-defined. In other words, it is not possible to encounter a situation depicted in \cref{fig:row reduction first traversal}, where the cheapest paths for different Steiner nodes cross over each other. By ``cross over'', we mean two paths share an edge $(a, b)$. One path goes from a to b, while the other goes from b to a.
    \begin{figure}[H]
        \centering
        \includegraphics[scale = .5]{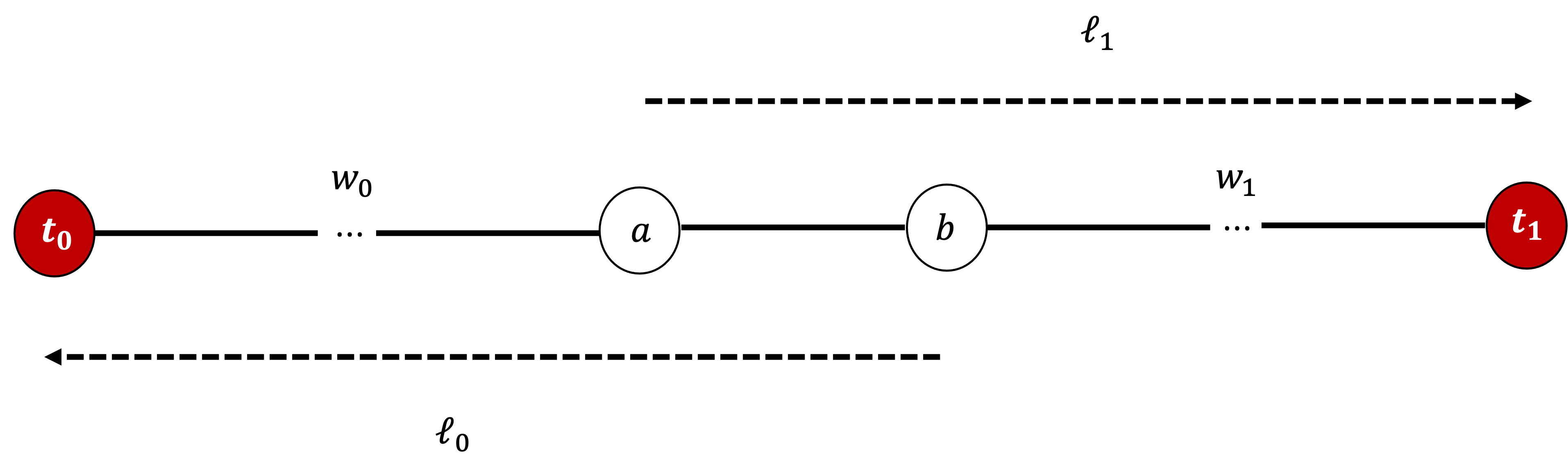}
        \caption{In $T_1 = \Steiner(G, S_1)$, $a$ and $b$ are Steiner nodes. $(a,b)$ denotes the edge between them with weight $q$. Among all terminal nodes, let $t_0$ and $t_1$ be the ones that have the cheapest paths $\ell_0$ and $\ell_1$ from $b$ and $a$ respectively. Ending at $t_0$, $w_0$ is the collection of edges on $\ell_0$ between $t_0$ and $a$. Ending at $t_1$, $w_1$ is the collection of edges on $\ell_1$ between $t_1$ and $b$.}
        \label{fig:row reduction first traversal}
    \end{figure}
    \label{lem: first traverse row reduc}
\end{lemma}

The symmetry of the graph traversal in \cref{lem: first traverse column reduc} leads to the validity of \cref{lem: first traverse row reduc}. Next, we adjust \cref{proc:select a pivot row} with a greedy heuristic. Since each row corresponds to a physical qubit, the removal of vertex $r$ must not disconnect $G$.

\begin{procedure}
To select a pivot row from $\vect{A}$, proceed as follows.
    \begin{enumerate}
    \item Among all non-cut vertices, select rows with the lowest Hamming weight to form $R_0$.
    \item From $R_0$, pick rows tied to vertices having the lowest average incident edge weight to form $R_1$. 
    \item Choose any row from $R_1$ as the pivot, noting its index as $r$. 
\end{enumerate}
\label{proc: select pivot row}
\end{procedure}

Lastly, we adjust \cref{proc:select a pivot column} by leveraging the noise-aware heuristic in \cref{proc: first traverse column reduc}. 
    
\begin{procedure}
    Given the pivot row $r$, to select a pivot column from $\vect{A}$, proceed as follows.
    \begin{enumerate}
    \item Among all columns in $\vect{A}$, find the set of columns that have a non-zero entry at row $r$. Let it be $C_0$.
    \item If there exists a basis vector in $C_0$, let it be the pivot column and note its index as $c$.
    \item Otherwise, for each column $c \in C_0$: (1) Build a Steiner tree $T_0 = \Steiner(G,S_0)$ with $r$ as its root. $S_0$ is the set of rows that have a parity of $1$ in $c$. $\lvert S_0\rvert > 1$. (2) Input $T_0$ to \cref{proc: first traverse column reduc}. 
    \item From $C_0$, select columns with the cheapest $\Cost$ to form $C_1$. 
    \item Choose any column from $C_1$ as the pivot, noting its index as $c$.
\end{enumerate}
\label{proc: select pivot column}
\end{procedure}


\section{Benchmark Results}
\label{sec:results}

IBM's fake backends mimic the behaviours of its quantum computers using system snapshots. They contain important information such as coupling maps, basic gates, qubit qualities (e.g., T1 and T2 time), and gate error rates. In our simulation of a noisy CNOT circuit, we consider the coupling map and use three backends of different sizes: fake Nairobi (7 qubits, \cref{fig:Nairobi}), fake Guadalupe (16 qubits, \cref{fig:Guadalupe}), and fake Cairo (27 qubits, \cref{fig:Cairo}). They are one of the most well-developed IBM's backends with relatively low CNOT gate error rates. Moreover, they have distinct topologies, which help us evaluate the adaptability of NAPermRowCol to different backends. In addition, they are commonly used for benchmarking different synthesis algorithms in the recent literature, so we use them to carry out in-depth comparisons between different CNOT synthesis algorithms. In this section, we discuss the simulation and benchmark results on the fake Nairobi backend with randomly generated CNOT circuits. In \cref{sec: compare different cost functions,sec: compare diff algorithms}, we report the results using the fake Guadalupe and Cairo backends. 

When a circuit consists of n qubits, it is of \textbf{width} n. The input dataset\footnote{\url{https://github.com/Aerylia/pyzx/tree/rowcol/circuits/steiner}} consists of CNOT circuits of widths 5, 7, and 16. Using the same circuit generation method as \cite{griendli2022,kissinger2020cnot}, we extend this dataset with circuits of widths 6 and 7. To create more variations of input circuits, we increase the CNOT count for circuits of the same width. For versions of important Python packages and our system specifications, please check out \cref{sec: version and hardware specifics}.


In \cref{subsec: compare functions}, for a synthesized CNOT circuit $\vect{C}_{syn}$, we compare the cost function $\Cost(\vect{C}_{syn})$ (\cref{def:generalization}) with other commonly used cost functions against the error probability $\Prob(\vect{C}_{syn}) = 1 - F_{avg}(\vect{C}_{syn})$. Due to the limited scalability of calculating $F_{avg}(\vect{C}_{syn})$, our comparison is restricted to $\vect{C}_{syn}$ of width no more than 7. Based on the simulation results, $\Cost(\vect{C}_{syn})$ fits $\Prob(\vect{C}_{syn})$ up to $10^{-3}$. In addition, its simple expression allows us to efficiently calculate the cost associated with a noisy CNOT circuit. To synthesize a large CNOT circuit, we use $\Cost$ to measure a circuit's noise levels and instruct the error mitigation strategy.

In \cref{subsec: compare algorithms}, we compare the performance of different synthesis algorithms in terms of the $\Cost$ metric and the \textbf{synthesized CNOT count} (i.e., the gate count of a synthesized CNOT circuit that is physically executable on NISQ hardware). NAPermRowCol performs significantly better than Qiskit transpiler, especially for circuits with a larger number of CNOT gates. Compared with algorithms that are noise-agnostic (PermRowCol and ROWCOL), NAPermRowCol lowers the synthesized CNOT count and is much cheaper for synthesizing random CNOT circuits of different sizes (i.e., different widths and the original CNOT counts). According to the benchmark results, NAPermRowCol provides an improved solution to the CNOT circuit synthesis problem. It not only reduces the CNOT counts, but also minimizes the overall error probabilities.


\subsection{Compare Different Cost Functions}
\label{subsec: compare functions}

Utilizing the Qiskit Python package, we compute the average gate fidelity of a noisy CNOT circuit via superoperator simulation. Since this requires a substantial amount of resources, we calculate $F_{avg}(\vect{C}_{syn})$ where $\vect{C}_{syn}$ has no more than 7 qubits. Next, compare $\Cost(\vect{C}_{syn})$ with two other cost functions to see how well they fit $\Prob(\vect{C}_{syn})$ respectively. $\Cost1(\vect{C}_{syn})$ is the cost function used in \cite{zhu2022physical}. $\Cost2(\vect{C}_{syn})$ is an alternative one based on the probability theory, i.e., the failure rate of $\vect{C}_{syn}$. $p_i$ is the error rate of the i-th noisy CNOT gate in $\vect{C}_{syn}$ and $n$ is the circuit width.
\begin{align}
    \text{\Cost}(\vect{C}_{syn}) &= 1-\prod_i (1-\alpha p_i),\qquad \alpha = 1+\frac{2^{n-2}-1}{2^n+1}.\nonumber\\
    \text{\Cost1}(\vect{C}_{syn}) &= \sum_i  p_i. \nonumber\\
    \text{\Cost2}(\vect{C}_{syn}) &= 1-\prod_i (1-p_i). \nonumber
\end{align}

\paragraph{Simulate noisy CNOT circuits with a small gate count}
To compare different cost functions against $\Prob(\vect{C}_{syn})$, consider CNOT circuits of width $n \in \{5,6,7\}$, with $m$ original CNOT count. $m=2^k$, $2 \leq k \leq 7,\; k \in \N$. For each combination of n and m, generate 100 random circuits. According to the topology in \cref{fig:Nairobi}, synthesize them with NAPermRowCol. Denote the synthesized circuit as $\vect{C}_{syn}$ and regroup them based on their CNOT count. For circuits of the same synthesized CNOT count, calculate $\Prob(\vect{C}_{syn})$, $\Cost(\vect{C}_{syn})$, $\Cost1(\vect{C}_{syn})$, and $\Cost2(\vect{C}_{syn})$ for each one of them and take the average. By ``small gate count'', we consider the synthesized CNOT count bounded by 40. For brevity, we drop $\vect{C}_{syn}$ from each expression in the remainder of this section. Besides, we use $\Prob$, $\Cost$, $\Cost1$, and $\Cost2$ to denote the value after taking the average for all synthesized circuits of the same CNOT count.

 In \cref{fig: cost comparison nairobi}, the comparison is conducted using two methods. In the left column, Root Mean Square Error (RMSE) is used to measure the average difference between the error probability and each cost function. For $\Cost$, $\Cost1$, and $\Cost2$, their RMSEs with respect to $\Prob$ correspond to the vertical error bars in green, blue, and orange. Suppose there are $N$ synthesized CNOT circuits with a fixed gate count. For $i\in \Z_N$, let $\vect{C}_i$ be each such circuit. The length of the green error bar is calculated as
\begin{equation}    
\sqrt{\frac{1}{N}\sum_{i=0}^{N-1}\bigl(\Prob(\vect{C}_i)-\Cost(\vect{C}_i)\bigr)^2 }.
\label{eq: errorbar}
\end{equation}

RMSE serves as a goodness-of-fit assessment and evaluates how accurately each cost function approximates $\Prob$. It ranges between 0 and positive infinity. As the cost function moves closer to $\Prob$, the approximation has less error, and thus has better precision. A value of 0 (almost impossible in practice) indicates a perfect fit to $\Prob$. RMSE is a simple metric that provides a straightforward interpretation of a cost function's overall error~\cite{armstrong2001evaluating}. 

For circuits of widths $5$, $6$, and $7$ with different CNOT counts, $\Cost$ demonstrates a minor deviation from $\Prob$. Compared with $\Cost1$ and $\Cost2$, it approximates $\Prob$ with higher precision. As shown in the left columns of \cref{fig: cost comparison Guadalupe,fig: cost comparison Cairo}, on the other IBM's backends, $\Cost$ consistently provides a much more accurate approximation for $\Prob$ than the other cost functions. These simulation results show that $\Cost$ is an accountable approximation for $\Prob$, despite varied topology and error distribution across different backends.

In the right column of \cref{fig: cost comparison nairobi}, the maximum distance between $\Prob$ and each cost function is used as an alternative metric for the goodness-of-fit assessment. Suppose there are $N$ synthesized CNOT circuits with gate count $M$. For $i\in \Z_N$, let $\vect{C}_i$ be each such circuit. The maximum distance between $\Prob$ and $\Cost$ is defined as
\begin{equation}    
d(\Prob, \Cost) = \max_{i \in \Z_N} \lvert \Prob(\vect{C}_i)-\Cost(\vect{C}_i) \rvert.
\label{eq: max distance}
\end{equation}

Since $d(\Prob, \Cost)$ has a wide range, we use a logarithmic scale for the y-axis to plot all data compactly. 
For circuits of widths $n \in \{5,6,7\}$, $\Cost$ fits $\Prob$ up to $10^{-3}$, much closer than $\Cost1$ and $\Cost2$. Similarly, in the right columns of \cref{fig: cost comparison Guadalupe,fig: cost comparison Cairo}, $\Cost$ is a tighter approximation for $\Prob$ than $\Cost1$ and $\Cost2$. It fits $\Prob$ up to $10^{-1}$ for all synthesized CNOT counts.

There is another interesting pattern between different cost functions. In the right column of \cref{fig: cost comparison nairobi}, when the synthesized CNOT count increases, $d(\Prob, \Cost)$ increases, so does $d(\Prob, \Cost1)$. However, $d(\Prob, \Cost2)$ increases when the synthesized CNOT count is lower than 20. Then it decreases when the synthesized CNOT counts get larger. This is because when the CNOT count is low, the accumulated CNOT error rate is also low. As the CNOT count increases, the growth of $\Cost2$ surpasses the growth of $\Prob$, so their distance is reduced. To see if $\Cost2$ will outperform $\Cost$, we use synthesized CNOT circuits with higher CNOT counts to investigate whether the distance between $\Prob$ and $\Cost2$ will eventually decrease to 0.


\begin{figure}[p]
    \centering
    \includegraphics[width=1\linewidth]{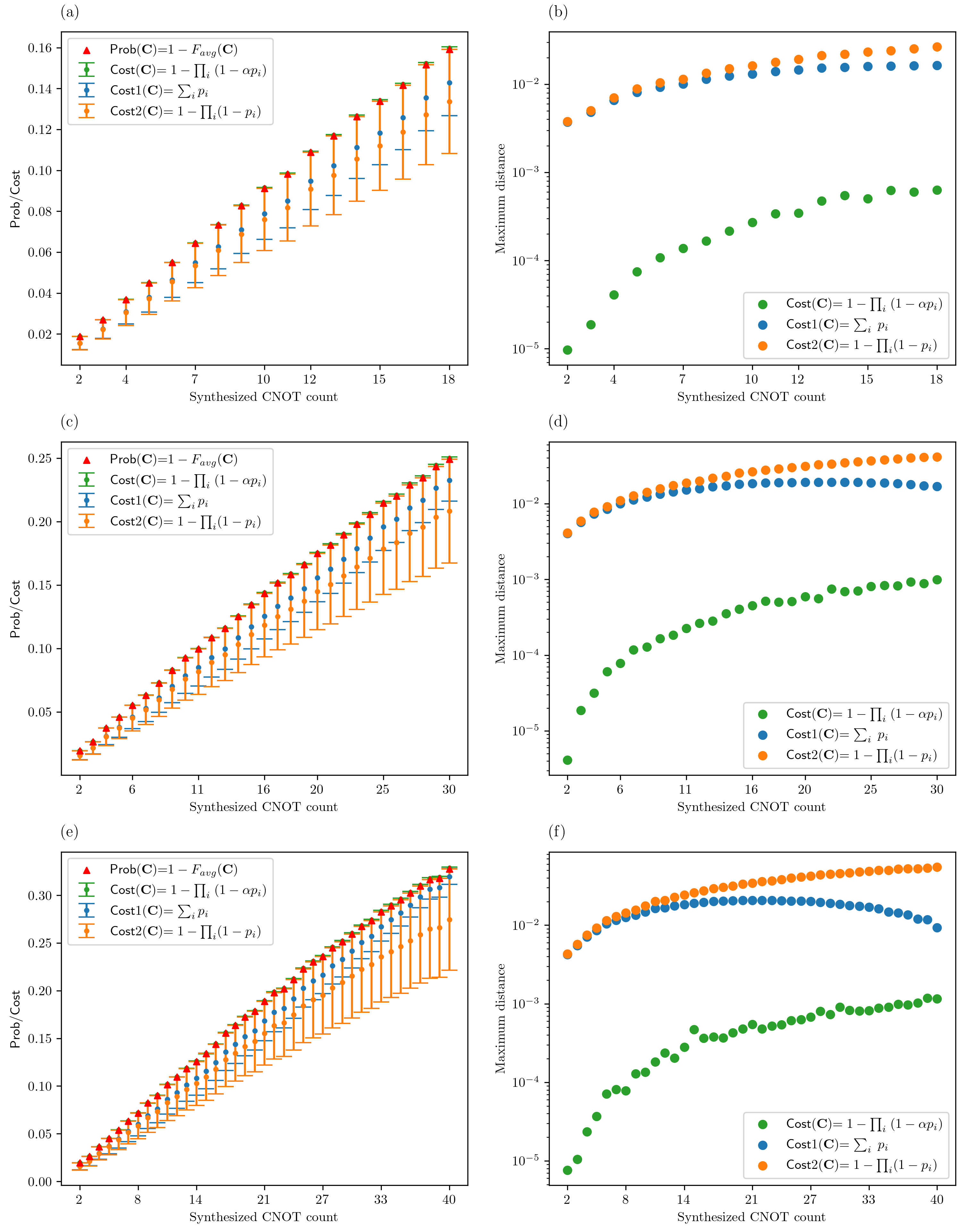}
    \caption{On IBM's fake Nairobi backend, compare the three cost functions against the error probability of a synthesized CNOT circuit. $p_i$ is the error rate of a noisy CNOT gate, $n$ is the circuit width, $\alpha = 1+\frac{2^{n-2}-1}{2^n+1}$. In the left column, figures (a), (c), and (e) report results about the synthesized CNOT circuits of widths 5, 6, and 7 respectively. In each figure, for circuits of the same synthesized CNOT count, a y-value corresponds to the average of the error probability or the average of a cost function. The length of an error bar measures the average difference between the error probability and each cost function. The longer error bar, the worse the approximation. In the right column, we compare the maximum distance between the error probability and each cost function. Figures (b), (d), and (f) report results about the synthesized CNOT circuits of widths 5, 6, and 7 respectively. Since $d(\Prob,\Cost)$ has a wide range, we use a logarithmic scale for the y-axis to see all the numbers easily.}
    \label{fig: cost comparison nairobi}
\end{figure}

\begin{figure}[!ht]
\begin{subfigure}[t]{.32\linewidth}
\includegraphics[scale = 0.3]{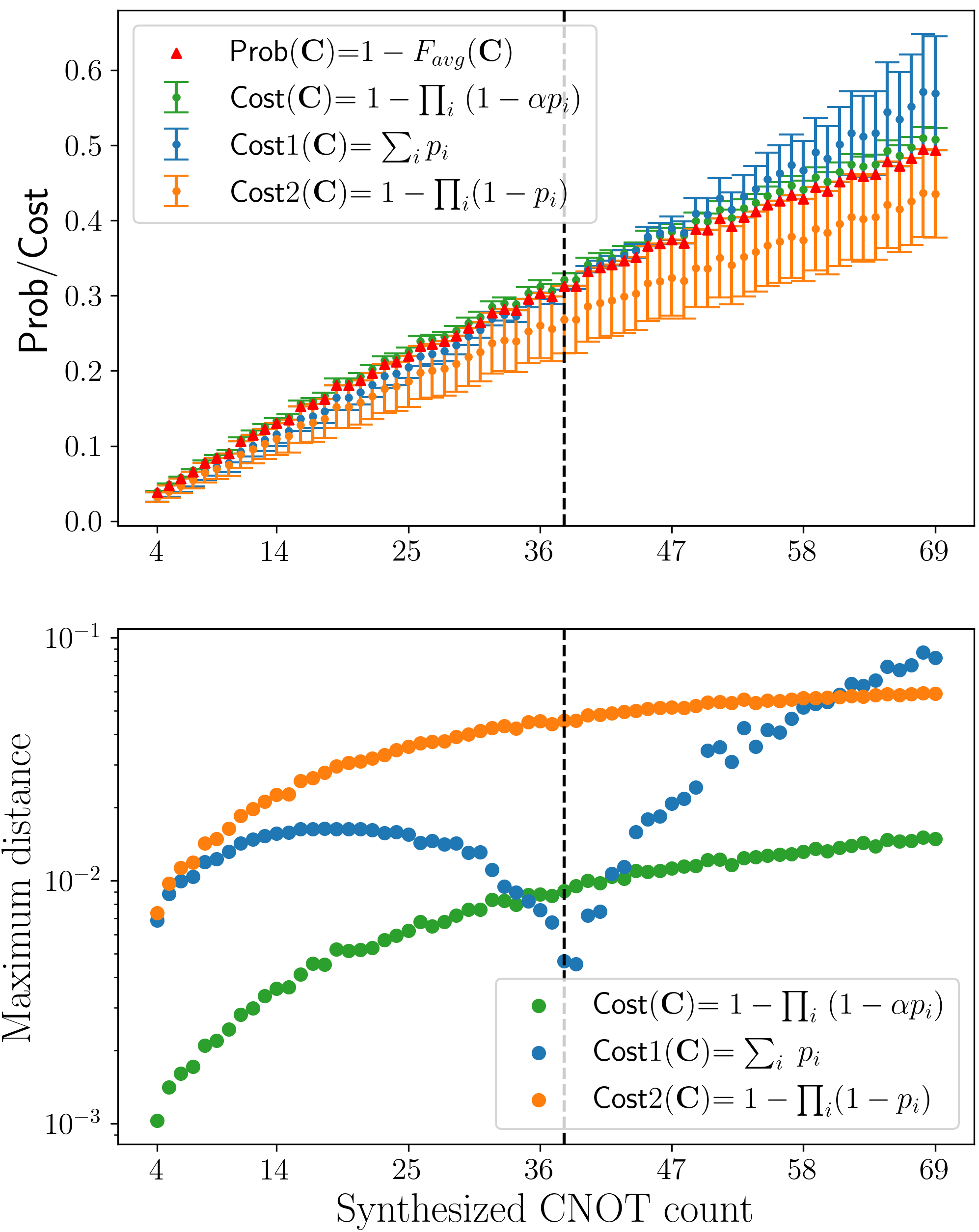}
\subcaption{Synthesized CNOT circuits of width 5.}
\label{fig: more cnot Nairobi 5qubit}
\end{subfigure}
\;
\begin{subfigure}[t]{.32\linewidth}
\includegraphics[scale = 0.3]{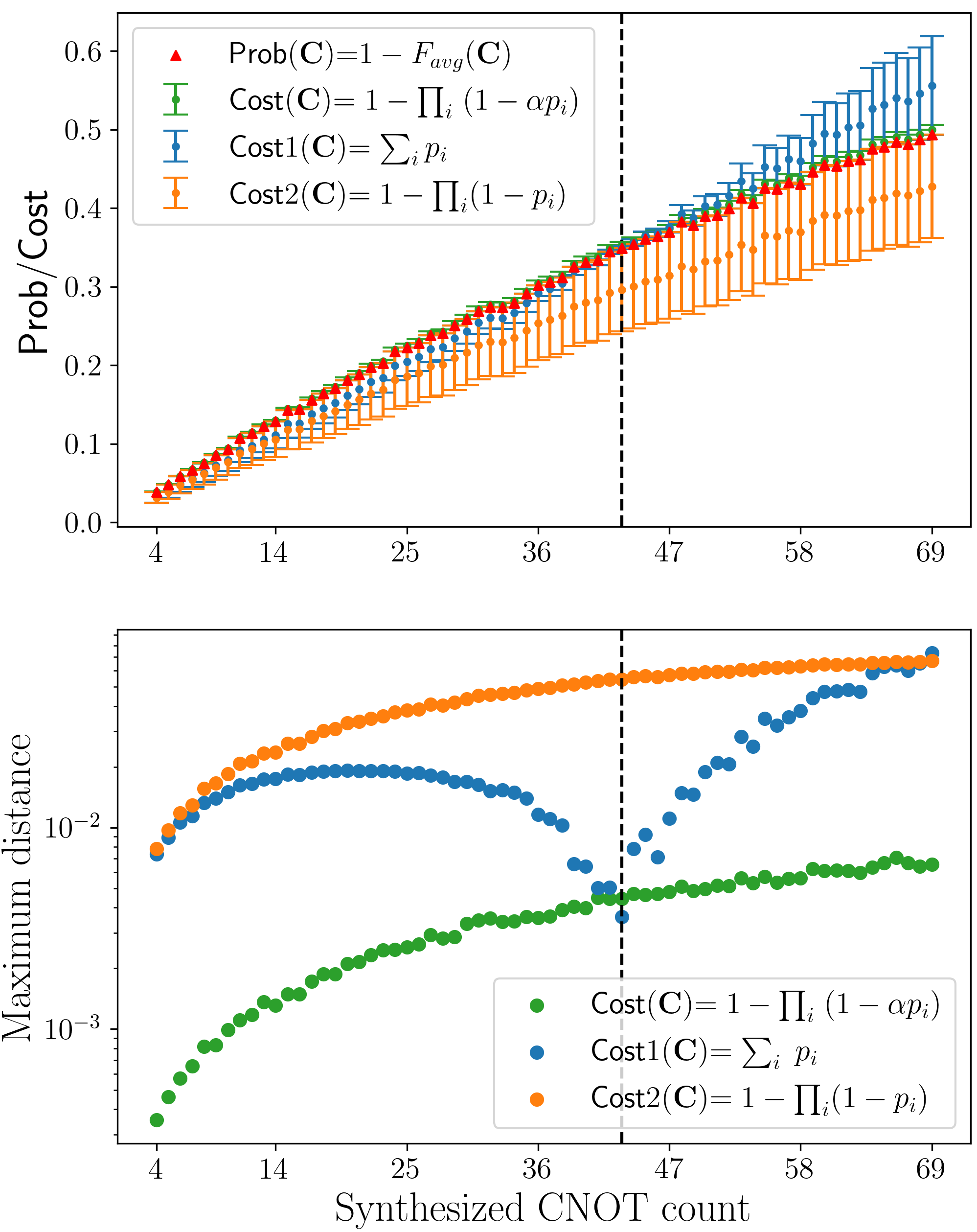}
\subcaption{Synthesized CNOT circuits of width 6.}
\label{fig: more cnot Nairobi 6qubit}
\end{subfigure}
\;
\begin{subfigure}[t]{.32\linewidth}
\includegraphics[scale = 0.3]{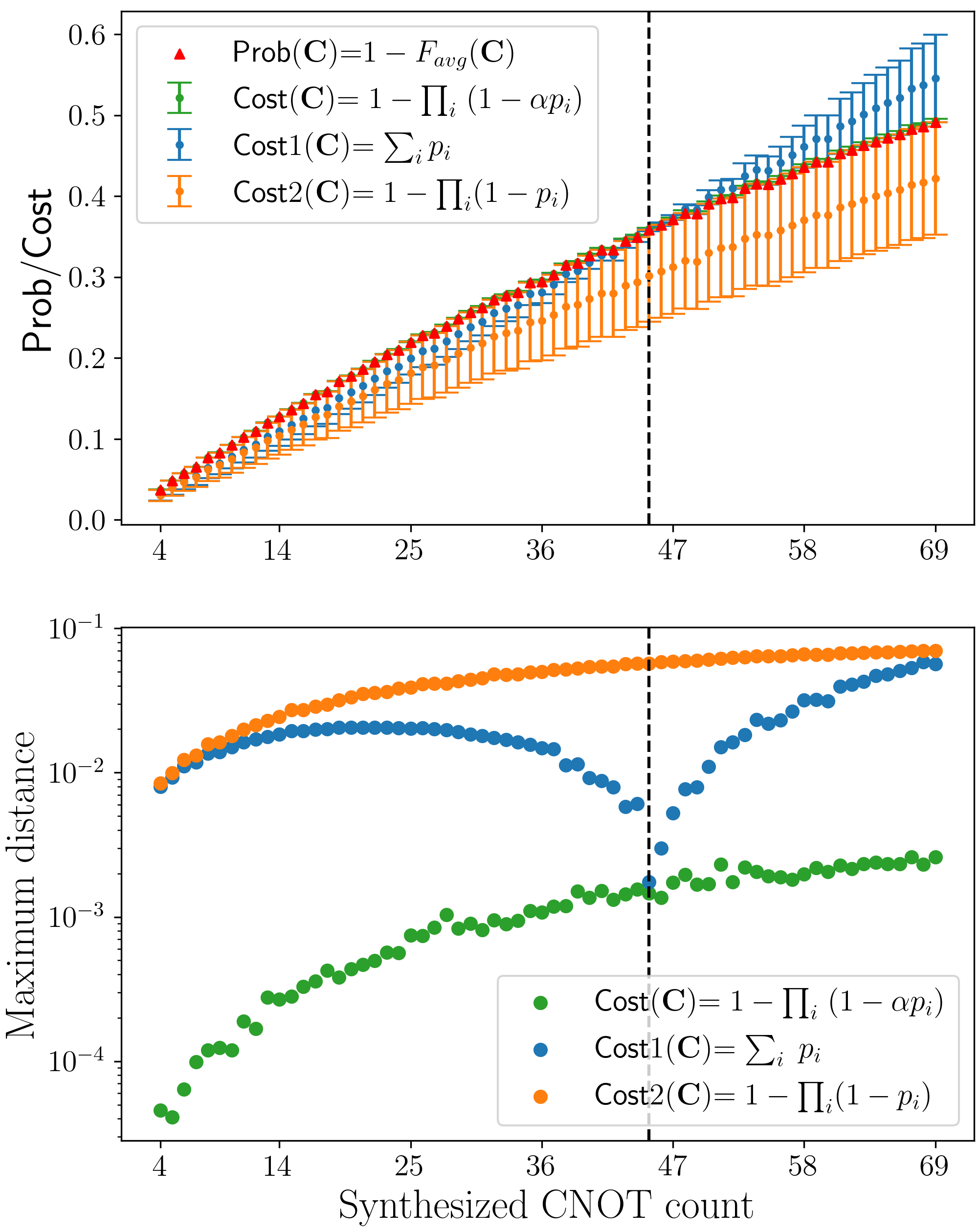}
\subcaption{Synthesized CNOT circuits of width 7.}
\label{fig: more cnot Nairobi 7qubit}
\end{subfigure}  
\caption{On IBM's fake Nairobi backend, compare the three cost functions against the error probability of a synthesized CNOT circuit with gate count ranging between 4 and 69. $p_i$ is the error rate of a noisy CNOT gate, $n$ is the circuit width, $\alpha = 1+\frac{2^{n-2}-1}{2^n+1}$. In each figure of the first row, for circuits of the same synthesized CNOT count, a y-value corresponds to the average of the error probability or the average of a cost function. The length of an error bar measures the average difference between the error probability and each cost function. The longer error bar, the worse the approximation. In the second row, we compare the maximum distance between the error probability and each cost function. Since $d(\Prob,\Cost)$ has a wide range, we use a logarithmic scale for the y-axis to see all the numbers easily. In each column, the black dashed line corresponds to a synthesized CNOT count of 38, 43, and 45. This denotes the point where $d(\Prob,\Cost1)$ attains the minimum value, which nearly coincides with $d(\Prob,\Cost)$. The extended input dataset enhances our confidence in the goodness of $\Cost$ as an approximation for $\Prob$.}
\label{fig: more cnots}
\end{figure}
\paragraph{Simulate noisy CNOT circuits with a larger gate count}
To randomly generate synthesized CNOT circuits, our previous method has its limitations. The synthesized CNOT count does not scale well with the number of input circuits to NAPermRowCol. For example, if we use circuits with 128 original CNOT count as input, no synthesized circuit has a CNOT count of more than 44. To this end, we use a topology-based method to generate a synthesized CNOT circuit with a larger gate count. By ``larger gate count'', we consider the synthesized circuit whose gate count is bounded by 70.

Let G be a hardware topology, $G = (V_G, E_G)$, $\lvert V_G\rvert = N$ and $\lvert E_G\rvert = M$. Compose a synthesized CNOT circuit of width $n$ with $m$ gates by considering only the allowed CNOT operations based on G, $n \leq N$. That is, the CNOTs whose control and target correspond to two adjacent vertices in G. Since a CNOT is asymmetric with its control and target, each edge $e=(u,v) \in E_G$ has an orientation. For a directed edge, let its head and tail be the control and target of a CNOT respectively.

\begin{enumerate}
    \item Let $\total = 0$, $\vect{C} = \emptyset$, $E_G'$ is the set of 2M directed edges of $E_G$.
    \item Randomly pick one directed edge $e \in E_G'$.
    \item While $\total  < m$:
    \begin{enumerate}[label=(\arabic*)]
        \item Right-append \vect{C} by $e$.
        \item $\total = \total + 1$.
        \item $e = e'$.
        \item Randomly pick one directed edge $e \in E'_G \setminus \{e'\}$. 
    \end{enumerate}
    \item Return \vect{C}.
\end{enumerate}

\cref{fig: more cnots} compare the goodness-of-fit of each cost function for the error probability of larger synthesized CNOT circuits. In the first row, for circuits of the same synthesized CNOT count, each error bar measures the average difference between the error probability and each cost function. For all circuit widths,  $\Cost1$ intersects with $\Prob$, before which the difference between $\Prob$ and $\Cost1$ decreases. This means $\Cost1$ provides a better fit for $\Prob$ when the synthesized CNOT count goes beyond a certain point. Such a changing point is annotated by the black dashed line in each column. It corresponds to a synthesized CNOT count of 38, 43, and 45. After the intersection, the difference between $\Prob$ and $\Cost1$ increases, indicating a worse fit for $\Prob$ when the synthesized CNOT count increases. In comparison, the difference between $\Prob$ and $\Cost$ remains small for all synthesized CNOT count. 

In the second row, we compare the maximum distance between the error probability and every cost function. In each column, the black dashed line from the first row denotes the point where $d(\Prob,\Cost1)$ attains the minimum value (between $10^{-3}$ and $10^{-2}$), which nearly coincides with $d(\Prob,\Cost)$. This implies that for certain synthesized CNOT count, $\Cost1$ provides a good approximation for $\Prob$. However, in general, $\Cost$ gives a more accurate approximation for circuits of different synthesized CNOT counts.

To summarize, on different IBM's backends and for CNOT circuits of different sizes, $\Cost$ fits $\Prob$ significantly better than the other cost functions. Moreover, it circumvents the computation complexity of calculating the average gate fidelity and shows remarkable scalability. Therefore, we use it to design the noise-aware CNOT synthesis algorithm NAPermRowCol and evaluate the performance of our error mitigation strategy.

\subsection{Compare Different CNOT Synthesis Algorithms}
\label{subsec: compare algorithms}

NAPermRowCol is inspired by the algorithm GENNS~\cite{zhu2022physical} and built upon the algorithm PermRowCol~\cite{griendli2022}. It reduces a noise-aware CNOT routing algorithm to a Steiner tree problem and lowers the CNOT count by factoring out SWAP gates after CNOT synthesis. Moreover, it uses noise-aware heuristics to make an informed decision at each reduction step. To gauge its performance, we benchmark it against the state-of-the-art CNOT synthesis algorithms according to the $\Cost$ metric and the synthesized CNOT count. Our comparative analysis is carried out between NAPermRowCol, PermRowCol, ROWCOL~\cite{wu2023optimization}, and Qiskit. Qiskit is short for ``Qiskit transpilation at optimization level 3''. It implements the SWAP-based heuristic algorithm SABRE~\cite{li2019tackling}. These algorithms take a logical CNOT circuit \vect{C} and an undirected edge-weighted connected graph $G=(V_G, E_G, \omega_G)$ as inputs and return a synthesized circuit with allowed CNOT operations. Although it would be ideal to include GENNS in our comparison, it fails upon an invalid initial qubit map. 
Consider an arbitrary logical qubit map to a connected subgraph of G. Since every other algorithm works with such a map, the incompatibility of GENNS makes it unsuitable for a fair comparison.

In this section, we analyze the benchmark results using IBM's fake Nairobi backend. In \cref{subsubsec:guadalupe,subsubsec:cairo}, we report results on IBM's fake Guadalupe and Cairo backends. For the Nairobi backend, the input dataset consists of CNOT circuits of width $n \in \{5,7\}$. For Guadalupe and Cairo backends, the input dataset is extended with CNOT circuits of width 16. Let $m$ be the original CNOT count of the input circuit. For all backends and circuit widths, $m=2^k$, $2 \leq k \leq 10$. For each combination $(n,m)$, randomly generate 100 circuits. All other experimental conditions remain the same across different backends. 

To find a good initial qubit map for each original CNOT circuit \vect{C}, we start by considering the one selected by Qiskit. It randomly generates a map for all logical qubits and synthesizes \vect{C} with SWAP gates. After repeating this several times, it finds the best map $\Phi$ and returns the synthesized circuit. If $\Phi$ maps logical qubits to connected quantum registers, we use it as the initial qubit map for ROWCOL, PermRowCol, and NAPermRowCol. Otherwise, these algorithms fail as they assume a connected topology. Alternatively, select a subset of connected vertices, $V_{G'} \subset V_G$, $\lvert V_{G'}\rvert = n$. Let $G'=(V_{G'}, E_{G'},\omega_{G'})$ be the induced subgraph of $G$. $\omega_{G'}(e) = \omega_{G}(e)$, $e \in E_{G'}$. Assign each logical qubit to a vertex in $V_{G'}$ in increasing order. That is, if $i < j$, $\Phi(i) < \Phi(j)$. For example, to map a 5-qubit logical circuit to a connected subgraph of \cref{fig:Nairobi}, let $V_{G'}=\{0,1,2,3,4\}$. Then the naive map $\Phi$ is described by \cref{ex: initial circuit placement}.

\begin{table}[H]
\centering
\begin{tabular}{|l|c|c|c|c|c|}
\hline
Logical qubit $\color{IMSGreen}i$  & $\color{IMSGreen}0$ & $\color{IMSGreen}1$ & $\color{IMSGreen}2$ & $\color{IMSGreen}3$ & $\color{IMSGreen}4$  \\ \hline
Vertex $\color{IMSBlue}j$ & \color{IMSBlue}0 & \color{IMSBlue}1 & \color{IMSBlue}2 & \color{IMSBlue}3& \color{IMSBlue}4 \\ \hline
\end{tabular}
\caption{For clarity, we use a green label to denote a logical qubit and a blue label to denote a physical qubit. The naive qubit map can be expressed as $\Phi=[\color{IMSBlue}0\color{black},\color{IMSBlue}1\color{black},\color{IMSBlue}2\color{black},\color{IMSBlue}3\color{black},\color{IMSBlue}4\color{black}]$.}
\label{ex: initial circuit placement}
\end{table}

After applying $\Phi$ to the logical qubits of \vect{C}, input it alongside $G'$ into the benchmarked algorithms. For each combination $(n,m)$, let $\vect{C}_i$ be the synthesized circuits returned by NAPermRowCol, $0 \leq i \leq 99$. Find the synthesized gate count of $\vect{C}_i$ and $\Cost(\vect{C}_i)$, then take the average for both of them. Repeat the same procedure for algorithms PermRowCol and ROWCOL. Since Qiskit handles the initial qubit mapping and returns a synthesized circuit $\vect{C}_i$ based on $G$, for each combination $(n,m)$, we can also average its synthesized gate count and $\Cost(\vect{C}_i)$, for $0 \leq i \leq 99$. To optimize the benchmark's runtime, we use the Python package Ray~\cite{Liang} to parallelize the execution of these four synthesis algorithms.

In \cref{sec:backends}, \cref{tab:CNOTNairobi,tab:CostNairobi} demonstrate benchmark results of synthesizing $5$- and $7$-qubit CNOT circuits on IBM's fake Nairobi backend. According to the $\Cost$ metric and the synthesized CNOT count, NAPermRowCol outperforms PermRowCol, ROWCOL, and Qiskit for all input circuits of different CNOT counts. This is also supported by the benchmark results on IBM's other backends. For more details, we encourage readers to check out \cref{tab:CNOTGuadalupe,tab:CostGuadalupe} and \cref{fig:Guadalupe5,fig:Guadalupe7} for CNOT synthesis on the fake Guadalupe backend, as well as \cref{tab:CNOTCairo,tab:CostCairo} and \cref{fig:Cairo5,fig:Cairo7} for CNOT synthesis on the fake Cairo backend. For the Nairobi backend, we discuss these benchmark results from two perspectives.

\paragraph{Compare the scalability of different synthesis algorithms}
Let $m$ be the original CNOT count. In \cref{fig:Nairobi5NCQ,fig:Nairobi5CostQ}, when $m$ is small, the performance of each algorithm is barely distinguishable. When $m$ increases, the synthesized CNOT count and the $\Cost$ of each algorithm increase. Among them, Qiskit shows the worst scalability. In \cref{fig:Nairobi5NCQ},
\begin{itemize}
    \item When $m \leq 8$, for each algorithm, its synthesized CNOT count is bounded by $12$, and its $\Cost$ is bounded by $0.12$. 
    \item When $m \geq 16$, Qiskit's synthesized CNOT count grows exponentially, while those of NAPermRowCol, PermRowCol, and ROWCOL remain nearly unchanged. They are bounded by 13, 14, and 20 respectively. 
    \item When $m \geq 256$, the synthesized CNOT count of these three algorithms is about 100 times fewer than that of Qiskit. 
\end{itemize}

Similarly, in \cref{fig:Nairobi5CostQ},

\begin{itemize}
    \item When $m \geq 16$, Qiskit's $\Cost$ explodes, while those of NAPermRowCol, PermRowCol, ROWCOL fluctuate around 0.11, 0.13, and 0.17 respectively.
    \item When $m \geq 256$, Qiskit's $\Cost$ saturates the maximum possible cost and reaches 1. In comparison, the $\Cost$ of NAPermRowCol, PermRowCol, and ROWCOL remains quite low. It is bounded by 0.12, 0.14, and 0.18 respectively, which is more than 5 times cheaper than Qiskit. 
\end{itemize}

In summary, NAPermRowCol, PermRowCol, and ROWCOL show impressive and comparable scalability when synthesizing a large $5$-qubit CNOT circuit.


\paragraph{Compare the relative performance of NAPermRowCol, PermRowCol, and ROWCOL}

Compared to \cref{fig:Nairobi5NCQ,fig:Nairobi5CostQ}, \cref{fig:Nairobi5NC,fig:Nairobi5Cost} get rid of the data related to Qiskit so that the remaining ones are distributed in a more compact area. They serve as the zoomed-in versions which allow us to compare the performance of NAPermRowCol, PermRowCol, and ROWCOL more closely. Compared to PermRowCol, NAPermRowCol reduces the synthesized CNOT count by about $16\%$, and it is about $13\%$ cheaper. Compared to ROWCOL, it reduces the synthesized CNOT count by about $42\%$, and it is about $48\%$ cheaper. In summary, both the $\Cost$ metric and the synthesized CNOT count show that NAPermRowCol outperforms PermRowCol and ROWCOL. 

These results also confirm the effectiveness of the noise-aware heuristics introduced in \cref{subsec: inner working}. In a reduction step, NAPermRowCol may take a detour to avoid expensive edges. Under both metrics, NAPermRowCol consistently beats PermRowCol. This means our greedy approach of selecting the cheapest path at each reduction step results in an additional improvement to optimize the overall synthesis results.

\begin{figure}[p]
\begin{subfigure}[t]{.5\textwidth}
  \centering
  \includegraphics[scale=0.42]{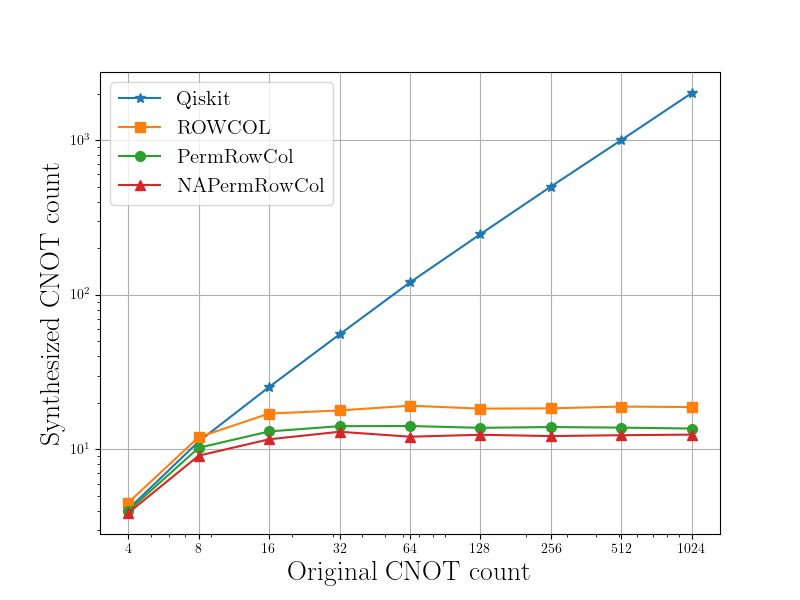}
  \subcaption{Compare NAPermRowCol with PermRowCol, ROWCOL, and Qiskit in terms of the synthesized CNOT count. Qiskit has the worst scalability when the original CNOT count grows exponentially.}
        \label{fig:Nairobi5NCQ}
\end{subfigure}
\quad
\begin{subfigure}[t]{.5\textwidth}
  \centering
  \includegraphics[scale=0.42]{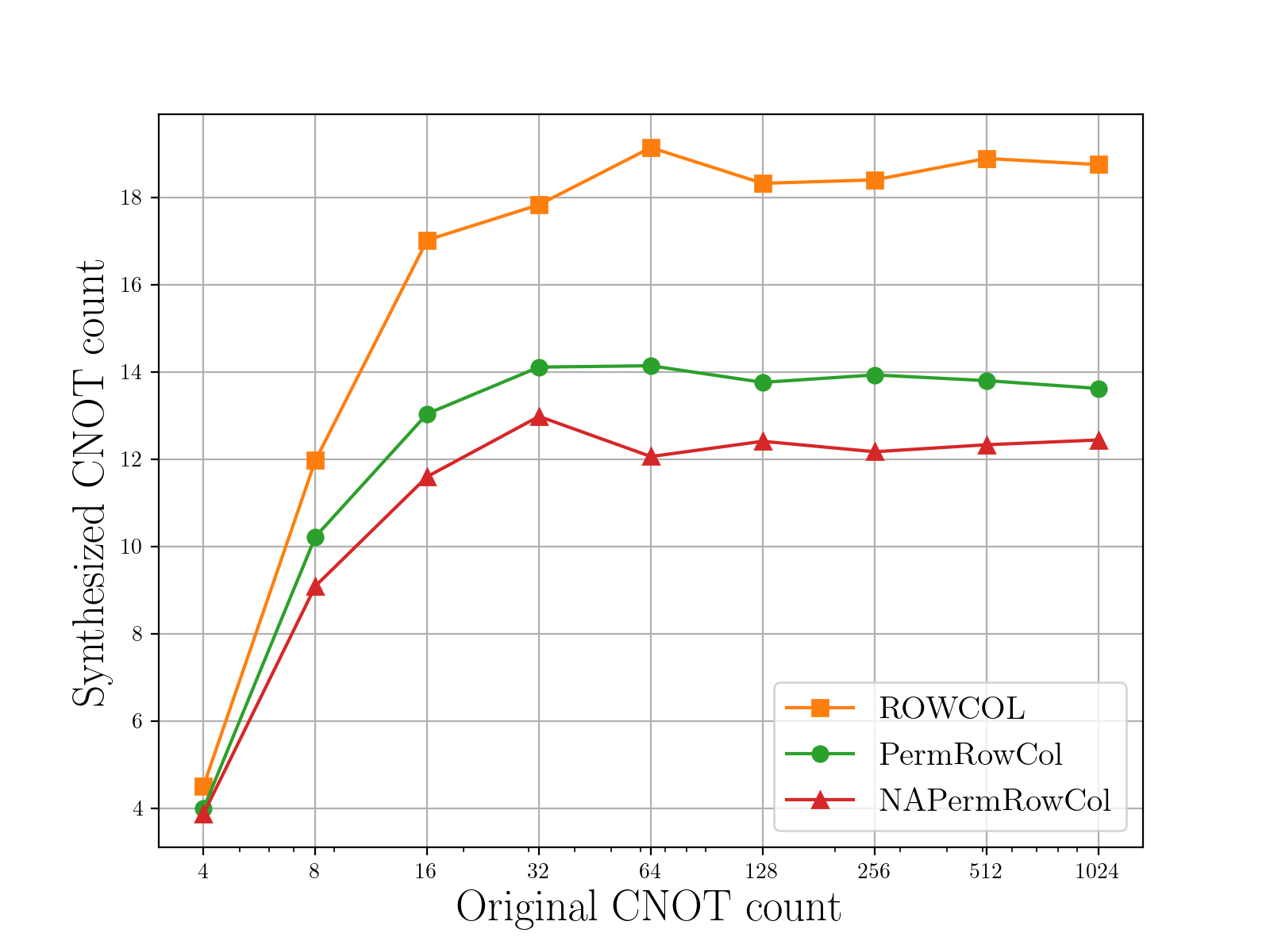}
  \subcaption{Compare NAPermRowCol with PermRowCol and ROWCOL in terms of the synthesized CNOT count. This is the zoomed-in version of the lefthand side.}
        \label{fig:Nairobi5NC}
\end{subfigure}
\begin{subfigure}[t]{.5\textwidth}
  \centering
  \includegraphics[scale=0.42]{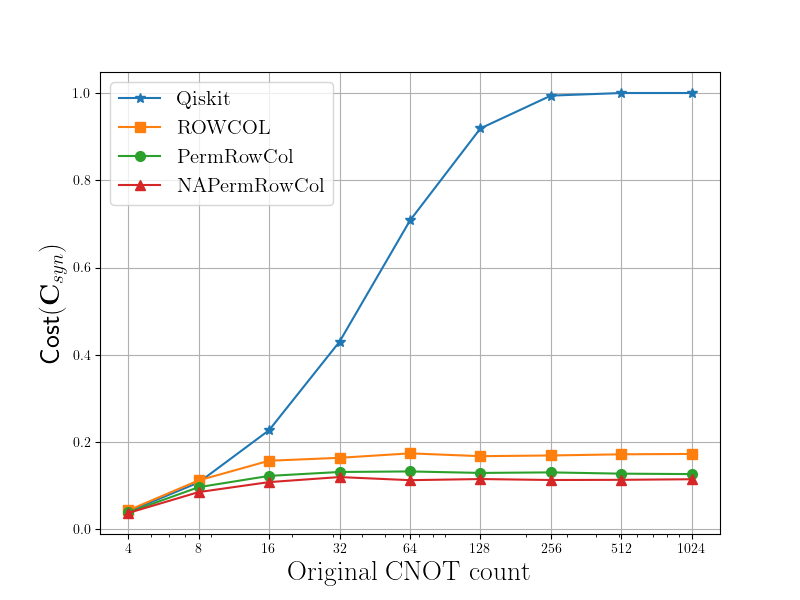}
  \subcaption{Compare NAPermRowCol with PermRowCol, ROWCOL, and Qiskit in terms of their respective costs. Qiskit has the worst scalability when the original CNOT count grows exponentially.}
        \label{fig:Nairobi5CostQ}
\end{subfigure}
\quad
\begin{subfigure}[t]{.5\textwidth}
  \centering
  \includegraphics[scale=0.42]{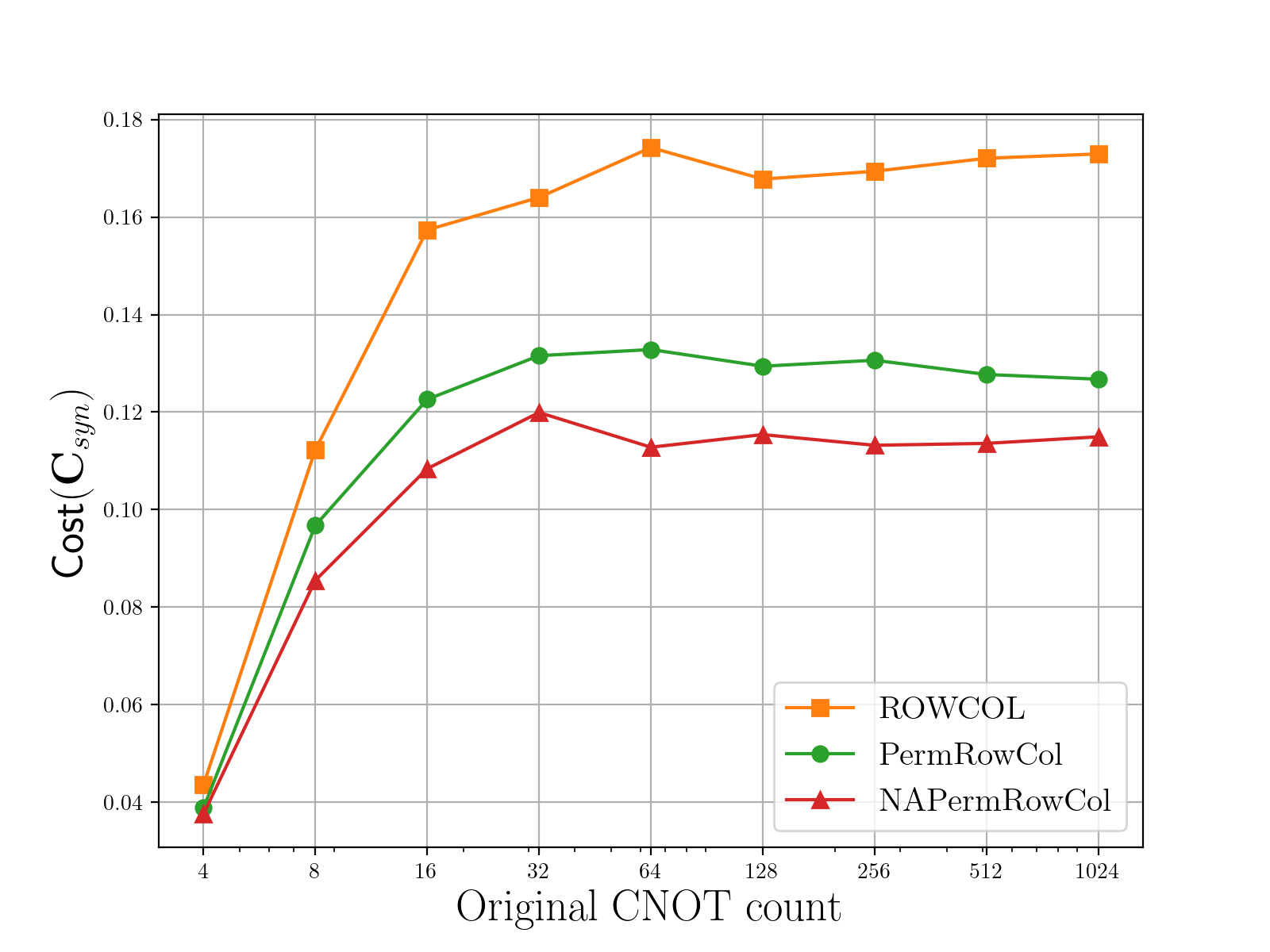}
  \subcaption{Compare NAPermRowCol with PermRowCol and ROWCOL in terms of their respective costs. This is the zoomed-in version of the lefthand side.}
        \label{fig:Nairobi5Cost}
\end{subfigure}
\caption{IBM's fake Nairobi backend hosts 7 qubits. We benchmark with its $5$-qubit connected subgraph and compare NAPermRowCol against three state-of-the-art CNOT synthesis algorithms. For each original CNOT count, input $100$ randomly generated CNOT circuits to each algorithm, obtain the synthesized CNOT circuits, then average their gate count and the circuit cost. The x-axis in each figure uses a logarithmic scale as the input gate count grows exponentially. The y-axis of \cref{fig:Nairobi5NCQ} uses a logarithmic scale, while the ones in \cref{fig:Nairobi5NC,fig:Nairobi5CostQ,fig:Nairobi5Cost} use a linear scale. Compared to \cref{fig:Nairobi5NCQ,fig:Nairobi5CostQ}, \cref{fig:Nairobi5NC,fig:Nairobi5Cost} get rid of the data related to Qiskit so that the remaining ones are distributed in a more compact area. They serve as the zoomed-in versions which allow us to compare the performance of NAPermRowCol, PermRowCol, and ROWCOL more closely. For all input circuits of different CNOT counts, NAPermRowCol outperforms other algorithms in terms of the synthesized CNOT count and circuit cost. It demonstrates remarkable scalability when the input circuit size grows exponentially.}
\label{fig:Nairobi5}
\end{figure}

\cref{fig:Nairobi7} demonstrate benchmark results of synthesizing a $7$-qubit CNOT circuit on IBM's fake Nairobi backend. As in \cref{fig:Nairobi5}, Qiskit shows the worst scalability when the original CNOT count increases. Under both metrics, NAPermRowCol outperforms PermRowCol, ROWCOL, and Qiskit for all input circuits of different CNOT counts. In addition, it shows that when synthesizing a wider CNOT circuit, the performance of each algorithm declines. Compare to \cref{fig:Nairobi5NC}, in \cref{fig:Nairobi7NC}, the maximum synthesized CNOT count of each algorithm has increased. For example, the synthesized CNOT count of NAPermRowCol is bounded by 31, but it is bounded by 12 when synthesizing a $5$-qubit CNOT circuit. Similarly, in \cref{fig:Nairobi7Cost}, the $\Cost$ of NAPermRowCol is bounded by 0.27, but it is bounded by 0.115 in \cref{fig:Nairobi5Cost}.

Moreover, the advantages of NAPermRowCol over PermRowCol and ROWCOL are less obvious. Compared to PermRowCol, NAPermRowCol reduces the synthesized CNOT count by about $13\%$, and it is about $11\%$ cheaper. Compared to ROWCOL, it reduces the synthesized CNOT count by about $35.5\%$, and it is about $26\%$ cheaper. In summary, the improvement of NAPermRowCol is suppressed when synthesizing a wider CNOT circuit. We can draw the same conclusion based on the benchmark results on IBM's other backends.  When the CNOT circuit has more than $15$ qubits, the performance between NAPermRowCol, PermRowCol, and ROWCOL is barely distinguishable. For more details, we encourage readers to check out \cref{fig:Guadalupe16,fig:Cairo16} for synthesizing $16$-qubit CNOT circuits on the fake Guadalupe and Cairo backends.

\begin{figure}[p]
\begin{subfigure}[t]{.5\textwidth}
  \centering
  \includegraphics[scale=0.42]{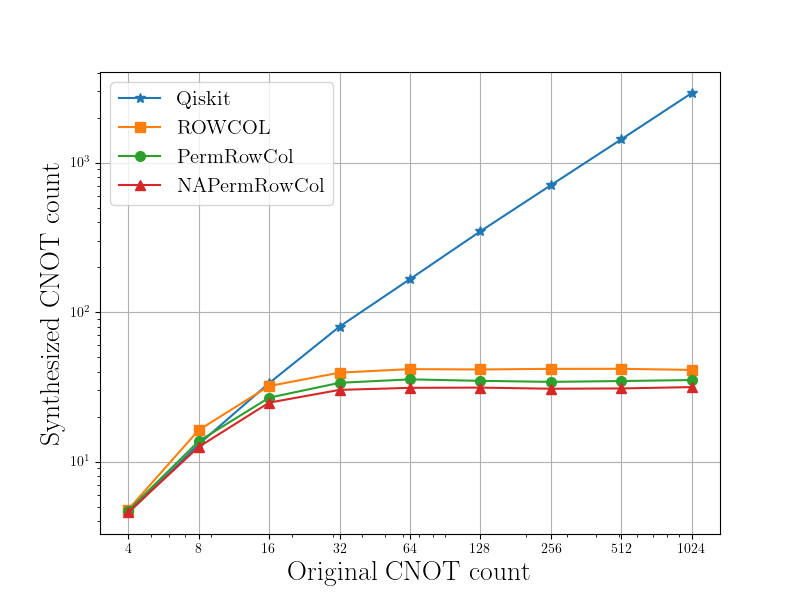}
  \subcaption{Compare NAPermRowCol with PermRowCol, ROWCOL, and Qiskit in terms of the synthesized CNOT count. Qiskit has the worst scalability when the original CNOT count grows exponentially.}
        \label{fig:Nairobi7NCQ}
\end{subfigure}
\quad
\begin{subfigure}[t]{.5\textwidth}
  \centering
  \includegraphics[scale=0.42]{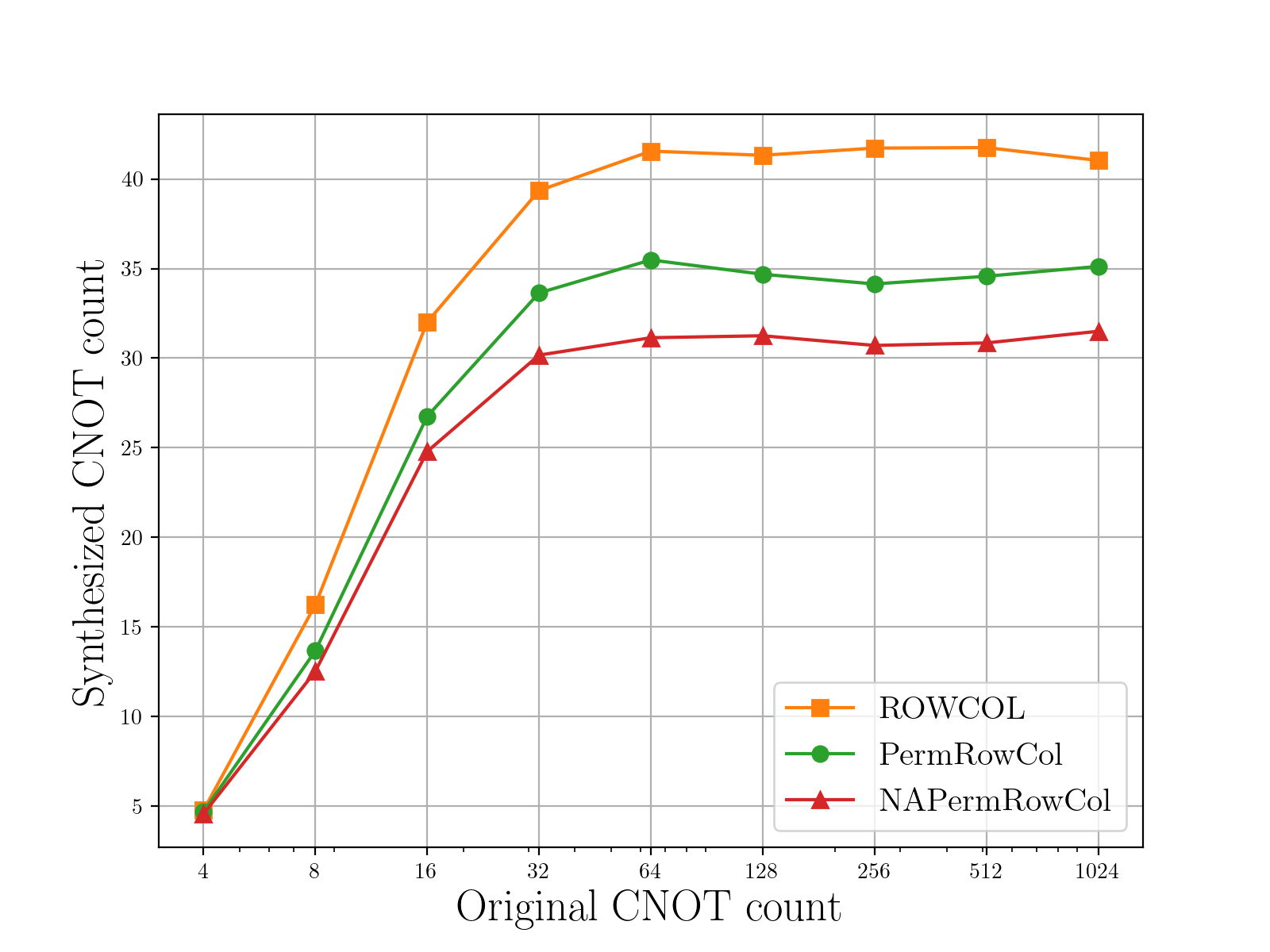}
  \subcaption{Compare NAPermRowCol with PermRowCol and ROWCOL in terms of the synthesized CNOT count. This is the zoomed-in version of the lefthand side.}
        \label{fig:Nairobi7NC}
\end{subfigure}
\begin{subfigure}[t]{.5\textwidth}
  \centering
  \includegraphics[scale=0.42]{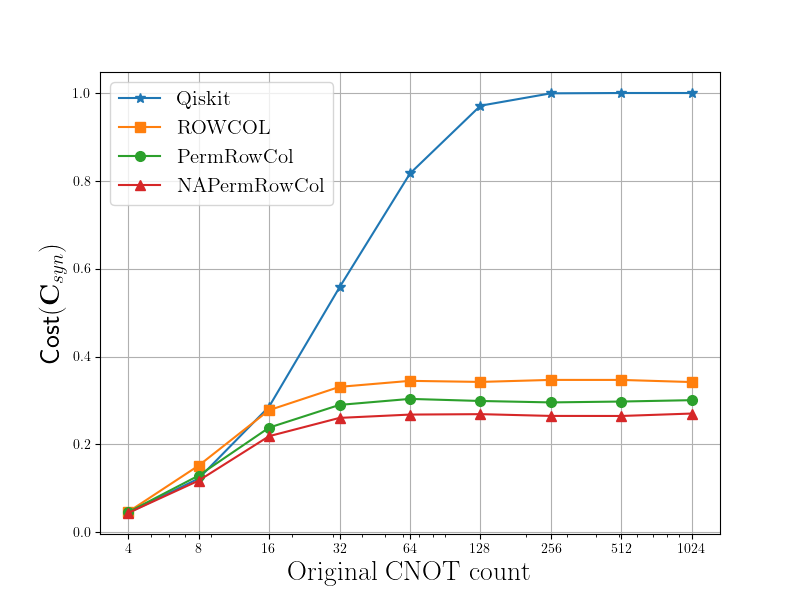}
  \subcaption{Compare NAPermRowCol with PermRowCol, ROWCOL, and Qiskit in terms of their respective costs. Qiskit has the worst scalability when the original CNOT count grows exponentially.}
        \label{fig:Nairobi7CostQ}
\end{subfigure}
\quad
\begin{subfigure}[t]{.5\textwidth}
  \centering
  \includegraphics[scale=0.42]{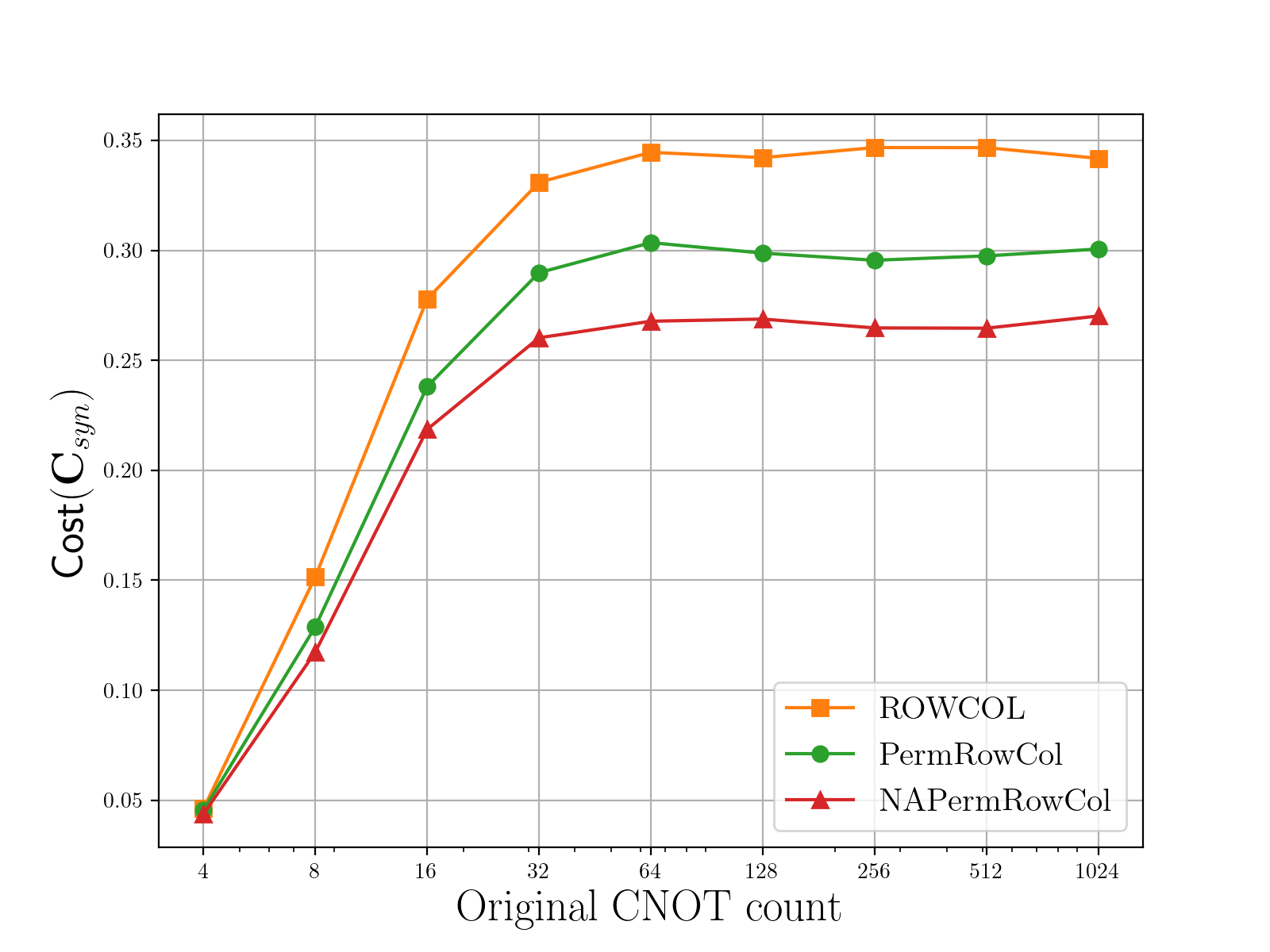}
  \subcaption{Compare NAPermRowCol with PermRowCol and ROWCOL in terms of their respective costs. This is the zoomed-in version of the lefthand side.}
        \label{fig:Nairobi7Cost}
\end{subfigure}
\caption{IBM's fake Nairobi backend hosts 7 qubits. We benchmark with its $7$-qubit connected subgraph and compare NAPermRowCol against three state-of-the-art CNOT synthesis algorithms. For each original CNOT count, input $100$ randomly generated CNOT circuits to each algorithm, obtain the synthesized CNOT circuits, then average their gate count and the circuit cost. The x-axis in each figure uses a logarithmic scale as the input gate count grows exponentially. The y-axis of \cref{fig:Nairobi7NCQ} uses a logarithmic scale, while the ones in \cref{fig:Nairobi7NC,fig:Nairobi7CostQ,fig:Nairobi7Cost} use a linear scale. Compared to \cref{fig:Nairobi7NCQ,fig:Nairobi7CostQ}, \cref{fig:Nairobi7NC,fig:Nairobi7Cost} get rid of the data related to Qiskit so that the remaining ones are distributed in a more compact area. They serve as the zoomed-in versions which allow us to compare the performance of NAPermRowCol, PermRowCol, and ROWCOL more closely. For input circuits of more than 16 CNOT counts, NAPermRowCol outperforms other algorithms in terms of the synthesized CNOT count and circuit cost. It demonstrates remarkable scalability when the input circuit size grows exponentially.}
\label{fig:Nairobi7}
\end{figure}

Finally, we compare the performance of NAPermRowCol with Qiskit, whose routing strategy uses ancillary qubits to connect subsets of non-adjacent qubits. In contrast, NAPermRowCol and other algorithms do not use ancilla and thus have much better scalability. In terms of the algorithm performance, NAPermRowCol is not as versatile as Qiskit, with a less-than-optimal strategy at each reduction step. Despite such an unfair comparison, NAPermRowCol consistently performs much better than Qiskit on all benchmarked backends, especially for CNOT circuits of large gate counts. This highlights the potential of NAPermRowCol and our error mitigation strategy to route a more complicated quantum circuits on various NISQ hardware.


\newpage
\section{Acknowledgement}
\label{sec:acknowledgement}
The diagrams were typeset using TikZit~\cite{tikz}. The authors would like to thank Christopher Wilson for consultations about superconducting quantum computers; Kohdai Kuroiwa for suggestions to improve the $\Cost$ function; Arianne Meijer - van de Griend and Ewan Murphy for discussing PermRowCol and its benchmark details; Gyungmin Cho and Chanung Park for fruitful discussions on approximating the average gate fidelity; Neil J. Ross, Peter Selinger, John van de Wetering, as well as anonymous reviewers at the Theory of Quantum Computation, Communication and Cryptography (TQC) and Quantum Physics and Logic (QPL) for helpful comments on an earlier version of this paper. 

MK and DK gratefully acknowledge the support of a National Research Foundation of Korea (NRF) grant funded by the Korean Government (MSIT) (No. 2019M3E4A1080144, No. 2019M3E4A1080145, No. 2019R1A5A1027055, RS-2023-00283291, SRC Center for Quantum Coherence in Condensed Matter RS-2023-00207732, and No. 2023R1A2C2005809) and a core center program grant funded by the Ministry of Education (No. 2021R1A6C101B418). SML and MM thank NTT Research for their financial and technical support. This work was supported in part by Canada’s NSERC. Research at IQC is supported in part by the Government of Canada through Innovation, Science and Economic Development Canada. 
Research at Perimeter Institute is supported in part by the Government of Canada through the Department of Innovation, Science and Economic Development Canada and by the Province of Ontario through the Ministry of Colleges and Universities.

\bibliographystyle{eptcs}
\bibliography{NAPRC}

\appendix
\section{Complementary Proofs}
\label{sec: proof of claim}
In this section, we provide complementary proof details for statements in \cref{subsubsec:average gate fidelity,sec:ApproximatedCostFunction}.

\begin{T3}
  \lemsim
\end{T3}   

\begin{proof}
    According to \cite{fortunato2002implementation,schumacher1996sending}, the process fidelity can be expressed as
\begin{align}
    F_{pro}(\mathcal{E}') = \sum_k\lvert \Tr[U^{\dagger}M_k]/t\rvert^2.
    \label{step: 1}
\end{align}
By direct computation, we rewrite \cref{step: 1} as
\begin{align}
    F_{pro}(\mathcal{E}') &= \frac{1}{t^2}\sum_{k} (\Tr[U^{\dagger}M_k])^*\Tr[U^{\dagger}M_k]\nonumber \\
    & = \frac{1}{t^2}\sum_{k} \Tr[U^{T}M_k^*]\Tr[U^{\dagger}M_k]\nonumber \\ 
    & = \frac{1}{t^2}\sum_{k} \Tr[(U^{T}M_k^*)\otimes (U^{\dagger}M_k)]\nonumber \\ 
    &= \frac{1}{t^2}\Tr\Bigl[\sum_{k}(U^{T}M_k^*)\otimes (U^{\dagger}M_k)\Bigr]\nonumber\\
    &= \frac{1}{t^2}\Tr\Bigl[\sum_{k}(U^{\dagger}M_k)^{*}\otimes (U^{\dagger}M_k)\Bigr]\nonumber\\
    & = \frac{1}{t^2}\Tr\Bigl[\sum_{k}M_k{'}^{*}\otimes M'_k\Bigr]. \label{step: 2}
\end{align}
By \cref{lem: superoperator not unitary}, since the action of $\mathcal{E}'$ is described by $\{M_k';\;\;\ M_k' = U^{\dagger}M_k\}$, $\sum_{k}(M_k^*\otimes M_k) = S_{\mathcal{E}'}$. Hence,
\[
F_{pro}(\mathcal{E}') = \frac{\Tr[S_{\mathcal{E}'}]}{t^2}.
\]
\end{proof}

\begin{T2}
  \lempara
\end{T2}  
\begin{proof}
Applying \cref{lem:trace tensor} with \cref{eqn:trace and epsilon,eqn:trace and epsilon intermediate}, we have
 \begin{align*}
     F_{avg}(\mathcal{E}_q) & =F_{avg}(\mathcal{E}_0 \otimes \mathcal{E}_1)= \frac{\Tr[S_{\mathcal{E}_0\otimes \mathcal{E}_1}]}{(d_0d_1)^2}\times \frac{d_0d_1}{d_0d_1+1} + \frac{1}{d_0d_1+1}\nonumber\\
     & =\frac{\Tr[S_{\mathcal{E}_0}]}{d_0^2}\frac{\Tr[S_{\mathcal{E}_1}]}{d_1^2}\times \frac{d_0d_1}{d_0d_1+1} + \frac{1}{d_0d_1+1}\nonumber\\
     & =  \Biggl(1-\frac{d_0+1}{d_0}p_0\Biggr)\Biggl(1-\frac{d_1+1}{d_1}p_1\Biggr)\frac{d_0d_1}{d_0d_1+1}+ \frac{1}{d_0d_1+1}\nonumber\\
     & = 1 - \Biggl(\frac{d_0+1}{d_0}p_0+\frac{d_1+1}{d_1}p_1\Biggr)\frac{d_0d_1}{d_0d_1+1}+\frac{(d_0+1)(d_1+1)}{d_0d_1+1}p_0p_1\nonumber\\
     & = 1 - \Biggl(\frac{d_0d_1+d_1+1-1}{d_0d_1+1}p_0+\frac{d_0d_1+d_0+1-1}{d_0d_1+1}p_1\Biggr)+\frac{d_0d_1+d_0+d_1+1}{d_0d_1+1}p_0p_1\nonumber\\
     & = 1-p_0 - \frac{d_1-1}{d_0d_1+1}p_0-p_1 - \frac{d_0-1}{d_0d_1+1}p_1+p_0p_1 +\frac{(d_0+d_1)p_0p_1}{d_0d_1+1}\nonumber\\
     & = 1-p_0 -p_1 + p_0p_1 + \frac{(1-d_1)p_0 + (1-d_0)p_1 + (d_0+d_1)p_0p_1}{d_0d_1+1}
\end{align*}
\end{proof}

\begin{T1}
  \lemtext
\end{T1}

\begin{proof}
According to \cref{cor:sup multiqubit}, for $k\in \Z_2$,
\begin{align}
    A_k = I - S_{\mathcal{E}_k} = I - \sum_{E_i \in \mathcal{B}_n}P^k_iS^k_{E_i}.
    \label{eqn: A}
\end{align}

Let $E^k_0$ be the $t^2 \times t^2$ identity matrix $I$ with probability $P^k_0$. \cref{eqn: A} can be expressed as
\begin{align}    
   A_k  = I - P^k_0E^k_0 - \sum_{E_i \in \mathcal{B}_n\setminus \{E_0^k\}} P_i^k S_{E_i} = (1 -P^k_0)I- \sum_{E_i \in \mathcal{B}_n\setminus \{I\}} P_i^k S_{E_i}.\label{eqn: B}
\end{align}
It follows that 
\begin{align}
    A_1A_0 &= \Bigl( \bigl(1 -P^1_0\bigr)I- \sum_{E_i \in \mathcal{B}_n\setminus \{I\}} P^1_i S_{E_i} \Bigr)\Bigl( \bigl((1 -P^0_0\bigr)I- \sum_{E_j \in \mathcal{B}_n\setminus \{I\}} P^0_j S_{E_j} \Bigr)\nonumber\\
    &= \bigl(1 -P^1_0\bigr)\bigl(1 -P^0_0\bigr)I - \bigl(1 -P^1_0\bigr)\sum_{E_j \in \mathcal{B}_n\setminus \{I\}} P^0_j S_{E_j} - (1 -P^0_0\bigr)\sum_{E_i \in \mathcal{B}_n\setminus \{I\}} P^1_i S_{E_i} + \sum_{\substack{E_i \in \mathcal{B}_n\setminus \{I\}\\E_j \in \mathcal{B}_n\setminus \{I\}}} \bigl(P^1_i S_{E_i}\bigr) \bigl(P^0_j S_{E_j}\bigr).
\end{align}
By linearity and \cref{rmk:trace},
\begin{align}
    \Tr[A_1A_0] = \Tr\Biggl[&\bigl(1 -P^1_0\bigr)\bigl(1 -P^0_0\bigr)I - \bigl(1 -P^1_0\bigr)\sum_{E_j \in \mathcal{B}_n\setminus \{I\}} P^0_j S_{E_j} - (1 -P^0_0\bigr)\sum_{E_i \in \mathcal{B}_n\setminus \{I\}} P^1_i S_{E_i}\nonumber \\ 
    & + \sum_{\substack{E_i \in \mathcal{B}_n\setminus \{I\}\\E_j \in \mathcal{B}_n\setminus \{I\}}} \bigl(P^1_i S_{E_i}\bigr) \bigl(P^0_j S_{E_j}\bigr)\Biggr]\nonumber\\
    =\Tr\Biggl[&\bigl(1 -P^1_0\bigr)\bigl(1 -P^0_0\bigr)I\Biggr] - \Tr\Biggl[\bigl(1 -P^1_0\bigr)\sum_{E_j \in \mathcal{B}_n\setminus \{I\}} P^0_j S_{E_j}\Biggr] - \Tr\Biggl[(1 -P^0_0\bigr)\sum_{E_i \in \mathcal{B}_n\setminus \{I\}} P^1_i S_{E_i}\Biggr] \nonumber\\
    + \Tr&\Biggl[\sum_{\substack{E_i \in \mathcal{B}_n\setminus \{I\}\\E_j \in \mathcal{B}_n\setminus \{I\}}} \bigl(P^1_i S_{E_i}\bigr) \bigl(P^0_j S_{E_j}\bigr)\Biggr].\nonumber\\
    = \bigl(1 -&P^1_0\bigr)\bigl(1 -P^0_0\bigr)t^2 - \bigl(1 -P^1_0\bigr)\sum_{E_j \in \mathcal{B}_n\setminus \{I\}}P_j^0\Tr[S_{E_j}] - \bigl(1 -P^0_0\bigr)\sum_{E_i \in \mathcal{B}_n\setminus \{I\}}P_i^1\Tr[S_{E_i}] \nonumber \\ 
    &  + \sum_{\substack{E_i \in \mathcal{B}_n\setminus \{I\}\\E_j \in \mathcal{B}_n\setminus \{I\}}}P^1_iP^0_j\Tr[S_{E_i}S_{E_j}].\nonumber\\
    = \bigl(1 -&P^1_0\bigr)\bigl(1 -P^0_0\bigr)t^2 + \sum_{\substack{E_i \in \mathcal{B}_n\setminus \{I\}\\E_j \in \mathcal{B}_n\setminus \{I\}}}P^1_iP^0_j\Tr[S_{E_i}S_{E_j}].
    \label{eqn: C}
\end{align}
For $1 \leq i,j \leq 4^n -1$, since $E_i$ and $E_j$ are unitary, by \cref{lem: superoperator unitary},
\begin{align}
    S_{E_i}S_{E_j} = (E_i^{*}\otimes E_i)(E_j^{*}\otimes E_j) = (E_i^{*}E_j^{*})\otimes(E_iE_j). \label{eqn: D}
\end{align}
Since 
\[
\Tr[(E_i^{*}E_j^{*})\otimes(E_iE_j)] = \Tr[E_i^{*}E_j^{*}]\Tr[E_iE_j] = \bigl(\Tr[E_jE_i]\bigr)^{*}\Tr[E_iE_j],
\]
combined with \cref{eqn: C,eqn: D,lem: helper}, we have
\begin{align}
    \Tr[A_1A_0] &= \bigl(1 -P^1_0\bigr)\bigl(1 -P^0_0\bigr)t^2 + \sum_{\substack{E_i \in \mathcal{B}_n\setminus \{I\}\\E_j \in \mathcal{B}_n\setminus \{I\}}}P^1_iP^0_j\bigl(\Tr[E_jE_i]\bigr)^{*}\Tr[E_iE_j].\nonumber\\
    &=\bigl(1 -P^1_0\bigr)\bigl(1 -P^0_0\bigr)t^2 + \sum_{i=1}^{4^n - 1}P^1_iP^0_it^2.\nonumber\\
    &= t^2\Bigl(\bigl(1 -P^1_0\bigr)\bigl(1 -P^0_0\bigr)+\sum_{i=1}^{4^n - 1}P^1_iP_i^0\Bigr).
    \label{eqn: E}
\end{align}

By \cref{eqn:trace and epsilon,rmk:trace}, 
\[
t^2\biggl(1-\frac{t+1}{t}p_k\biggr) = \Tr\biggl[\sum_{E_i \in \mathcal{B}_n}P^k_iS_{E^k_i}\biggr] = P_0^k\Tr[S_{I}] + \sum_{i=1}^{4^n - 1}P^k_i\Tr[S_{E^k_i}] = t^2P_0^k.
\]

It follows that
\begin{align}
    P_0^k = 1 - \frac{t+1}{t}p_k, \qquad \sum_{i=1}^{4^n - 1}P^k_i = 1 - P_0^k = \frac{t+1}{t}p_k.
    \label{eqn: G}
\end{align}

Define a function $f(P_i^0, P_i^1)$ with domain $D$ as follows.
\[
f(P_i^0, P_i^1) = \sum_{i=1}^{4^n - 1}P_i^1P_i^0, \qquad D=\Biggl\{P_i^k;\;\; \sum_{i=1}^{4^n - 1}P^k_i = \frac{t+1}{t}p_k,\;\; k\in\Z_2, \;\; P_i^k \geq 0,\;\; 1 \leq i \leq 4^n - 1 \Biggr\}.
\]

By \cref{eqn: G}, \cref{eqn: E} can be simplified as 
\begin{align}
    \Tr[A_1A_0] = t^2\Biggl(\frac{(t+1)^2}{t^2}p_0p_1+\sum_{i=1}^{4^n - 1}P^1_iP_i^0\Biggr)=(t+1)^2p_0p_1 + t^2f(P_i^0, P_i^1).
    \label{eqn: F}
\end{align}

To bound \cref{eqn: F}, our problem is reduced to identifying the maximum and minimum values of $f$, $f_{\max}$ and $f_{\min}$, on $D$. By the extreme value theorem, $f_{\max}$ and $f_{\min}$ are either the extremum points of $f$ or located on the boundary of $D$. Using the Lagrange Multiplier method~\cite{LarangeM}, the extremum point is reached when all $P_i^k$ have the same values. That is, for all $1 \leq i \leq 4^n - 1$
\[
P_i^0 = \frac{(t+1)p_0}{t(t^2-1)}, \qquad P_i^1 = \frac{(t+1)p_1}{t(t^2-1)}.
\]

Since $t^2 = 4^n$,
\begin{align}
    f_{ex} = \sum_{i=1}^{4^n - 1} \Biggl(\frac{(t+1)p_0}{t(t^2-1)}\Biggr) \Biggl(\frac{(t+1)p_1}{t(t^2-1)}\Biggr) = \frac{(t+1)^2}{(t^2-1)t^2}p_0p_1.
    \label{eqn: extremum of f}
\end{align}

For $k \in \Z_2$, $0 \leq P_i^k \leq \frac{t+1}{t}p_k$. The boundary of $D$ is reached when for all $1 \leq i \leq t^2 - 1$, $P_i^1P_i^0 = 0$. Thus $f = 0$. Or for some $1 \leq i \leq t^2 - 1$, $P_i^k = \frac{t+1}{t}p_k$. Thus
\begin{align*}
    f  = \sum_{i=1}^{t^2 - 1}P_i^1P_i^0 = \Biggl(\frac{t+1}{t}p_0\Biggr) \Biggl(\frac{t+1}{t}p_1\Biggr) + 0 = \frac{(t+1)^2}{t^2}p_0p_1
\end{align*}

Therefore, 
\begin{align}
    f_{\min} = 0, \qquad f_{\max} = \frac{(t+1)^2}{t^2}p_0p_1.
    \label{eqn: max min of f}
\end{align}

Combining \cref{eqn: max min of f,eqn: F}, we have
\begin{align*}
     (t+1)^2p_0p_1 \leq \Tr[A_1A_0]\leq 2(t+1)^2p_0p_1. \label{eqn: final trace}
\end{align*}
\end{proof}

\section{Package Version and Hardware Specifics}
\label{sec: version and hardware specifics}

To generate simulation and benchmark results, we utilize two distinct hardware environments and various Python packages. The important ones are documented here. For other packages, please refer to the GitHub repository\footnote{\url{https://github.com/Minyoung-Kim1110/NoiseAwareSynthesis}}. Qiskit Aer is installed from the source code\footnote{\url{https://qiskit.org/ecosystem/aer/getting_started.html}}.

\begin{table}[H]
\centering
\begin{tabular}{|l|l|}
\hline
\textbf{Python packages} & \textbf{Version} \\\hline 
Qiskit & 0.45.2\\\hline
Qiskit Aer & 0.13.2\\\hline
pyzx & 0.7.3 \\\hline 
\textbf{System information} & \textbf{Hardware Specifics} \\\hline 
Python & 3.8.17 \\\hline 
Operating System & Ubuntu 22.04.2 LTS  \\\hline 
CPU & AMD EPYC 7763 64-Core Processor @ 2.45GHz\\\hline 
Memory & 1.23 TB\\\hline 
\end{tabular}
\caption{The simulation results reported in \cref{subsec: compare functions,sec: compare different cost functions} are generated under these specifications.}
\end{table}

\begin{table}[H]
\centering
\begin{tabular}{|l|l|}
\hline
\textbf{Python packages} & \textbf{Version} \\\hline 
Qiskit & 0.44.3\\\hline
Qiskit Terra & 0.25.3\\\hline
Qiskit Aer & 0.13.0\\\hline
\textbf{System information} & \textbf{Hardware Specifics} \\\hline 
Python & 3.8.18 \\\hline 
Operating System & Window 10 Education \\\hline 
CPU & AMD Ryzen Threadripper 3970X 32-core @3.69 GHz\\\hline 
Memory & 256 GB\\\hline 
\end{tabular}
\caption{The benchmark results reported in \cref{subsec: compare algorithms,sec: compare diff algorithms} are generated under these specifications.}
\end{table}

Note that here, "pyzx" is not a package. It is a module from the "rowcol" branch in the GitHub repo\footnote{\url{https://github.com/Aerylia/pyzx}}.
\section{Backends for Simulation and Benchmarking}
\label{sec:backends}
We use the empirical data from IBM's fake Nairobi, Guadalupe, and Cairo backends to compare different cost functions and the performance of different CNOT synthesis algorithms.

\subsection{IBM's Fake Nairobi Backend}
\begin{figure}[H]
\begin{subfigure}[t]{.67\textwidth}
    \centering
    \scalebox{0.8}{\begin{tikzpicture}
	\begin{pgfonlayer}{nodelayer}
		\node [style=new style 0] (0) at (-12.5, 6.25) {$0$};
		\node [style=new style 0] (2) at (-8, 6.25) {$1$};
		\node [style=new style 0] (3) at (-3.25, 6.25) {$2$};
		\node [style=new style 0] (9) at (-8, -2.25) {$5$};
		\node [style=new style 0] (10) at (-3.25, -2.25) {$6$};
		\node [style=new style 0] (14) at (-8, 2) {$3$};
		\node [style=new style 0] (18) at (-12.5, -2.25) {$4$};
	\end{pgfonlayer}
	\begin{pgfonlayer}{edgelayer}
		\draw [style=edge] (0) to (2);
		\draw [style=edge] (2) to (3);
		\draw [style=edge] (9) to (10);
		\draw [style=edge] (2) to (14);
		\draw [style=edge] (14) to (9);
		\draw (18) to (9);
	\end{pgfonlayer}
\end{tikzpicture}}
    \vspace{1.5 em}
    \subcaption{Each vertex corresponds to a physical qubit. Each edge represents a CNOT gate that can be performed on the qubits corresponding to its endpoints.}
        \label{fig:Nairobi topology}
\end{subfigure}
\quad
\begin{subfigure}[t]{.33\textwidth}
\centering
\begin{tabular}{|l|l|}
\hline
Edge    & Edge Weight\\ \hline
(0, 1) & 0.00777\\ \hline
(1, 2) & 0.00607\\ \hline
(1, 3) & 0.00792\\ \hline
(3, 5) & 0.01016\\ \hline
(4, 5) & 0.00619\\ \hline
(5, 6) & 0.00918\\ \hline
\end{tabular}
\caption{For $e = (u, v) \in E_G$, $\omega_G(e)$ is the error rate of coupling physical qubits u and v.}
\label{tab: noises nairobi}
\end{subfigure}
    \caption{$G = (V_G, E_G, \omega_G)$ is the connectivity graph of IBM's fake Nairobi backend. It is an undirected edge-weighted connected graph. $\lvert V_G\rvert = 7$, $\omega_G: E_G \rightarrow \{x \in \R; \; 0 \leq x < 1\}$.}
    \label{fig:Nairobi}
\end{figure}

\subsection{IBM's Fake Guadalupe Backend}

\begin{figure}[H]
\begin{subfigure}[t]{.67\textwidth}
    \centering
    \scalebox{0.8}{\begin{tikzpicture}
	\begin{pgfonlayer}{nodelayer}
		\node [style=new style 0] (0) at (-12.25, 5.75) {$0$};
		\node [style=new style 0] (2) at (-8.5, 5.75) {$1$};
		\node [style=new style 0] (3) at (-4.75, 5.75) {$4$};
		\node [style=new style 0] (6) at (2.75, 5.75) {$10$};
		\node [style=new style 0] (7) at (6.5, 5.75) {$12$};
		\node [style=new style 0] (8) at (10.25, 5.75) {$15$};
		\node [style=new style 0] (9) at (-8.5, -2) {$3$};
		\node [style=new style 0] (10) at (-4.75, -2) {$5$};
		\node [style=new style 0] (11) at (-1, -2) {$8$};
		\node [style=new style 0] (12) at (2.75, -2) {$11$};
		\node [style=new style 0] (13) at (6.5, -2) {$14$};
		\node [style=new style 0] (14) at (-8.5, 2) {$2$};
		\node [style=new style 0] (15) at (6.5, 2) {$13$};
		\node [style=new style 0] (16) at (-1, 9.75) {$6$};
		\node [style=new style 0] (17) at (-1, 5.75) {$7$};
		\node [style=new style 0] (18) at (-1, -6.25) {$9$};
	\end{pgfonlayer}
	\begin{pgfonlayer}{edgelayer}
		\draw [style=edge] (0) to (2);
		\draw [style=edge] (2) to (3);
		\draw [style=edge] (7) to (8);
		\draw [style=edge] (6) to (7);
		\draw [style=edge] (9) to (10);
		\draw [style=edge] (11) to (12);
		\draw [style=edge] (10) to (11);
		\draw [style=edge] (12) to (13);
		\draw [style=edge] (2) to (14);
		\draw [style=edge] (14) to (9);
		\draw [style=edge] (7) to (15);
		\draw [style=edge] (15) to (13);
		\draw [style=edge] (3) to (17);
		\draw [style=edge] (16) to (17);
		\draw [style=edge] (17) to (6);
		\draw [style=edge] (11) to (18);
	\end{pgfonlayer}
\end{tikzpicture}}
    \vspace{3.7 em}
    \subcaption{Each vertex corresponds to a physical qubit. Each edge represents a CNOT gate that can be performed on the qubits corresponding to its endpoints.}
        \label{fig:Guadalupe topology}
\end{subfigure}
\quad
\begin{subfigure}[t]{.33\textwidth}
\centering
\begin{tabular}{|l|l|}
\hline
Edge     & Edge Weight \\ \hline
(0, 1) & 0.009690\\ \hline
(1, 2) & 0.015158\\ \hline
(1, 4) & 0.007311\\ \hline
(2, 3) & 0.013654\\ \hline
(3, 5) & 0.012821\\ \hline
(4, 7) & 0.011911\\ \hline
(5, 8) & 0.008868\\ \hline
(6, 7) & 0.006946\\ \hline
(7, 10) & 0.006762\\ \hline
(8, 9) & 0.012718\\ \hline
(8, 11) & 0.009196\\ \hline
(10, 12) & 0.019895\\ \hline
(11, 14) & 0.010583\\ \hline
(12, 13) & 0.007202\\ \hline
(12, 15) & 0.007804\\ \hline
(13, 14) & 0.012091\\ \hline
\end{tabular}
\caption{For $e = (u, v) \in E_G$, $\omega_G(e)$ is the error rate of coupling physical qubits u and v.}
\label{tab: noises guadalupe}
\end{subfigure}
    \caption{$G = (V_G, E_G, \omega_G)$ is the connectivity graph of IBM's fake Guadalupe backend. It is an undirected edge-weighted connected graph. $\lvert V_G\rvert = 16$, $\omega_G: E_G \rightarrow \{x \in \R; \; 0 \leq x < 1\}$.}
    \label{fig:Guadalupe}
\end{figure}

\subsection{IBM's Fake Cairo Backend}
\begin{figure}[H]
\begin{subfigure}[t]{1\textwidth}
    \centering
    \scalebox{0.8}{\begin{tikzpicture}
	\begin{pgfonlayer}{nodelayer}
		\node [style=new style 0] (0) at (-12.25, 5.75) {$0$};
		\node [style=new style 0] (2) at (-8.5, 5.75) {$1$};
		\node [style=new style 0] (3) at (-4.75, 5.75) {$4$};
		\node [style=new style 0] (6) at (2.75, 5.75) {$10$};
		\node [style=new style 0] (7) at (6.5, 5.75) {$12$};
		\node [style=new style 0] (9) at (-8.5, -2) {$3$};
		\node [style=new style 0] (10) at (-4.75, -2) {$5$};
		\node [style=new style 0] (11) at (-1, -2) {$8$};
		\node [style=new style 0] (12) at (2.75, -2) {$11$};
		\node [style=new style 0] (13) at (6.5, -2) {$14$};
		\node [style=new style 0] (14) at (-8.5, 2) {$2$};
		\node [style=new style 0] (15) at (6.5, 2) {$13$};
		\node [style=new style 0] (16) at (-1, 9.75) {$6$};
		\node [style=new style 0] (17) at (-1, 5.75) {$7$};
		\node [style=new style 0] (18) at (-1, -6.25) {$9$};
		\node [style=new style 0] (21) at (10.25, 5.75) {$15$};
		\node [style=new style 0] (22) at (17.75, 5.75) {$21$};
		\node [style=new style 0] (23) at (21.5, 5.75) {$23$};
		\node [style=new style 0] (25) at (10.25, -2) {$16$};
		\node [style=new style 0] (26) at (14, -2) {$19$};
		\node [style=new style 0] (27) at (17.75, -2) {$22$};
		\node [style=new style 0] (28) at (21.5, -2) {$25$};
		\node [style=new style 0] (30) at (21.5, 2) {$24$};
		\node [style=new style 0] (31) at (14, 9.75) {$17$};
		\node [style=new style 0] (32) at (14, 5.75) {$18$};
		\node [style=new style 0] (33) at (14, -6.25) {$20$};
		\node [style=new style 0] (34) at (25.25, -2) {$26$};
	\end{pgfonlayer}
	\begin{pgfonlayer}{edgelayer}
		\draw [style=edge] (0) to (2);
		\draw [style=edge] (2) to (3);
		\draw [style=edge] (6) to (7);
		\draw [style=edge] (9) to (10);
		\draw [style=edge] (11) to (12);
		\draw [style=edge] (10) to (11);
		\draw [style=edge] (12) to (13);
		\draw [style=edge] (2) to (14);
		\draw [style=edge] (14) to (9);
		\draw [style=edge] (7) to (15);
		\draw [style=edge] (15) to (13);
		\draw [style=edge] (3) to (17);
		\draw [style=edge] (16) to (17);
		\draw [style=edge] (17) to (6);
		\draw [style=edge] (11) to (18);
		\draw [style=edge] (22) to (23);
		\draw [style=edge] (26) to (27);
		\draw [style=edge] (25) to (26);
		\draw [style=edge] (27) to (28);
		\draw [style=edge] (23) to (30);
		\draw [style=edge] (30) to (28);
		\draw [style=edge] (21) to (32);
		\draw [style=edge] (31) to (32);
		\draw [style=edge] (32) to (22);
		\draw [style=edge] (26) to (33);
		\draw (7) to (21);
		\draw (13) to (25);
		\draw (28) to (34);
	\end{pgfonlayer}
\end{tikzpicture}}
    \vspace{1.5 em}
    \subcaption{Each vertex corresponds to a physical qubit. Each edge represents a CNOT gate that can be performed on the qubits corresponding to its endpoints.}
        \label{fig:Cairo topology}
\end{subfigure}

\vspace{2 em}

\begin{subfigure}[h]{1\textwidth}
   \centering
\begin{tabular}{|l|l|ll}
\hline
Edge     & Edge Weight & \multicolumn{1}{l|}{Edge}     & \multicolumn{1}{l|}{Edge Weight} \\ \hline
(0, 1) & 0.025728& \multicolumn{1}{l|}{(12, 15)} & \multicolumn{1}{l|}{0.008980} \\\hline
(1, 2) & 0.006662& \multicolumn{1}{l|}{(13, 14)} & \multicolumn{1}{l|}{0.004904} \\\hline
(1, 4) & 0.011427& \multicolumn{1}{l|}{(14, 16)} & \multicolumn{1}{l|}{0.005036} \\\hline
(2, 3) & 0.011890& \multicolumn{1}{l|}{(15, 18)} & \multicolumn{1}{l|}{0.005864} \\\hline
(3, 5) & 0.005375& \multicolumn{1}{l|}{(16, 19)} & \multicolumn{1}{l|}{0.007042} \\\hline
(4, 7) & 0.016432& \multicolumn{1}{l|}{(17, 18)} & \multicolumn{1}{l|}{0.009230} \\\hline
(5, 8) & 0.004620& \multicolumn{1}{l|}{(18, 21)} & \multicolumn{1}{l|}{0.005924} \\\hline
(6, 7) & 0.014319& \multicolumn{1}{l|}{(19, 20)} & \multicolumn{1}{l|}{0.007014} \\\hline
(7, 10) & 0.022012& \multicolumn{1}{l|}{(19, 22)} & \multicolumn{1}{l|}{0.005040} \\\hline
(8, 9) & 0.006167& \multicolumn{1}{l|}{(21, 23)} & \multicolumn{1}{l|}{0.008903} \\\hline
(8, 11) & 0.053568& \multicolumn{1}{l|}{(22, 25)} & \multicolumn{1}{l|}{0.023629} \\\hline
(10, 12) & 0.006628& \multicolumn{1}{l|}{(23, 24)} & \multicolumn{1}{l|}{0.003967} \\\hline
(11, 14) & 0.013671& \multicolumn{1}{l|}{(24, 25)} & \multicolumn{1}{l|}{0.023295} \\\hline
(12, 13) & 0.014007& \multicolumn{1}{l|}{(25, 26)} & \multicolumn{1}{l|}{0.028715} \\\hline
\end{tabular}
\caption{For $e = (u, v) \in E_G$, $\omega_G(e)$ is the error rate of coupling physical qubits u and v.}
\label{tab:noises cairo}
\end{subfigure}
\caption{$G = (V_G, E_G, \omega_G)$ is the connectivity graph of IBM's fake Cairo backend. It is an undirected edge-weighted connected graph. $\lvert V_G\rvert = 16$, $\omega_G: E_G \rightarrow \{x \in \R; \; 0 \leq x < 1\}$. }
    \label{fig:Cairo}
\end{figure}
\hfill\hfill


\section{Compare Different Cost Functions on IBM's Fake Backends}
\label{sec: compare different cost functions}
\begin{figure}[H]
    \centering
    \includegraphics[width=0.82\linewidth]{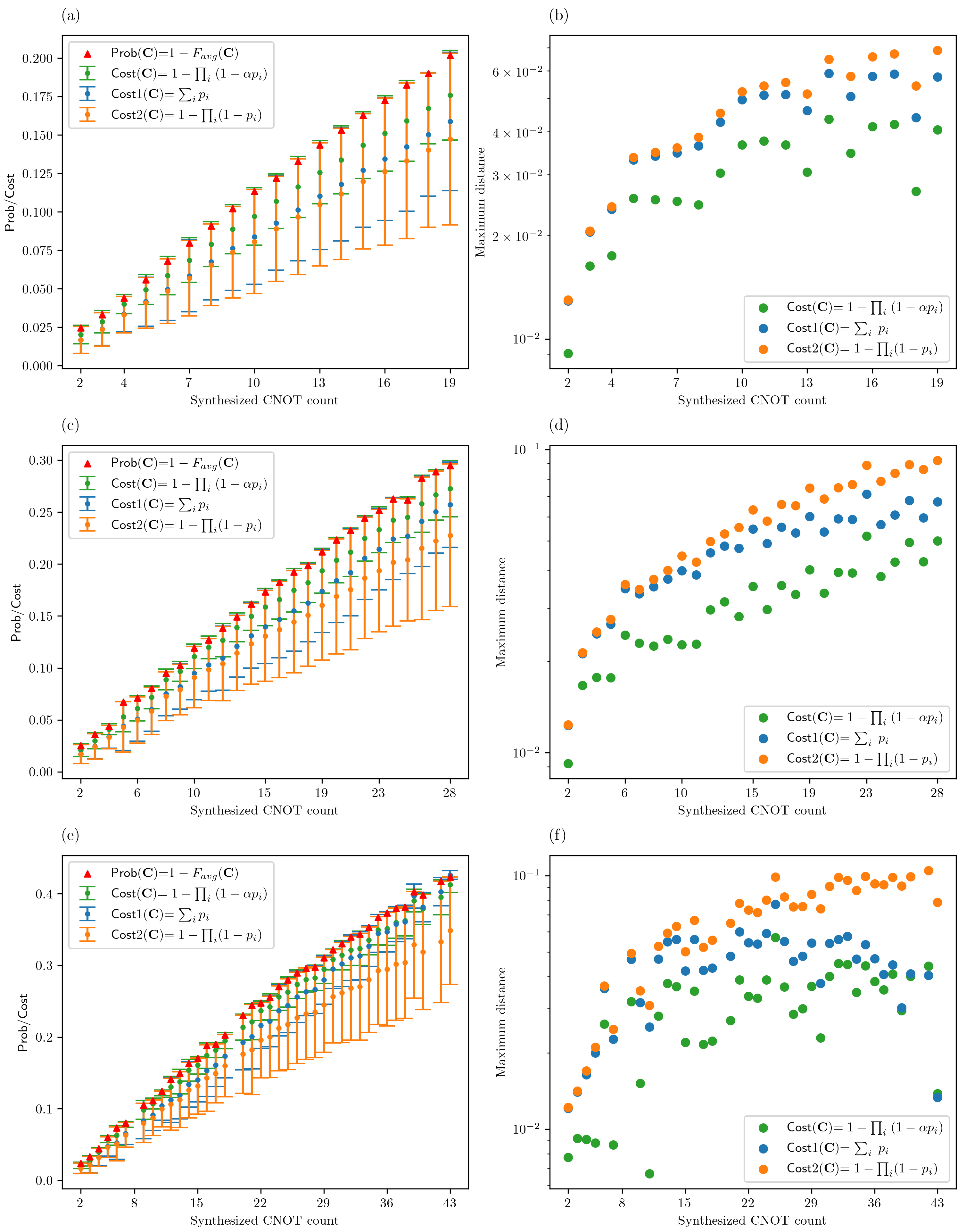}
    \caption{On IBM's fake Guadalupe backend, compare the three cost functions against the error probability of a synthesized CNOT circuit. $p_i$ is the error rate of a noisy CNOT gate, $n$ is the circuit width, $\alpha = 1+\frac{2^{n-2}-1}{2^n+1}$. In the left column, figures (a), (c), and (e) report results about the synthesized CNOT circuits of widths 5, 6, and 7 respectively. In each figure, for circuits of the same synthesized CNOT count, a y-value corresponds to the average of the error probability or the average of a cost function. The length of an error bar measures the average difference between the error probability and each cost function. The higher the value, the worse the approximation. In the right column, we compare the maximum distance between the error probability and each cost function. Figures (b), (d), and (f) report results about the synthesized CNOT circuits of widths 5, 6, and 7 respectively. Since $\max\lvert\Prob - \Cost\rvert$ has a wide range, we use a logarithmic scale for the y-axis to see all the numbers easily.}
    \label{fig: cost comparison Guadalupe}
\end{figure}

\begin{figure}[H]
    \centering
    \includegraphics[width=0.82\linewidth]{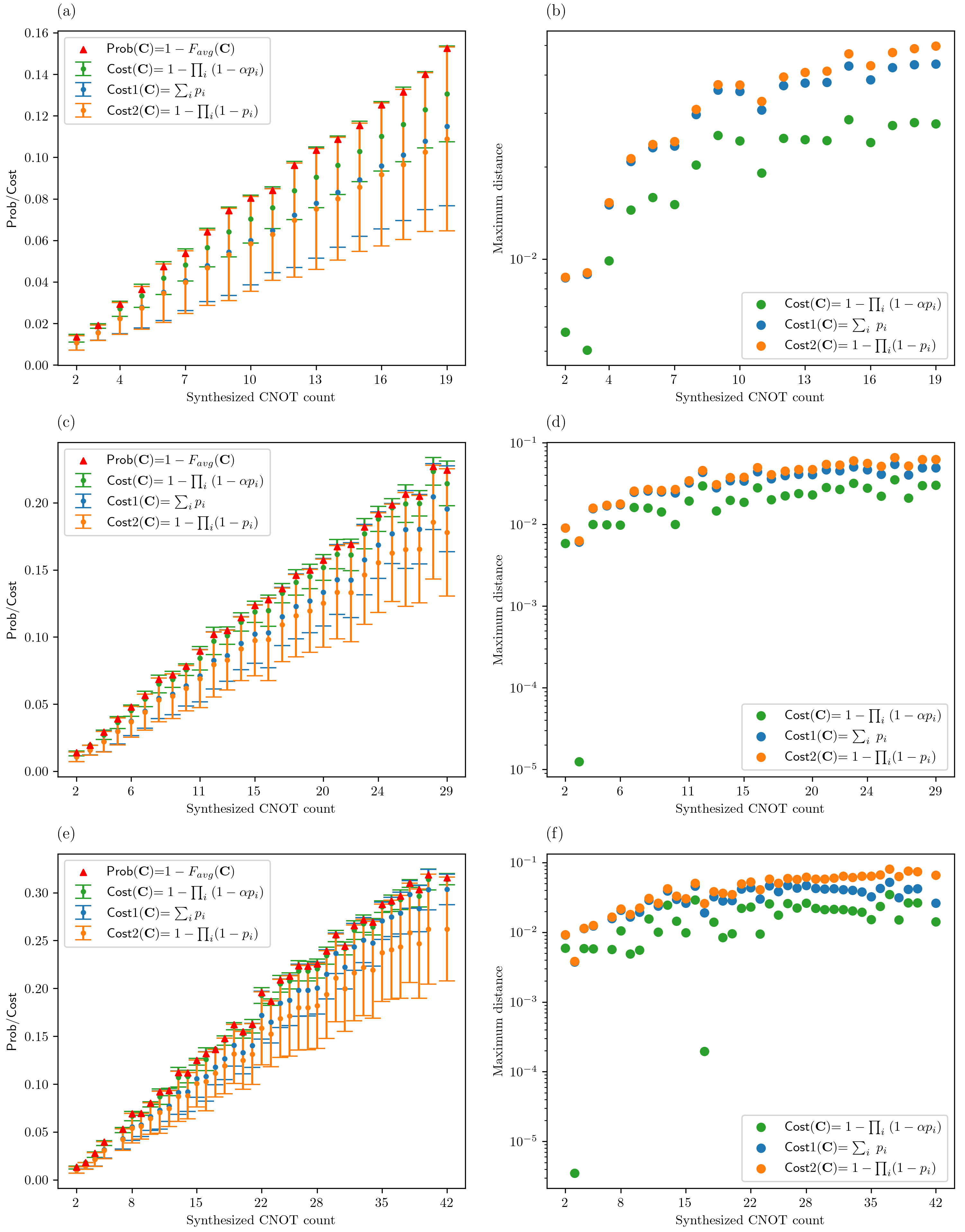}
    \caption{On IBM's fake Cairo backend, compare the three cost functions against the error probability of a synthesized CNOT circuit. $p_i$ is the error rate of a noisy CNOT gate, $n$ is the circuit width, $\alpha = 1+\frac{2^{n-2}-1}{2^n+1}$. In the left column, figures (a), (c), and (e) report results about the synthesized CNOT circuits of widths 5, 6, and 7 respectively. In each figure, for circuits of the same synthesized CNOT count, a y-value corresponds to the average of the error probability or the average of a cost function. The length of an error bar measures the average difference between the error probability and each cost function. The higher the value, the worse the approximation. In the right column, we compare the maximum distance between the error probability and each cost function. Figures (b), (d), and (f) report results about the synthesized CNOT circuits of widths 5, 6, and 7 respectively. Since $\max\lvert\Prob - \Cost\rvert$ has a wide range, we use a logarithmic scale for the y-axis to see all the numbers easily.}
    \label{fig: cost comparison Cairo}
\end{figure}


\section{Compare Different CNOT Synthesis Algorithms on IBM's Fake Backends}
\label{sec: compare diff algorithms}

\subsection{Benchmark on IBM's Fake Nairobi Backend}
\begin{table}[H]
\centering
\begin{tabular}{c|c|l|l|l|l}
\hline
Circuit & Original  & \multirow{2}{3em}{Qiskit} & \multirow{2}{4em}{ROWCOL} & \multirow{2}{5em}{PermRowCol} & \multirow{2}{7em}{NAPermRowCol} \\ Width & CNOT Count & & & & \\\hline
\multirow{9}{2em}{5} & 4 & 4.09\quad (5.68\%) & 4.51\quad (16.54\%) & 3.99\quad (3.10\%) & 3.87\\
 & 8 & 11.22\quad (23.43\%) & 11.97\quad (31.68\%) & 10.22\quad (12.43\%) & 9.09\\
 & 16 & 25.28\quad (117.93\%) & 17.02\quad (46.72\%) & 13.04\quad (12.41\%) & 11.60\\
 & 32 & 55.34\quad (326.35\%) & 17.83\quad (37.37\%) & 14.11\quad (8.71\%) & 12.98\\
 & 64 & 120.14\quad (896.19\%) & 19.14\quad (58.71\%) & 14.14\quad (17.25\%) & 12.06\\
 & 128 & 246.41\quad (1885.58\%) & 18.32\quad (47.62\%) & 13.76\quad (10.88\%) & 12.41\\
 & 256 & 500.59\quad (4013.31\%) & 18.40\quad (51.19\%) & 13.93\quad (14.46\%) & 12.17\\
 & 512 & 1001.44\quad (8021.98\%) & 18.89\quad (53.20\%) & 13.80\quad (11.92\%) & 12.33\\
 & 1024 & 2017.34\quad (16116.56\%) & 18.75\quad (50.72\%) & 13.62\quad (9.49\%) & 12.44\\
\hline
\multirow{9}{2em}{7} & 4 & 4.53\quad (-0.88\%) & 4.77\quad (4.38\%) & 4.70\quad (2.84\%) & 4.57\\
 & 8 & 13.05\quad (4.15\%) & 16.24\quad (29.61\%) & 13.67\quad (9.10\%) & 12.53\\
 & 16 & 33.47\quad (35.01\%) & 32.01\quad (29.12\%) & 26.75\quad (7.91\%) & 24.79\\
 & 32 & 79.96\quad (165.03\%) & 39.35\quad (30.43\%) & 33.64\quad (11.50\%) & 30.17\\
 & 64 & 166.18\quad (433.65\%) & 41.55\quad (33.43\%) & 35.48\quad (13.94\%) & 31.14\\
 & 128 & 347.20\quad (1011.04\%) & 41.32\quad (32.22\%) & 34.68\quad (10.98\%) & 31.25\\
 & 256 & 706.32\quad (2199.97\%) & 41.72\quad (35.85\%) & 34.14\quad (11.17\%) & 30.71\\
 & 512 & 1431.80\quad (4541.17\%) & 41.75\quad (35.33\%) & 34.57\quad (12.06\%) & 30.85\\
 & 1024 & 2920.34\quad (9170.92\%) & 41.03\quad (30.25\%) & 35.12\quad (11.49\%) & 31.50\\
\hline
 \end{tabular}
\caption{IBM's fake Nairobi backend hosts $7$ qubits. We benchmark with its 5- and 7-qubit connected subgraph and compare NAPermRowCol with other state-of-the-art CNOT synthesis algorithms in terms of the synthesized CNOT count. Qiskit is short for ``Qiskit transpilation at optimization level 3''. It implements the SWAP-based heuristic algorithm SABRE. For each row (i.e., the original CNOT count), input $100$ randomly generated CNOT circuits to each algorithm listed in the table header, then calculate the average synthesized CNOT count. The value in each bracket shows the percentage difference compared to the results of NAPermRowCol.}
\label{tab:CNOTNairobi}
\end{table}

\begin{table}[H]
\centering
\begin{tabular}{c|c|l|l|l|l}
\hline
Circuit& Original  & \multirow{2}{3em}{Qiskit} & \multirow{2}{4em}{ROWCOL} & \multirow{2}{5em}{PermRowCol} & \multirow{2}{7em}{NAPermRowCol} \\ Width & CNOT Count & & & & \\\hline
\multirow{9}{2em}{5} & 4 & 0.0407\quad (8.21\%) & 0.0436\quad (15.93\%) & 0.0388\quad (3.29\%) & 0.0376\\
 & 8 & 0.1075\quad (25.72\%) & 0.1123\quad (31.34\%) & 0.0968\quad (13.17\%) & 0.0855\\
 & 16 & 0.2273\quad (109.72\%) & 0.1574\quad (45.18\%) & 0.1226\quad (13.12\%) & 0.1084\\
 & 32 & 0.4301\quad (258.69\%) & 0.1640\quad (36.78\%) & 0.1316\quad (9.72\%) & 0.1199\\
 & 64 & 0.7079\quad (527.59\%) & 0.1743\quad (54.48\%) & 0.1328\quad (17.76\%) & 0.1128\\
 & 128 & 0.9189\quad (696.34\%) & 0.1678\quad (45.42\%) & 0.1294\quad (12.15\%) & 0.1154\\
 & 256 & 0.9940\quad (778.12\%) & 0.1694\quad (49.64\%) & 0.1306\quad (15.41\%) & 0.1132\\
 & 512 & 1.0000\quad (780.36\%) & 0.1721\quad (51.48\%) & 0.1277\quad (12.45\%) & 0.1136\\
 & 1024 & 1.0000\quad (770.07\%) & 0.1729\quad (50.48\%) & 0.1267\quad (10.28\%) & 0.1149\\
\hline
\multirow{9}{2em}{7} & 4 & 0.0437\quad (-0.23\%) & 0.0461\quad (5.45\%) & 0.0455\quad (3.95\%) & 0.0438\\
 & 8 & 0.1222\quad (4.08\%) & 0.1516\quad (29.13\%) & 0.1289\quad (9.79\%) & 0.1174\\
 & 16 & 0.2858\quad (30.77\%) & 0.2778\quad (27.11\%) & 0.2381\quad (8.92\%) & 0.2186\\
 & 32 & 0.5576\quad (114.27\%) & 0.3309\quad (27.16\%) & 0.2898\quad (11.37\%) & 0.2602\\
 & 64 & 0.8167\quad (205.02\%) & 0.3446\quad (28.68\%) & 0.3035\quad (13.34\%) & 0.2678\\
 & 128 & 0.9711\quad (261.34\%) & 0.3421\quad (27.31\%) & 0.2987\quad (11.16\%) & 0.2687\\
 & 256 & 0.9993\quad (277.49\%) & 0.3467\quad (30.98\%) & 0.2955\quad (11.62\%) & 0.2647\\
 & 512 & 1.0000\quad (277.93\%) & 0.3467\quad (31.02\%) & 0.2975\quad (12.42\%) & 0.2646\\
 & 1024 & 1.0000\quad (270.19\%) & 0.3418\quad (26.53\%) & 0.3006\quad (11.27\%) & 0.2701\\
\hline
 \end{tabular}
\caption{IBM's fake Nairobi backend hosts $7$ qubits. We benchmark with its 5- and 7-qubit connected subgraph and compare NAPermRowCol with other state-of-the-art CNOT synthesis algorithms in terms of the $\Cost$ metric. For each row (i.e., the original CNOT count), input $100$ randomly generated CNOT circuits to each algorithm listed in the table header, then calculate the average $\Cost$. The value in each bracket shows the percentage difference compared to the results of NAPermRowCol.}
\label{tab:CostNairobi}
\end{table}
\subsection{Benchmark on IBM's Fake Guadalupe Backend}
\label{subsubsec:guadalupe}

\begin{table}[H]
\centering
\begin{tabular}{c|c|l|l|l|l}
\hline
Circuit& Original  & \multirow{2}{3em}{Qiskit} & \multirow{2}{4em}{ROWCOL} & \multirow{2}{5em}{PermRowCol} & \multirow{2}{7em}{NAPermRowCol} \\ Width & CNOT Count & & & & \\\hline
\multirow{9}{2em}{5} & 4 & 4.07\quad (3.04\%) & 4.34\quad (9.87\%) & 3.92\quad (-0.76\%) & 3.95\\
 & 8 & 11.08\quad (23.11\%) & 12.03\quad (33.67\%) & 10.19\quad (13.22\%) & 9.00\\
 & 16 & 25.43\quad (116.43\%) & 16.94\quad (44.17\%) & 13.18\quad (12.17\%) & 11.75\\
 & 32 & 56.35\quad (334.13\%) & 18.37\quad (41.53\%) & 14.23\quad (9.63\%) & 12.98\\
 & 64 & 121.11\quad (893.52\%) & 18.93\quad (55.29\%) & 13.91\quad (14.11\%) & 12.19\\
 & 128 & 247.08\quad (1913.69\%) & 18.15\quad (47.92\%) & 13.69\quad (11.57\%) & 12.27\\
 & 256 & 504.43\quad (3932.21\%) & 18.50\quad (47.88\%) & 14.03\quad (12.15\%) & 12.51\\
 & 512 & 1005.74\quad (8198.18\%) & 18.75\quad (54.70\%) & 13.76\quad (13.53\%) & 12.12\\
 & 1024 & 2023.69\quad (16063.66\%) & 18.63\quad (48.80\%) & 13.89\quad (10.94\%) & 12.52\\
\hline
\multirow{9}{2em}{7} & 4 & 3.85\quad (4.90\%) & 3.86\quad (5.18\%) & 3.79\quad (3.27\%) & 3.67\\
 & 8 & 11.96\quad (0.76\%) & 14.98\quad (26.20\%) & 12.56\quad (5.81\%) & 11.87\\
 & 16 & 31.29\quad (29.46\%) & 31.01\quad (28.30\%) & 26.63\quad (10.18\%) & 24.17\\
 & 32 & 79.05\quad (147.26\%) & 41.06\quad (28.43\%) & 35.99\quad (12.57\%) & 31.97\\
 & 64 & 168.18\quad (426.88\%) & 43.09\quad (34.99\%) & 35.87\quad (12.37\%) & 31.92\\
 & 128 & 354.01\quad (1003.52\%) & 43.33\quad (35.07\%) & 35.35\quad (10.19\%) & 32.08\\
 & 256 & 726.68\quad (2161.69\%) & 42.16\quad (31.22\%) & 35.00\quad (8.93\%) & 32.13\\
 & 512 & 1483.26\quad (4462.47\%) & 42.65\quad (31.19\%) & 34.56\quad (6.31\%) & 32.51\\
 & 1024 & 3023.32\quad (9202.52\%) & 42.69\quad (31.35\%) & 35.86\quad (10.34\%) & 32.50\\
\hline
\multirow{9}{2em}{16} & 4 & 3.98\quad (1.02\%) & 5.61\quad (42.39\%) & 3.94\quad (0.00\%) & 3.94\\
 & 8 & 8.90\quad (-9.28\%) & 11.95\quad (21.81\%) & 9.24\quad (-5.81\%) & 9.81\\
 & 16 & 33.17\quad (-36.61\%) & 70.54\quad (34.80\%) & 57.85\quad (10.55\%) & 52.33\\
 & 32 & 98.06\quad (-27.84\%) & 170.99\quad (25.82\%) & 146.18\quad (7.56\%) & 135.90\\
 & 64 & 252.75\quad (13.75\%) & 247.43\quad (11.35\%) & 228.18\quad (2.69\%) & 222.20\\
 & 128 & 577.35\quad (134.48\%) & 274.54\quad (11.50\%) & 258.17\quad (4.85\%) & 246.23\\
 & 256 & 1261.06\quad (408.53\%) & 275.65\quad (11.16\%) & 260.65\quad (5.11\%) & 247.98\\
 & 512 & 2644.47\quad (964.69\%) & 274.35\quad (10.46\%) & 262.67\quad (5.75\%) & 248.38\\
 & 1024 & 5465.55\quad (2117.08\%) & 274.42\quad (11.32\%) & 262.28\quad (6.39\%) & 246.52\\
\hline
 \end{tabular}
\caption{IBM's fake Guadalupe backend hosts $16$ qubits. We benchmark with its 5-, 7-, and 16-qubit connected subgraph and compare NAPermRowCol with other state-of-the-art CNOT synthesis algorithms in terms of the synthesized CNOT count. Qiskit is short for ``Qiskit transpilation at optimization level 3''. It implements the SWAP-based heuristic algorithm SABRE. For each row (i.e., the original CNOT count), input $100$ randomly generated CNOT circuits to each algorithm listed in the table header, then calculate the average synthesized CNOT count. The value in each bracket shows the percentage difference compared to the results of NAPermRowCol.}
\label{tab:CNOTGuadalupe}
\end{table}

\begin{table}[H]
\centering
\begin{tabular}{c|c|l|l|l|l}
\hline
Circuit& Original  & \multirow{2}{3em}{Qiskit} & \multirow{2}{4em}{ROWCOL} & \multirow{2}{5em}{PermRowCol} & \multirow{2}{7em}{NAPermRowCol} \\ Width & CNOT Count & & & & \\\hline
\multirow{9}{2em}{5} & 4 & 0.0546\quad (6.56\%) & 0.0567\quad (10.59\%) & 0.0513\quad (0.06\%) & 0.0512\\
 & 8 & 0.1379\quad (30.06\%) & 0.1425\quad (34.41\%) & 0.1231\quad (16.06\%) & 0.1060\\
 & 16 & 0.2876\quad (113.88\%) & 0.1961\quad (45.82\%) & 0.1564\quad (16.26\%) & 0.1345\\
 & 32 & 0.5378\quad (255.14\%) & 0.2163\quad (42.82\%) & 0.1718\quad (13.46\%) & 0.1514\\
 & 64 & 0.7987\quad (467.07\%) & 0.2175\quad (54.45\%) & 0.1640\quad (16.43\%) & 0.1408\\
 & 128 & 0.9579\quad (581.91\%) & 0.2081\quad (48.15\%) & 0.1600\quad (13.92\%) & 0.1405\\
 & 256 & 0.9985\quad (586.23\%) & 0.2145\quad (47.43\%) & 0.1685\quad (15.83\%) & 0.1455\\
 & 512 & 1.0000\quad (622.77\%) & 0.2148\quad (55.25\%) & 0.1601\quad (15.71\%) & 0.1384\\
 & 1024 & 1.0000\quad (592.21\%) & 0.2152\quad (48.94\%) & 0.1645\quad (13.86\%) & 0.1445\\
\hline
\multirow{9}{2em}{7} & 4 & 0.0520\quad (5.94\%) & 0.0514\quad (4.81\%) & 0.0506\quad (3.25\%) & 0.0491\\
 & 8 & 0.1518\quad (4.33\%) & 0.1837\quad (26.24\%) & 0.1555\quad (6.89\%) & 0.1455\\
 & 16 & 0.3474\quad (28.99\%) & 0.3403\quad (26.35\%) & 0.3000\quad (11.40\%) & 0.2693\\
 & 32 & 0.6571\quad (92.43\%) & 0.4266\quad (24.92\%) & 0.3846\quad (12.63\%) & 0.3415\\
 & 64 & 0.8974\quad (163.33\%) & 0.4428\quad (29.93\%) & 0.3869\quad (13.54\%) & 0.3408\\
 & 128 & 0.9911\quad (191.90\%) & 0.4384\quad (29.12\%) & 0.3782\quad (11.38\%) & 0.3395\\
 & 256 & 0.9999\quad (198.17\%) & 0.4291\quad (27.95\%) & 0.3734\quad (11.36\%) & 0.3354\\
 & 512 & 1.0000\quad (195.00\%) & 0.4307\quad (27.05\%) & 0.3661\quad (8.00\%) & 0.3390\\
 & 1024 & 1.0000\quad (194.05\%) & 0.4357\quad (28.11\%) & 0.3808\quad (11.97\%) & 0.3401\\
\hline
\multirow{9}{2em}{16} & 4 & 0.0515\quad (0.96\%) & 0.0685\quad (34.22\%) & 0.0511\quad (0.00\%) & 0.0511\\
 & 8 & 0.1140\quad (-6.28\%) & 0.1435\quad (17.96\%) & 0.1182\quad (-2.83\%) & 0.1216\\
 & 16 & 0.3663\quad (-25.21\%) & 0.6028\quad (23.07\%) & 0.5273\quad (7.64\%) & 0.4898\\
 & 32 & 0.7394\quad (-11.03\%) & 0.8983\quad (8.09\%) & 0.8601\quad (3.49\%) & 0.8310\\
 & 64 & 0.9689\quad (2.28\%) & 0.9667\quad (2.05\%) & 0.9560\quad (0.91\%) & 0.9473\\
 & 128 & 0.9996\quad (3.82\%) & 0.9775\quad (1.52\%) & 0.9720\quad (0.95\%) & 0.9629\\
 & 256 & 1.0000\quad (3.84\%) & 0.9780\quad (1.56\%) & 0.9726\quad (0.99\%) & 0.9630\\
 & 512 & 1.0000\quad (3.71\%) & 0.9773\quad (1.36\%) & 0.9736\quad (0.98\%) & 0.9642\\
 & 1024 & 1.0000\quad (3.77\%) & 0.9776\quad (1.44\%) & 0.9734\quad (1.01\%) & 0.9636\\
\hline
 \end{tabular}
\caption{IBM's fake Guadalupe backend hosts $16$ qubits. We benchmark with its 5-, 7-, and 16-qubit connected subgraph and compare NAPermRowCol with other state-of-the-art CNOT synthesis algorithms in terms of the $\Cost$ metric. For each row (i.e., the original CNOT count), input $100$ randomly generated CNOT circuits to each algorithm listed in the table header, then calculate the average synthesized CNOT count. The value in each bracket shows the percentage difference compared to the results of NAPermRowCol.}
\label{tab:CostGuadalupe}
\end{table}

\begin{figure}[p]
\begin{subfigure}[t]{.5\textwidth}
  \centering
  \includegraphics[scale=0.42]{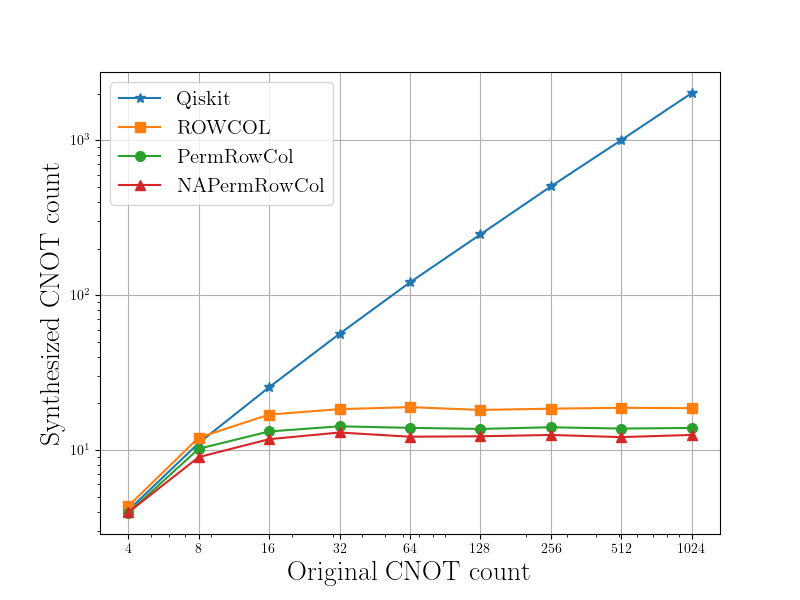}
  \subcaption{Compare NAPermRowCol with PermRowCol, ROWCOL, and Qiskit in terms of the synthesized CNOT count. Qiskit has the worst scalability when the original CNOT count grows exponentially.}
        \label{fig:Guadalupe5NCQ}
\end{subfigure}
\quad
\begin{subfigure}[t]{.5\textwidth}
  \centering
  \includegraphics[scale=0.42]{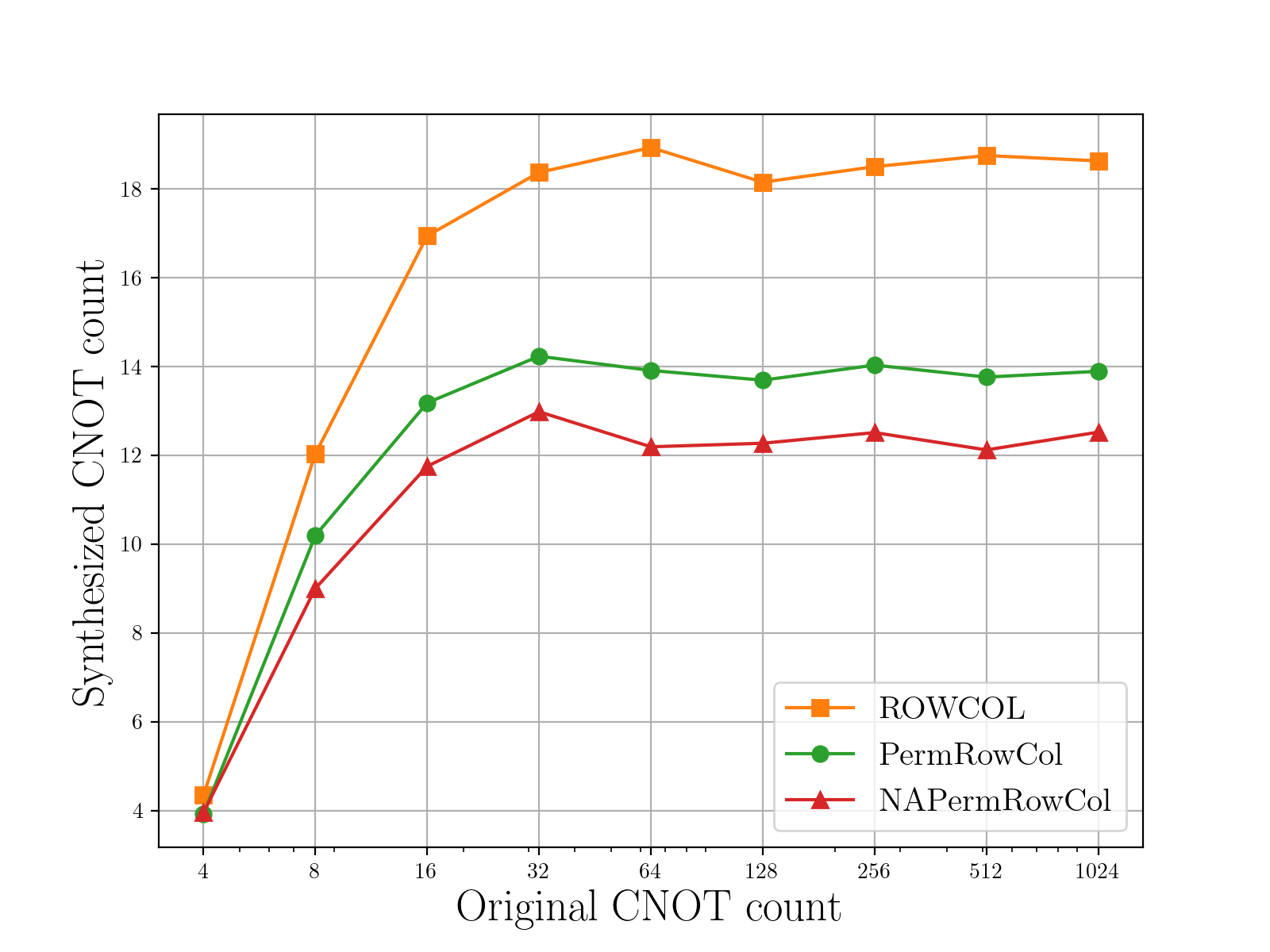}
  \subcaption{Compare NAPermRowCol with PermRowCol and ROWCOL in terms of the synthesized CNOT count. This is the zoomed-in version of the lefthand side.}
        \label{fig:Guadalupe5NC}
\end{subfigure}
\begin{subfigure}[t]{.5\textwidth}
  \centering
  \includegraphics[scale=0.42]{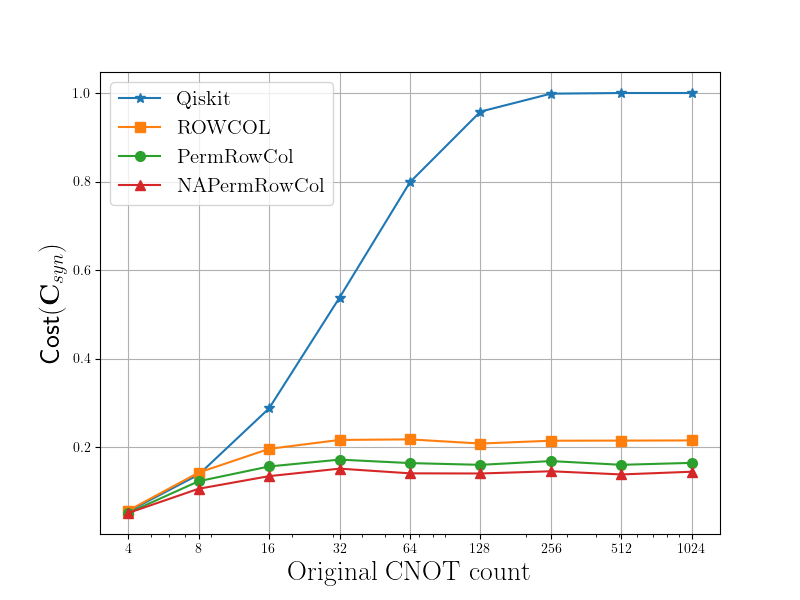}
  \subcaption{Compare NAPermRowCol with PermRowCol, ROWCOL, and Qiskit in terms of their respective costs. Qiskit has the worst scalability when the original CNOT count grows exponentially.}
        \label{fig:Guadalupe5CostQ}
\end{subfigure}
\quad
\begin{subfigure}[t]{.5\textwidth}
  \centering
  \includegraphics[scale=0.42]{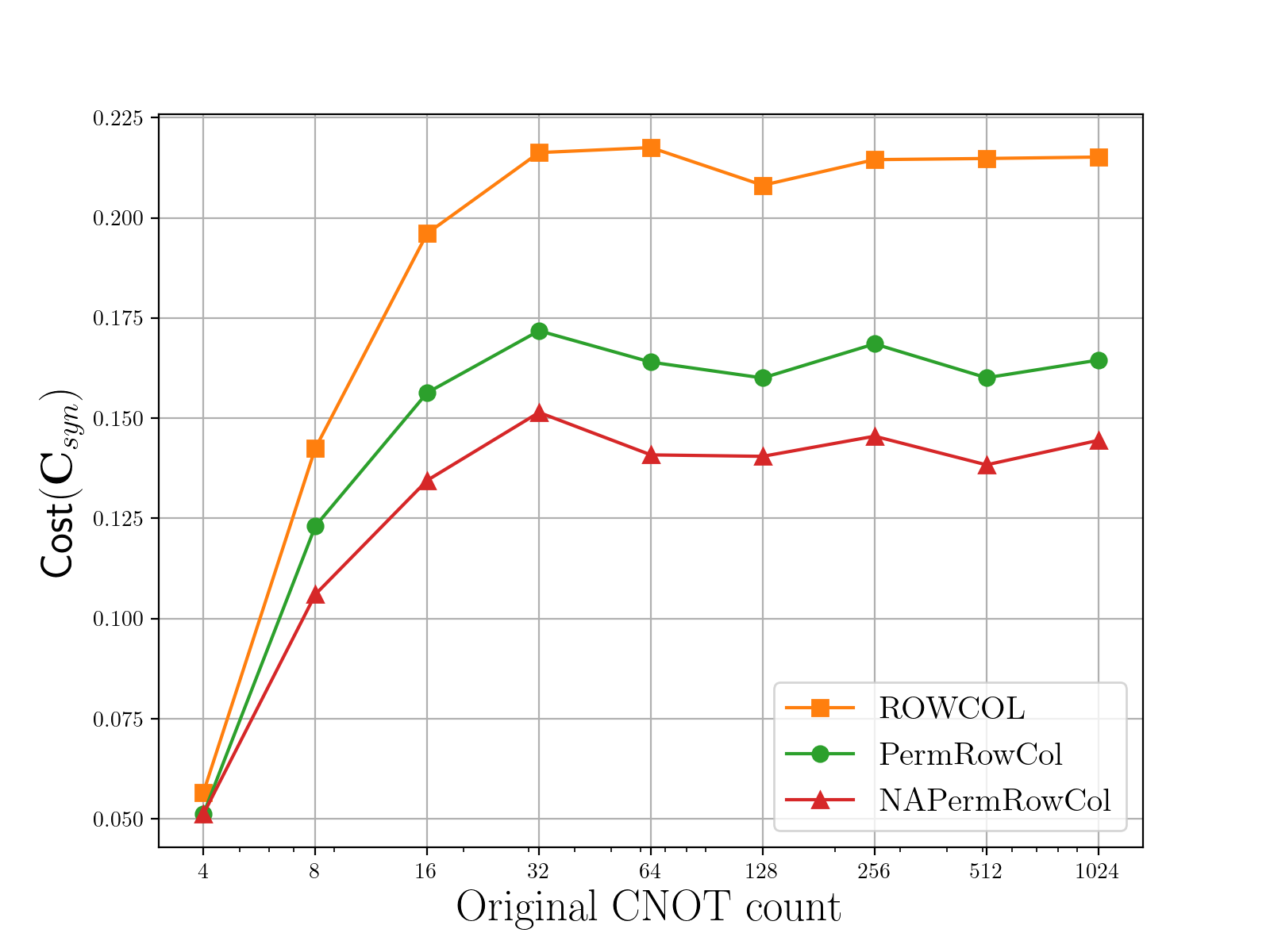}
  \subcaption{Compare NAPermRowCol with PermRowCol and ROWCOL in terms of their respective costs. This is the zoomed-in version of the lefthand side.}
        \label{fig:Guadalupe5Cost}
\end{subfigure}
\caption{IBM's fake Guadalupe backend hosts 16 qubits. We benchmark with its $5$-qubit connected subgraph and compare NAPermRowCol against three state-of-the-art CNOT synthesis algorithms. For each original CNOT count, input $100$ randomly generated CNOT circuits to each algorithm, obtain the synthesized CNOT circuits, then average their gate count and the circuit cost. The x-axis in each figure uses a logarithmic scale as the input gate count grows exponentially. The y-axis of \cref{fig:Guadalupe5NCQ} uses a logarithmic scale, while the one in \cref{fig:Guadalupe5NC,fig:Guadalupe5CostQ,fig:Guadalupe5Cost} uses a linear scale. Compared to \cref{fig:Guadalupe5NCQ,fig:Guadalupe5CostQ}, \cref{fig:Guadalupe5NC,fig:Guadalupe5Cost} get rid of the data related to Qiskit so that the remaining ones are distributed in a more compact area. They serve as the zoomed-in versions which allow us to compare the performance of NAPermRowCol, PermRowCol, and ROWCOL more closely. For all input circuits of different CNOT counts, NAPermRowCol outperforms other algorithms in terms of the synthesized CNOT count and circuit cost. It demonstrates remarkable scalability when the input circuit size grows exponentially.}
\label{fig:Guadalupe5}
\end{figure}

\begin{figure}[H]
\begin{subfigure}[t]{.5\textwidth}
  \centering
  \includegraphics[scale=0.42]{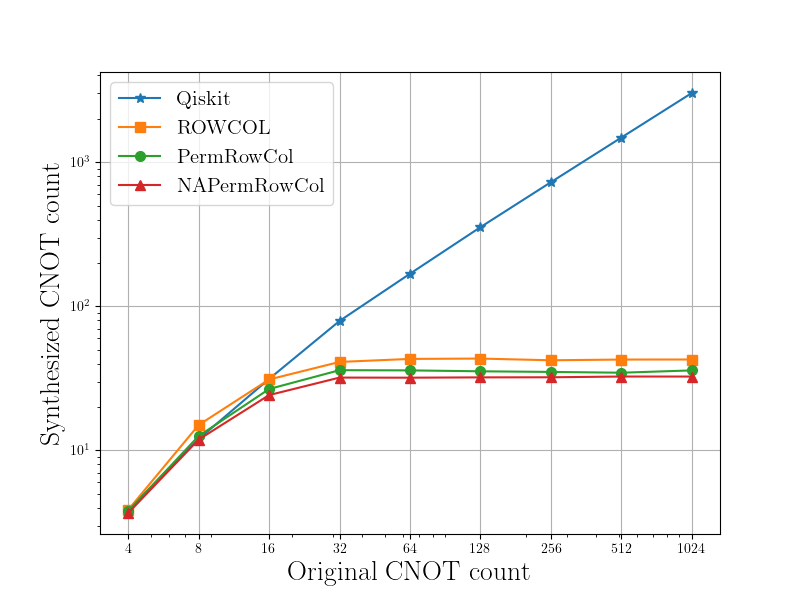}
  \subcaption{Compare NAPermRowCol with PermRowCol, ROWCOL, and Qiskit in terms of the synthesized CNOT count. Qiskit has the worst scalability when the original CNOT count grows exponentially.}
        \label{fig:Guadalupe7NCQ}
\end{subfigure}
\quad
\begin{subfigure}[t]{.5\textwidth}
  \centering
  \includegraphics[scale=0.42]{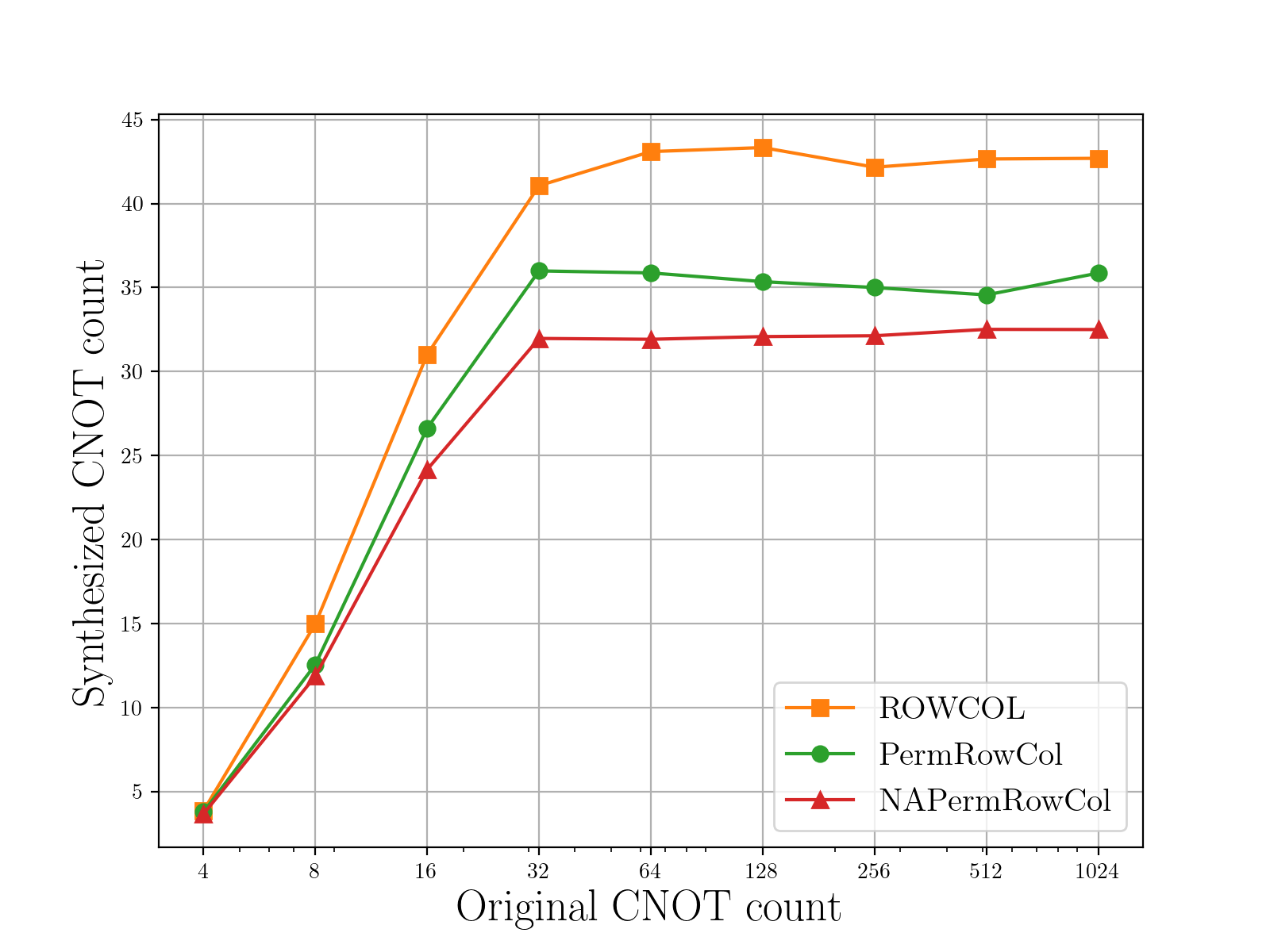}
  \subcaption{Compare NAPermRowCol with PermRowCol and ROWCOL in terms of the synthesized CNOT count. This is the zoomed-in version of the lefthand side.}
        \label{fig:Guadalupe7NC}
\end{subfigure}
\begin{subfigure}[t]{.5\textwidth}
  \centering
  \includegraphics[scale=0.42]{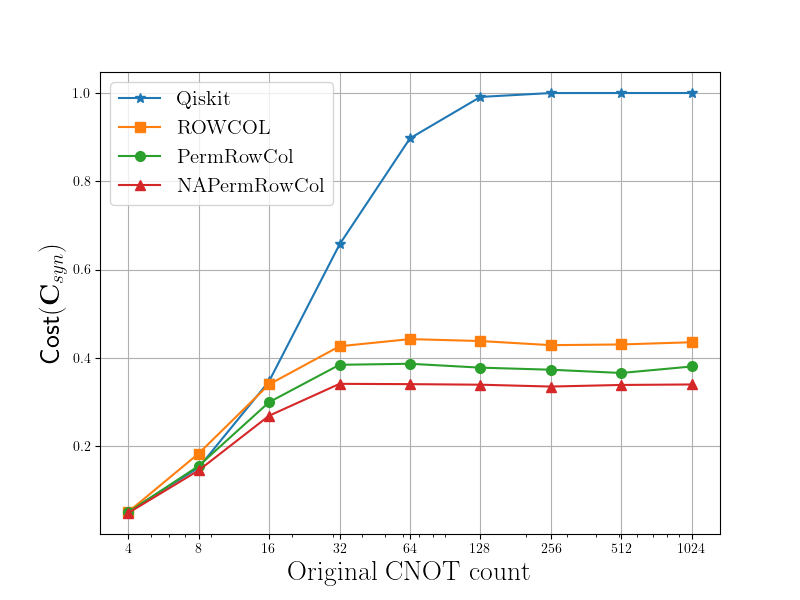}
  \subcaption{Compare NAPermRowCol with PermRowCol, ROWCOL, and Qiskit in terms of their respective costs. Qiskit has the worst scalability when the original CNOT count grows exponentially.}
        \label{fig:Guadalupe7CostQ}
\end{subfigure}
\quad
\begin{subfigure}[t]{.5\textwidth}
  \centering
  \includegraphics[scale=0.42]{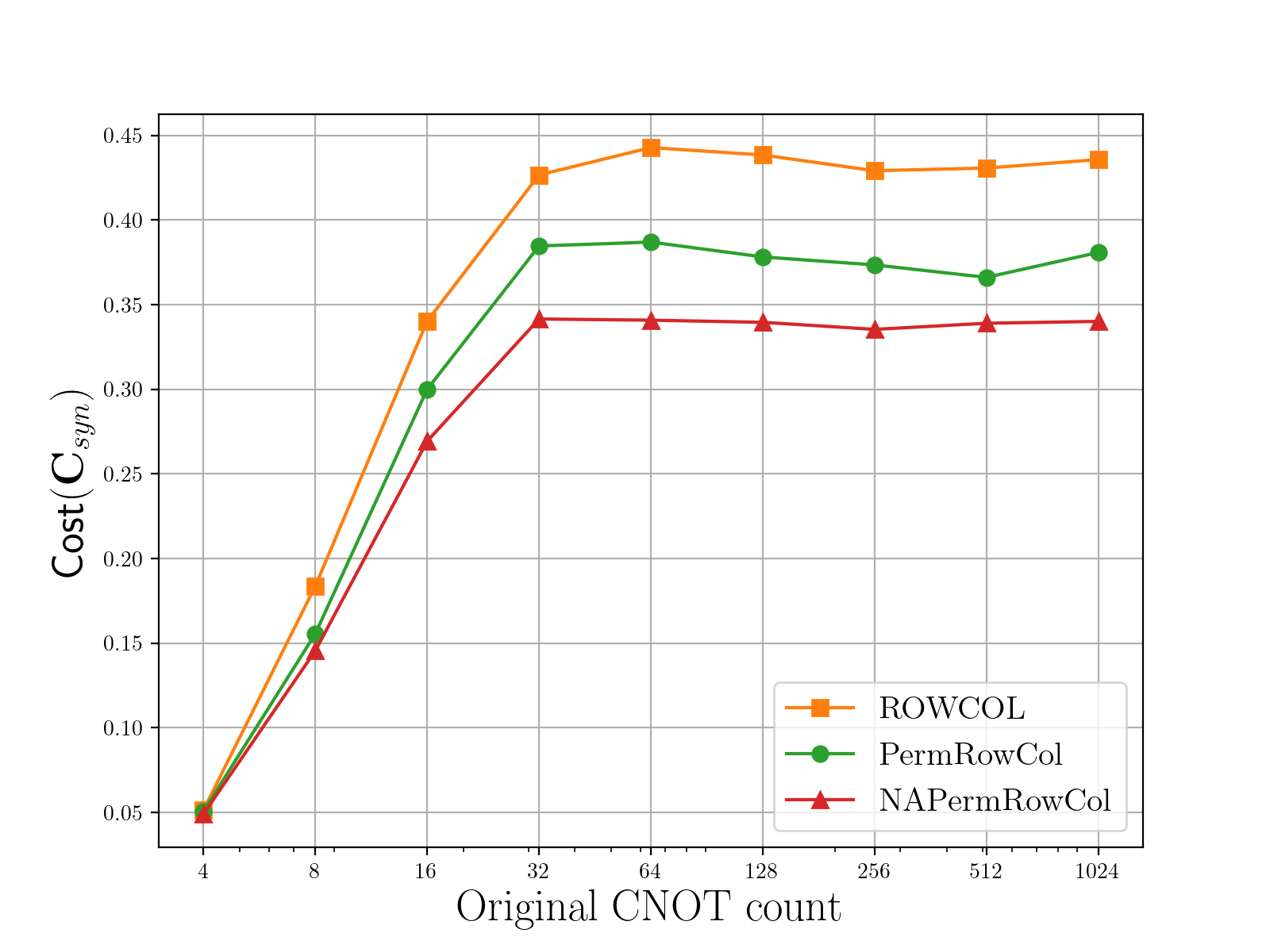}
  \subcaption{Compare NAPermRowCol with PermRowCol and ROWCOL in terms of their respective costs. This is the zoomed-in version of the lefthand side.}
        \label{fig:Guadalupe7Cost}
\end{subfigure}
\caption{IBM's fake Guadalupe backend hosts 16 qubits. We benchmark with its $7$-qubit connected subgraph and compare NAPermRowCol against three state-of-the-art CNOT synthesis algorithms. For each original CNOT count, input $100$ randomly generated CNOT circuits to each algorithm, obtain the synthesized CNOT circuits, then average their gate count and the circuit cost. The x-axis in each figure uses a logarithmic scale as the input gate count grows exponentially. The y-axis of \cref{fig:Guadalupe7NCQ} uses a logarithmic scale, while the one in \cref{fig:Guadalupe7NC,fig:Guadalupe7CostQ,fig:Guadalupe7Cost} uses a linear scale. Compared to \cref{fig:Guadalupe7NCQ,fig:Guadalupe7CostQ}, \cref{fig:Guadalupe7NC,fig:Guadalupe7Cost} get rid of the data related to Qiskit so that the remaining ones are distributed in a more compact area. They serve as the zoomed-in versions which allow us to compare the performance of NAPermRowCol, PermRowCol, and ROWCOL more closely. For input circuits of more than 16 CNOT counts, NAPermRowCol outperforms other algorithms in terms of the synthesized CNOT count and circuit cost. It demonstrates remarkable scalability when the input circuit size grows exponentially.}
\label{fig:Guadalupe7}
\end{figure}

\begin{figure}[H]
\begin{subfigure}[t]{.5\textwidth}
  \centering
  \includegraphics[scale=0.42]{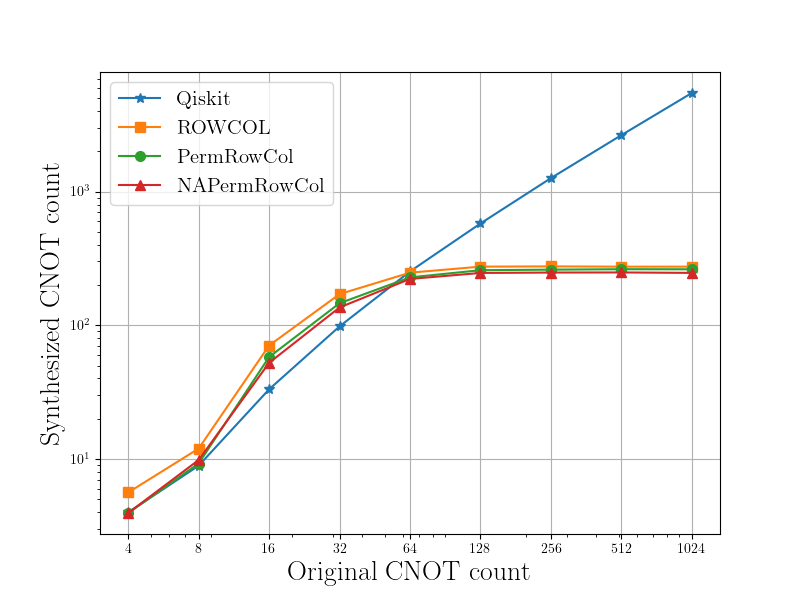}
  \subcaption{Compare NAPermRowCol with PermRowCol, ROWCOL, and Qiskit in terms of the synthesized CNOT count. Qiskit has the worst scalability when the original CNOT count grows exponentially}
        \label{fig:Guadalupe16NCQ}
\end{subfigure}
\quad
\begin{subfigure}[t]{.5\textwidth}
  \centering
  \includegraphics[scale=0.42]{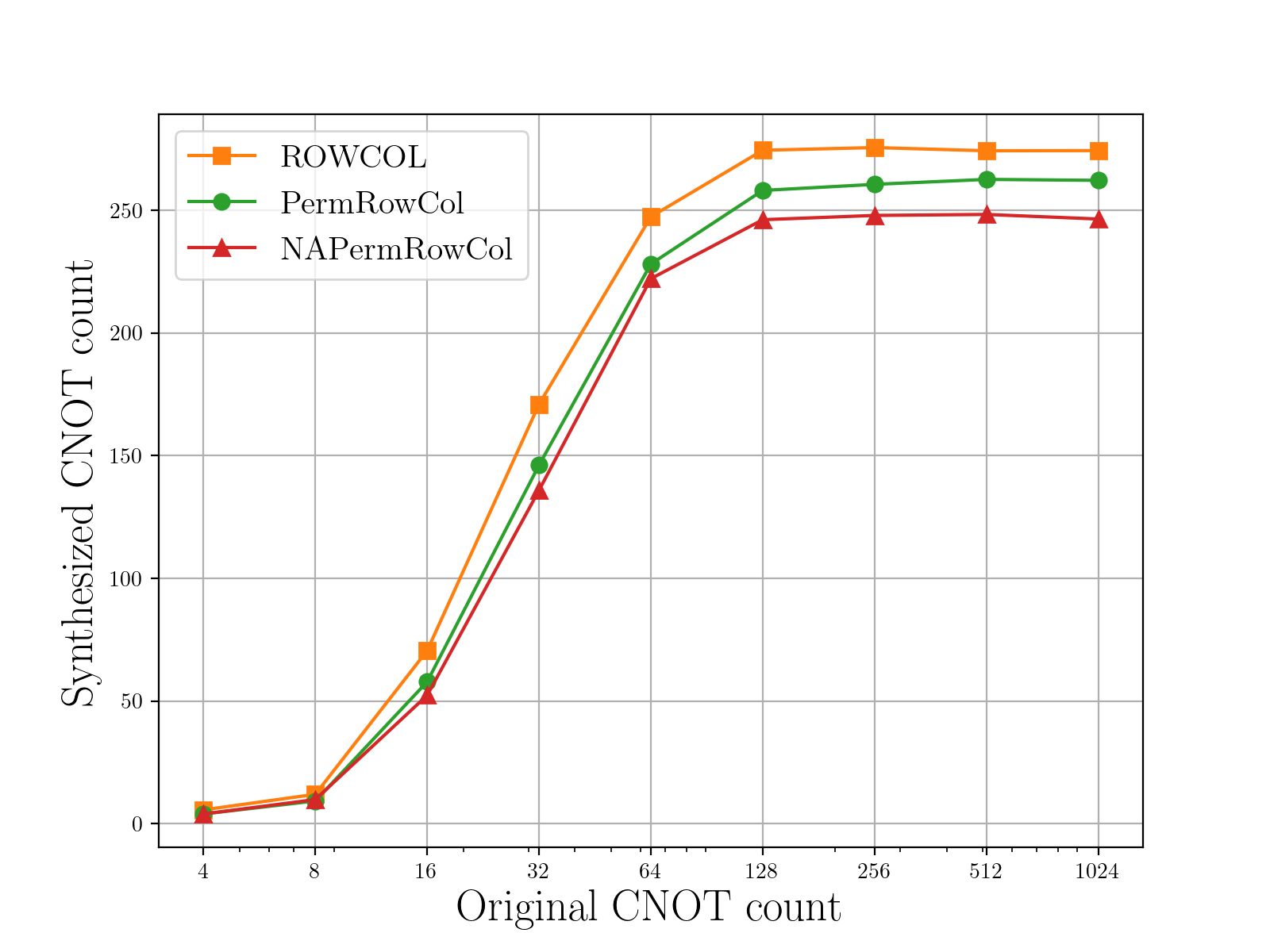}
  \subcaption{Compare NAPermRowCol with PermRowCol and ROWCOL in terms of the synthesized CNOT count. This is the zoomed-in version of the lefthand side.}
        \label{fig:Guadalupe16NC}
\end{subfigure}
\begin{subfigure}[t]{.5\textwidth}
  \centering
  \includegraphics[scale=0.42]{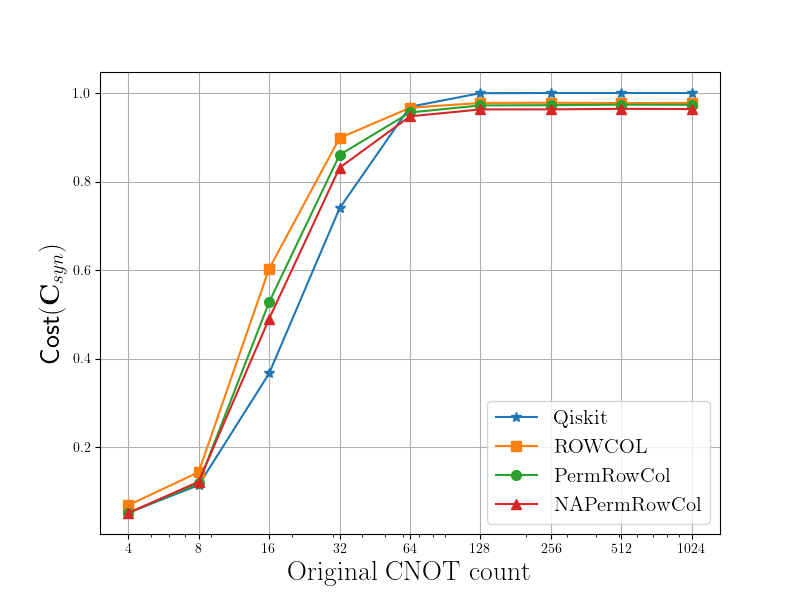}
  \subcaption{Compare NAPermRowCol with PermRowCol, ROWCOL, and Qiskit in terms of their respective costs. Qiskit has the worst scalability when the original CNOT count grows exponentially.}
        \label{fig:Guadalupe16CostQ}
\end{subfigure}
\quad
\begin{subfigure}[t]{.5\textwidth}
  \centering
  \includegraphics[scale=0.42]{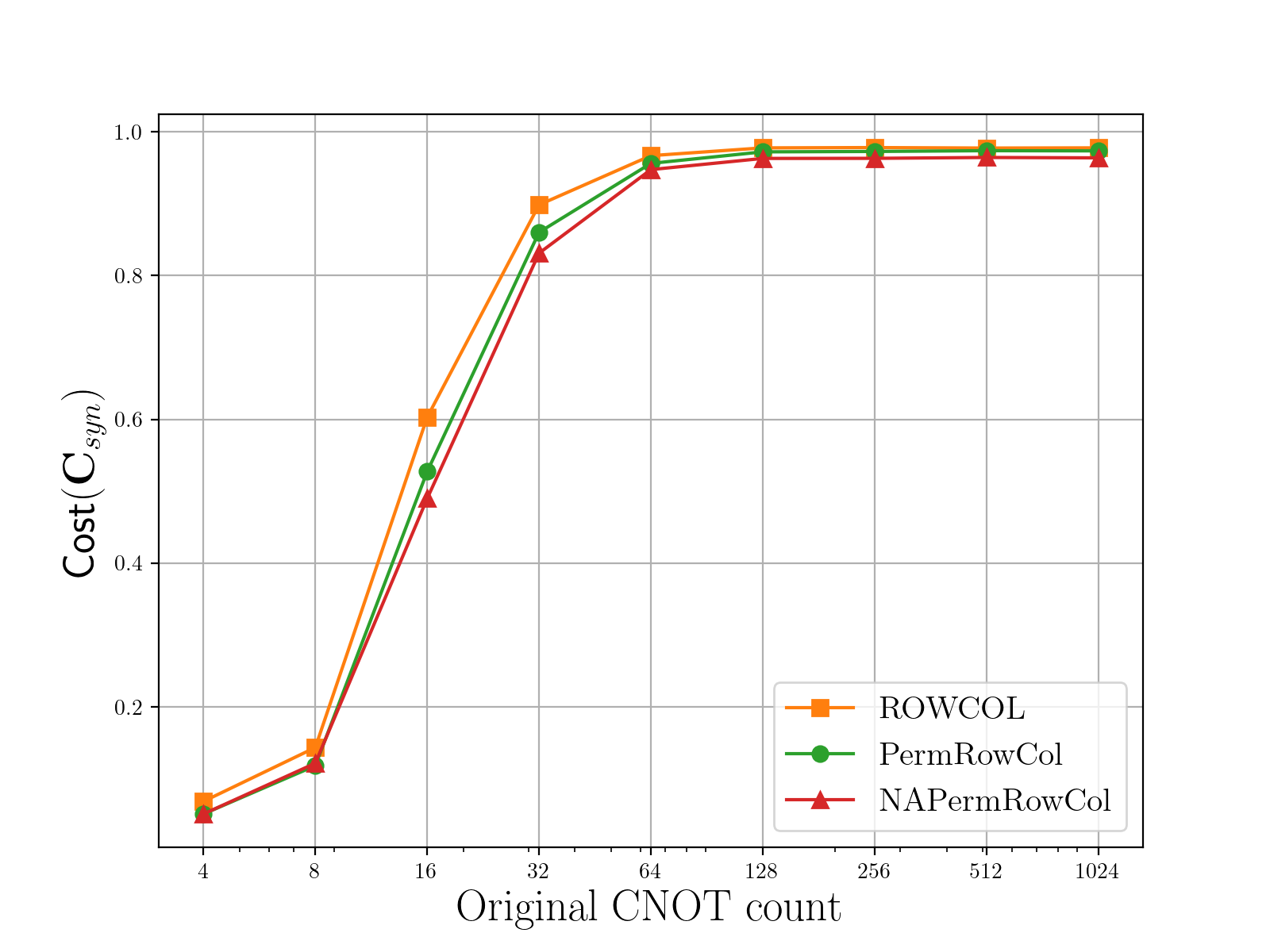}
  \subcaption{Compare NAPermRowCol with PermRowCol and ROWCOL in terms of their respective costs. This is the zoomed-in version of the lefthand side.}
        \label{fig:Guadalupe16Cost}
\end{subfigure}
\caption{IBM's fake Guadalupe backend hosts 16 qubits. We benchmark with its $16$-qubit connected subgraph and compare NAPermRowCol against three state-of-the-art CNOT synthesis algorithms. For each original CNOT count, input $100$ randomly generated CNOT circuits to each algorithm, obtain the synthesized CNOT circuits, then average their gate count and the circuit cost. The x-axis in each figure uses a logarithmic scale as the input gate count grows exponentially. The y-axis of \cref{fig:Guadalupe16NCQ} uses a logarithmic scale, while the one in \cref{fig:Guadalupe16NC,fig:Guadalupe16CostQ,fig:Guadalupe16Cost} uses a linear scale. Compared to \cref{fig:Guadalupe16NCQ,fig:Guadalupe16CostQ}, \cref{fig:Guadalupe16NC,fig:Guadalupe16Cost} get rid of the data related to Qiskit so that the remaining ones are distributed in a more compact area. They serve as the zoomed-in versions which allow us to compare the performance of NAPermRowCol, PermRowCol, and ROWCOL more closely. For input circuits of more than 64 CNOT counts, NAPermRowCol outperforms other algorithms in terms of the synthesized CNOT count and circuit cost. It demonstrates remarkable scalability when the input circuit size grows exponentially.}
\label{fig:Guadalupe16}
\end{figure}

\subsection{Benchmark on IBM's Fake Cairo Backend}
\label{subsubsec:cairo}

\begin{table}[H]
\centering
\begin{tabular}{c|c|l|l|l|l}
\hline
Circuit& Original  & \multirow{2}{3em}{Qiskit} & \multirow{2}{4em}{ROWCOL} & \multirow{2}{5em}{PermRowCol} & \multirow{2}{7em}{NAPermRowCol} \\ Width & CNOT Count & & & & \\\hline
\multirow{9}{2em}{5} & 4 & 4.13\quad (8.12\%) & 4.61\quad (20.68\%) & 4.05\quad (6.02\%) & 3.82\\
 & 8 & 11.28\quad (26.03\%) & 11.75\quad (31.28\%) & 10.30\quad (15.08\%) & 8.95\\
 & 16 & 25.41\quad (117.18\%) & 17.09\quad (46.07\%) & 13.04\quad (11.45\%) & 11.70\\
 & 32 & 55.58\quad (341.81\%) & 17.85\quad (41.89\%) & 13.89\quad (10.41\%) & 12.58\\
 & 64 & 120.31\quad (899.25\%) & 19.00\quad (57.81\%) & 14.00\quad (16.28\%) & 12.04\\
 & 128 & 246.29\quad (1881.42\%) & 17.91\quad (44.09\%) & 13.76\quad (10.70\%) & 12.43\\
 & 256 & 501.19\quad (3988.01\%) & 18.45\quad (50.49\%) & 13.80\quad (12.56\%) & 12.26\\
 & 512 & 1002.96\quad (7942.98\%) & 18.89\quad (51.48\%) & 13.77\quad (10.43\%) & 12.47\\
 & 1024 & 2021.27\quad (16031.44\%) & 18.56\quad (48.12\%) & 13.54\quad (8.06\%) & 12.53\\
\hline
\multirow{9}{2em}{7} & 4 & 3.85\quad (3.49\%) & 3.85\quad (3.49\%) & 3.82\quad (2.69\%) & 3.72\\
 & 8 & 11.86\quad (3.58\%) & 14.60\quad (27.51\%) & 12.53\quad (9.43\%) & 11.45\\
 & 16 & 30.90\quad (26.33\%) & 31.52\quad (28.86\%) & 26.34\quad (7.69\%) & 24.46\\
 & 32 & 76.55\quad (152.39\%) & 40.04\quad (32.01\%) & 34.34\quad (13.22\%) & 30.33\\
 & 64 & 164.84\quad (420.66\%) & 42.13\quad (33.07\%) & 35.48\quad (12.07\%) & 31.66\\
 & 128 & 342.74\quad (973.07\%) & 41.74\quad (30.68\%) & 35.19\quad (10.18\%) & 31.94\\
 & 256 & 708.43\quad (2156.87\%) & 41.70\quad (32.84\%) & 33.54\quad (6.85\%) & 31.39\\
 & 512 & 1452.64\quad (4495.51\%) & 40.42\quad (27.87\%) & 34.83\quad (10.19\%) & 31.61\\
 & 1024 & 2925.63\quad (8982.99\%) & 40.28\quad (25.05\%) & 34.12\quad (5.93\%) & 32.21\\
\hline
\multirow{9}{2em}{16} & 4 & 3.98\quad (1.53\%) & 4.40\quad (12.24\%) & 3.94\quad (0.51\%) & 3.92\\
 & 8 & 8.82\quad (-2.97\%) & 9.65\quad (6.16\%) & 8.96\quad (-1.43\%) & 9.09\\
 & 16 & 32.69\quad (-38.10\%) & 60.23\quad (14.05\%) & 55.92\quad (5.89\%) & 52.81\\
 & 32 & 102.45\quad (-20.72\%) & 153.30\quad (18.63\%) & 135.98\quad (5.23\%) & 129.22\\
 & 64 & 264.28\quad (21.73\%) & 245.55\quad (13.10\%) & 226.74\quad (4.44\%) & 217.11\\
 & 128 & 603.25\quad (141.35\%) & 278.36\quad (11.37\%) & 260.97\quad (4.41\%) & 249.95\\
 & 256 & 1294.44\quad (416.74\%) & 275.55\quad (10.00\%) & 264.12\quad (5.44\%) & 250.50\\
 & 512 & 2694.88\quad (972.50\%) & 276.54\quad (10.06\%) & 263.68\quad (4.94\%) & 251.27\\
 & 1024 & 5527.37\quad (2111.39\%) & 274.74\quad (9.92\%) & 261.06\quad (4.44\%) & 249.95\\
\hline
 \end{tabular}
\caption{IBM's fake Cairo backend hosts $27$ qubits. We benchmark with its 5-, 7-, and 16-qubit connected subgraph and compare NAPermRowCol with other state-of-the-art CNOT synthesis algorithms in terms of the synthesized CNOT count. Qiskit is short for ``Qiskit transpilation at optimization level 3''. It implements the SWAP-based heuristic algorithm SABRE. For each row (i.e., the original CNOT count), input $100$ randomly generated CNOT circuits to each algorithm listed in the table header, then calculate the average synthesized CNOT count. The value in each bracket shows the percentage difference compared to the results of NAPermRowCol.}
\label{tab:CNOTCairo}
\end{table}

\begin{table}[H]
\centering
\begin{tabular}{c|c|l|l|l|l}
\hline
Circuit& Original  & \multirow{2}{3em}{Qiskit} & \multirow{2}{4em}{ROWCOL} & \multirow{2}{5em}{PermRowCol} & \multirow{2}{7em}{NAPermRowCol} \\ Width & CNOT Count & & & & \\\hline
\multirow{9}{2em}{5} & 4 & 0.0547\quad (8.11\%) & 0.0609\quad (20.35\%) & 0.0559\quad (10.47\%) & 0.0506\\
 & 8 & 0.1605\quad (34.55\%) & 0.1561\quad (30.82\%) & 0.1451\quad (21.60\%) & 0.1193\\
 & 16 & 0.3100\quad (105.96\%) & 0.2211\quad (46.86\%) & 0.1713\quad (13.81\%) & 0.1505\\
 & 32 & 0.5277\quad (254.01\%) & 0.2188\quad (46.77\%) & 0.1709\quad (14.64\%) & 0.1491\\
 & 64 & 0.7903\quad (435.71\%) & 0.2315\quad (56.92\%) & 0.1774\quad (20.24\%) & 0.1475\\
 & 128 & 0.9433\quad (547.99\%) & 0.2085\quad (43.24\%) & 0.1666\quad (14.44\%) & 0.1456\\
 & 256 & 0.9956\quad (594.40\%) & 0.2147\quad (49.77\%) & 0.1666\quad (16.20\%) & 0.1434\\
 & 512 & 1.0000\quad (547.47\%) & 0.2398\quad (55.27\%) & 0.1807\quad (16.97\%) & 0.1544\\
 & 1024 & 1.0000\quad (523.26\%) & 0.2386\quad (48.74\%) & 0.1808\quad (12.71\%) & 0.1604\\
\hline
\multirow{9}{2em}{7} & 4 & 0.0548\quad (3.36\%) & 0.0549\quad (3.59\%) & 0.0542\quad (2.15\%) & 0.0530\\
 & 8 & 0.1632\quad (9.82\%) & 0.1933\quad (30.11\%) & 0.1660\quad (11.72\%) & 0.1486\\
 & 16 & 0.3549\quad (28.35\%) & 0.3618\quad (30.86\%) & 0.3077\quad (11.30\%) & 0.2765\\
 & 32 & 0.6530\quad (96.39\%) & 0.4239\quad (27.49\%) & 0.3849\quad (15.76\%) & 0.3325\\
 & 64 & 0.8849\quad (160.36\%) & 0.4419\quad (30.01\%) & 0.3914\quad (15.15\%) & 0.3399\\
 & 128 & 0.9854\quad (189.61\%) & 0.4356\quad (28.01\%) & 0.3792\quad (11.45\%) & 0.3403\\
 & 256 & 0.9997\quad (185.68\%) & 0.4611\quad (31.77\%) & 0.3972\quad (13.51\%) & 0.3500\\
 & 512 & 1.0000\quad (194.61\%) & 0.4340\quad (27.86\%) & 0.3854\quad (13.53\%) & 0.3394\\
 & 1024 & 1.0000\quad (192.17\%) & 0.4313\quad (26.01\%) & 0.3769\quad (10.11\%) & 0.3423\\
\hline
\multirow{9}{2em}{16} & 4 & 0.0575\quad (0.96\%) & 0.0632\quad (11.03\%) & 0.0573\quad (0.56\%) & 0.0569\\
 & 8 & 0.1257\quad (-0.32\%) & 0.1339\quad (6.14\%) & 0.1264\quad (0.17\%) & 0.1262\\
 & 16 & 0.4049\quad (-17.47\%) & 0.5615\quad (14.46\%) & 0.5343\quad (8.91\%) & 0.4906\\
 & 32 & 0.7800\quad (-3.71\%) & 0.8810\quad (8.75\%) & 0.8478\quad (4.65\%) & 0.8101\\
 & 64 & 0.9811\quad (3.71\%) & 0.9747\quad (3.03\%) & 0.9647\quad (1.97\%) & 0.9460\\
 & 128 & 0.9998\quad (3.99\%) & 0.9809\quad (2.02\%) & 0.9757\quad (1.48\%) & 0.9614\\
 & 256 & 1.0000\quad (3.61\%) & 0.9810\quad (1.64\%) & 0.9777\quad (1.30\%) & 0.9651\\
 & 512 & 1.0000\quad (3.37\%) & 0.9837\quad (1.68\%) & 0.9790\quad (1.20\%) & 0.9674\\
 & 1024 & 1.0000\quad (4.11\%) & 0.9775\quad (1.76\%) & 0.9723\quad (1.22\%) & 0.9605\\
\hline
 \end{tabular}
\caption{IBM's fake Cairo backend hosts $27$ qubits. We benchmark with its 5-, 7-, and 16-qubit connected subgraph and compare NAPermRowCol with other state-of-the-art CNOT synthesis algorithms in terms of the $\Cost$ metric. For each row (i.e., the original CNOT count), input $100$ randomly generated CNOT circuits to each algorithm listed in the table header, then calculate the average synthesized CNOT count. The value in each bracket shows the percentage difference compared to the results of NAPermRowCol.}
\label{tab:CostCairo}
\end{table}

\begin{figure}[H]
\begin{subfigure}[t]{.5\textwidth}
  \centering
  \includegraphics[scale=0.42]{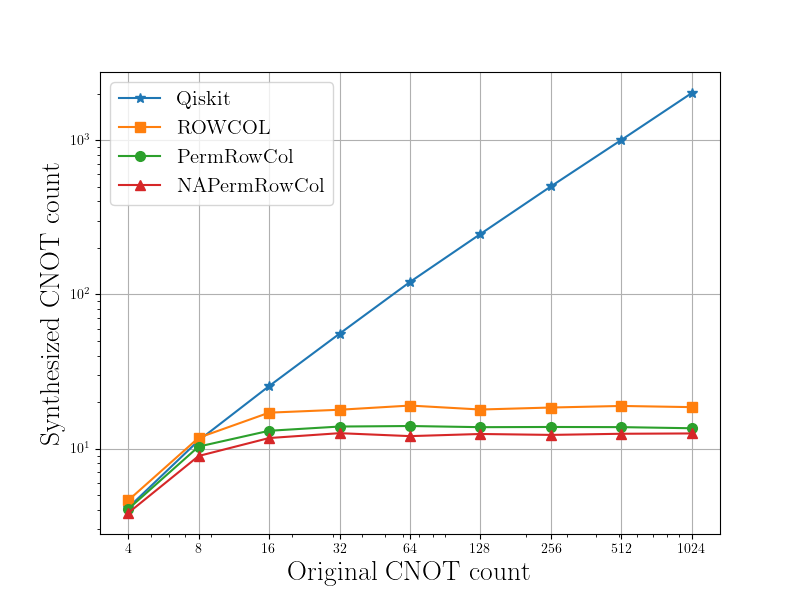}
  \subcaption{Compare NAPermRowCol with PermRowCol, ROWCOL, and Qiskit in terms of the synthesized CNOT count. Qiskit has the worst scalability when the original CNOT count grows exponentially.}
        \label{fig:Cairo5NCQ}
\end{subfigure}
\quad
\begin{subfigure}[t]{.5\textwidth}
  \centering
  \includegraphics[scale=0.42]{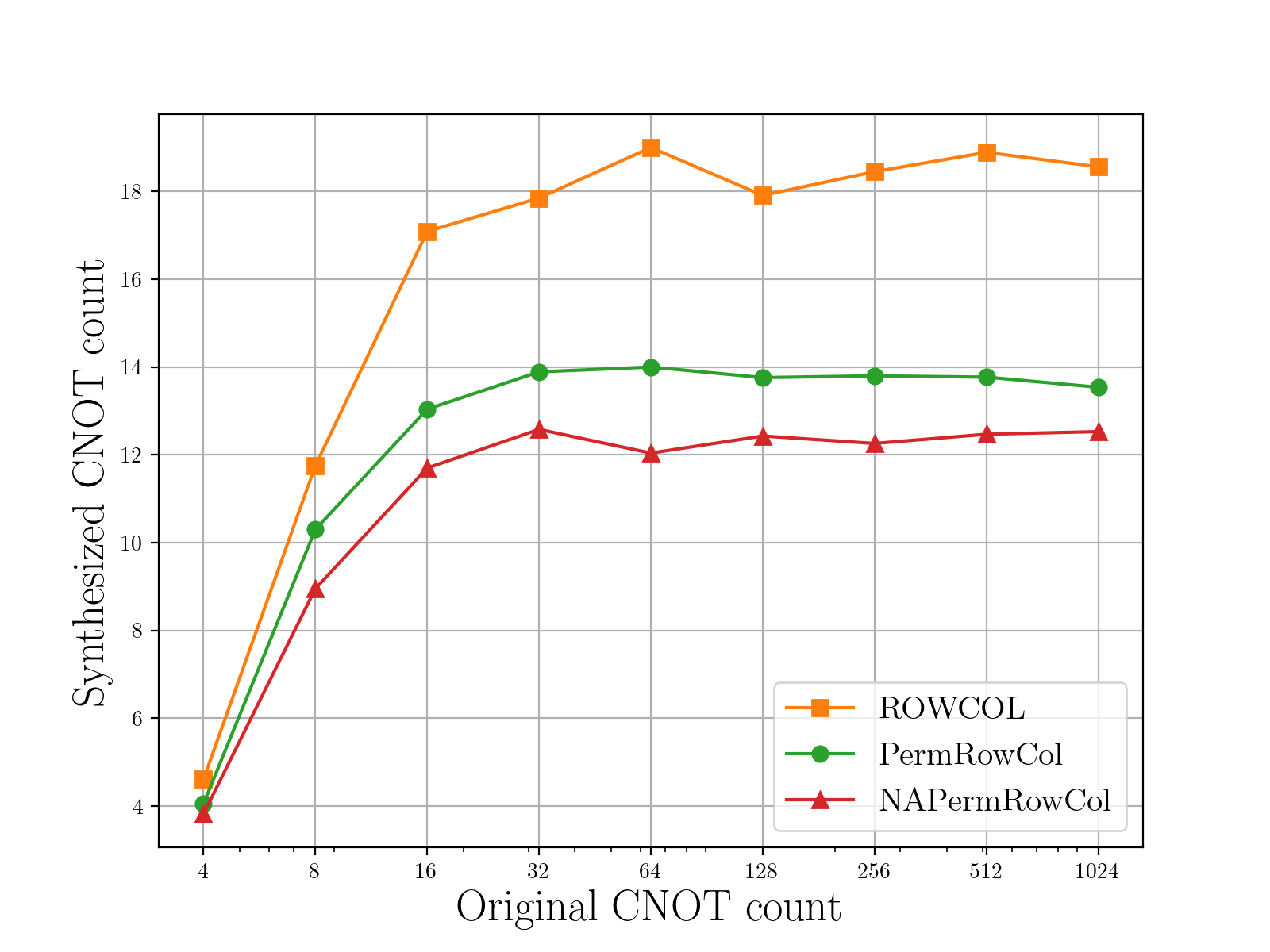}
  \subcaption{Compare NAPermRowCol with PermRowCol and ROWCOL in terms of the synthesized CNOT count. This is the zoomed-in version of the lefthand side.}
        \label{fig:Cairo5NC}
\end{subfigure}
\begin{subfigure}[t]{.5\textwidth}
  \centering
  \includegraphics[scale=0.42]{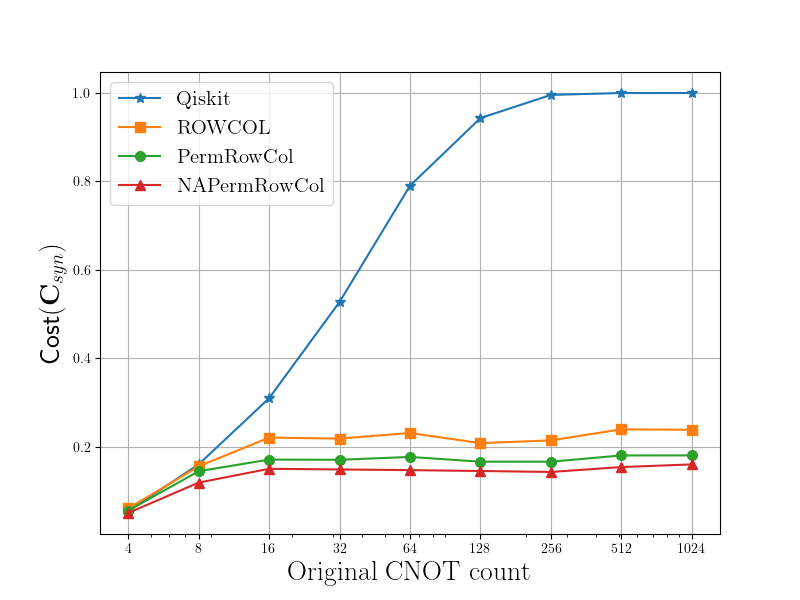}
  \subcaption{Compare NAPermRowCol with PermRowCol, ROWCOL, and Qiskit in terms of their respective costs. Qiskit has the worst scalability when the original CNOT count grows exponentially.}
        \label{fig:Cairo5CostQ}
\end{subfigure}
\quad
\begin{subfigure}[t]{.5\textwidth}
  \centering
  \includegraphics[scale=0.42]{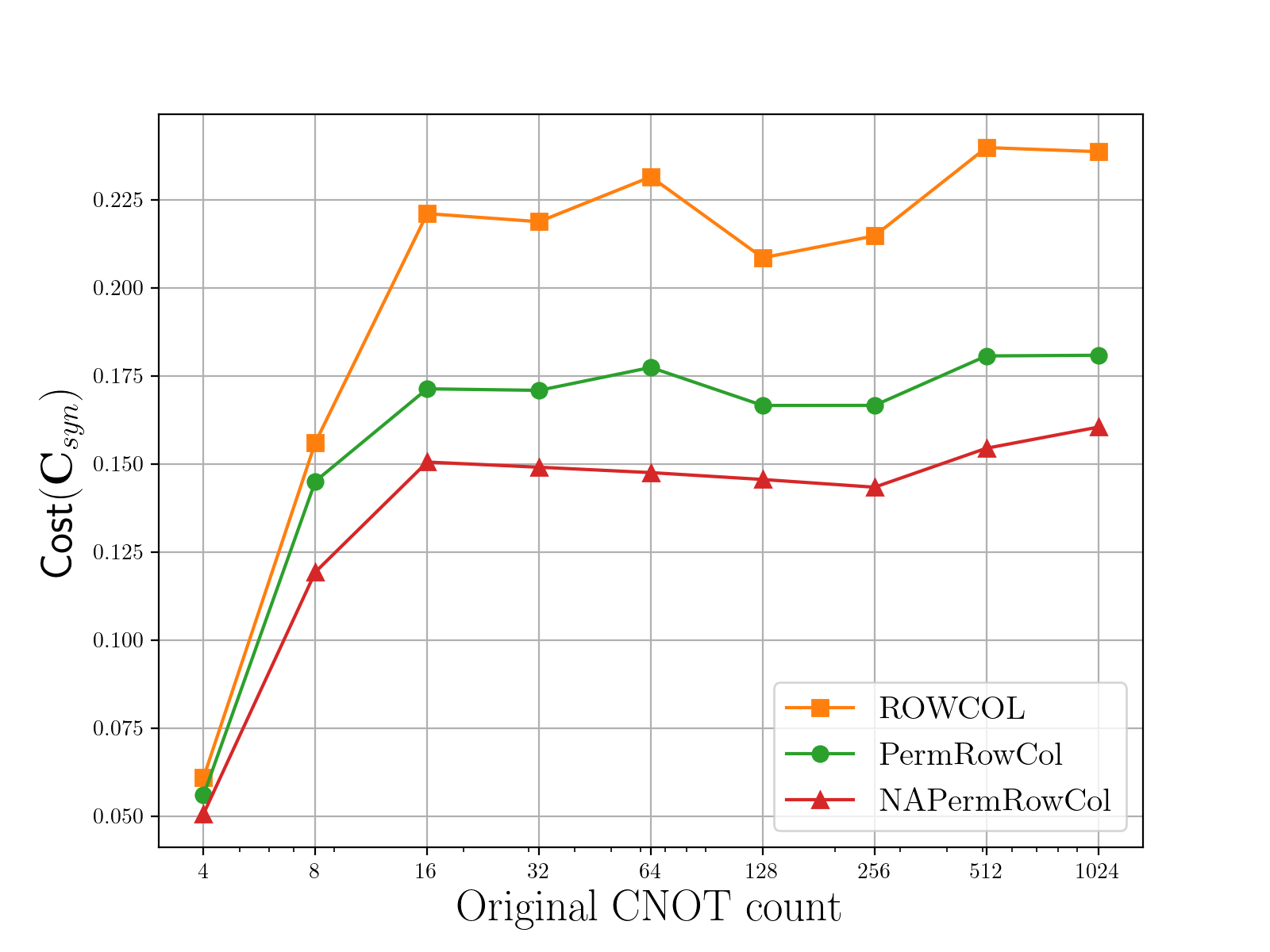}
  \subcaption{Compare NAPermRowCol with PermRowCol and ROWCOL in terms of their respective costs. This is the zoomed-in version of the lefthand side.}
        \label{fig:Cairo5Cost}
\end{subfigure}
\caption{IBM's fake Cairo backend hosts 27 qubits. We benchmark with its $5$-qubit connected subgraph and compare NAPermRowCol against three state-of-the-art CNOT synthesis algorithms. For each original CNOT count, input $100$ randomly generated CNOT circuits to each algorithm, obtain the synthesized CNOT circuits, then average their gate count and the circuit cost. The x-axis in each figure uses a logarithmic scale as the input gate count grows exponentially. The y-axis of \cref{fig:Cairo5NCQ} uses a logarithmic scale, while the one in \cref{fig:Cairo5NC,fig:Cairo5CostQ,fig:Cairo5Cost} uses a linear scale. Compared to \cref{fig:Cairo5NCQ,fig:Cairo5CostQ}, \cref{fig:Cairo5NC,fig:Cairo5Cost} get rid of the data related to Qiskit so that the remaining ones are distributed in a more compact area. They serve as the zoomed-in versions which allow us to compare the performance of NAPermRowCol, PermRowCol, and ROWCOL more closely. For all input circuits of different CNOT counts, NAPermRowCol outperforms other algorithms in terms of the synthesized CNOT count and circuit cost. It demonstrates remarkable scalability when the input circuit size grows exponentially.}
\label{fig:Cairo5}
\end{figure}

\begin{figure}[H]
\begin{subfigure}[t]{.5\textwidth}
  \centering
  \includegraphics[scale=0.42]{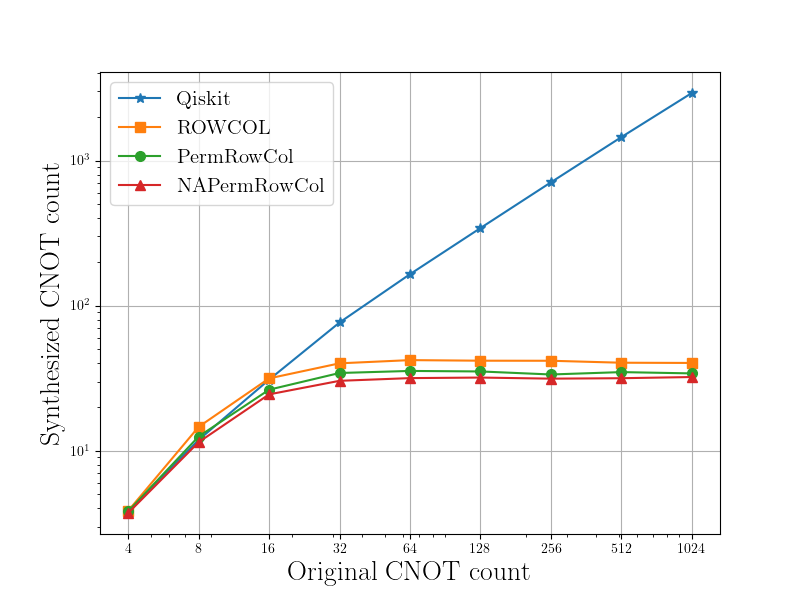}
  \subcaption{Compare NAPermRowCol with PermRowCol, ROWCOL, and Qiskit in terms of the synthesized CNOT count. Qiskit has the worst scalability when the original CNOT count grows exponentially.}
        \label{fig:Cairo7NCQ}
\end{subfigure}
\quad
\begin{subfigure}[t]{.5\textwidth}
  \centering
  \includegraphics[scale=0.42]{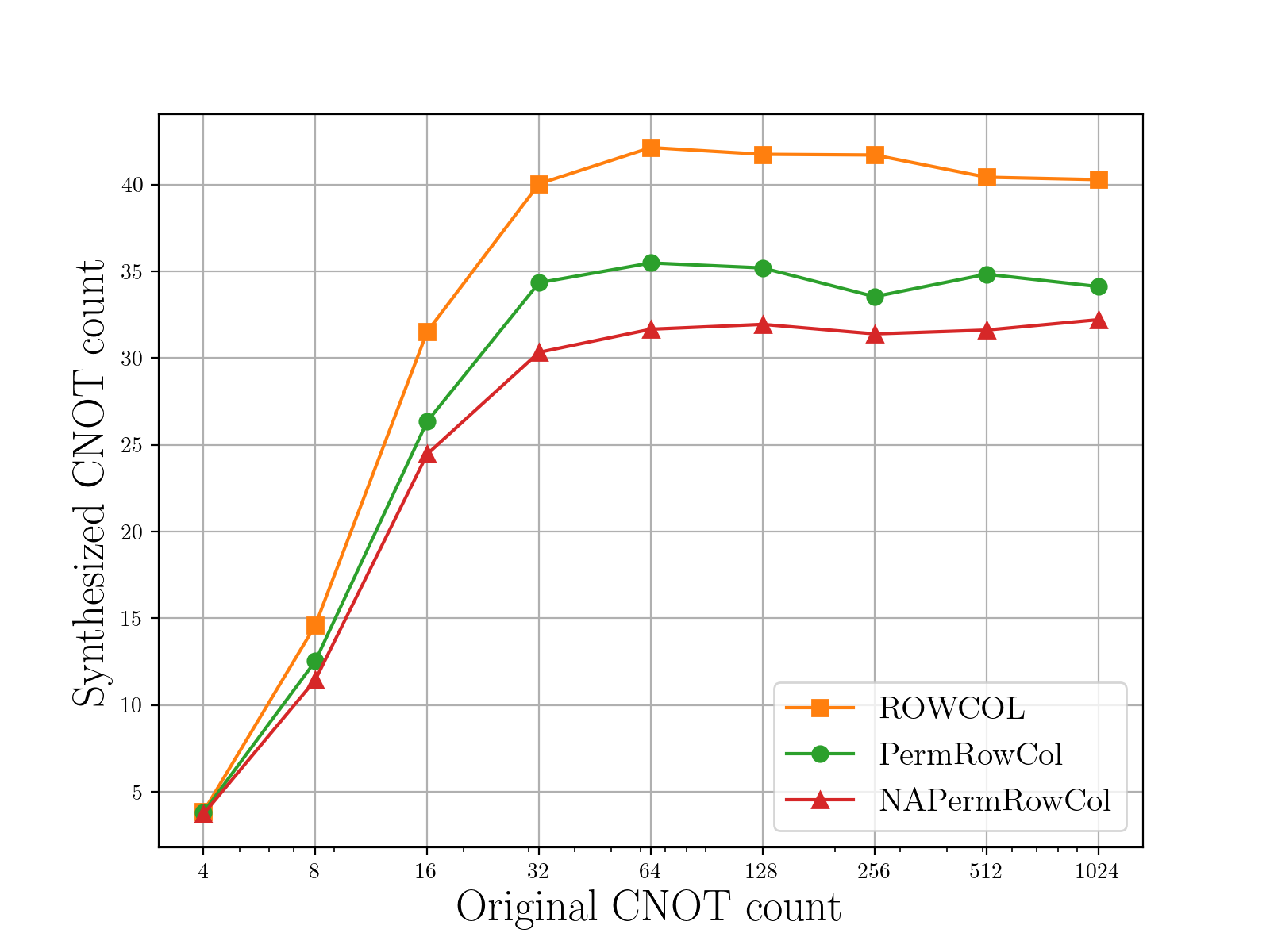}
  \subcaption{Compare NAPermRowCol with PermRowCol and ROWCOL in terms of the synthesized CNOT count. This is the zoomed-in version of the lefthand side.}
        \label{fig:Cairo7NC}
\end{subfigure}
\begin{subfigure}[t]{.5\textwidth}
  \centering
  \includegraphics[scale=0.42]{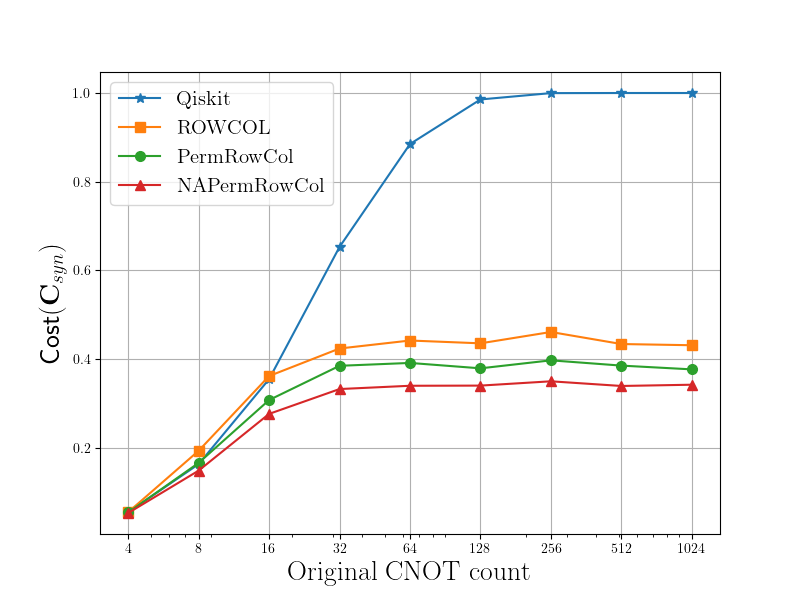}
  \subcaption{Compare NAPermRowCol with PermRowCol, ROWCOL, and Qiskit in terms of their respective costs. Qiskit has the worst scalability when the original CNOT count grows exponentially.}
        \label{fig:Cairo7CostQ}
\end{subfigure}
\quad
\begin{subfigure}[t]{.5\textwidth}
  \centering
  \includegraphics[scale=0.42]{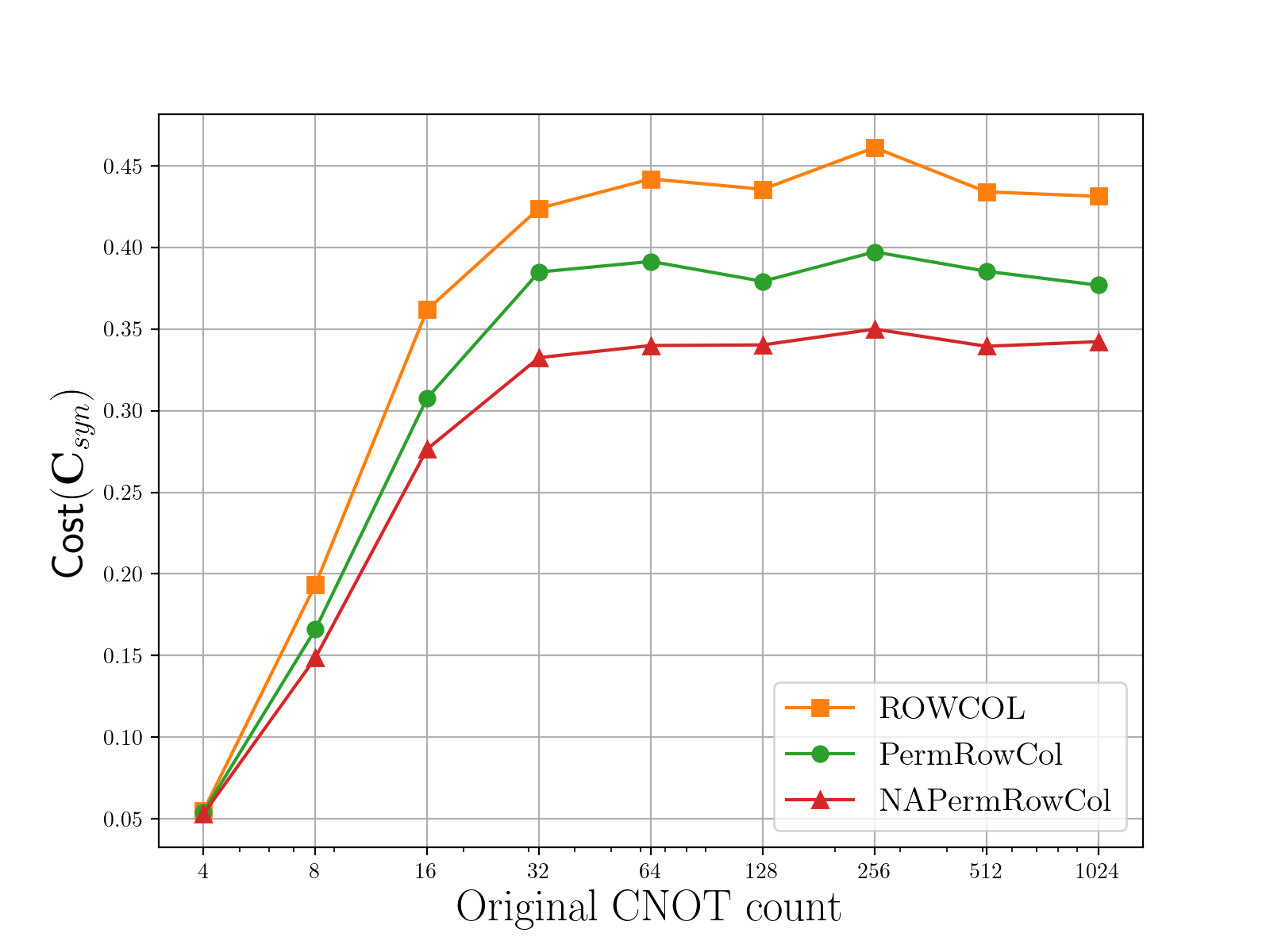}
  \subcaption{Compare NAPermRowCol with PermRowCol and ROWCOL in terms of their respective costs. This is the zoomed-in version of the lefthand side.}
        \label{fig:Cairo7Cost}
\end{subfigure}
\caption{IBM's fake Cairo backend hosts 27 qubits. We benchmark with its $7$-qubit connected subgraph and compare NAPermRowCol against three state-of-the-art CNOT synthesis algorithms. For each original CNOT count, input $100$ randomly generated CNOT circuits to each algorithm, obtain the synthesized CNOT circuits, then average their gate count and the circuit cost. The x-axis in each figure uses a logarithmic scale as the input gate count grows exponentially. The y-axis of \cref{fig:Cairo7NCQ} uses a logarithmic scale, while the one in \cref{fig:Cairo7NC,fig:Cairo7CostQ,fig:Cairo7Cost} uses a linear scale. Compared to \cref{fig:Cairo7NCQ,fig:Cairo7CostQ}, \cref{fig:Cairo7NC,fig:Cairo7Cost} get rid of the data related to Qiskit so that the remaining ones are distributed in a more compact area. They serve as the zoomed-in versions which allow us to compare the performance of NAPermRowCol, PermRowCol, and ROWCOL more closely. For all input circuits of different CNOT counts, NAPermRowCol outperforms other algorithms in terms of the synthesized CNOT count and circuit cost. It demonstrates remarkable scalability when the input circuit size grows exponentially.}
\label{fig:Cairo7}
\end{figure}

\begin{figure}[H]
\begin{subfigure}[t]{.5\textwidth}
  \centering
  \includegraphics[scale=0.42]{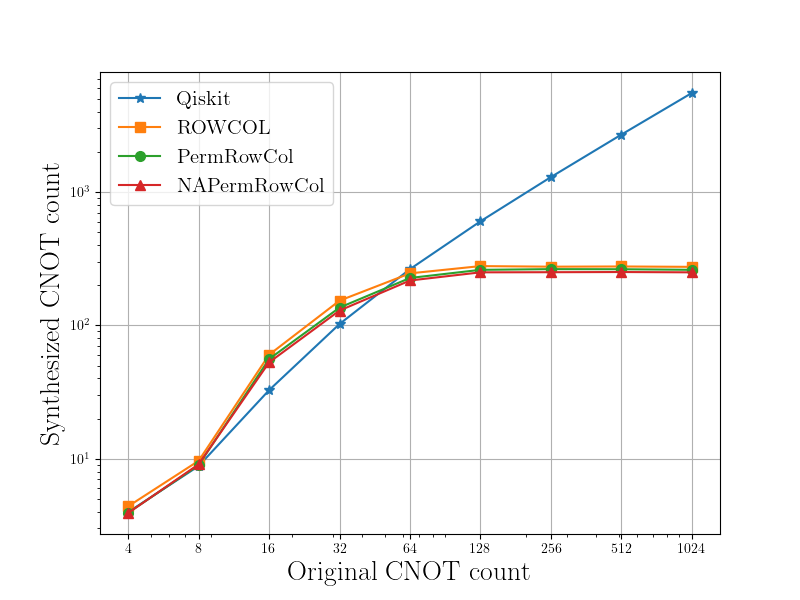}
  \subcaption{Compare NAPermRowCol with PermRowCol, ROWCOL, and Qiskit in terms of the synthesized CNOT count. Qiskit has the worst scalability when the original CNOT count grows exponentially.}
        \label{fig:Cairo16NCQ}
\end{subfigure}
\quad
\begin{subfigure}[t]{.5\textwidth}
  \centering
  \includegraphics[scale=0.42]{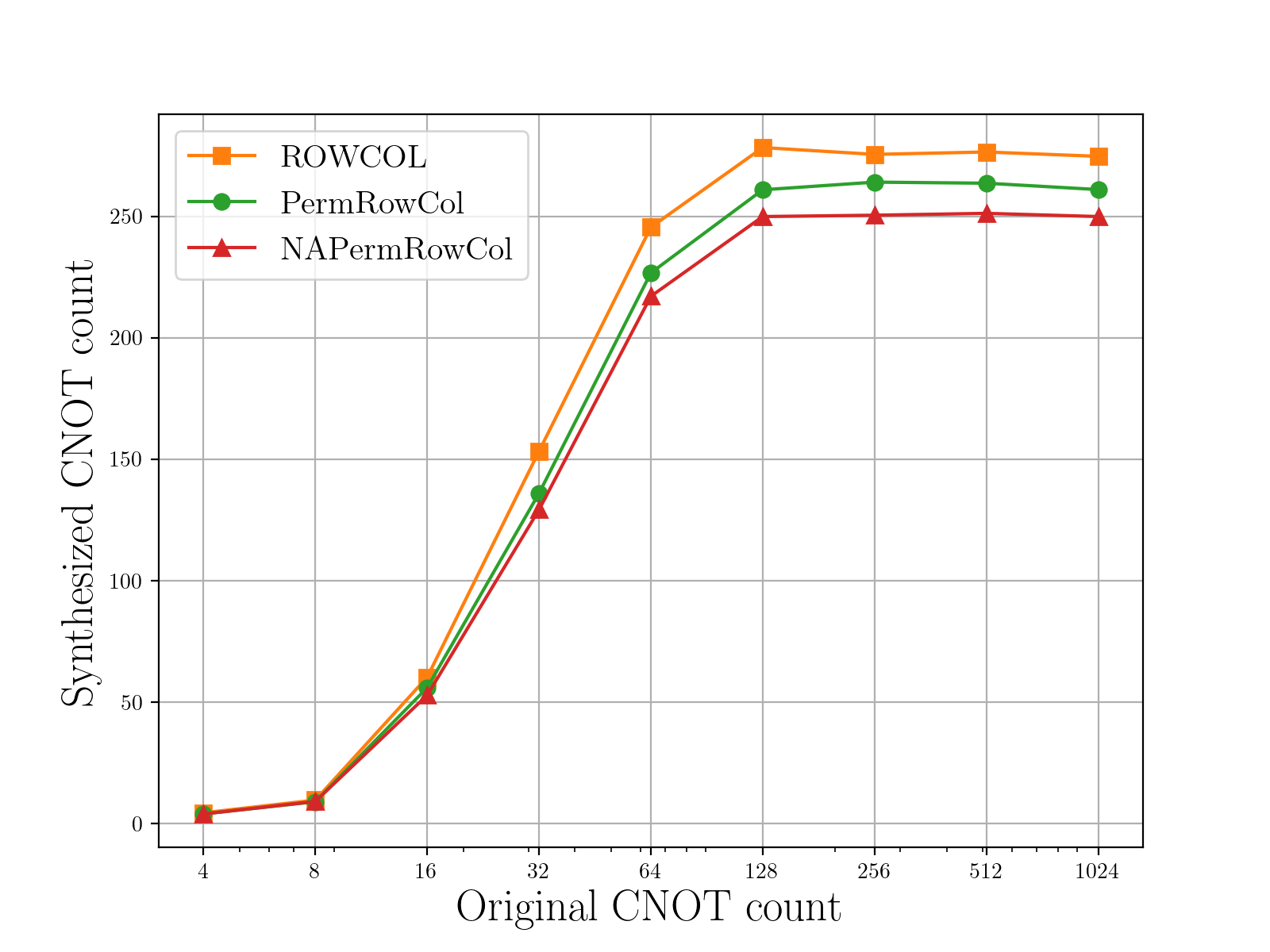}
  \subcaption{Compare NAPermRowCol with PermRowCol and ROWCOL in terms of the synthesized CNOT count. This is the zoomed-in version of the lefthand side.}
        \label{fig:Cairo16NC}
\end{subfigure}
\begin{subfigure}[t]{.5\textwidth}
  \centering
  \includegraphics[scale=0.42]{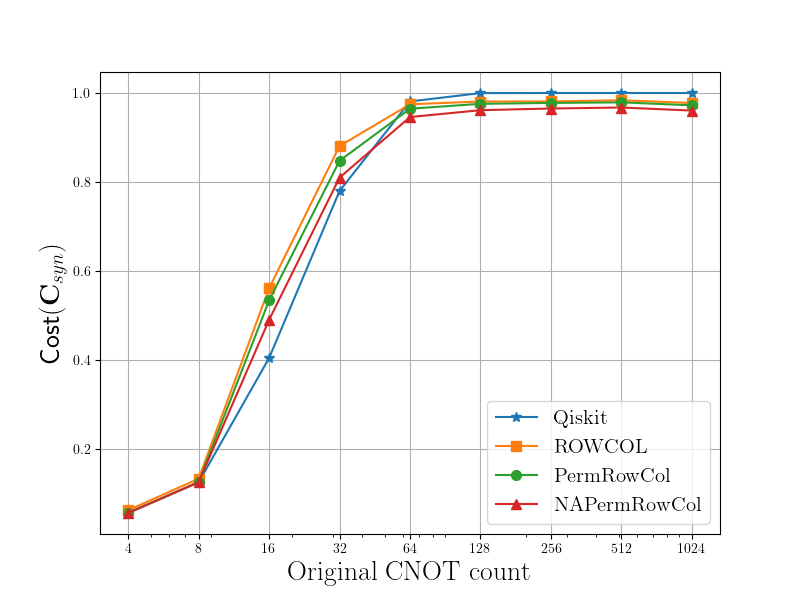}
  \subcaption{Compare NAPermRowCol with PermRowCol, ROWCOL, and Qiskit in terms of their respective costs. Qiskit has the worst scalability when the original CNOT count grows exponentially.}
        \label{fig:Cairo16CostQ}
\end{subfigure}
\quad
\begin{subfigure}[t]{.5\textwidth}
  \centering
  \includegraphics[scale=0.42]{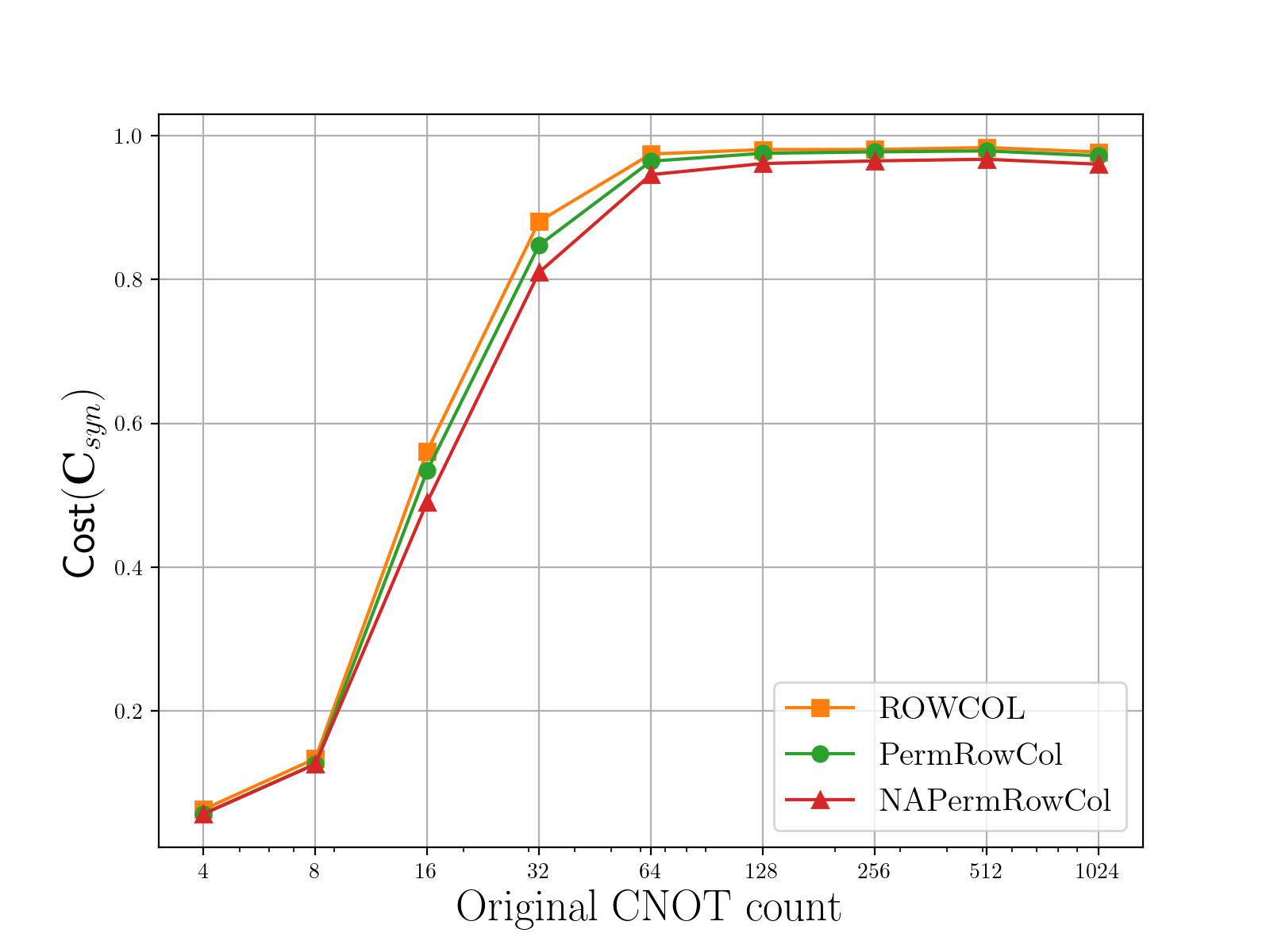}
  \subcaption{Compare NAPermRowCol with PermRowCol and ROWCOL in terms of their respective costs. This is the zoomed-in version of the lefthand side.}
        \label{fig:Cairo16Cost}
\end{subfigure}
\caption{IBM's fake Cairo backend hosts 27 qubits. We benchmark with its $16$-qubit connected subgraph and compare NAPermRowCol against three state-of-the-art CNOT synthesis algorithms. For each original CNOT count, input $100$ randomly generated CNOT circuits to each algorithm, obtain the synthesized CNOT circuits, then average their gate count and the circuit cost. The x-axis in each figure uses a logarithmic scale as the input gate count grows exponentially. The y-axis of \cref{fig:Cairo16NCQ} uses a logarithmic scale, while the one in \cref{fig:Cairo16NC,fig:Cairo16CostQ,fig:Cairo16Cost} uses a linear scale. Compared to \cref{fig:Cairo16NCQ,fig:Cairo16CostQ}, \cref{fig:Cairo16NC,fig:Cairo16Cost} get rid of the data related to Qiskit so that the remaining ones are distributed in a more compact area. They serve as the zoomed-in versions which allow us to compare the performance of NAPermRowCol, PermRowCol, and ROWCOL more closely. For input circuits of more than 64 CNOT counts, NAPermRowCol outperforms other algorithms in terms of the synthesized CNOT count and circuit cost. It demonstrates remarkable scalability when the input circuit size grows exponentially.}
\label{fig:Cairo16}
\end{figure}

\end{document}